%% file: thesis.tex
\documentclass[11pt,twoside]{report}

\usepackage[headheight=14pt,inner=4cm,outer=2.5cm]{geometry} 
\usepackage{setspace}
\usepackage{datetime}
\usepackage{fancyhdr}

\usepackage{amssymb,amsmath,amsthm}
\usepackage{makecell}
\usepackage{hyperref,lineno,mathtools,url,verbatim}
\usepackage{tkz-euclide}
\usepackage{nccmath}
\usepackage{thmtools,thm-restate}
\usepackage{subcaption}
\usepackage{enumitem}
\usepackage{fourier}
\usepackage{microtype}
\usepackage{float}
\usepackage{tikz}
\DeclareMathAlphabet{\mathcal}{OMS}{cmsy}{m}{n}
\usetikzlibrary{calc,trees,positioning,arrows,chains,shapes.geometric,%
    decorations.pathreplacing,decorations.pathmorphing,shapes,%
    matrix,shapes.symbols, trees,backgrounds}
      \tikzstyle{blockzm1} = [rectangle, draw, fill=white, 
    text width=1.5em, text centered, rounded corners, minimum height=1.5em, minimum width = 1.5em]
    
      \tikzstyle{blockzm2} = [rectangle, draw, fill=white, 
    text width=1.0em, text centered, rounded corners, minimum height=1.0em, minimum width = 1.0em]
    
      \tikzstyle{blockzm1r} = [rectangle, draw, fill=red!70, 
    text width=1.5em, text centered, rounded corners, minimum height=1.5em, minimum width = 1.5em]
      \tikzstyle{blockzm1o} = [rectangle, draw, fill=orange!80, 
    text width=1.5em, text centered, rounded corners, minimum height=1.5em, minimum width = 1.5em]
      \tikzstyle{blockzm1y} = [rectangle, draw, fill=yellow!50, 
    text width=1.5em, text centered, rounded corners, minimum height=1.5em, minimum width = 1.5em]
      \tikzstyle{blockzm1g} = [rectangle, draw, fill=green!60, 
    text width=1.5em, text centered, rounded corners, minimum height=1.5em, minimum width = 1.5em]
        \tikzstyle{blockzm1p} = [rectangle, draw, fill=blue!50, 
    text width=1.5em, text centered, rounded corners, minimum height=1.5em, minimum width = 1.5em]
    
   \tikzstyle{blockz} = [rectangle, draw, fill=white, 
    text width=1.75em, text centered, rounded corners, minimum height=1.75em, minimum width = 1.75em]
     \tikzstyle{blockzflexi} = [rectangle, draw, fill=white, 
    text centered, rounded corners]
    \tikzstyle{blockz2} = [rectangle, draw, fill=white, 
    text width=2em, text centered, rounded corners, minimum height=2em, minimum width = 2em]
        \tikzstyle{blockzL} = [rectangle, draw, fill=white, 
    text width=4em, text centered, rounded corners, minimum height=3em, minimum width = 4em]
    
\tikzstyle{vertex}=[circle,fill=black!25,minimum size=20pt,inner sep=0pt]
\tikzstyle{selected vertex} = [vertex, fill=red!24]
\tikzstyle{edge} = [draw,thick,-]
\tikzstyle{weight} = [font=\small]
\tikzstyle{selected edge} = [draw,line width=5pt,-,red!50]
\tikzstyle{ignored edge} = [draw,line width=5pt,-,black!20]
    
\tikzset{
  treenode/.style = {align=center, inner sep=0pt, text centered,
    font=\sffamily},
  arn_n/.style = {treenode, circle, white, font=\sffamily\bfseries, draw=black,
    fill=black, text width=1.5em},
  arn_r/.style = {treenode, circle, red, draw=red, 
    text width=1.5em, very thick},
  arn_x/.style = {treenode, rectangle, draw=black,
    minimum width=0.5em, minimum height=0.5em}
}

\usetikzlibrary{matrix}
\usetikzlibrary{quotes,angles}
\usetikzlibrary{arrows.meta}

\tikzstyle{vertex}=[circle,fill=black!25,minimum size=20pt,inner sep=0pt]
\tikzstyle{selected vertex} = [vertex, fill=red!24]
\tikzstyle{edge} = [draw,thick,-]
\tikzstyle{weight} = [font=\small]
\tikzstyle{selected edge} = [draw,line width=5pt,-,red!50]
\tikzstyle{ignored edge} = [draw,line width=5pt,-,black!20]
\usepackage{tabularx}

\tikzset{
>=stealth',
  punktchain/.style={
    rectangle, 
    rounded corners, 
    draw=black, very thick,
    text width=10em, 
    minimum height=3em, 
    text centered, 
    on chain},
  line/.style={draw, thick, <-},
  element/.style={
    tape,
    top color=white,
    bottom color=blue!50!black!60!,
    minimum width=8em,
    draw=blue!40!black!90, very thick,
    text width=10em, 
    minimum height=3.5em, 
    text centered, 
    on chain},
  every join/.style={->, thick,shorten >=1pt},
  decoration={brace},
  tuborg/.style={decorate},
  tubnode/.style={midway, right=2pt},
}

\usepackage{tikz-cd}
\usepackage{algorithm}
\usepackage{algcompatible}

\usepackage{pgfplots}
\pgfplotsset{grid style={red}}
\usepackage{filecontents}

\usepackage{hyperref} 

\tikzstyle{blockz} = [rectangle, draw, fill=blue!20,  text width=2em, text centered, rounded corners, minimum height=2em, minimum width = 2em]

\numberwithin{algorithm}{section}

\newtheorem{thmm}[algorithm]{Theorem}
\theoremstyle{definition}
\newtheorem{dfn}[algorithm]{Definition}
\newtheorem{exa}[algorithm]{Example}
\newtheorem{cexa}[algorithm]{Counterexample}
\newtheorem{prob}[algorithm]{Problem}
\newtheorem{rem}[algorithm]{Remark}
\theoremstyle{plain}

\newtheorem{lem}[algorithm]{Lemma}
\newtheorem{pro}[algorithm]{Proposition}

\newtheorem{cor}[algorithm]{Corollary}

\numberwithin{equation}{thmm}

\usepackage{mathtools}
\DeclarePairedDelimiter\ceil{\lceil}{\rceil}
\DeclarePairedDelimiter\floor{\lfloor}{\rfloor}

\newcommand{\li}[1]{\#\mathrm{Line}[#1]}
\newcommand{\Es}{\mathcal{E}}
\newcommand{\C}{\mathcal{C}}
\newcommand{\AP}{\mathcal{P}}
\newcommand{\Child}{\mathrm{Children}}
\newcommand{\Desc}{\mathrm{Descendants}}
\newcommand{\nxt}{\mathrm{Next}}

\newcommand{\R}{\mathbb{R}}

\newcommand{\Z}{\mathbb{Z}}
\newcommand{\PS}{\mathbb{P}}

\newcommand{\V}{\mathbb{V}}
\newcommand{\W}{\mathbb{W}}
\newcommand{\I}{\mathbb{I}}
\newcommand{\T}{\mathcal{T}}
\newcommand{\Sd}{\mathcal{S}}
\newcommand{\al}{\alpha}
\newcommand{\be}{\beta}
\newcommand{\de}{\delta}
\newcommand{\De}{\Delta}
\newcommand{\ep}{\epsilon}

\newcommand{\la}{\lambda}

\newcommand{\NN}{\mathrm{NN}}

\newcommand{\rad}{\mathrm{diam}}
\newcommand{\diam}{\mathrm{diam}}

\newcommand{\bs}{\hfill $\blacksquare$}

\newcommand{\id}{\mathrm{id}}
\newcommand{\life}{\mathrm{life}}
\newcommand{\MG}{\mathrm{MG}}
\newcommand{\PD}{\mathrm{PD}}
\newcommand{\HD}{\mathrm{HD}}
\newcommand{\BD}{\mathrm{BD}}
\newcommand{\GH}{\mathrm{GH}}
\newcommand{\ID}{\mathrm{ID}}
\newcommand{\SL}{\mathrm{SL}}
\newcommand{\MST}{\mathrm{MST}}
\newcommand{\birth}{\mathrm{birth}}
\newcommand{\death}{\mathrm{death}}

\usepackage{array} 
\usepackage{varwidth} 
\newcolumntype{M}{>{\begin{varwidth}{3cm}}l<{\end{varwidth}}}

\graphicspath{{./figures/}{../figures/}} 

\makeatletter
\title{A new compressed cover tree for k-nearest neighbour search and the
stable-under-noise mergegram of a point cloud} \let\Title\@title
\author{Yury Elkin} \let\Author\@author
\makeatother

\begin{document}
\pagestyle{empty}
\include{title}
\clearpage

\pagenumbering{roman}
\pagestyle{fancy}
\fancyhead{} 
\fancyhead[LO,RE]{\Title}
\fancyhead[LE,RO]{\Author}
\fancyfoot{} 
\fancyfoot[C]{\thepage}

\begin{singlespace}

\addcontentsline{toc}{chapter}{Abstract} 
\include{./chapters/abstract}

\clearpage

\addcontentsline{toc}{chapter}{Acknowledgements} 
\include{./chapters/acknowledgements}
\clearpage

\tableofcontents
\addcontentsline{toc}{chapter}{Contents} 
\clearpage

\listoffigures
\addcontentsline{toc}{chapter}{List of Figures} 
\clearpage

\listoftables
\addcontentsline{toc}{chapter}{List of Tables} 
\clearpage


\end{singlespace}
\cleardoublepage

\pagenumbering{arabic}
\pagestyle{fancy}
\fancyhead{} 
\fancyhead[RO,LE]{\thepage}
\fancyhead[LO]{\nouppercase\leftmark}
\fancyhead[RE]{\Author}
\fancyfoot{} 

\include{./chapters/chapter1}
\include{./chapters/chapter2}

\include{./chapters/KNNUpdated}
\include{./chapters/MSTUpdated}

\include{./chapters/chapter5}

\include{./chapters/conclusions}
\bibliographystyle{plainurl}
\addcontentsline{toc}{chapter}{Bibliography} 
\bibliography{referencesUpdated}


\end{document}

%% file: title.tex
\begin{titlepage}
	\centering
	\vspace*{1cm}
	\includegraphics[width=65mm]{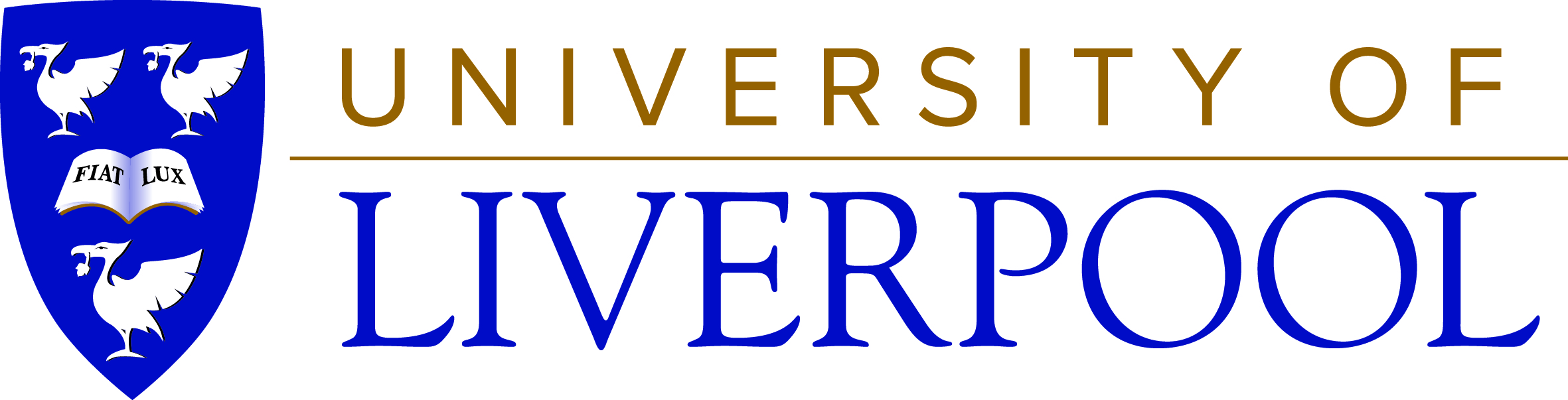}\par\vspace{1cm}
		\vspace{2cm}
	{\huge \Title\par}
	\vspace{5cm}
	{Thesis submitted in accordance with the requirements of the University of Liverpool for the  degree  of  Doctor  in  Philosophy by  \par}
	\vspace{1cm}
	{\textbf{\Author}\par}
	\vfill

	{\large \monthname\, \the\year\par}
\end{titlepage}

%% file: chapters/abstract.tex
\chapter*{Abstract}
\thispagestyle{fancy}

The analysis of data sets mathematically representable as finite metric spaces plays a significant role in every scientific study. In this thesis we focus on constructing new effective algorithms in the area of computational geometry that can be effectively deployed for the study and classification of large data sets prevalent in natural science, economic analysis, medicine, environmental protection etc.




\medskip
\noindent 
The first major contribution of this thesis is a new near-linear time algorithm, that resolves the classical problem of finding $k$-nearest neighbors (KNN) to of query set $Q$ in a larger reference set $R$, where $Q$ and $R$ both belong to some metric space $X$.
This project was inspired by the work of Beygelzimer, Kakade, and Langford in ICML 2006 that attempted to show that such problem is resolvable for $k=1$ having a near-linear time complexity. 
However, in 2015 it was pointed out that the proof of their time complexity might contain mistakes, which has been ascertained in this thesis by showing that the proposed proof does not withstand a concrete counterexample.
An important application of the KNN algorithm is a KNN graph on a finite metric space $R$ whose edge set is formed by connecting every point $p \in R$ with its $k$-nearest neighbors. The KNN graph finds its application in areas of data-skeletonization, where it can serve as an initial skeleton of the data set, or in cluster analysis, where connected components of the KNN graph can represent the clusters.

\medskip
\noindent
Another application of the the KNN algorithm is Minimum spanning tree (MST), which is an efficient way to visualize any unstructured data while knowing only distances, for example any metric graph connecting abstract data points. Although many efficient algorithms for the MST in metric spaces have been devised, there existed only one past attempt to justify a near-linear time complexity in the size of a given metric space. In 2010 March, Ram, and Gray claimed that MST of any finite metric space can be built in a parametrized near-linear time. 
In this work we have demonstrated, with multiple counterexamples, that the attempted proof was incorrect by showing that one of its step fails.
Encouraged by the results of the work of 2010 this thesis produces a new algorithm that is based on Boruvka algorithm, which is combined with the KNN method to resolve the metric MST problem in a near-linear time complexity.

\medskip
\noindent
In the thesis final chapter the MST algorithm is applied in the computation of a new isometry invariant mergegram of Topological Data Analysis (TDA). 
TDA quantifies topological shapes hidden in unorganized data such as clouds of unordered points.
In the $0$-dimensional ($0$D) case, the distance-based persistence is determined by a single-linkage (SL) clustering of a finite set in a metric space.
Equivalently, the $0D$ persistence captures only edge lengths of a Minimum Spanning Tree (MST).
Both the SL dendrogram and the MST are unstable under perturbations of points.
In this thesis, we define the new stable-under-noise mergegram which outperforms
previous isometry invariants on a classification of point clouds.

\medskip
\noindent
In conclusion, the developed fast algorithms of this thesis can cater to a vast varieties of tasks in data science and beyond. The newly proposed corrected time complexity analysis of KNN and MST not only rectifies the past issues in their theoretical justifications but also gives a way to fix analogous issues in other similar methods based on the cover tree data structure.



\bigskip



%% file: chapters/acknowledgements.tex
\chapter*{Acknowledgements}

First of all, I would like to thank my primary supervisor, Dr. Vitaliy Kurlin, for
giving me the opportunity to deeply study in a "high-voltage" field of mathematics about which I remain to be passionate as ever.
His incredible experience has been
inspiring and encouraging to me.
His persistent moral support during all these years was a significant
factor to my motivation.
I am so grateful for all our fruitful conversations that greatly aided to my research results!

I would also like to thank my second supervisor, Dr. Marja Kankaanrinta, for all her guidance,
her patience, and for sharing her knowledge with me. Her feedback has been fundamental
to my research and steered me steadily the right course.

Finally, I can not put in words how much I owe to my family, for their commitment to assist, uplift and back me up.

%% file: chapters/chapter1.tex
\chapter{Introduction} 
\label{ch:intro} 

\section{Motivation}

Topological Data Analysis (TDA) is a recent modern branch of mathematics that appeared from various works in applied topology and computational geometry during the early 2000s. Although one can trace back geometric approaches for data analysis quite far in the past, TDA
started as a field with the pioneering works of Edelsbrunner et al. \cite{edelsbrunner2000topological} and Zomorodian
and Carlsson \cite{zomorodian2005computingTwo} in persistent homology and was later popularized in a landmark paper in 2009
Carlsson \cite{carlsson2009topology}. TDA is mainly motivated by the idea that topology and geometry provide a
a powerful approach to access information about the structure of data \cite{chazal2016structure}.

\medskip

\noindent
TDA aims to provide adequate mathematical, statistical, and algorithmic methods to derive, analyze and exploit the underlying complex topological and geometric structures in given data. The data is often represented as point clouds in Euclidean or more complex metric spaces. During last years, gigantic efforts have been made to provide TDA with powerful data structures and algorithms that are now implemented and ready to use, this involves skeletonization methods such as \cite{singh2007topological} Mapper, Homologically Persistent Skeleton \cite{smith2021skeletonisation,kurlin2015one},  and fast approximate skeleton \cite{elkin2021fast}, as well as libraries allowing us to compute persistent homology such as the Gudhi library (C++ and Python) Cl{\'e}ment Maria et al \cite{maria2014gudhi}, as well as a set of efficient tools that can be used in combination or complementary to other data sciences tools, such as deep learning  \cite{kim2020efficient}.

\medskip

\noindent
 There now exist a large variety of methods and techniques inspired by topological and geometric advances. Providing a complete analysis of all these existing advances is beyond the scope
of this Ph.D. thesis. However, most of them rely on the following typical
pipeline pattern that constitutes as a backbone for topological data analysis:

\begin{enumerate}
\item  The input is assumed to be a finite set of points coming with a notion of distance. This distance can be induced by the metric in the underlying space
(e.g. the Euclidean metric when the data are embedded in $\R^m$ ) or come as an arbitrarily defined metric for instance described as a pairwise distance matrix. The definition of the metric on the data is
commonly given as an input advised by the application. It should be noted that the choice of metric might be crucial for bringing to light the thought-provoking topological and geometric features of the data.

 \item  A “continuous” shape is constructed on top of the data to reformulate the task in terms of topology
or geometry. This is often a CW-complex or a nested family of complexes,
called a filtration, that highlights the formation of the data at different scales. 
The complexes can be seen as a higher-dimensional analog of neighboring graphs that are
in normal circumstances built on top of data which is the case in the classical data analysis or learning procedures. The
main challenge of this step is to come up with an appropriate structure that will outline the essential information.

\item Topological or geometric information is extracted from the output of step 2. This can result in full reconstruction, typically a skeleton of the shape, from which geometric features can be easily extracted or, in approximate summaries, or so-called algebraic invariants. The latter highlights relevant information, such as persistence diagrams that are isometric invariant of the original data and that can be often quickly computed in a modern GPU-parallel way for example using a Ripser algorithm \cite{1908.02518}.
Beyond getting interesting topological/geometric information or its interpretation the goal is to prove its relevance, in particular its stability with respect to noise in the input data.

\item The summaries produced in step 3 may generate additional sets of features and descriptors of the initial data points. They can be used to improve our understanding of the data, for example, via visualization, such as using the new TDAView visualization technique \cite{walsh2020tdaview} or TTK \cite{tierny2017topology}. Or they can be combined with other features and used further in machine learning tasks, for example, a neural network can be built on persistence diagrams using a topological landscape layer \cite{kim2020efficient} or Perslay-library \cite{carriere2019perslay} containing multiple customized permutation invariants.
\end{enumerate}

\newpage

\section{Thesis structure}

The central objective of this thesis is to describe new computational methods for the analysis of finite metric spaces, as well as parametrized near-linear time complexities for $k$-nearest neighbors search, construction of minimum spanning tree, as well as newly a developed object, mergegram. The research is carried out in the framework of space partitioning data structures and topological data analysis (TDA), for which we develop and effectively utilize novel computational geometry techniques. All worst-case time complexity estimates in this work are measured by the random-access machine model \cite{Gordon_Bell1971-gy}. This work breaks into following four chapters:
\begin{itemize}
\item Chapter \ref{ch:preliminaries} introduces the basic concepts and definitions.
\item Chapter \ref{ch:knn} solves the $k$-nearest neighborhood problem.
\item Chapter \ref{ch:mst} is dedicated to the minimum spanning tree problem.
\item Chapter \ref{ch:mergegram} introduces a new invariant of finite metric spaces, mergegram.
\end{itemize}
To avoid confusion we formally define the meaning of \emph{incorrectness of proof}. 
\begin{dfn}
We say that a result (Theorem, Proposition, Lemma ...) has an incorrect proof, if it contains a claim and there exists an example for which the claim can not be satisfied. 
\end{dfn}

\noindent 
It should be noted that if a presented proof in the literature for some result is incorrect it does not necessary follow that the result is incorrect since there might exist another correct proof for it.

\subsection{A new compressed cover tree for the $k$-nearest neighbor search}

Chapter \ref{ch:knn} is dedicated to the exact $k$-nearest neighbor problem in any metric space. 
In a modern formulation, the problem is to find all $k\geq 1$ nearest neighbors in a finite reference set $R$ for all points from a finite query set $Q$.
Both sets live in an ambient space $X$ with a given distance $d$ satisfying all metric axioms.
The simplest example is $X=\R^n$ with the Euclidean metric, where a query set $Q$ can be a single point or a subset of a larger set $R$. 

\medskip
\noindent
The approach taken to resolve the $k$-nearest neighbors problem is inspired by \cite{beygelzimer2006cover}.
Publication \cite{beygelzimer2006cover} introduced new concepts: cover tree, which stored sets of $R$ into a leveled tree-type data structure and a parameter, expansion constant $c(R)$ that was defined as a growth-bound for any finite metric space $R$.
The main contribution of \cite{beygelzimer2006cover} was an attempt to prove two new results regarding time upper bound time complexity estimates of the cover tree construction algorithm and the nearest neighbors algorithm.
Theorem 6 of \cite{beygelzimer2006cover} claimed a construction time of $O(c(R)^{O(1)} \cdot |R|\log|R|)$ for cover tree on any finite metric space $R$. Theorem 5 claimed that the first nearest neighbor of any query point $q \in Q$ can be found in time $O(c(R)^{O(1)} \cdot \log|R|)$ assuming that the cover tree data structure was built on the reference set $R$.  However, in this Ph.D. thesis both of the results have been shown to have incorrect proof: 
Counterexample~\ref{cexa:construction_algorithm_of_original_cover_tree} shows that a step in the proof of Theorem 6 and 
Counterexample~\ref{cexa:original_all_nearest_neighbors_algorithm} shows a step in the proof of Theorem 5 were incorrect.

\medskip 
\noindent
In this Ph.D. thesis Definition~\ref{dfn:cover_tree_compressed} introduces a new concept: compressed cover tree, which is a simplified version of a cover tree. For its construction, Algorithm~\ref{alg:cover_tree_k-nearest_construction_whole} is employed, which is an adaption of \cite[Algorithm~2]{beygelzimer2006cover} to compressed cover trees.
Corollary~\ref{cor:construction_time_KR} proves that the upper bound time complexity estimate of the new construction method is bounded by $O(c(R)^{O(1)} \cdot |R| \cdot \log|R|)$. This thesis introduces a new method for finding $k$-nearest neighbors Algorithm~\ref{alg:cover_tree_k-nearest}. Algorithm~\ref{alg:cover_tree_k-nearest} is a modification of \cite[Algorithm~1]{beygelzimer2006cover}, where an additional stoppage is added to prove a near-logarithmic time complexity for a $k$-nearest neighbors search of a single query point $q \in Q$. The main result, Theorem~\ref{thm:knn_KR_time} gives an upper bound time complexity of 
$O\big(c(R)^{O(1)}\cdot\log(k) \cdot (\log(|R|) + k)\big)$
for the newly introduced $k$-nearest neighbors method.

\medskip
\noindent

Other important results of this section are:

\begin{itemize}

\item Definition~\ref{dfn:expansion_constant} introduces new dimensionality constant $c_m(R)$. 
Theorem \ref{thm:normed_space_exp_constant} shows that if $R$ is a finite subset of $\R^{n}$, then $c_m(R) \leq 2^{n}$.

\item Theorem~\ref{thm:construction_time} and Corollary~\ref{cor:cover_tree_knn_miniziminzed_constant_time} give an upper bound for time complexity of the compressed cover tree construction algorithm and the $k$-nearest neighbors algorithm by using aspect ratio $\Delta(R)$ of Definition \ref{dfn:radius+d_min} and minimized expansion constant $c_m(R)$ of Definition~\ref{dfn:expansion_constant} as parameters. 

\item Theorem~\ref{thm:approximate_k_nearestneighbors} gives an upper bound for time complexity of $(1+\epsilon)$-approximate $k$-nearest neighbors method, Algorithm \ref{alg:cover_tree_k-nearest_approximate}. 

\item Counterexample \ref{cexa:dualtreecode} shows that the time complexity result \cite[Theorem~3.1]{ram2009linear} of the dual-tree approach to the nearest neighbor problem had an incorrect proof as well.

\end{itemize}

\begin{table}
\caption{This table contains list of theorems that used the cover tree datastructure and were proven incorrectly, as well as new algorithms and new theorems that are intended to rectify the past mistakes.}
\begin{tabular}{|V{6.0cm}|V{6cm}|V{2.9cm}|V{3.0cm}|}
\hline
Original result &  Counterexample  & New Algorithm   & New Theorem \\
\hline
$k$-nn search \cite[Theorem~5]{beygelzimer2006cover} & Counterexample \ref{cexa:original_all_nearest_neighbors_algorithm} & Algorithm \ref{alg:cover_tree_k-nearest_construction_whole} & Corollary \ref{cor:construction_time_KR} \\
\hline
Cover-tree construction \cite[Theorem~6]{beygelzimer2006cover} & Counterexample \ref{cexa:construction_algorithm_of_original_cover_tree} & Algorithm \ref{alg:cover_tree_k-nearest} & Theorem \ref{thm:knn_KR_time} \\
\hline
Paired-tree $k$-nn search \cite[Theorem~3.1]{ram2009linear} & Counterexample \ref{cexa:dualtreeproof} & \makecell{Ryan Curtin \\ \cite[Algorithm~2,3]{curtin2015plug} } & \makecell{Ryan Curtin \\
Weaker result \\ \cite[Theorem~2]{curtin2015plug}}\\
\hline
\makecell{Construction of $\MST$ \\ \cite[Theorem~5.4]{march2010fast}} & \makecell{Counterexample \ref{cexa:mst}, \\ Counterexample \ref{cexa:mst_component}} &  \makecell{Single-tree \\ Algorithm \ref{alg:single_tree_boruvka}} & Weaker result Corollary \ref{cor:single_cover_tree_mst_time} \\
\hline       
\end{tabular}
\end{table}
  

\subsection{Fast algorithm for MST based on the compressed cover tree}

Chapter \ref{ch:mst} is dedicated to metric minimum spanning tree problem. 
A minimum spanning tree on a finite metric space $(R,d)$ is a tree $\MST(R)$ with vertex set $R$ and a minimum total length of edges among all possible spanning trees of $R$. In the metric minimum spanning tree problem the goal is to construct a near-linear time algorithm that constructs $\MST(R)$ on any finite metric space $R$ in a near-linear time.

\medskip

\noindent 
The approach taken to resolve the metric minimum spanning tree problem is inspired by \cite{march2010fast}.
In \cite[Algorithm~1]{march2010fast} there was an attempt to solve the metric MST problem by exploiting a cover tree datastructure and a dual-tree nearest neighbors algorithm \cite[Algorirthm~1]{ram2009linear} in a near linear time.
However, the proof of time complexity result \cite[Theorem~5.1]{march2010fast} was exhibited to be incorrect in Counterexample \ref{cexa:mst} and Counterexample~\ref{cexa:mst_component} .

 
 \medskip
 
\noindent 
To resolve the metric minimum spanning tree problem we propose new method, \cite[Algorithm~1]{march2010fast}, which is a single-tree version of \cite[Algorithm~3]{march2010fast}.
Theorem~\ref{thm:dtb_correctness} proves the correctness of Algorithm~\ref{alg:cover_tree_connected_components}.
Theorem~\ref{thm:single_cover_tree_mst_time} shows that the time complexity of Algorithm~\ref{alg:cover_tree_connected_components} is:
$ O\Big ((c_m(R))^{4+ \ceil{\log_2(\rho(R))}} \cdot \log_2(\Delta(R)) \cdot |R| \cdot \log(|R|) \cdot \alpha(|R|)\Big ), $
where $\rho(R)$ is the ratio: maximal edge length divided by the minimal edge length of $\MST(R)$ and $\Delta(R)$ is the
aspect ratio of Definition~\ref{dfn:radius+d_min}.

\subsection{Mergegram extends the 0-dimensional persistence}



Chapter \ref{ch:mergegram} introduces a new isometry invariant of a point cloud, mergegram. Mergegram is an extension of the $0$-dimensional persistence diagram, in such a way that it encodes a dendrogram built on $X$ into multiset of $\R^{2}$ by projecting every link of the dendrogram into a pair (the link start, the link end). It will be proved that the mergegram is stable under noise in the sense that a topological space and its noisy point sample have close persistence diagrams. The stability of mergegram shows that it is a suitable method for recognition of noisy point clouds.

\medskip

\noindent
One of the main aims of Chapter  \ref{ch:mergegram} is to solve an isometry recognition problem: given a set of different objects, the goal is to find out which of the objects have similar shapes. In the specific case, where objects are point clouds, the goal is to classify them as classes modulo isometries that preserve distances between points.
In practice, we are given a dataset of $n$ shapes, each of which belongs to one of $m$ shape classes.
The goal is to build an algorithm that takes as input a shape that is isometric to one of the $m$ shape classes and outputs its correct class. 

\medskip


\noindent
All the main results of Chapter \ref{ch:mergegram} were published in: \cite{elkin2020mergegram} and \cite{elkin2021isometry}. Here is the summary of new contributions to Topological Data Analysis:

\begin{itemize}

\item Definition~\ref{dfn:mergegram} introduces the concept of a mergegram for any dendrogram of clustering.
\item Theorem~\ref{thm:mergegram_to_0D_persistence} justifies that the mergegram of a single-linkage dendrogram is strictly stronger than the 0D persistence of a distance-based filtration of sublevel sets.
\item Theorem~\ref{thm:mergegram-to-dendrogram} showing how to reconstruct a dendrogram of a single-linkage clustering from a mergegram in general position.
\item Theorem~\ref{thm:stability_mergegram} proves bottleneck stability property of mergegram.
\item Theorem~\ref{thm:complexity} shows that the mergegram can be computed in a near-linear time.
\item Sections \ref{sec:experiments_mg_f} and \ref{sec:experiments_mg} shows experementally that mergegram performs better than $0$-dimensional persistence diagram or a diagram made of distances to two closes neighbors $\NN(2)$ for the shape recognition problems.

\end{itemize}

\section{Publications}
Counterexamples of Chapter~\ref{ch:knn} are published in \cite{2208.09447}. Rest of the chapter will be published later. Chapter~\ref{ch:mst} will be published later. Chapter~\ref{ch:mergegram} is published in \cite{elkin2020mergegram} and \cite{elkin2021isometry}.

%% file: chapters/chapter2.tex
\chapter{Basic definitions and past results} 
\label{ch:preliminaries} 

\section{Metric spaces}

Topological and geometric features can be usually found only in continuous structures and therefore as real data usually represent finite sets of observation it is impossible to directly inherit any topological information that the input contains. One of the natural ways to approach this problem is to connect data points that are "close" to each other. This unmasks global continuous shapes associated with
 the data. To specify the notion of closeness between data points we use the standard notion of metric introduced in Definition \ref{dfn:metricSpace}.
 
\begin{dfn}[Metric space]
\label{dfn:metricSpace}
Metric space  $(X, d)$ is a set $X$ with a function $d : X \times X \rightarrow \R$ ,
called a distance, such that for any triplet $x,y,z \in X $ the following properies of
\begin{enumerate}
\item Identity: $d(x, y) \geq 0 \text{ and } d(x, y) = 0 \text{ if and only if } x = y$
\item Symmetry: $d(x, y) = d(y, x)$ and,
\item Triangle inequality: $d(x, z) \leq d(x, y) + d(y, z)$
\end{enumerate}
are satisfied.
\end{dfn}

\noindent

\begin{dfn}[Hausdorff distance]
\label{dfn:HausdorffDistance}
Let $A,B$ be two non-empty subsets of a metric space $(X,d)$. We define their Hausdorff distance $d_H(A,B)$ by
$$d_H (A, B) = \max\{\sup_{b \in B} d(b, A), \sup_{a \in A} d(a, B)\}$$
\end{dfn}

\begin{dfn}[$\delta$-sparse subset]
\label{dfn:deltaSparsness}
We say that a subset $C$ of metric space $(X,d)$ is $\delta$-sparse if $d(a,b) \geq \delta$ for all $a,b \in C$.
\end{dfn}


%



\noindent 
From the TDA perspective, Hausdorff distance provides a helpful way to specify the distance between two sets which are subsets of the same metric space. However, in some cases, one has to compare data sets that belong to different metric spaces. In this case, we can use the Gromov-Hausdorff distance which generalizes Hausdorff distance in a way where it can be used for the comparison of pairs of compact metric spaces. 

\begin{dfn}[Isometric embedding]
Let $(X,d)$ and $(Y,u)$ be metric spaces.
Injection $\phi:X \rightarrow Y$ is a \emph{isometric embedding} if $d(a,b) = u(\phi(a), \phi(b))$ for any $a,b \in X$.   
\end{dfn}

%
%


\begin{dfn}[Gromov-Hausdorff distance]
\label{dfn:GromovHausdorffDistance}
The Gromov-Hausdorff distance $d_{GH}(X,Y)$ measured on two compact metric spaces $X,Y$ is the infinium $d_H(f(X),f(Y))$ of all metric space $M$ and isometric embeddings $f:X \rightarrow M$, $g:Y \rightarrow M$.
\end{dfn}

\begin{figure}[H]
 \centering
\begin{tikzpicture}[      
        every node/.style={anchor=south west,inner sep=0pt},
        x=1mm, y=1mm,
      ]   
      
	\node(fig10) at (50,40)
	{\includegraphics[scale=1.0]{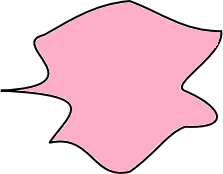}};
	\node(fig11) at (95,45)
	{\includegraphics[scale=1.0]{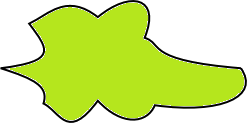}};

	 \node (fig1) at (70,5)
       {\includegraphics[scale=1.0]{Pictures/Shape0wb.png}};
     \node[opacity=0.7,rotate = 45] (fig2) at (85,0)
       {\includegraphics[scale=1.0]{Pictures/Shape1wb.png}};

     \node (fig1) at (0,23)
       {\includegraphics[scale=1.0]{Pictures/Shape0wb.png}};
     \node[opacity=0.7] (fig2) at (3,23)
       {\includegraphics[scale=1.0]{Pictures/Shape1wb.png}};  
       
      \node (txt1) at (0,40) {$A$};
       \node (txt2) at (20,20) {$B$};
       \node[shape = circle] (plc1) at (19,39){\small b};
       \node[shape = circle] (plc2) at (21,54){\small a};
       \path[draw, <->] (plc1) -- node[midway, right]{$d_H(A,B)$} (plc2);
       
        \node[shape = circle] (plc3) at (98.5,17){\small b};
       \node[shape = circle] (plc4) at (108.5,13.5){\small a};
       \path[draw, <->] (plc3) -- node[midway, below = 0.75cm]{$d_{GH}(A,B)$} (plc4);
       \node (txt3) at (70,55) {$A$};
       \node (txt4) at (110,55) {$B$};
\end{tikzpicture}
\caption{This figure illustrates Hausdorff distance between $A$ and $B$ on \textbf{left} side and their 
 Gromov Hausdorff distance on \textbf{right} side.}
\label{Gromof-Hausdorff}

\end{figure}

\begin{exa}[Hausdorff distance and Gromov-Hausdorff distance]
 Let $A,B$ of Figure \ref{Gromof-Hausdorff} be sets located in the same plane of $\R^2$. In the picture the distance of points $a$, $b$ equals $d_H(A,B)$, because $b \in B$ is the nearest neighbor of $a$ which produces maximal distance to its nearest neighbor over all points in $A$. On the other hand, while computing the Gromov-Hausdorff distance we assume that $A$ and $B$ might belong to distinct spaces. Thus $A$ can be moved and rotated to reduce its distance to $B$. Because we are allowed to perform isometry operations, while computing Gromov-Hausdorff distance on both spaces, it is evident that $d_{GH}(A,B) \leq d_H(A,B)$.

\end{exa}


\section{Graphs}

This section formally defines concept of graph.

\begin{dfn}[Graph]
\label{dfn:graph}
An undirected graph $G$ is a pair $(\mathbb{V}(G),\mathbb{E}(G))$, where:
\begin{itemize}
\item $\mathbb{V}(G)$ is a finite set of vertices.
\item $\mathbb{E}(G)$ is a finite set of unordered pairs $(a,b)$ of \hspace{0.05cm} $\mathbb{V}(G) \times \mathbb{V}(G)$ 
\end{itemize}
We can extend any graph $G$ into \emph{a weighted graph} by introducing a function $w:\mathbb{E}(G) \rightarrow \R_{+}$, that measures the length of each edge.
\end{dfn}
\noindent 
Abstract graph defined above does not represent directly any geometrical shape. However, this notion can be developed into a \emph{embedded} graph which contains the features of abstract graph, but which now exposes explicit geometrical info.

\begin{dfn}[Embedded graph]
Given a metric space $(X,d)$ define embedded graph $G$ to be a graph having
\begin{itemize}
\item Vertex set \hspace{0.05cm} $\mathbb{V}(G)$ is a subset of $X$.
\item For each edge $(u,v) \in \mathbb{E}(G)$ there exists a path in $X$ which connects $a$ to $b$.   
\end{itemize}
In a special case where $X = \mathbb{R}^m$ we say that a graph is a \emph{straight-line} graph if every path in the edge set $E$ is a straight line. 
\end{dfn}

\section{Heap}
\label{sec:BinaryHeap}

Heap is a tree-based data structure which was designed to store elements with a priority in an effective manner. 
Although there are many implementation of heap data structure for simplicity in this section we focus only on binary heap data structure.

\medskip
\noindent 
First introduced in 1964  binary heap \cite{williams1964algorithm} was first concrete implementation of heap data structure, that is effective at storing data with additional priority property. The binary heap data structure allows adding elements, removing the minimal element, and decreasing the value of any given element in logarithmic time.  To define a binary heap, we first define the notion of a complete binary tree.

\begin{figure}[H]
\centering
\begin{tikzpicture}[-,>=stealth',level/.style={sibling distance = 6cm/#1,
  level distance = 2.0cm}] 
\node [align=center, circle, white, draw=black,
    fill=black, text width=1.5em] {1}
    child{ node [align=center, circle, white, draw=black,
    fill=black, text width=1.5em] {9} 
            child{ node [align=center, circle, white, draw=black,
    fill=black, text width=1.5em] {13} 
            	child{ node [align=center, circle, white, draw=black,
    fill=black, text width=1.5em] {29} edge from parent node[above left]
                         {}} 
							child{ node [align=center, circle, white, draw=black,
    fill=black, text width=1.5em] {94}}
            }
            child{ node [align=center, circle, white, draw=black,
    fill=black, text width=1.5em] {17}
							child{ node [align=center, circle, white, draw=black,
    fill=black, text width=1.5em] {18}}
							child{ node [align=center, circle, white, draw=black,
    fill=black, text width=1.5em] {23}}
            }                            
    }
    child{ node [align=center, circle, white, draw=black,
    fill=black, text width=1.5em] {3}
            child{ node [align=center, circle, white, draw=black,
    fill=black, text width=1.5em] {11} 
							child{ node [align=center, circle, white, draw=black,
    fill=black, text width=1.5em] {36}}
							child{ node [align=center, circle, white, draw=black,
    fill=black, text width=1.5em] {39}}
            }
            child{ node [align=center, circle, white, draw=black,
    fill=black, text width=1.5em] {51}
							child{ node [align=center, circle, white, draw=black,
    fill=black, text width=1.5em] {79}}
							child{ node [align=center, circle, white, draw=black,
    fill=black, text width=1.5em] {99}}
            }
		}
; 
\end{tikzpicture}

\begin{tikzpicture}[font=\ttfamily,
array/.style={matrix of nodes,nodes={draw, minimum size=7mm, fill=green!30},column sep=-\pgflinewidth, row sep=0.5mm, nodes in empty cells,
row 1/.style={nodes={draw=none, fill=none, minimum size=5mm}},
row 1 column 1/.style={nodes={draw}}}]

\matrix[array] (array) {
1 & 2 & 3 & 4 & 5 & 6 & 7 & 8 & 9 & 10 & 11 & 12 & 13 & 14 & 15 \\
1  & 9  &  3 &  13 &  17 & 11  &  51 & 29  & 94  & 18 & 23 & 36 & 39 & 79 & 99  \\};

\begin{scope}[on background layer]
\fill[green!10] (array-1-1.north west) rectangle (array-1-15.south east);
\end{scope}

\draw (array-1-1.north)--++(90:3mm) node [above] (first) {First index};
\draw (array-1-15.east)--++(0:3mm) node [right]{Indices};
%
\end{tikzpicture}

\caption{a binary heap built on a finite subset of integers; on the \textbf{top:} its binary tree representation; on the \textbf{bottom:} its array-representation.}

\end{figure}
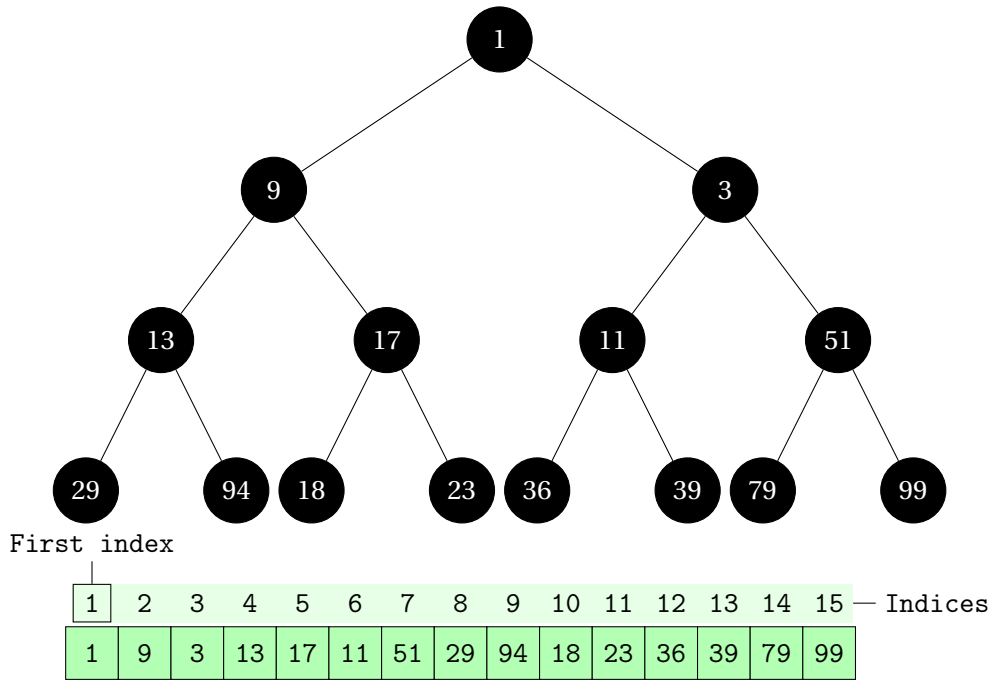

\begin{dfn}[Binary search tree]
 Given a binary tree, assign an index $i$ to each of its nodes. To start, assign the root an index $i=1$. Now suppose a node has been assigned an index
$i$. Then assign its left child (if it exists) the index $2i$, and assign to its right child
(if it exists) the value $2i+1$. Then we say that the tree is \emph{complete} iff the set of assigned indices
is a \emph{contiguous} set; meaning that, if $i$ and $k$ are assigned indices, with $i < k$,  then so is $j$, for every $i< j < k$.
\end{dfn}

\begin{dfn}
\label{dfn:BinaryHeap}
Let $(C, \leq)$ be a completely ordered set. A {\emph{min heap}} (respectively, \emph{max heap}) is a complete binary tree $T$ on the set $C$, which satisfies the following condition: If $p \in T$, and $q\in T$ is a child of $p$, then the element stored at $q$ is larger (respectively, is smaller) or equal to the element stored at $p$.
\end{dfn}
%

\begin{lem}[Binary heap complexities]
\label{lem:BinaryHeap}
A min binary heap $T$ built on an ordered set $(C, \leq)$ has the following complexities:
\begin{itemize}
\item Adding any element $a \in C$ to a binary tree can be done in $\log(|C|)$ time.
\item Finding the minimal of the set $C$ can be done in $O(1)$ in time.
\item Removing the minimal element of the set $C$ can be performed in $O(\log(|C|))$ time. 
\end{itemize}
\end{lem}
\begin{proof}
The proof is available in \cite[Section~6.5]{Cormen1990}.
\end{proof}

\noindent
In the explicit implementation, a binary heap often has a tree-like structure, where the minimal element can always be found at the root of the tree.

\section{Hierarchical clustering}
This section formally defines concepts related to hierarchical clustering.
Hierarchial clustering will be used in the definition of mergegram, Chapter~\ref{ch:mergegram}.

\begin{dfn}[partition set $\PS(A)$]
\label{dfn:partition}
For any set $A$, a \emph{partition} of $A$ is a finite collection of non-empty disjoint subsets $A_1,\dots,A_k\subset A$ whose union is $A$.
The \emph{single-block} partition of $A$ consists of the set $A$ itself.
The \emph{partition set} $\PS(A)$ consists of all partitions of $A$.
\bs
\end{dfn}

\noindent 
The partition set $\PS(A)$ of the abstract set $A=\{0,1,2\}$ consists of the five partitions
$$(\{0\},\{1\},\{2\}),\quad
(\{0,1\},\{2\}),\quad
(\{0,2\},\{1\}),\quad
(\{1,2\},\{0\}),\quad
(\{0,1,2\}).$$
For example, the collections $(\{0\},\{1\})$ and $(\{0,1\},\{0,2\})$ are not partitions of $A$.
\medskip

\noindent
The definition of dendrogram was first presented in {\cite[Section ~3.1 (4)]{carlsson2010characterization}}.
\begin{dfn}[{\cite[Section ~3.1]{carlsson2010characterization}}]
\label{dfn:dendrogramOriginal}
Let $(A, \theta)$ be a pair , where $A$ is a finite set and $\theta:[0,+\infty) \rightarrow \mathbb{P}(X)$ is a function from positive real numbers to the power set of $X$. We say that $(A, \theta)$ is a dendrogram if it satisfies the following conditions:
\begin{enumerate}

\item $\theta(0) = \{\{x_0\}, ... , \{x_n\}\}$, where $x_i \in A$ for all $i \in [n]$.
\item There exists $t_0$ s.t $\theta(t)$ is the single block partition for all $t \geq t_0$ i.e. $\theta(t) = \{\{x \mid x \in A\}\}$.
\item If $r \leq s$ then $\theta(r)$ refines $\theta(t)$ i.e. for every $U \in \theta_r $ there exists $V \in \theta_s$ satisfying $U \subseteq V$. 
\item Technical condition {\cite[Section ~3.1 (4)]{carlsson2010characterization}} For all $r$ there exists $\epsilon > 0$ s.t. $\theta(r) = \theta(t)$ for $t \in [r,r+\epsilon]$ .
\end{enumerate}

\end{dfn}

\noindent
 As an important type of dendrogram, we define a single linkage clustering dendrogram in Definition \ref{dfn:sl_clustering}.

\begin{dfn}[single-linkage clustering]
\label{dfn:sl_cluster}
Let $A$ be a finite set in a metric space $X$ with a distance $d:X\times X\to[0,+\infty)$ and let $s \in R_{+}$ be a scale. Define $\De_{SL}(A;s)$ to be single linkage clustering at scale $s$, in such a way that
any points $a,b\in A$  belong to \emph{SL cluster} if and only if there is a finite sequence $a=a_1,\dots,a_m=b\in A$ such that any two successive points have a distance at most $s$, i.e. $d(a_i,a_{i+1})\leq s$ for $i=1,\dots,m-1$. 
\bs
\end{dfn}
For $s=0$, any point $a\in A$ forms a singleton cluster $\{a\}$.
\begin{dfn}[single-linkage dendrogram]
\label{dfn:sl_clustering}
Representing each cluster from $\De_{SL}(A;s)$ over all $s\geq 0$ by one point, we get the \emph{single-linkage dendrogram} $\De_{SL}(A)$ visualizing how clusters merge, see Figure.~\ref{fig:SLDendrogram}.
\bs
\end{dfn}

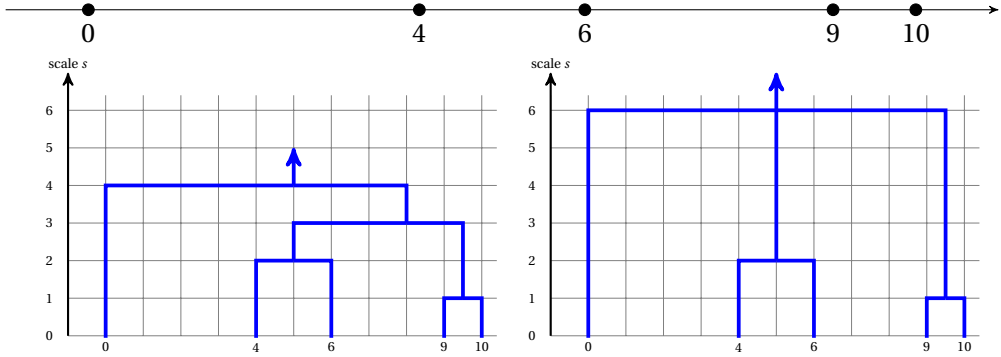
\begin{figure}[H]
\centering
\begin{tikzpicture}[scale = 1.1]
  \draw[->] (-1,0) -- (11,0) node[right]{} ;
   \foreach \x/\xtext in {0, 4, 6, 9, 10}
    \draw[shift={(\x,0)}] (0pt,2pt) -- (0pt,-2pt) node[below] {$\xtext$};   
   \filldraw (0,0) circle (2pt);
   \filldraw (4,0) circle (2pt);
   \filldraw (6,0) circle (2pt);
   \filldraw (9,0) circle (2pt);
   \filldraw (10,0) circle (2pt);
\end{tikzpicture}

\begin{tikzpicture}[thick,scale=0.50, every node/.style={transform shape}][sloped]
\draw[style=help lines,step = 1] (-1,0) grid (10.4,6.4);
\draw [->] (-1,0) -- (-1,7) node[above] {scale $s$};
\foreach \i in {0,1,...,6}{ \node at (-1.5,\i) {\i}; }
\node (a) at (0,-0.3) {0};
\node (b) at (4,-0.3) {4};
\node (c) at (6,-0.3) {6};
\node (d) at (9,-0.3) {9};
\node (e) at (10,-0.3) {10};
\node (x) at (5,5) {};
\node (de) at (9.5,1){};
\node (bc) at (5.0,2){};
\node (bcde) at (8.0,3){};
\node (all) at (5.0,4){};
\draw [line width=0.5mm, blue ] (a) |- (all.center);
\draw [line width=0.5mm, blue ] (b) |- (bc.center);
\draw [line width=0.5mm, blue ] (c) |- (bc.center);
\draw [line width=0.5mm, blue ] (d) |- (de.center);
\draw [line width=0.5mm, blue ] (e) |- (de.center);
\draw [line width=0.5mm, blue ] (de.center) |- (bcde.center);
\draw [line width=0.5mm, blue ] (bc.center) |- (bcde.center);
\draw [line width=0.5mm, blue ] (bcde.center) |- (all.center);
\draw [line width=0.5mm, blue ] [->] (all.center) -> (x.center);
\end{tikzpicture}
\hspace*{2mm}
\begin{tikzpicture}[thick,scale=0.50, every node/.style={transform shape}][sloped]
\draw[style=help lines,step = 1] (-1,0) grid (10.4,6.4);
\draw [->] (-1,0) -- (-1,7) node[above] {scale $s$};
\foreach \i in {0,1,...,6}{ \node at (-1.5,\i) {\i}; }
\node (a) at (0,-0.3) {0};
\node (b) at (4,-0.3) {4};
\node (c) at (6,-0.3) {6};
\node (d) at (9,-0.3) {9};
\node (e) at (10,-0.3) {10};
\node (x) at (5,7) {};
\node (de) at (9.5,1){};
\node (bc) at (5.0,2){};
\node (bcde) at (8.0,3){};
\node (all) at (5.0,6){};
\draw [line width=0.5mm, blue ] (a) |- (all.center);
\draw [line width=0.5mm, blue ] (b) |- (bc.center);
\draw [line width=0.5mm, blue ] (c) |- (bc.center);
\draw [line width=0.5mm, blue ] (d) |- (de.center);
\draw [line width=0.5mm, blue ] (e) |- (de.center);
\draw [line width=0.5mm, blue ] (bc.center) |- (all.center);
\draw [line width=0.5mm, blue ] (de.center) |- (all.center);
\draw [line width=0.5mm, blue ] [->] (all.center) -> (x.center);
\end{tikzpicture}

\caption{\textbf{Top}: the 5-point cloud $A = \{0,4,6,9,10\}\subset\R$.
\textbf{Bottom left}:  Single-linkage dendrogram $\De_{SL}(A)$ from Definition~\ref{dfn:sl_clustering}.
\textbf{Bottom right}: Complete-linkage dendrogram $\De_{CL}(A)$ from Example \ref{exa:cl_clustering}.}
\label{fig:SLDendrogram}
%

\end{figure}

\begin{exa}[{\cite[Section ~3.2.2]{carlsson2010characterization}}]
\label{exa:cl_clustering}

Let $(A,d)$ be a finite metric space. Let $l:\PS(A) \times \PS(A) \rightarrow \R$ a function satisfying  $l(U,V) = \max_{(p,q) \in U \times V}d(p,q)$ for all $U,V \in \PS(A)$. For every $r \in \R$ consider equivalence relation $\backsim_r$ on $\PS(A)$, given by $B \backsim B'$ if there exists a sequence in   $B = B_1, ..., B_s = B'$ satisfying $l(B_k, B_{k+1}) \leq r$ for all $k \in {1,...,s-1}$. Consider the sequences $r_1, r_2, ... \in [0,\infty)$ and $\gamma_1, \gamma_2, ... \in \PS(A)$ given by $\gamma_1 = A$,
and recursively for $i \geq 1$ by $\gamma_{i+1} = \gamma_i / \backsim_{r_i}$, where 
$$r_i = \min\{ l(B, B') \mid B, B' \in \gamma_i, B \neq B'\}. $$
Finally, we define $\triangle_{CL}(A,r) = \gamma_{i(r)}$, where $i(r) = \max\{i \mid r_i \leq r\}$.


\end{exa}

\noindent 
Any dendrogram can be represented as an ultra-metric space. This representation will be used to prove the stability result of single-linkage dendrograms first presented in \cite{carlsson2010characterization}.

\begin{dfn}{{\cite[Section~3.3.1 ]{carlsson2010characterization}}}
Let $(A,\theta)$ be a dendrogram over $A$. Let us define a symmetric map $u: A \times A \rightarrow \R_{+}$ given by
$u(x,y) = \min \{ r \in \R_{+} \mid x,y \text{ belong to the same partition of } \theta(r)\}. $
We say that $(A, u)$ is the \emph{ultra-metric space} corresponding to a dendrogram $(A, \theta)$. 
\end{dfn}

\begin{thmm}[{\cite[Proposition~2]{carlsson2010characterization}}]
\label{thm:SingleLinkageDendrogramStability}
Let $(X,d_X)$ and $(Y,d_Y)$ be two finite metric spaces and let $(X,u_X)$ and $(Y,u_Y)$ be the two finite ultrametric space corresponding to single linkage dendrograms built on $(X,d_X)$ and $(Y,d_Y)$, respectively. Then 
$$d_{GH}((X,u_X), (Y, u_Y)) \leq d_{GH}((X,d_X), (Y, d_Y)).$$
Here, $d_{GH}$ stands for the Gromov-Hausdorff distance of Definition \ref{dfn:GromovHausdorffDistance}. 

\end{thmm}

\section{Persistent homology}

This section introduces the key concepts from the thorough review by Chazal et al. \cite{chazal2016structure} 
that will be used in Chapter~\ref{ch:mergegram}. 


\begin{dfn}[persistence module $\V$]
\label{dfn:persistence_module}
A \emph{persistence module} $\mathbb{V}$ over the real numbers $\mathbb{R}$ is a family of vector spaces $V_t$, $t\in \mathbb{R}$ with linear maps $v^t_s:V_s \rightarrow V_t$, $s\leq t$ such that $v^t_t$ is the identity map on $V_t$ and the composition is respected: $v^t_s \circ v^s_r = v^t_r$ for any $r \leq s \leq t$.
\bs
\end{dfn}

\noindent
The set of real numbers can be considered as a category  $\mathbb{R}$ in the following sense.
The objects of $\R$ are all real numbers. 
Any two real numbers such that $a\leq b$ define a single morphism $a\to b$.
The composition of morphisms $a\to b$ and $b \to c$ is the morphism $a \leq c$. 
In this language, a persistence module is a functor from $\R$ to the category of vector spaces.
\medskip

\noindent
A basic example of $\V$ is an interval module.
An interval $J$ between points $p<q$ in the line $\R$ can be one of the following types: closed $[p,q]$, open $(p,q)$ and half-open or half-closed $[p,q)$ and $(p,q]$.
It is convenient to encode types of endpoints by $\pm$ superscripts as follows:
$$[p^-,q^+]:=[p,q],\quad
[p^+,q^-]:=(p,q),\quad
[p^+,q^+]:=(p,q],\quad
[p^-,q^-]:=[p,q).$$

\noindent
The endpoints $p,q$ can also take the infinite values $\pm\infty$, but without superscripts.

\begin{exa}[interval module $\I(J)$]
\label{exa:interval_module}
For any interval $J\subset\R$, the \emph{interval module} $\I(J)$ is the persistence module defined by the following vector spaces $I_s$ and linear maps $i_s^t:I_s\to I_t$
$$I_s=\left\{ \begin{array}{ll} 
\Z_2, & \mbox{ for } s\in J, \\
0, & \mbox{ otherwise }; 
\end{array} \right.\qquad
i_s^t=\left\{ \begin{array}{ll} 
\id, & \mbox{ for } s,t\in J, \\
0, & \mbox{ otherwise }
\end{array} \right.\mbox{ for any }s\leq t.$$
\end{exa}
\medskip

\noindent
The direct sum $\W=\mathbb{U}\oplus\V$ of persistence modules $\mathbb{U},\V$ is defined  as the persistence module with the vector spaces $W_s=U_s\oplus V_s$ and linear maps $w_s^t=u_s^t\oplus v_s^t$.
\medskip

\noindent
We illustrate the abstract concepts above using geometric constructions of Topological Data Analysis.
Let $f:X\to\R$ be a continuous function on a topological space.
Its \emph{sublevel} sets $X_s^f=f^{-1}((-\infty,s])$ form nested subspaces $X_s^f\subset X_t^f$ for any $s\leq t$.
The inclusions of the sublevel sets respect compositions similarly to a dendrogram $\De$ in Definition~\ref{dfn:dendrogram}.
\medskip

\noindent
On a metric space $X$ with with a distance function $d:X\times X\to[0,+\infty)$, a typical example of a function $f:X\to\R$ is the distance to a finite set of points $A\subset X$. 
More specifically, for any point $p\in X$, let $f(p)$ be the distance from $p$ to (a closest point of) $A$.
For any $r\geq 0$, the preimage $X_r^f=f^{-1}((-\infty,r])=\{q\in X \mid d(q,A)\leq r\}$ is the union of closed balls that have the radius $r$ and centers at all points $p\in A$.
For example, $X_0^f=f^{-1}((-\infty,0])=A$ and $X_{+\infty}^f=f^{-1}(\R)=X$.
\medskip

\noindent
If we consider any continuous function $f:X\to\R$, we have the inclusion $X_s^f\subset X_r^f$ for any $s\leq r$.
Hence all sublevel sets $X_s^f$ form a nested sequence of subspaces within $X$.
The above construction of a \emph{filtration} $\{X_s^f\}$ can be considered as a functor from $\R$ to the category of topological spaces.  
Below we discuss the most practically used case of dimension 0.

\begin{exa}[persistent homology]
\label{exa:persistent_homology}
For any topological space $X$,  the 0-dimensional \emph{homology} $H_0(X)$ is the vector space (with coefficients $\Z_2$) generated by all connected components of $X$.
Let $\{X_s\}$ be any \emph{filtration} of nested spaces, e.g. sublevel sets $X_s^f$ based on a continuous function $f:X\to\R$.
The inclusions $X_s\subset X_r$ for $s\leq r$ induce the linear maps between homology groups $H_0(X_s)\to H_0(X_r)$ and define the \emph{persistent homology} $\{H_0(X_s)\}$, which satisfies the conditions of a persistence module from Definition~\ref{dfn:persistence_module}.
\bs
\end{exa}
\medskip

\noindent 
If $X$ is a finite set of $m$ points, then $H_0(X)$ is the direct sum $\Z_2^m$ of $m$ copies of $\Z_2$.  The persistence modules that can be decomposed as direct sums of interval modules can be described in a very simple combinatorial way by persistence diagrams of dots in $\R^2$.

\begin{dfn}[persistence diagram $\PD(\V)$]
\label{dfn:persistence_diagram}
Let a persistence module $\V$ be decomposed as a direct sum of interval modules from Example~\ref{exa:interval_module} : $\V\cong\bigoplus\limits_{l \in L}\I(p^{*}_l,q^{*}_l)$, where $*$ is $+$ or $-$.
The \emph{persistence diagram} $\PD(\V)$ is the multiset 
$\PD(\mathbb{V}) = \{(p_l,q_l) \mid l \in L \} \setminus \{p=q\}\subset\R^2$.
\bs
\end{dfn}
\medskip

\noindent
The 0-dimensional persistent homology of a space $X$ with a continuous function $f:X\to\R$ will be denoted by $\PD\{H_0(X_s^f)\}$.

\begin{dfn}[a homomorphism of a degree $\de$ between persistence modules]
\label{dfn:homo_modules}
Let $\mathbb{U}$ and $\mathbb{V}$ be persistent modules over $\mathbb{R}$. 
A \emph{homomorphism} $\mathbb{U}\to\V$ of \emph{degree} $\delta\in\R$ is a collection of linear maps $\phi_t:U_t \rightarrow V_{t+\delta}$, $t \in \mathbb{R}$, such that the diagram commutes for all $s \leq t$. 
\begin{figure}[H]
\centering
\begin{tikzpicture}[scale=1.0]
  \matrix (m) [matrix of math nodes,row sep=3em,column sep=4em,minimum width=2em]
  {
     U_s & U_t \\
     V_{s+\delta} & V_{t+\delta} \\};
  \path[-stealth]
    (m-1-1) edge node [left] {$\phi_s$} (m-2-1)
            edge [-] node [above] {$u^t_s$} (m-1-2)
    (m-2-1.east|-m-2-2) edge node [above] {$v^{t+\delta}_{s+\delta}$}
            node [above] {} (m-2-2)
      (m-1-2) edge node [right] {$\phi_t$} (m-2-2);
\end{tikzpicture}
\end{figure}

\noindent 
Let $\text{Hom}^\delta(\mathbb{U},\mathbb{V})$ be all homomorphisms $\mathbb{U}\rightarrow \mathbb{V}$  of degree $\delta$.
Persistence modules $\mathbb{U},\V$ are \emph{isomorphic} if they have inverse homomorphisms $\mathbb{U}\to\V\to\mathbb{U}$ of degree $0$.
\bs
\end{dfn}

\noindent 
For a persistence module $\V$ with maps $v_s^t:V_s\to V_t$, the simplest example of a homomorphism of a degree $\de\geq 0$
 is $1_{\V}^{\de}:\V\to\V$ defined by the maps $v_s^{s+\de}$, $t\in\R$.
So $v_s^t$ defining the structure of $\V$ shift all vector spaces $V_s$ by the difference $\de=t-s$.
\medskip

\noindent 
The concept of interleaved modules below is an algebraic generalization of a geometric perturbation of a set $X$ in terms of (the homology of) its sublevel sets $X_s$.

\begin{dfn}[interleaving distance ID]
\label{dfn:interleaving_distance}
Persistence modules $\mathbb{U}$ and $\mathbb{V}$ are called $\delta$-interleaved if there are homomorphisms $\phi\in \text{Hom}^\delta(\mathbb{U},\mathbb{V})$ and $\psi \in \text{Hom}^\delta(\mathbb{V},\mathbb{U}) $ such that $\phi\circ\psi = 1_{\mathbb{V}}^{2\de} \text{ and } \psi\circ\phi = 1_{\mathbb{U}}^{2\de}$.
The \emph{interleaving distance} between the persistence modules $\mathbb{U}$ and $\mathbb{V}$ is 
$\ID(\mathbb{U},\V)=\inf\{\de\geq 0 \mid \mathbb{U} \text{ and } \mathbb{V} \text{ are } \delta\text{-interleaved} \}.$
\bs
\end{dfn}

\noindent 
If $f,g:X\to\R$ are continuous functions such that $||f-g||_{\infty}\leq\de$ in the $L_{\infty}$-distance, the modules $H_k\{f^{-1}(-\infty,s]\}$, $H_k\{g^{-1}(-\infty,s]\}$ are $\de$-interleaved for any $k$ \cite{cohen2007stability}.
The last conclusion extends to persistence diagrams for the bottleneck distance below.

\begin{dfn}[bottleneck distance BD]
\label{dfn:bottleneck_distance}
Let multisets $C,D$ contain finitely many points $(p,q)\in\R^2$, $p<q$, of finite multiplicity and all diagonal points $(p,p)\in\R^2$ of infinite multiplicity.
For $\de\geq 0$, a $\de$-matching is a bijection $h:C\to D$ such that $|h(a)-a|_{\infty}\leq\de$ in the $L_{\infty}$-distance for any point $a\in C$.
The \emph{bottleneck} distance between persistence modules $\mathbb{U},\V$ is $\BD(\mathbb{U},\mathbb{V}) = \text{inf}\{ \delta \mid \text{ there is a }\delta\text{-matching between } \PD(\mathbb{U}), \PD(\mathbb{V})\}$. 
\bs
\end{dfn}

\noindent 
The original stability of persistence for sequences of sublevel sets found by Herbert Edelsbrunner \cite{cohen2007stability} was extended as Theorem~\ref{thm:stability_persistence} to $q$-tame persistence modules. 
A persistence module $\V$ is $q$-tame if any non-diagonal square in the persistence diagram $\PD(\V)$ contains only finitely many of points, see \cite[section~2.8]{chazal2016structure}.  
Any finitely decomposable persistence module is $q$-tame.
  
\begin{thmm}[stability of persistence modules]\cite[isometry theorem~4.11]{chazal2016structure}
\label{thm:stability_persistence}
 Let $\mathbb{U}$ and $\mathbb{V}$ be q-tame persistence modules. Then $\ID(\mathbb{U},\mathbb{V}) = \BD(\PD(\mathbb{U}),\PD(\mathbb{V}))$,
 where $\ID$ is the interleaving distance, $\BD$ is the bottleneck distance between persistence modules.
\bs
\end{thmm}

%% file: chapters/KNNUpdated.tex
\chapter{New compressed cover tree for $k$-nearest neighbor search} 
\label{ch:knn} 

Section \ref{app:challenge_trees} and Section \ref{sec:challenges_paired_tree} were accepted to Proceedings of TopoInVis 2022 \cite{2208.09447}.

\section{The $k$-nearest neighbor search and overview of results}
\label{sec:intro_knn}

\begin{figure}
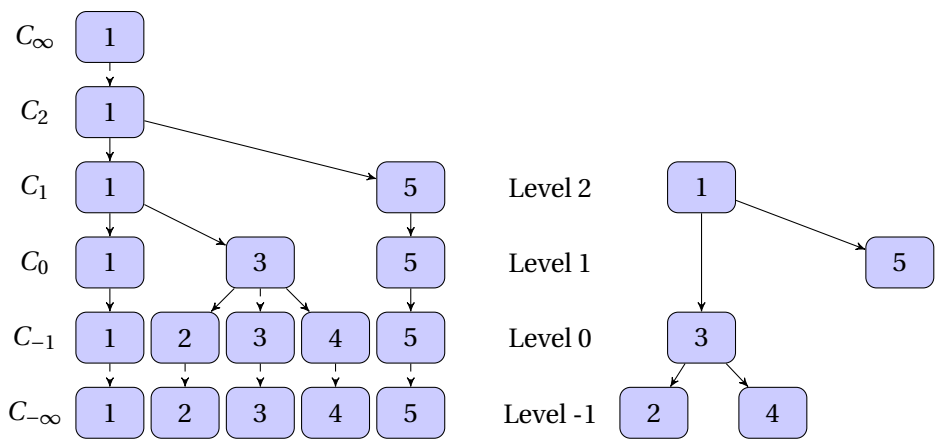

	\centering
	\input images/cover-tree.txt 
	\hspace{0.5cm}
	\input images/explicit_cover_tree.tex 
	\caption{\textbf{Left:} an implicit cover tree from \cite[Section~2]{beygelzimer2006cover} for a finite set of reference points $R = \{1,2,3,4,5\}$. 
		\textbf{Right:} a new compressed cover tree in Definition~\ref{dfn:cover_tree_compressed} corrects the past worst-case complexity results for $k$-nearest neighbors search in $R$.}
	\label{fig:implicitcompressed}
\end{figure}

Given a query set $Q$ and a reference set $R$, which are subsets of some ambient space $X$ with a distance $d$ satisfying all metric axioms the $k$-nearest neighbors problem intends to discover all $k\geq 1$ nearest neighbors in the reference set $R$ for all points from the query set $Q$.
The simplest example is $X$ being $\R^n$ with the Euclidean metric and the query set $Q$ can be a single point or a subset of a larger set $R$. 

\medskip
\noindent
The \emph{exact} $k$-nearest neighbor problem intends to find exact (true) $k$-nearest neighbors of every query point $q$. 
Another probabilistic version of the $k$-nearest neighbor search \cite{har2006fast,manocha2007empirical} aims to find exact $k$-nearest neighbors with a given probability. It should be noted that the probabilistic $k$-nearest neighbor problem can be simplified to $k$ instances of 1-nearest-neighbors problem by splitting $R$ into $k$ subsets $R_{1}, ..., R_{k}$ and searching for nearest neighbors in each of the subsets.

\medskip
\noindent
On the other hand, the approximate version \cite{arya1993approximate,krauthgamer2004navigating,andoni2018approximate, wang2021comprehensive} looks for its $\epsilon$-approximate neighbor $r\in R$ satisfying $d(q,r) \leq (1+\epsilon)d(q,\NN(q))$ for every query point $q \in Q$, where $\epsilon>0$ is fixed and $\NN(q)$ is the exact first nearest neighbor of $q$. 

\medskip
\noindent
\textbf{Spacial data structures.}
It is well known that the time complexity of a brute-force approach of finding all 1st nearest neighbors of points from $Q$ within $R$ is proportional to the product $|Q|\cdot|R|$ of the sizes of $Q, R$.
Already by the mid of 1970s real data was big enough to motivate faster algorithms and sophisticated data structures.
One of the first spacial data structures, a \emph{quadtree} \cite{finkel1974quad}, hierarchically indexes a reference set $R\subset\R^2$ by subdividing its bounding box (a root) into four smaller boxes (children), which are recursively subdivided until final boxes (leaf nodes) contain only a small number of reference points.
A generalization of the quadtree to $\R^n$ exposes an exponential dependence of its computational complexity on $n$, because the $n$-dimensional box is subdivided into $2^n$ smaller boxes.
\medskip

\noindent
The first attempt to overcome this dimensionality curse was the $kd$-tree \cite{bentley1975multidimensional} that subdivides a subset of the reference set $R$ at every  recursion step into two subsets instead of $2^n$ subsets.
Since then more advanced new algorithms utilizing spatial data structures have positively impacted various related research areas such as a minimum spanning tree \cite{bentley1978fast}, range search \cite{pelleg1999accelerating},  $k$-means clustering  \cite{pelleg1999accelerating}, and ray tracing \cite{fussell1988fast}.
The spacial data structures for finding nearest neighbors in the chronological order are $k$-means tree \cite{fukunaga1975branch}, $R$ tree \cite{beckmann1990r}, ball tree \cite{omohundro1989five}, $R^*$ tree \cite{beckmann1990r}, vantage-point tree \cite{yianilos1993data}, TV trees \cite{lin1994tv}, X trees \cite{berchtold1996x}, principal axis tree \cite{mcnames2001fast}, spill tree \cite{liu2004investigation}, cover tree \cite{beygelzimer2006cover}, cosine tree \cite{holmes2008quic}, max-margin tree \cite{ram2012nearest}, cone tree \cite{ram2012maximum} and others.


\begin{dfn}[$k$-nearest neighbor set $\NN_k$]
	\label{dfn:kNearestNeighbor}
For any fixed point $q\in Q$, let $d_1\leq\dots\leq d_{|R|}$ be ordered
 distances from $q$ to all points of $R$, where $|R|$ is the number of points in $R$.
For any $k\geq 1$, the $k$-nearest neighbor sets $\NN_k(q;R)$ consists of all points $u\in R$ with $d(q,u)\leq d_k$. 
\bs
	\bs
\end{dfn}

\noindent
For $Q=R=\{0,1,2,3\}$, the point $q=1$ has ordered distances $d_1=0<d_2=1=d_3<d_4=2$. 
The nearest neighbor sets are $\NN_1(1;R)=\{1\}$,
$\NN_2(1;R)=\{0,1,2\}=\NN_3(1;R)$, $\NN_4(1;R)=R$. 
So 0 can be a 2nd neighbor of 1, then 2 becomes a 3rd neighbor of 1, or these neighbors of $0$ can be found in a different order.


\begin{prob}[$k$-nearest neighbor search]
	\label{pro:knn}
Let $Q,R$ be finite subsets of query and reference points in a metric space $(X,d)$.
For any fixed $k\geq 1$, design an algorithm to find $k$ distinct points from $\NN_k(q;R)$ for all $q\in Q$ so that the parametrized worst-case time complexity of the algorithm is near-linear in $\max\{|Q|,|R|\}$, where hidden constants may depend on structures of $Q,R$ but not on their sizes $|Q|,|R|$. 
\bs
\end{prob}

\noindent

\noindent
\textbf{Expansion constants}. 
Recall that $|\bar B(p,t)|$ is the number (if finite) of points in closed ball $\bar B(p,t)$.


\begin{dfn}[Locally finite space]
	\label{dfn:locally_finite_space}
	We say that a subset $R$ of a metric space $(X,d)$ is locally finite, if for all $p \in X$ and $t \in \R_{+}$ set $\bar{B}(p,t) \cap R$ is finite. 
	\bs
\end{dfn}
\noindent
The following definition of the expansion constant is borrowed from \cite{beygelzimer2006cover}. It also introduces the new minimized expansion constant. Minimized expansion constant is a discrete analog of doubling dimension.

\begin{dfn}[expansion constants $c,c_m$]
	\label{dfn:expansion_constant}
	Let $R$ be a locally finite set in a metric space $X$. 
	The \emph{expansion constant} $c(R)$ is the smallest $c(R)\geq 2$ such that $|\bar{B}(p,2t)|\leq c(R) \cdot |\bar{B}(p,t)|$ for any $p\in R$ and $t\geq 0$ \cite{beygelzimer2006cover}.  
	The \emph{minimized expansion constant} $c_m(R) = \inf\limits_{0 < \xi}\inf\limits_{R\subseteq A\subseteq X}\sup\limits_{p \in A,t > \xi}\dfrac{|\bar{B}(p,2t) \cap A|}{|\bar{B}(p,t) \cap A|}$ where $A$ is a locally finite set which covers $R$.
	\bs
\end{dfn}

\begin{lem}[properties of $c_m$]
\label{lem:expansion_constant_property}
For any finite sets $R\subseteq A$ in a metric space, the following inequalities hold: $c_m(R) \leq c_m(A)$, $c_m(R) \leq c(R)$.
	\bs 
\end{lem}

\noindent
Note that both $c(R), c_m(R)$ is always defined when $R$ is finite. 
We will show that a single outlier point can make the expansion constant $c(R)$ as large as $O(|R|)$.
The set $R=\{1,2,\dots,n,2n+1\}$ of $|R|=n+1$ points has $c(R)=n+1$ because $\bar B(2n+1;n)=\{2n+1\}$ is a single point, while $\bar B(2n+1;2n)=R$ is the full set of $n+1$ points. 
On the other hand the same set $R$ can be extended to a larger uniform set $A=\{1,2,\dots,2n-1,2n\}$ whose expansion constant $c(A)=2$, therefore the minimized constant of the original set $R$ becomes much smaller: $c_m(R) \leq c(A)=2<c(R)=n+1$.
\medskip

\noindent
The constant $c$ from \cite{beygelzimer2006cover} equals to $2^{\text{dim}_{KR}}$ from \cite[Section~2.1]{krauthgamer2004navigating}. 
In \cite[Section~1.1]{krauthgamer2004navigating} the doubling dimension $2^{\text{dim}}$ is defined as a minimum value $\rho$ such that any set $X$ can be covered by $2^{\rho}$ sets whose diameters are half of the diameter of $X$.
The work \cite{krauthgamer2004navigating} proves that  $2^{\text{dim}} \leq 2^{n}$ for any subset of $\R^n$.
In Theorem \ref{thm:normed_space_exp_constant} it is shown that $c_m(R) \leq 2^{n}$ for any a finite subset $R$ of a normed vector space $\R^{n}$. It is not a surprise because $c_m(R)$ mimics $2^{\text{dim}}$.

\medskip
\noindent
\textbf{Navigating nets}.
In 2004, \cite[Theorem~2.7]{krauthgamer2004navigating} claimed that a navigating net can be constructed in time 
$O\big(2^{O(\text{dim}_{KR}(R)} |R| (\log|R|) \log(\log|R|)\big)$ and
all $k$-nearest neighbors of a query point $q$ can be found in time $O(2^{O(\text{dim}_{KR}(R \cup \{q\})}(k + \log|R|)$, where $\text{dim}_{KR}(R \cup \{q\})$ is the expansion constant defined above. 
All proofs and pseudo-codes were unfortunately omitted.
The authors didn't reply to our request for details. 
\medskip

\noindent
\textbf{Modified navigating nets} \cite{cole2006searching} were used in 2006 to claim the time $O(\log(n) + (1/\epsilon)^{O(1)})$ for the $(1+\epsilon)$-approximate neighbors.
All proofs and pseudo-codes were left out, also for the construction of the modified navigating net for the claimed time $O(|R| \cdot \log(|R|))$. 
\medskip


\noindent
\textbf{Cover trees}. 
In 2006, \cite{beygelzimer2006cover} introduced a cover tree inspired by the navigating nets \cite{krauthgamer2004navigating}. 
This cover tree was designed to prove a worst-case bound for the nearest neighbor search in terms of the size $|R|$ of a reference set $R$ and the expansion constant $c(R)$ of Definition \ref{dfn:expansion_constant}. 
Assume that a cover tree is already constructed on set $R$. Then \cite[Theorem~5]{beygelzimer2006cover} claims that nearest neighbor of any query point $q \in Q$ could be found in time $O(c(R)^{12} \cdot \log|R|)$. In 2015, \cite[Section~5.3]{curtin2015improving} pointed out that the proof of \cite[Theorem~5]{beygelzimer2006cover} contains a crucial gap, now have been confirmed by a specific dataset in Counterexample~\ref{cexa:original_all_nearest_neighbors_algorithm}. The time complexity result of the cover tree construction algorithm \cite[Theorem~6]{beygelzimer2006cover} had a similar issue, the gap of which is exposed rigorously in Counterexample~\ref{cexa:construction_algorithm_of_original_cover_tree}.

\medskip
\noindent 
\textbf{Further studies in cover trees.} A noteworthy paper on cover trees \cite{kollar2006fast} introduced a new probabilistic algorithm for the nearest neighbor search, as well as corrected the pseudo-code of the cover tree construction algorithm of \cite[Algorithm~2]{beygelzimer2006cover}.
Later in 2015, a new, more efficient implementation of cover tree was introduced in \cite{izbicki2015faster}. However, no new time-complexity results were proven. A study \cite{jahanseir2016transforming} explored connections between modified navigating nets \cite{cole2006searching} and cover trees \cite{beygelzimer2006cover}.
 Multiple papers \cite{beygelzimer2006coverExtend, ram2009linear, curtin2015plug} studied possibility of solving $k$-nearest neighbor problem (Problem \ref{pro:knn}) by using cover tree on both, the query set and the reference set, for further details see Section~\ref{sec:challenges_paired_tree}.


\begin{table}
	\label{table:KR:knearest}
	\centering
	\caption{Results for exact $k$-nearest neighbors of one query point $q \in X$ using hidden classic expansion constant $c(R)$ of Definition \ref{dfn:expansion_constant} or KR-type constant $2^{\text{dim}_{KR}}$  \cite[Section~2.1]{krauthgamer2004navigating} and assuming that all data structures are already built. Note that $2^{\text{dim}_{KR}}$ corresponds to $c(R)^{O(1)}$ }
	\begin{tabular}{|V{3.0cm}|V{5cm}|V{3cm}|V{3.0cm}|}
		\hline
		Data structure, reference     & time complexity   & space  & proofs \\
		\hline
		Navigating nets \cite{krauthgamer2004navigating} & $O\big(2^{O(\text{dim}_{KR})}(\log(|R|) + k)\big)$ for $k\geq 1$ \cite[Theorem~2.7]{krauthgamer2004navigating} & $O(2^{O(\text{dim})} \cdot |R|)$ & Not available \\
		\hline
		Cover tree \cite{beygelzimer2006cover}   &  $O\big(c(R)^{O(1)}\log(|R|)\big)$ for $k=1$, \cite[Theorem~5]{beygelzimer2006cover} & $O(|R|)$ & Counterexample \ref{cexa:original_all_nearest_neighbors_algorithm} shows that the proof of Theorem 5 is incorrect \\ 
		\hline
		Compressed cover tree [ Definition \ref{dfn:cover_tree_compressed} ]&     $O\big(c(R)^{O(1)} \cdot \log(k) \cdot (\log(|R|) + k)\big)$  & $O(|R|)$, Lem \ref{lem:linear_space_cover_tree}      & Theorem \ref{thm:knn_KR_time} \\
		\hline            
	\end{tabular}
\end{table}

\begin{table}
	\label{table:KR:construction}
	\centering
	\caption{Results for building data structures with hidden classic expansion constant $c(R)$ of Definition \ref{dfn:expansion_constant} or KR-type constant $2^{\text{dim}_{KR}}$ \cite[Section~2.1]{krauthgamer2004navigating}}
	\begin{tabular}{|V{3.0cm}|V{5cm}|V{2.9cm}|V{3.0cm}|}
		\hline
		Data structure, reference     & time complexity   & space & proofs \\
		\hline
		Navigating nets \cite{krauthgamer2004navigating} & $O\big(2^{O(\text{dim}_{KR})} \cdot |R| \log(|R|) \log(\log(|R|) )\big)$, \cite[Theorem~2.6]{krauthgamer2004navigating} & $O(2^{O(\text{dim})}|R|)$ & Not available \\
		\hline
		Cover tree \cite{beygelzimer2006cover}   &  $O(c(R)^{O(1)} \cdot |R| \cdot \log(|R|))$, \cite[Theorem~6]{beygelzimer2006cover} & $O(|R|)$ & Counterexample \ref{cexa:construction_algorithm_of_original_cover_tree} shows that the proof of Theorem 6 is incorrect \\ 
		\hline
		Compressed cover tree [dfn \ref{dfn:cover_tree_compressed}] &     $O\big(c(R)^{O(1)} \cdot |R| \cdot \log(R) \big)$  &  $O(|R|)$, Lem \ref{lem:linear_space_cover_tree}     & Corollary \ref{cor:construction_time_KR} \\
		\hline       
	\end{tabular}
\end{table}

\begin{table}
	\centering
	\caption{Results for exact $k$-nearest neighbors of one point $q$ using hidden $c_m(R)$ or dimensionality constant $2^{\text{dim}}$ \cite[Section~1.1]{krauthgamer2004navigating}  assuming that all structures are built.}
	\begin{tabular}{|V{3.0cm}|V{5cm}|V{2.5cm}|V{3.0cm}|}
		\hline
		Data structure, reference     & time complexity & space & proofs \\
		\hline
		Navigating nets \cite{krauthgamer2004navigating} & $O\big(2^{O(\text{dim})} \cdot \log(\Delta) + |\bar{B}(q,O(1) \cdot d(q,R))|\big)$, for $k = 1$, \cite[Theorem~2.3]{krauthgamer2004navigating} & $O(2^{O(\text{dim})}\cdot|R|)$  & a proof outline in \cite[Theorem~2.3]{krauthgamer2004navigating} \\
		\hline
		Compressed cover tree [dfn \ref{dfn:cover_tree_compressed}] &     $O\big(c_m(R)^{O(1)} \cdot \log(k) \cdot (\log(|\Delta|) + |\bar{B}(q,O(1) \cdot d_k(q,R))| ) \big )$  & $O(|R|)$, Lem \ref{lem:linear_space_cover_tree}     & Corollary~\ref{cor:cover_tree_knn_miniziminzed_constant_time} \\
		\hline          
	\end{tabular}
	\label{table:dim:knearest}
\end{table}

\begin{table}
	\centering
	\caption{Building data structures with hidden $c_m(R)$ or dimensionality constant $2^{\text{dim}}$ \cite[Section~1.1]{krauthgamer2004navigating}}
	\begin{tabular}{|V{3.0cm}|V{5.5cm}|V{3cm}|V{3.0cm}|}
		\hline
		Data structure, reference     & time complexity   & space & proofs \\
		\hline
		Navigating nets \cite{krauthgamer2004navigating} & $O\big(2^{O(\text{dim})} \cdot |R| \cdot \log(\Delta) \cdot \log(\log((\Delta ))\big)$
		& $O(2^{O(\text{dim})} \cdot |R|)$  & \cite[Theorem~2.5]{krauthgamer2004navigating} \\
		\hline
		Compressed cover tree [dfn \ref{dfn:cover_tree_compressed}] & $O\big(c_m(R)^{O(1)}\cdot|R|\log(\Delta(|R|))\big)$    & $O(|R|)$, Lem \ref{lem:linear_space_cover_tree}     & Theorem \ref{thm:construction_time}  \\
		\hline           
	\end{tabular}
	\label{table:dim:construction}  
\end{table}

\noindent
\textbf{Compressed cover tree.}
In this work we correct the past issues of the standard single-tree cover tree approach \cite{beygelzimer2006cover} by applying a new compressed cover tree $\T(R)$ of Definition \ref{dfn:cover_tree_compressed}, which can be constructed on any finite set $R$ with a metric $d$. 
In Theorem~\ref{thm:construction_time_KR} it will be shown that a compressed cover tree $\T(R)$ can be built in time
$O(c_m(R)^8 \cdot c(R)^2 \cdot \log_2(|R|) \cdot |R|)$.

\medskip
 \noindent
The past gap in time complexity result \cite[Theorem~1]{beygelzimer2006cover} of nearest neighborhood search is tackled by introducing a new method, Algorithm \ref{alg:cover_tree_k-nearest} which differs from the original method \cite[Algorithm~1]{beygelzimer2006cover} by having an additional block of code.
The extra block eliminates the issue of having too many successive iterations, in case where the query point $q$ is disproportionately far away from the remaining candidate set $R_i$ on some level $i$.
By means of the new block, Lemma \ref{lem:knn_depth_bound} guarantees that the number of iterations of Algorithm \ref{alg:cover_tree_k-nearest} is bounded by $O(c(R)^2\log_2(|R|))$. 
This new lemma replaces the old result \cite[Lemma~4.3]{beygelzimer2006cover}, which had a similar bound for the number of explicit levels of a cover tree, see Definition~\ref{dfn:explicit_depth_for_compressed_cover_tree}, but couldn't be used to estimate the number of iterations of \cite[Algorithm~1]{beygelzimer2006cover} due to Counterexample~\ref{cexa:original_all_nearest_neighbors_algorithm}. 
\medskip

\noindent
Assume that a compressed cover tree $\T(R)$ is already constructed on the reference set $R$. 
The main result of this work Theorem \ref{thm:knn_KR_time} shows that $k$-nearest neighbors of a query node $q$ can be found in time of
	$$O\Big ( c(R \cup \{q\})^2 \cdot \log_2(k) \cdot \big((c_m(R))^{10}  \cdot \log_2(|R|) + c(R \cup \{q\}) \cdot k\big) \Big).$$
Recall that $c(R)$ can potentially become as large as $O(|R|)$ when $R$ is not uniformly distributed.
Our second result Corollary~\ref{cor:cover_tree_knn_miniziminzed_constant_time} bounds the time complexity
of the new $k$-nearest neighborhood algorithm through making use of only minimized expansion constant $c_m(R)$ of Definition \ref{dfn:expansion_constant} and the aspect ratio $\Delta(R)$ of Definition \ref{dfn:radius+d_min} as parameters. The advantage of this is that the parameters are less dependent on the noise of the data set, in most cases $\Delta(R)$ is relatively small and $c_m(R)$ depends mostly on the dimension of the ambient space $X$. It is shown that $k$-nearest neighbors of $q$ in a reference set $R$ can be computed in time of
$$O\Big ((c_m(R))^{10} \cdot \log_2(k) \cdot \log_2(\Delta(R)) + |\bar{B}(q, 5d_k(q,R))| \cdot \log_2(k) \Big ),$$
 where $d_k(q,R)$ is the distance of $q$ to its $k$th nearest neighbor and $|\bar{B}(q, 5d_k(q,R))|$ is the number of points in the ball 
$\bar{B}(q, 5d_k(q,R))$.



\section{Compressed cover tree}
\label{sec:cover_tree}

This section introduces a new data structure, compressed cover tree in Definition \ref{dfn:cover_tree_compressed} that will be used to solve Problem \ref{pro:knn}. This section also contains important results regarding expansion constants: Lemma \ref{lem:packing} and Lemma \ref{lem:growth_bound}.
Given a $\delta$-spare finite metric space $R$, Lemma \ref{lem:packing} shows that the number of points of $R$ in a closed ball $\bar{B}(p,t)$ can be bound by $c_m(S)^{\mu}$, where $\mu$ depends on ratio $\frac{t}{\delta}$.
On the other hand Lemma \ref{lem:growth_bound} can be used to show that if there exists points $p,q$ in a finite metric space $R$ satisfying $2r < d(p,q) \geq 3r$ for some $r \in \R$, then $|bar{B}(q,4r)| \geq (1 + \frac{1}{c(R)^2})|\bar{B}(q,r)|$.
 
\begin{dfn}[a compressed cover tree $\T(R)$]
	\label{dfn:cover_tree_compressed}
	Let $R$ be a finite set in a metric space $(X,d)$. 
	\emph{A compressed cover tree} $\T(R)$ has the vertex set $R$ with a root $r \in R$ and a \emph{level} function $l : R \rightarrow \Z$ satisfying the conditions below.
	\medskip
	
	\noindent
	(\ref{dfn:cover_tree_compressed}a)
	\emph{Root condition} :  
	the level of the root node $r$ is $l(r) \geq 1 +  \max\limits_{p \in R \setminus \{r\}}l(p)$.
	
	\noindent
	(\ref{dfn:cover_tree_compressed}b)
	\emph{Covering condition} : 
	for every non-root node $q \in R\setminus \{r\}$, we select a unique \emph{parent} $p$ and a level $l(q)$ such that $d(q,p) \leq 2^{l(q)+1}$ and $l(q) < l(p)$; 
	this parent node $p$ has a single link to its  \emph{child} node $q$ in the tree $\T(R)$. 
	\medskip
	
	\noindent
	(\ref{dfn:cover_tree_compressed}c)
	\emph{Separation condition} : 
	for $i \in \Z$, the \emph{cover set} 
	$C_i = \{p \in R \mid l(p) \geq i\}$ has
	$d_{\min}(C_i) = \min\limits_{p \in C_{i}}\min\limits_{q \in C_{i}\setminus \{p\}} d(p,q) > 2^{i}$.
	\medskip
	
	\noindent
	Since there is a 1-1 map between all points of $R$ and all nodes of $\T(R)$, the same notation $p$ can refer to a point in the set $R$ or to a node of the tree $\T(R)$.  
	Set $l_{\max} = 1 +  \max\limits_{p \in R \setminus \{r\}}l(p) $ and $l_{\min} = \min\limits_{p \in R}l(p)$.
	For any node $p\in\T(R)$, $\Child(p)$ denotes the set consisting of all children of $p$, including $p$ itself.
	For any node $p \in\T(R)$, define the \emph{node-to-root} path as a unique sequence of nodes $w_0,\dots,w_m$ such that $w_0 = p$, $w_m$ is the root and $w_{j+1}$ is the parent of $w_{j}$ for $j=0,...,m-1$. 
	A node $q \in\T(R)$ is a \emph{descendant} of a node $p$ if $p$ is in the node-to-root path of $q$. 
	A node $p$ is an \emph{ancestor} of $q$ if $q$ is in the node-to-root path of $p$. 
	Let $\Desc(p)$ be the set of all descendants of $p$, including $p$. 
	\bs
\end{dfn} 

\begin{lem}[Linear space of $\T(R)$]
\label{lem:linear_space_cover_tree}
Let $(R,d)$ be a finite metric space. Then any cover tree $\T(R)$ from Definition \ref{dfn:cover_tree_compressed} takes $O(|R|)$ space. 
\end{lem}
\begin{proof}
Since $\T(R)$ is a tree , both its vertex set and its edge set contain at most $|R|$ nodes. Therefore $\T(R)$ takes at most $O(|R|)$ space. 
\end{proof}

\begin{figure}[h]
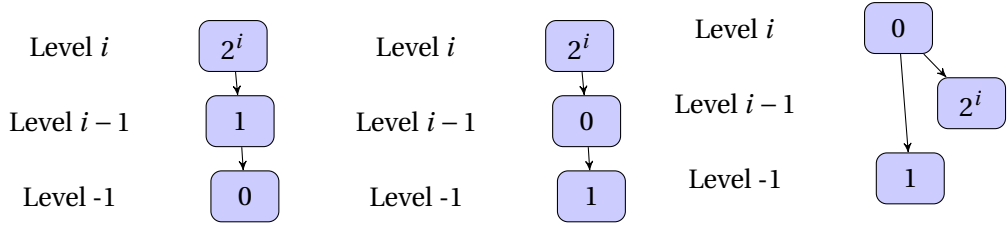

	\centering
	\begin{subfigure}{.30\textwidth}
		\centering
		\input images/easy_tree_one.tex
		\label{fig:cover_tree_variant_one}
		
	\end{subfigure}
	\begin{subfigure}{.30\textwidth}
		\centering
		\input images/easy_tree_two.tex
		\label{fig:cover_tree_variant_two}
	\end{subfigure}
	\begin{subfigure}{.30\textwidth}
		\centering
		\input images/easy_tree_three.tex
		\label{fig:cover_tree_variant_three}
	\end{subfigure}
	\caption{
		Compressed cover trees $\T(R)$ from Definition~\ref{dfn:cover_tree_compressed} for $R = \{0,1,2^{i}\}$. 
	}
	\label{fig:cover_tree_easy_example}
\end{figure}

\begin{figure}
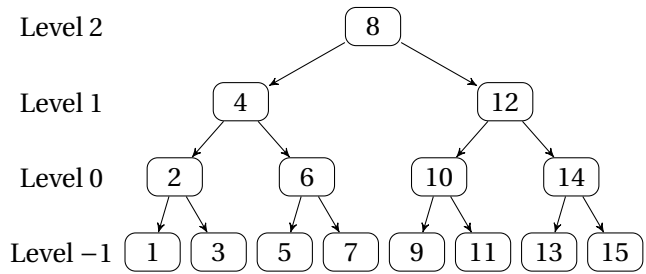

	\centering
	\input images/CoverTreeLongExample/good_tree_example_0.tex
	\caption{Compressed cover tree $\T(R)$ on the set $R$ in Example \ref{exa:cover_tree_big} with root $16$. }
	\label{fig:cover_tree_big}
\end{figure}

\begin{exa}[$\T(R)$ in Fig.~\ref{fig:cover_tree_big}]
	\label{exa:cover_tree_big}
	Let $(\R, d = |x-y|)$ be the real line with euclidean metric.
	Let $R = \{1,2,3,...,15\} $ be its finite subset. 
	Fig.~\ref{fig:cover_tree_big} shows a compressed cover tree on the set $R$ with the root $r=8$. 
	The cover sets of $\T(R)$ are $ C_{-1} = \{1,2,3,...,15\}$, $C_0 = \{2,4,6,8,10,12,14\}$, $C_{1} = \{4,8,12\}$ and $C_{2} = \{8\}$.
	We check the conditions of Definition \ref{dfn:cover_tree_compressed}.
	\begin{itemize}
		\item Root condition $(\ref{dfn:cover_tree_compressed}a)$: 
		since $\max_{p \in R \setminus \{8\}}d(p, 8) = 7$ and $\ceil{\log_2(7)} - 1= 2$, the root can have the level $l(8) = 2$.
		\item Covering condition (\ref{dfn:cover_tree_compressed}b) : for any $i \in {-1,0,1,2}$, let $p_i$ be arbitrary point having $l(p_i) = i$. Then we have 
		$d(p_{-1}, p_{0}) = 1 \leq 2^{0}$, 
		$d(p_0, p_1) = 2 \leq 2^{1}$ and  
		$d(p_1, p_2) = 4 \leq 2^{2}$.
		\item  Separation condition (\ref{dfn:cover_tree_compressed}c) : $d_{\min}(C_{-1}) = 1 > \frac{1}{2} = 2^{-1}$, $d_{\min}(C_{0}) = 2 > 1 = 2^{0}, d_{\min}(C_{1}) = 4 > 2 = 2^{1}$. 
		\bs
	\end{itemize}
\end{exa}

\noindent
A cover tree was defined in \cite[Section~2]{beygelzimer2006cover} as a tree version of a navigating net from \cite[Section ~ 2.1]{krauthgamer2004navigating}. 
For any index $i \in \Z\cup \{\pm\infty\}$, the level $i$ set of this cover tree coincides with the cover set $C_i$ above, which can have nodes at different levels in Definition~\ref{dfn:cover_tree_compressed}. 
Any point $p \in C_i$ has a single parent in the set $C_{i+1}$, which satisfied conditions (\ref{dfn:cover_tree_compressed}b,c). 
\cite[Section~2]{beygelzimer2006cover} referred to this original tree as an implicit representation of a cover tree.
Such a tree in Figure \ref{fig:tripleexample} (left) contains infinitely many repetitions of every point $p\in R$ in long branches and will be called an \emph{implicit cover tree}.
\medskip

\noindent
Since an implicit cover tree is formally infinite, for practical implementations, the authors of \cite{beygelzimer2006cover} had to use another version that they named an explicit representation of a cover tree. 
We call this version an \emph{explicit cover tree}. 
Here is the full defining quote at the end of \cite[Section~2]{beygelzimer2006cover}: "The explicit representation of the tree coalesces all nodes in which the only child is a self-child". 
In an explicit cover tree, if a subpath of every node-to-root path consists of all identical nodes without other children, all these identical nodes collapse to a single node, see Figure \ref{fig:tripleexample} (middle). 
\medskip

\noindent
Since an explicit cover tree still contains repeated points, Definition~\ref{dfn:cover_tree_compressed} is well-motivated by the aim to include every point only once, which saves memory and simplifies all subsequent algorithms, see Fig.~\ref{fig:tripleexample} (right).

\begin{figure}
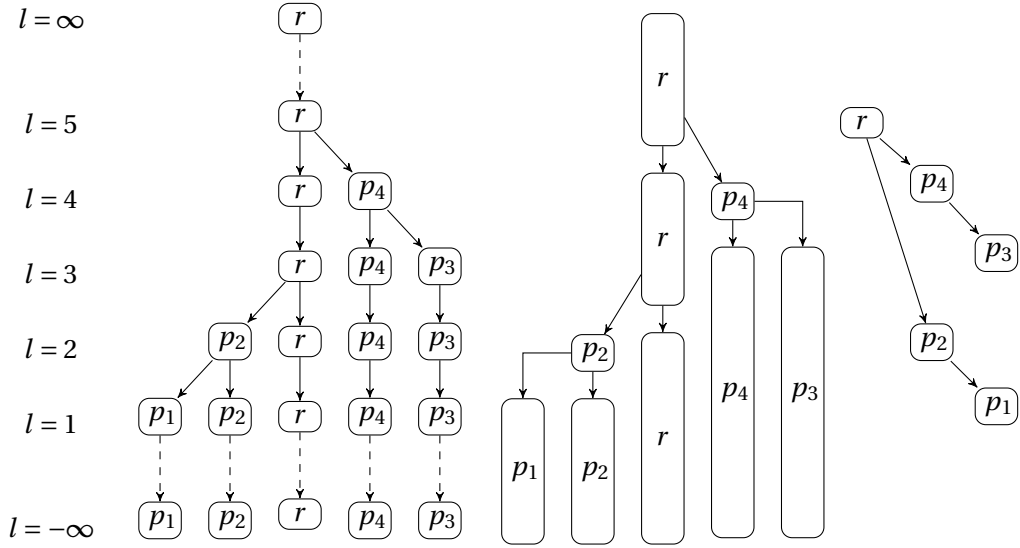

	\centering
	\input images/tripleExample.tex
	\caption{A comparison of past cover trees and a new tree in Example \ref{exa:implicitexplicitexample}. \textbf{Left:} an implicit cover tree contains infinite repetitions. \textbf{Middle:} an explicit cover tree. \textbf{Right:} a compressed cover tree from Definition \ref{dfn:cover_tree_compressed} includes each point once.  }
	\label{fig:tripleexample}
\end{figure}

\begin{exa}[a short train line tree]
	\label{exa:implicitexplicitexample}
	Let $G$ be the unoriented metric graph consisting of two vertices $r,q$ connected by three different edges $e,h,g$ of lengths $|e| = 2^6$ , $|h| = 2^{3}$ , $|g| = 1$. Let $p_{4}$ be the middle point of the edge $e$. 
	Let $p_{3}$ be the middle point of the subedge $(p_4 , q)$. 
	Let $p_{2}$ be the middle point of the edge $h$.
	Let $p_{1}$ be the middle point of the subedge $(p_{2}, q)$. 
	Let $R = \{p_1, p_2,p_3,p_4,r\}$. 
	We construct a compressed cover tree $\T(R)$ by choosing the level $l(p_i) = i$ and by setting the root $r$ to be the parent of both $p_2$ and $p_4$, $p_4$ to be the parent of $p_{3}$, and $p_{2}$ to be the parent of $p_{1}$. 
	Then $\T(R)$ satisfies all the conditions of Definition \ref{dfn:cover_tree_compressed}, see a comparison of the three cover trees in Fig.~\ref{fig:tripleexample}.
	\bs
\end{exa}


\begin{figure}[h]
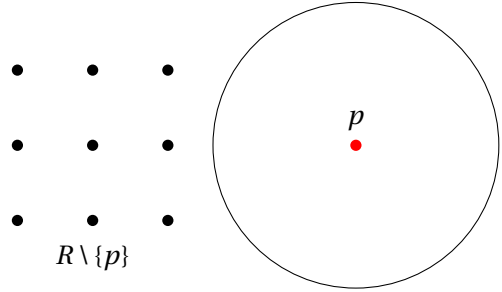

	\centering
	\input images/outlierconstruction.tex
	\caption{Example~\ref{exa:outlierconstruction} describes a set $R$ with a big expansion constant $c(R)$. 
		Let $R\setminus \{p\}$ be a finite subset of a unit square lattice in $\R^2$, but a point $p$ is located far away from $R\setminus \{p\}$ at a distance larger than $\diam(R \setminus \{p\})$.
		Definition \ref{dfn:expansion_constant} implies that $c(R) = |R|$. }
	\label{fig:outlierconstruction}
\end{figure}

\noindent
Even a single outlier point can make the expansion constant big. 
Consider set $R = \{1,2,...,n-1,2n\}$ for some $n \in \Z_{+}$. 
Since $|\bar{B}(2n,n)| = 1$ and $|\bar{B}(2n,n)| = |R|$, we have $c(R) = |R|$. 
Example~\ref{exa:outlierconstruction} shows that expansion constant of a set $R$ can be as big as $|R|$. 

\begin{exa}[one outlier can make the expansion constant big]
	\label{exa:outlierconstruction}
	Let $R$ be a finite metric space and $p \in R$ satisfy $d(p,R \setminus \{t\}) > \diam(R \setminus \{p\})$. 
	Since $\bar{B}(p, 2d(p,R \setminus \{t\}) = R$ , $\bar{B}(p, d(p,R \setminus \{t\}) = \{p\}$, we get $c(R) = N$, see Fig.~\ref{fig:outlierconstruction}. 
	\bs
\end{exa}

\noindent
Example~\ref{exa:minimized_normal_expansion_constant} shows that the minimized expansion can be significantly smaller than the original expansion constant. 

\begin{exa}[minimized expansion constants]
	\label{exa:minimized_normal_expansion_constant}
	Let $(\R, d)$ be the Euclidean line. 
	For an integer $n>10$, consider the finite sets $R = \{2^{i} \mid i \in [1,n]\}$ and let $Q = \{i \mid i \in [1,2^n]\}$. 
	If $0<\epsilon < 10^{-9}$, then
	$\bar{B}(2^n, 2^{n-1} - \epsilon) = \{2^n\} $ and $\bar{B}(2^n, 2(2^{n-1} - \epsilon)) = R$, so $c(R) = n$.
	For any $q \in Q$ and any $t \in \R$, we have the balls $\bar{B}(q,t) = \mathbb{Z} \cap [q - t, q + t]$ and 
	$\bar{B}(q,2t) = \mathbb{Z} \cap [q - 2t, q + 2t]$, so $c(Q) \leq 4$.  
	Lemma \ref{lem:expansion_constant_property} implies that $c_m(R) \leq c_m(Q) \leq c(Q) \leq 4$.
	\bs 
\end{exa}

\noindent
Lemma~\ref{lem:compressed_cover_tree_descendant_bound} provides an upper bound for a distance between a node and its descendants.  

\begin{lem}[a distance bound on descendants]
	\label{lem:compressed_cover_tree_descendant_bound}
	Let $R$ be a finite subset of an ambient space $X$ with a metric $d$. 
	In a compressed cover tree $\T(R)$, let $q$ be any descendant of a node $p$. Let the node-to-root path $S$ of $q$ contain a node $u$ satisfying $u \in \Child(p) \setminus \{p\}$. Then $d(p,q) \leq 2^{l(u) + 2} \leq 2^{l(p) + 1}$. 
	\bs
\end{lem}
\begin{proof}
	Let $(w_0, ..., w_m)$ be a subpath of the node-to-root path for $w_0 = q$ , $w_{m-1} = u$, $w_m = p$. 
	Then $d(w_{i}, w_{i+1}) \leq 2^{l(w_i) + 1}$ for any $i$. 
	The first required inequality follows from the triangle inequality below: 
	$$    d(p,q) \leq \sum^{m-1}_{j = 0}d(w_j, w_{j+1})  \leq \sum^{m-1}_{j = 0}2^{l(w_j) + 1} \leq   \sum_{t = l_{\min}}^{l(u) + 1}2^{t}\leq  2^{l(u) + 2} $$
	Finally, $l(u) \leq l(p) - 1$ implies that $d(p,q) \leq 2^{l(p)+1}$.
\end{proof}

\noindent
Lemma~\ref{lem:packing} uses the idea of \cite[Lemma~1]{curtin2015plug} to show that if $S$ is a $\delta$-sparse subset of a metric space $X$, then $S$ has at most $(c_m(S))^\mu$ points in the ball $\bar{B}(p,r)$, where $c_m(S)$ is the minimized expansion constant of $S$, while $\mu$ depends on $\delta,r$. 

\begin{figure}
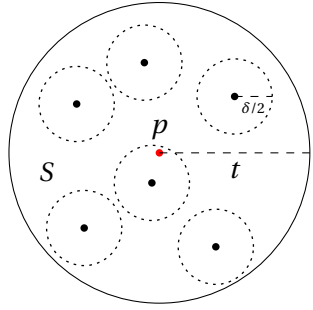

	\centering
	\input  \input images/packingLemmaIllustration.tex 
	\caption{This volume argument proves Lemma~\ref{lem:packing}. By using an expansion constant, we can find an upper bound for the number of smaller balls of radius $\frac{\delta}{2}$ that can fit inside a larger $\bar{B}(p, t)$. }
	\label{fig:packingLemma}
\end{figure}

\begin{lem}[packing]
	\label{lem:packing}
	Let $S$ be a finite $\delta$-sparse set in a metric space $(X,d)$, so $d(a,b) > \delta$ for all $a,b \in S$. 
	Then, for any point $ p \in X$ and any radius $t > \delta$, we have
	$|\bar{B}(p, t)  \cap S | \leq (c_m(S))^{\mu}$, where $\mu = \lceil \log_2(\frac{4t}{\delta} + 1) \rceil $.
	\bs
\end{lem}
\begin{proof}
	Assume that $d(p,q) > t$ for any point $q \in S$. 
	Then $\bar{B}(p, t) \cap S = \emptyset$ and the lemma holds trivially. Otherwise $\bar{B}(p, t) \cap S$ is non-empty. 
	By Definition~\ref{dfn:expansion_constant} of a minimized expansion constant, for any $\epsilon > 0$, we can always find a set $A$ such that $S \subseteq A \subseteq X$ , $\xi = \frac{\delta}{4}$ and
	\begin{ceqn}
		\begin{equation}
			\label{eqa:dfn_of_exp_constant}
			|B(q,2s) \cap A| \leq (c_m(S) + \epsilon) \cdot | B(q,s) \cap A|
		\end{equation}
	\end{ceqn}
	for any $q \in A$ and $s > \xi = \frac{\delta}{4}$. Note that for any $u \in \bar{B}(p,t) \cap S$ we have $\bar{B}(u, \frac{\delta}{2}) \subseteq  \bar{B}(u, t + \frac{\delta}{2})$. 
	Therefore, for any point $q \in \bar{B}(p,t) \cap S$, we get
	$$\bigcup_{u \in \bar{B}(p, t) \cap S}\bar{B}(u, \frac{\delta}{2}) \subseteq  \bar{B}(p,t + \frac{\delta}{2}) \subseteq \bar{B}(q, 2t +  \frac{\delta}{2})$$
	Since all the points of $S$ were separated by $\delta$, we have
	\begin{equation*}
		\label{eqa:packing_zero}
		| \bar{B}(p, t) \cap S| \cdot \min_{u \in \bar{B}(p, t) \cap S}| \bar{B}(u, \frac{\delta}{2}) \cap A|  \leq \sum_{u \in \bar{B}(p, t) \cap S} | \bar{B}(u, \frac{\delta}{2}) \cap A |\leq | \bar{B}(q, 2t + \frac{\delta}{2}) \cap A |
	\end{equation*}
	In particular, by setting $q = \mathrm{argmin}_{a \in S \cap \bar{B}(p,t)}| \bar{B}(a, \frac{\delta}{2})|$, we get:
	\begin{ceqn}
		\begin{equation}
			\label{eqa:packing_one}
			| \bar{B}(p, t) \cap S| \cdot |  \bar{B}(q, \frac{\delta}{2}) \cap A| \leq | \bar{B}(q, 2t + \frac{\delta}{2}) \cap A |
		\end{equation}
	\end{ceqn}
	Inequality~(\ref{eqa:dfn_of_exp_constant}) applied $\mu$ times 
	for the radii $s_i = \dfrac{2t + \frac{\delta}{2}}{2^{i}} $, $i = 1,...,\mu$, implies that:
	\begin{ceqn}
		\begin{equation} 
			\label{eqa:packing_two}
			|\bar{B}(q,2t + \frac{\delta}{2}) \cap A| \leq (c_m(S) + \epsilon)^{\mu}|\bar{B}(q, \dfrac{2t + \frac{\delta}{2}}{2^{ \mu}}) \cap A |  \leq  (c_m(S) + \epsilon)^{ \mu}|\bar{B}(q, \frac{\delta}{2}) \cap A|. 
		\end{equation}
	\end{ceqn}
	By combining inequalities (\ref{eqa:packing_one}) and (\ref{eqa:packing_two}), we get
	$$| \bar{B}(p,t) \cap S  |\leq \dfrac{|\bar{B}(q, 2t + \frac{\delta}{2})  \cap A |}{|\bar{B}(q, \frac{\delta}{2})  \cap A|} \leq (c_m(S)+\epsilon)^{\mu}.$$
	The required inequality is obtained by letting $\epsilon \rightarrow 0$.\end{proof}

\noindent
\cite[Section~1.1]{krauthgamer2004navigating} defined dim($X$) of a space $(X,d)$ as the minimum number $m$ such that every set $U \subseteq X$ can be covered by $2^{m}$ sets whose diameter is a half of the diameter of $U$. 
If $U$ is finite, an easy application of Lemma \ref{lem:packing} for $\delta = \frac{r}{2}$ shows that
$\text{dim}(X) \leq \sup_{A \subseteq X}(c_m(A))^4 \leq \sup_{A \subseteq X}\inf_{A \subseteq B \subseteq X}(c(B))^4,$
where $A$ and $B$ are finite subsets of $X$. 
\medskip

\noindent
Let $T(R)$ be an implicit cover tree of \cite{beygelzimer2006cover} on a finite set $R$. 
\cite[Lemma~4.1]{beygelzimer2006cover} showed that the number of children of any node $p \in T(R)$ has the upper bound $(c(R))^4$. 
Lemma~\ref{lem:compressed_cover_tree_width_bound} generalizes \cite[Lemma~4.1]{beygelzimer2006cover} for a compressed cover tree.

\begin{lem}[width bound]
	\label{lem:compressed_cover_tree_width_bound}
	Let $R$ be a finite subset of a metric space $(X,d)$. 
	For any compressed cover tree $\T(R)$, any node $p$ we have  
	$$\{q \in \Child(p) \mid l(q) = i\} \cup \{p\} \leq (c_m(R))^4, $$ where $c_m(R)$ is the minimized expansion constant of the set $R$.
	\bs
\end{lem}
\begin{proof}
	By the covering condition of $\T(R)$, any child $q$ of $p$ located on the level $i$ has $d(q,p) \leq 2^{i+1}$. 
	Then the number of children of the node $p$ at level $i$ at most $|\bar{B}(p,2^{i+1})|$. 
	The separation condition in Definition~\ref{dfn:cover_tree_compressed} implies that the set $C_i$ is a $2^{i}$-sparse subset of $X$. 
	We apply Lemma \ref{lem:packing} for $t = 2^{i+1}$ and $\delta = 2^{i}$. 
	Since $4 \cdot \frac{t}{\delta} + 1 \leq 4 \cdot 2 + 1 \leq 2^4$, we get $|\bar{B}(q,2^{i+1}) \cap C_i|  \leq (c_m(C_{i}))^4$. Lemma \ref{lem:expansion_constant_property} implies that $(c_m(C_{i}))^4  \leq (c_m(R))^4 $, so the upper bound is proved.
\end{proof}

%

\begin{lem}[growth bound]
	\label{lem:growth_bound}
	Let $(A,d)$ be a finite metric space, let $q \in A$ be an arbitrary point and let $r \in \R$ be a real number.
	Let $c(A)$ be the expansion constant from Definition \ref{dfn:expansion_constant}.
	If there exists a point $p \in A$ such that $2r < d(p,q) \leq 3r$, then $|\bar{B}(q, 4r)|  \geq  (1+\frac{1}{c(A)^2}) \cdot |\bar{B}(q, r)|$. 
\end{lem}
\begin{proof}
	Since $\bar{B}(q,r) \subset \bar{B}(p,3r + r)$, we have
	$|\bar{B}(q,r)| \leq |\bar{B}(q,4r)| \leq c(A)^2 \cdot |\bar{B}(p,r)|$.
	And since $\bar{B}(p,r)$ and $\bar{B}(q,r)$ are disjoint and are subsets of $\bar{B}(q,4r)$, we have
	$$|\bar{B}(q, 4r)| \geq |\bar{B}(q,r)| + |\bar{B}(p,r)| \geq |\bar{B}(q,r)| + \frac{|\bar{B}(q,r)|}{c(A)^2} \geq (1 + \frac{1}{c(A)^2}) \cdot |\bar{B}(q,r)|,$$
	which proves the claim.
\end{proof}
\begin{lem}[extended growth bound]
	\label{lem:growth_bound_extension}
	Let $(A,d)$ be a finite metric space, let $q \in A$ be an arbitrary point. Let $p_1, ..., p_n$ be a sequence of distinct points in $R$, in such a way that  for all $i \in \{2,...,n\}$ we have $4 \cdot d(p_i, q) \leq d(p_{i+1},q)$.
	Then $$|\bar{B}(q,\frac{4}{3} \cdot d(q,p_n))| \geq (1+\frac{1}{c(A)^2})^{n} \cdot |\bar{B}(q,\frac{1}{3} \cdot d(q,p_1))|.$$
\end{lem}
\begin{proof}
	Let us prove this by induction. In basecase $n = 1 $ define $r = \frac{d(q,p_{m})}{3}$. 
	Now by Lemma \ref{lem:growth_bound} we have 
	$$|\bar{B}(q, \frac{4}{3}d(q,p_1)) | = |\bar{B}(q, 4r)|  \geq  (1+\frac{1}{c(A)^2}) \cdot |\bar{B}(q, r)| = |\bar{B}(q, \frac{1}{3}d(q,p_1)) |.$$

	\medskip
	
	Assume now that the claim holds for some $i = m$ and let $p_1, ..., p_{m+1}$ be a sequence satisfying the condition of Lemma \ref{lem:growth_bound_extension}. 
	By induction assumption we have $ |\bar{B}(q, \frac{4}{3}d(q,p_m)) | \geq (1+\frac{1}{c(A)^2})^m \cdot |\bar{B}(q, \frac{1}{3}d(q,p_1)) |$. Let us pick $r = \frac{d(q,p_{m+1})}{3}$. Then we have:
	\begin{ceqn}
		\begin{align*}
			|\bar{B}(q,\frac{4}{3} \cdot d(q,p_{m+1}))| & \geq (1+\frac{1}{c(A)^2}) \cdot |\bar{B}(q,\frac{1}{3} \cdot d(q,p_{m+1}))| \\
			&\geq (1+\frac{1}{c(A)^2}) \cdot |\bar{B}(q,\frac{4}{3} \cdot d(q,p_{m}))| \\
			&\geq (1+\frac{1}{c(A)^2}) \cdot (1+\frac{1}{c(A)^2})^m \cdot |\bar{B}(q,\frac{1}{3} \cdot d(q,p_{1}))| \\
			&\geq (1+\frac{1}{c(A)^2})^{m+1} \cdot |\bar{B}(q,\frac{1}{3} \cdot d(q,p_{1}))|
		\end{align*}
	\end{ceqn}
	which proves the claim. 
\end{proof}

\begin{lem}
	\label{lem:hard_function_bound}
	For every $x \in \R$ satisfying $x \geq 2$, the following inequality holds:
	$$x^2 \geq \frac{1}{\log_2(1 + \frac{1}{x^2})} .$$
\end{lem}
\begin{proof}
Let $\ln$ be natural logarithm. Note first that for any $u > 0$ we have:
$$\frac{u}{u+1} = \int^u_0 \frac{dt}{u+1} \leq \int^u_0 \frac{dt}{t+1} = \ln(u+1).$$
By setting $u = \frac{1}{x^2} > 0$ we get:
$$\log_2(1 + \frac{1}{x^2}) = \frac{\ln(\frac{1}{x^2})}{\ln(2)} \geq \frac{1}{\ln(2)} \cdot \frac{\frac{1}{x^2}}{\frac{1}{x^2}+1} 
= \frac{1}{\ln(2)} \cdot \frac{1}{x^2+1}. $$
Let us now show that for $x \geq 2$ we have: $\frac{1}{\ln(2)} \cdot \frac{1}{x^2+1} \geq \frac{1}{x^2} $.
Note first that $4 \geq \frac{\ln(2)}{1 - \ln(2)}$. Since $x \geq 2$ we have $x^2 \geq \frac{\ln(2)}{1 - \ln(2)}$.
Therefore $x^2 - \ln(2) \cdot x^2 \geq \ln(2)$ and $x^2 \geq \ln(2) \cdot (1 + x^2)$. 
It follows that $\frac{1}{\ln(2)}\frac{1}{1+x^2} \geq \frac{1}{x^2}$, which proves the claim.




\end{proof}

\begin{dfn}[the height of a compressed cover tree]
	\label{dfn:depth}
	For a compressed cover tree $\T(R)$ on a finite set $R$,
	the \emph{height set} is $H(\T(R))=\{ i \mid C_{i-1}\neq C_{i}\}\cup \{l_{\max},l_{\min}\}$.
	The size $|H(\T(R))|$ of this set is called the \emph{height} of $\T(R)$.
	\bs
\end{dfn}

\noindent
The new concept of the height $|H(\T)|$ will justify a near-linear parameterized worst-case complexity in Theorem~\ref{thm:knn_KR_time}.
By condition~(\ref{dfn:cover_tree_compressed}b), the height $|H(\T(R))|$ counts the number of levels $i$ whose cover sets $C_i$ include new points that were absent on higher levels.
Then $|H(\T)|\leq|R|$ as any point can be alone at its own level.

\begin{dfn}
	\label{dfn:radius+d_min}
	For any finite metric set $R$ with a metric $d$, the \emph{diameter} is $$\rad(R) = \max\limits_{p \in R}\max\limits_{q \in R}d(p,q).$$
	The \emph{aspect ratio} \cite{krauthgamer2004navigating} is $\Delta(R) = \dfrac{\rad(R)}{d_{\min}(R)}$.
	\bs
\end{dfn}

\begin{lem}
	\label{lem:depth_bound}
	Any finite set $R$ has the bound $|H(\T(R))|\leq 1+\log_2(\Delta(R))$.
	\bs
\end{lem}
\begin{proof}[\textbf{Proof of Lemma}~\ref{lem:depth_bound}]
	We have $|H(\T(R))|\leq l_{\max} - l_{\min}+1$ by Definition~\ref{dfn:depth}. 
	We estimate $l_{\max} - l_{\min}$ as follows.
	\medskip
	
	\noindent
	Let $p \in R$ be a point such that $\rad(R) = \max_{q \in R}d(p,q)$. 
	Then $R$ is covered by the closed ball $\bar B(p; \rad(R))$.
	Hence the cover set $C_i$ at the level $i=\log_2(\rad(R))$ consists of a single point $p$. 
	The separation condition in Definition~\ref{dfn:cover_tree_compressed} implies that 
	$l_{\max}\leq \log_2(d_{\max}(R))$.
	Since any distinct points $p,q \in R$ have $d(p,q)\geq d_{\min}(R)$, the covering condition implies that no new points can enter the cover set $C_i$ at the level $i=[\log_2(d_{\min}(R))]$, so $l_{\min}\geq\log_2(d_{\min}(R))$.
	Then
	$|H(\T(R))| \leq 1+l_{\max} - l_{\min} \leq 
	1+\log_2(\frac{\rad(R)}{d_{\min}(R)})$.
\end{proof}

\noindent
If the aspect ratio $\Delta(R) = O(\text{Poly}(|R|))$ polynomially depends on the size $|R|$, then $|H(\T(R))| \leq O(\log(|R|))$.
Lemma \ref{lem:growth_bound} corresponds \cite[Lemma~4.2]{beygelzimer2006cover} with slightly modified notation.

\section{Minimized expansion constant in Euclidean space $\R^n$}
\label{sec:minimized_exp_constant}

The main result of this section is Theorem \ref{thm:normed_space_exp_constant},
in which we will show that for any finite subset $R$ of normed vector space $(\R^n, \Vert \cdot \Vert )$, the minimized expansion constant from 
Definition \ref{dfn:expansion_constant} satisfies 
$$c_m(R)  = \inf\limits_{0 < \xi}\inf\limits_{R\subseteq A\subseteq \R^{n}}\sup\limits_{p \in A,t > \xi}\dfrac{|\bar{B}(p,2t) \cap A|}{|\bar{B}(p,t) \cap A|} \leq 2^{n}.$$ 
The proof of Theorem \ref{thm:normed_space_exp_constant} is based on the volume argument.
We briefly explain the idea before giving the proof later.
For this purpose, we recall the definition of the Lebesgue measure in Definition \ref{dfn:lebesgue_measure}.

\medskip
\noindent
In Definition \ref{dfn:voronoi_region} we define concepts of grid $G(\delta) = \delta \cdot \Z^{n}$ and cubic regions 
$\bar{V}(p,\delta) = p + [-\frac{\delta}{2}, \frac{\delta}{2}]^n$. For every $\delta > 0$ we define grid extension $U(\delta)$ of $R$ 
as set $U(\delta) = (G(\delta) \setminus f(R)) \cup R$, where $f:R \rightarrow G(\delta)$ is used to replace points of $R$ with their nearest neighbors in grid $G(\delta)$. 

\medskip
\noindent
Note that $\xi$ in the definition of $c_m(R)$ acts as a low bound for radius $t > \xi$.
Let $\gamma > 0$ be a constant, that depends on dimension $n$ and norm $\Vert \cdot \Vert$.
In Lemma \ref{lem:u_bounds} it is shown that if $\delta$ satisfies
$0 < \delta <  \frac{\xi}{\gamma}$, then for any $p \in U(\delta)$ and $t > \xi$ we can bound $|\bar{B}(p,t) \cap U(\delta)|$ as follows:
$$\frac{\mu( \bar{B}(p,t - \delta \cdot \gamma))}{\delta^n} \leq |\bar{B}(p,t) \cap U(\delta)| \leq \frac{\mu( \bar{B}(p,t + \delta \cdot \gamma))}{\delta^n}. $$
 Therefore 
$$\frac{|\bar{B}(p,2t) \cap U(\delta)|}{|\bar{B}(p,t) \cap U(\delta)|} \leq \frac{\mu( \bar{B}(p,2t + \delta \cdot \gamma))}{\mu( \bar{B}(p,t - \delta \cdot \gamma))}.$$
Now since this inequality is satisfied for any $\delta > 0$, we can choose arbitrary dense grid extension $U(\delta)$.
It will be shown that when $\delta \rightarrow 0$ , then $$\frac{\mu( \bar{B}(p,2t + \delta \cdot \gamma))}{\mu( \bar{B}(p,t - \delta \cdot \gamma))} \rightarrow 2^{n},$$
we conclude that $c_m(R) \leq 2^{n} $.

\begin{dfn}[Normed vector space on $\R^n$, {\cite{Rudin1990-tp}}]
	\label{dfn:normed_rn_space}
	Let $\R^{n}$ be an $n$-dimensional space. A norm is a function $\Vert \cdot \Vert:\R^{n} \rightarrow \R$ satisfying the following properties:
	\begin{enumerate}
		\item Non-negativity, $\Vert x \Vert \geq 0$.
		\item It is positive on nonzero vectors, that is $\Vert x \Vert = 0$ implies $x = 0$.
		\item For every vector $x \in \R^{n}$, and every scalar $a \in \R$: $\Vert a \cdot x \Vert = |a| \cdot \Vert x \Vert$.
		\item The triangle inequality holds for every $x \in \R^{n}$ and $y \in \R^{n}$, $\Vert x + y \Vert \leq \Vert x \Vert + \Vert y \Vert$.
	\end{enumerate}
	A norm induces a metric by the formula $d(x,y) = \Vert x - y \Vert$. For every $i \in \{1,...,n\}$ let $e_i$ be a unit vector of $\R^n$ i.e. $e_i(i) = 1$ and $e_i(j) = 0$ for all $j \in \{1,...,n\} \setminus \{i\}$. Define $\rho = \max_{i \in \{1,...,n\}}\Vert e_i \Vert$. \bs 
\end{dfn}

\begin{dfn}[Lebesgue outer measure, {\cite[Section~2.A]{Jones2000-db}}]
	\label{dfn:lebesgue_measure}
	Let $\R^{n}$ be an $n$-dimensional space. Define $n$-dimensional interval as
	$$I = \{x \in \R^{n} \mid a_i \leq x_i \leq b_i, i = 1,...,n\} = [a_1,b_1] \times ... \times [a_n,b_n],$$
	with sides parallel to the coordinate axes. Define Lebesgue outer measure $\mu^{*}:\{A \mid A \subseteq \R^{n}\} \rightarrow [0,\infty) \cup 
	\{\infty\}$ 
	of interval $I$ as 
	$\mu^{*}(I) = (b_1-a_1) \cdot ... \cdot (b_n-a_n)$. The Lebesgue $\mu$ measure of a set $A \subseteq \R^n$ is defined as:
	$$\mu^{*}(A) = \inf_{A} \{\sum^{\infty}_{i = 0}\mu^{*}(I_i) \mid A \subseteq \cup^{\infty}_{i = 0}I_i\}, $$
	where the infinium is taken over all covering of $A$ by countably many intervals $I_i, i = 1,2...$.
	If set $E \subseteq \R^{n}$ satisfies $\mu^{*}(A) = \mu^{*}(A \cap E) + \mu^{*}(A \setminus E)$
	for all $A \subseteq \R^{n}$, then $E$ is lebesgue-measurable and we set $\mu(E) = \mu^{*}(E)$. 
	\bs
\end{dfn}
\noindent
It should be noted that all open sets and closed sets , as well as compact sets are Lebesgue-measurable. 

\begin{lem}[Basic properties of Lebesgue measure, {\cite[Section~2.A]{Jones2000-db}}]
	\label{lem:basic_properties_lebesgue_measure}
	A Lebesgue outer measure $\mu^{*}$ of Definition \ref{dfn:lebesgue_measure} satisfies the following conditions:
	\begin{enumerate}
		\item $\mu^{*}(\emptyset) = 0,$
		\item $\mu^{*}(A) \leq \mu^{*}(B)$ whenever $A \subseteq B \subseteq \R^{n}$ and
		\item $\mu^{*}( \cup^{\infty}_{i = 1} \mu^{*}(A_i)) \leq \sum^{\infty}_{i=1}\mu^{*}(A_i) $.
	\end{enumerate}
\end{lem}

\begin{lem}[Lebesgue measure scale property, {\cite[Section~3.B]{Jones2000-db}}]
	\label{lem:scale_property_lebesgue_measure}
	Let $\mu$ be Lebesgue measure on normed vector space $(\R^{n},d)$, then for any $p \in \R^{n}$ and $t \in \R_{+}$ we have:
	$\mu(\bar{B}(p,t)) =  t^{n} \cdot \mu (\bar{B}(p,1) )$.
\end{lem}

\begin{dfn}[Grid and Cubic region]
	\label{dfn:voronoi_region}
	Let $\R^{n}$ be a normed vector space
	and let $\delta \in \R$. 
	Define $\delta$-grid on $\R^{n}$ as the set
	$G(\delta) = \{ t \cdot \delta  \mid t \in \Z^{n} \}$. For any $p \in \R^{n}$ define its open cubic region $V(p,\delta) \subseteq \R^{n}$ as the set $\{ p + u \mid u \in (-\frac{\delta}{2},\frac{\delta}{2})^n \}$ and closed cubic region $\bar{V}(p,\delta) \subseteq \R^{n}$ as $\{ p + u \mid u \in [-\frac{\delta}{2},\frac{\delta}{2}]^n \}$.
\end{dfn}
\noindent
Note that the union $\cup_{p \in G(\delta)}V(p, \delta)$ covers whole set $\R^{n}$. 

\begin{lem}[Cubic regions are separate]
	\label{lem:open_voronoi_region_separation}
	In conditions of Definition \ref{dfn:voronoi_region} let $p, q \in G(\delta)$ be distinct points. 
	Then their cubic regions are separate i.e. $V(p,\delta) \cap V(q, \delta) = \emptyset$.
\end{lem}
\begin{proof}
	Assume contrary that there exists $a \in V(p,\delta) \cap V(q, \delta)$, then $|a(i) - p(i)| < \frac{\delta}{2}$ and $|a(i) - q(i)| < \frac{\delta}{2}$
	for all $i \in \{1,...,n\}$. Since $p \neq q$, there exists index $j$, such that $p(j) \neq q(j)$. By definition of grid $G(\delta)$ it follows that
	$|p(j) - q(j)| \geq \delta$. However, by triangle inequality we have 
	$$|p(j) - q(j)| \leq |p(j) - a(j) | + |q(j) - a(j) |< \frac{\delta}{2} + \frac{\delta}{2} = \delta,$$
	which is a contradiction. Therefore $V(p,\delta) \cap V(q, \delta) = \emptyset$. 
	
\end{proof}

\begin{lem}
	\label{lem:diam_voronoi_region}
	Let $\R^n$ be a normed vector space of Definition \ref{dfn:normed_rn_space}. 
	Let $\delta \in \R$ and let $G(\delta)$ be a grid of Definition \ref{dfn:voronoi_region}. 
	Let $p \in G(\delta)$ and let $q \in V(p,\delta)$, then $d(p,q) \leq \frac{ n \cdot \delta \cdot \rho}{2}$
\end{lem}
\begin{proof}
	Let $\gamma \in \R$ be such that $q = p + \gamma$. By condition (3) of Definition \ref{dfn:normed_rn_space} we have
	$\Vert \gamma(i) \Vert \leq \Vert e_i \Vert \cdot \frac{\delta }{2} \leq \frac{\delta \cdot \rho}{2}$ for all $i \in \{1,...,n\}$.
	By the definition of norm and triangle inequality we have:
	$$d(p,q) = \Vert p -q \Vert = \Vert\gamma\Vert \leq \sum^n_{i = 1} \Vert \gamma(i) \Vert \leq \frac{n \cdot \delta \cdot \rho}{2}.$$
\end{proof}
\noindent
If $(\R^{n},d)$ is a normed vector space , then its metric $d(\cdot, \cdot)$ satisfies $d(x,y) = \Vert x - y \Vert$.
\begin{lem}[Existence of covering grid ]
	\label{lem:covering_grid_existence}
	Let $R$ be a  finite subset of normed vector space $(\R^{n}, d)$. 
	Then for any $\delta \in \R$ having $\delta < \frac{d_{\min}(R)}{n \cdot \rho}$, then any map
	$f: R \rightarrow G(\delta)$ which maps $p \in R$ to one of its nearest neighbor in $G(\delta)$ is a well-defined injection.
\end{lem}
\begin{proof}
	Let $f$ be an arbitrary map $f: R \rightarrow G(\delta)$ mapping point $p \in R$ to one of its nearest neighbors. 
	This map is clearly well-defined. Let us now show that it is injective. Assume that $x,y \in R$ are such that
	$f(x) = f(y)$. Then by triangle inequality and Lemma \ref{lem:diam_voronoi_region} we have:
	$$d(x,y) \leq d(x,p) + d(p,y) \leq n \cdot \delta \cdot \rho < d_{\min}(R),$$
	it follows that $x = y$. Therefore map $f$ is injective.

	
\end{proof}

\begin{lem}
	\label{lem:covering_grid_inequalities}
	Let $R$ be a finite subset of normed space $(\R^{n},d)$, let $\rho$ be as in Definition \ref{dfn:normed_rn_space} and let $\delta \in \R$ be such that $0 < \delta < \frac{d_{\min}(R)}{n \cdot \rho}$.
	Let $ p \in R$ be arbitrary point and let $t > \frac{n \cdot \delta \cdot \rho}{2}$ be a real number.
	Then there exists a set $U(\delta)$ satisfying $R \subseteq U(\delta)$ and
	$$|G(\delta) \cap \bar{B}(p , t - \frac{n \cdot \delta \cdot \rho}{2}) | \leq |U(\delta) \cap \bar{B}(p, t)| 
	\leq |G(\delta) \cap \bar{B}(p , t + \frac{n \cdot \delta \cdot \rho}{2}) |$$
\end{lem}
\begin{proof}
	Let $f: R \rightarrow G(\delta)$ be an injection from Lemma \ref{lem:covering_grid_existence}, which maps every $q \in R$ to one of its nearest neighbors in $G(\delta)$.
	Define $U(\delta) = (G(\delta) \setminus f(R)) \cup R$.
	Let us first show that $$g:U(\delta) \cap \bar{B}(p, t) \rightarrow G(\delta) \cap \bar{B}(p , t + \frac{n \cdot \delta \cdot \rho}{2}),$$
	defined by $g(q) = f(q)$, if $q \in R$ and $g(q) = q$, if $q \notin R$,  is an injection. Let us show first that the map $g$ is well-defined, if $q \notin R$, the claim is trivial. Let $q \notin R$, then by triangle inequality 
	$d(g(q),p) \leq  d(q,p) + d(g(q), q) \leq t + \frac{n \cdot \delta \cdot \rho}{2}$. Assume now that $g(a) = g(b)$ for some $a, b \in U(\delta) \cap \bar{B}(p,t)$. Let us first show that either $a,b$ both belong to $R$ or neither of $a,b$ belong to $R$. Assume contrary that $a \in R$ and $b \notin R$.
	Since $b \notin R$ we have $b \in G(\delta) \setminus f(R)$. On the other hand since $h(a) = h(b)$ we have $f(a) = b$, therefore $b \in f(R)$, which is a contradiction. If both, $a$ and $b$ belong to $R$ we have $a = b$, similarly if $a ,b \notin R$ we have $a = b$ by injectivity of function $f$.
	Therefore we have now shown that $g$ is well-defined injection. 
	It follows $|U(\delta) \cap \bar{B}(p, t)| \leq  |G(\delta) \cap \bar{B}(p , t + \frac{n \cdot \delta \cdot \rho}{2})|$.
	Let us now show that map 
	$$ h : G(\delta) \cap \bar{B}(p , t - \frac{n \cdot \delta \cdot \rho}{2}) \rightarrow U(\delta) \cap \bar{B}(p,t),$$
	defined by $h(q) = f^{-1}(q)$, if $q \in f(R)$ and $h(q) = q$, if $q \notin f(R)$ is well-defined injection. 
	Let us first show that the map is well-defined. Let $q \in G(\delta) \cap \bar{B}(p , t - \frac{n \cdot \delta \cdot \rho}{2})$, if $q \notin f(R)$ the claim is satisfied trivially. If $q \in f(R)$, then by definition $d(h(q), q) \leq \frac{n \cdot \delta \cdot \rho}{2}$. By using triangle inequality we obtain: $$d(p, h(q)) \leq d(p,q) + d(q,h(q)) \leq t - \frac{n \cdot \delta \cdot \rho}{2} + \frac{n \cdot \delta \cdot \rho}{2} \leq t. $$
	Therefore $h(q) \in U(\delta) \cap \bar{B}(p,t)$. 
	
	\medskip
	
	\noindent
	Let us now show that $h$ is an injection. Let $a,b \in G(\delta) \cap \bar{B}(p , t - \frac{n \cdot \delta}{2}) $ assume that $h(a) = h(b)$, let us show that $a = b$. Let us first show that either $a,b \in f(R)$ or neither of $a,b$ belong to $f(R)$.
	Assume contrary that $a \in f(R)$ and $b \notin f(R)$. Then $h(a) = h(b)$ 
	implies that $f^{-1}(a) = b$. Since $f^{-1}(a) \in R$, we have $b \in R$. Since $b \in G(\delta)$, it follows that $f(b) = b$, which is a contradiction since $b \notin f(R)$. Assume now that $a, b \in f(R)$, then the claim follows by noting that $f^{-1}$ is injection. If $a, b \notin f(R)$ , then claim follows by noting that $h(a) = a$ and $h(b) = b$. Therefore map $h$ is injection. It follows that 
	$|G(\delta) \cap \bar{B}(p , t - \frac{n \cdot \delta \cdot \rho}{2})| \leq |U(\delta) \cap \bar{B}(p,t)|.$

	
\end{proof}

\begin{lem}
	\label{lem:covering_grid_inclusions}
	Let $R$ be a finite subset of normed vector space $\R^{n}$ and let $\delta \in \R$.
	For any $p \in G(\delta)$ recall that $V(p,\delta)$ is Minkowski sum $p + (-\frac{\delta}{2},\frac{\delta}{2})^n$.
	Define $$\bar{W}(p, t, \delta) = \bigcup_{q \in \bar{B}(p,t) \cap G(\delta)} \bar{V}(q,\delta).$$
	Then for any $p \in R$ and $t > \frac{n \cdot \delta \cdot \rho}{2}$ we have:
	$$\bar{B}(p, t - \frac{n \cdot \delta \cdot \rho}{2}) \subseteq \bar{W}(p,t,\delta) \subseteq \bar{B}(p, t + \frac{n \cdot \delta \cdot \rho}{2}).$$
\end{lem}
\begin{proof}
	Let $u \in \bar{B}(p, t-\frac{n \cdot \delta \cdot \rho}{2})$ be an arbitrary point. Since $\{\bar{V}(q,\delta) \mid q \in G(\delta)\}$ covers $R$ it follows that there exists
	$a \in G(\delta)$ such that $u \in \bar{V}(a, \delta)$. By triangle inequality we obtain:
	$$d(a,p) \leq d(a,u) + d(u,p) \leq \frac{n \cdot \delta \cdot \rho}{n} + t - \frac{n \cdot \delta \cdot \rho}{n} \leq t.  $$
	It follows that $\bar{V}(w,\delta) \in \bar{W}(p,t)$, therefore $p \in \bar{W}(p,t,\delta)$. We have $\bar{B}(p, t - \frac{n \cdot \delta \cdot \rho}{2}) \subseteq \bar{W}(p,t,\delta).$
	Let $u \in \bar{W}(p,t,\delta)$, then there exists $a \in G(\delta)$ such that $u \in \bar{V}(a, \delta)$ and $\bar{V}(a, \delta) \in \bar{W}(p,t)$. By triangle inequality we obtain: 
	$$d(u,p) \leq d(u,a) + d(a,p) \leq \frac{n \cdot \delta \cdot \rho}{n} + t . $$
	It follows that $u \in \bar{B}(p, t + \frac{n \cdot \delta \cdot \rho}{2})$. Therefore 
	$\bar{W}(p,t,\delta) \subseteq \bar{B}(p, t + \frac{n \cdot \delta \cdot \rho}{2}) $ which proves the claim.
	
	
\end{proof}

\begin{lem}[Countable additivity, {\cite[Section~2.A]{Jones2000-db}}]
	\label{lem:measure_additivty}
	Assume that $A_i \subseteq \R^{n}$, $i = 1,2,...,$ are pairwise disjoint i.e. $A_i \cap A_j = \emptyset$ for all $i \neq j$ Lebesgue-measurable sets. Then $$\mu(\bigcup^{\infty}_{i = 0}A_i) = \sum^{\infty}_{i = 1}\mu(A_i) .$$
\end{lem}
\begin{lem}[Lebesgue measure of $\bar{W}(p,t,\delta)$]
\label{lem:lebesgue_measure_w}
	In notations of Lemma \ref{lem:covering_grid_inclusions}
	let $\mu$ be a Lebesgue measure on $R$ from Definition \ref{dfn:lebesgue_measure}, then $\mu(\bar{W}(p,t,\delta)) = \delta^n \cdot |\bar{B}(p,t) \cap G(\delta)|$.
\end{lem}
\begin{proof}
	Define $W(p, t, \delta) = \bigcup\limits_{q \in \bar{B}(p,t) \cap G(\delta)} V(q,\delta).$
	Recall that for all $p \in \R^{n}$ and $\delta > 0$ set $\bar{V}(p,t)$ is a closed $n-$dimensional interval and ${V}(p,t)$ is an open $n$-dimensional interval. Therefore we have $\mu(\bar{V}(p,t)) = \mu(V(p,t))$. Since $\bar{V}(p,t)$ is a closed interval, it follows that $\mu(\bar{V}(p,t)) = \delta^{n}$. 
	Since all the sets of $W$ are separate we can use Lemma \ref{lem:measure_additivty} to obtain:
	$$\mu(W(p,t,\delta))  = \sum_{A \in W(p,t)} \mu(A) =  \sum_{A \in \bar{W}(p,t)} \mu(A) = \delta^{n} \cdot |\bar{B}(p,t) \cap G(\delta)|$$
	By Lemma \ref{lem:basic_properties_lebesgue_measure} (2), since $W(p,t,\delta) \subseteq \bar{W}(p,t,\delta)$ we obtain $\mu(\cup\bar{W}(p,t)) \geq \delta^{n} \cdot |\bar{B}(p,t) \cap G(\delta)|$. On the other hand, by Lemma \ref{lem:basic_properties_lebesgue_measure} (3) we obtain 
	$$\mu(\bar{W}(p,t,\delta)) \leq \sum_{A \in \bar{W}(p,t)} \mu(A) = \delta^{n} \cdot |\bar{B}(p,t) \cap G(\delta)| $$
	Therefore we have shown that $\mu(\bar{W}(p,t,\delta)) = \delta^{n} \cdot |\bar{B}(p,t) \cap G(\delta)|$.
\end{proof}

\begin{lem}[Set $U(\delta)$ bounds]
	\label{lem:u_bounds}
	Let $\R^{n}$ be a normed vector space.
	Let $R \subseteq \R^{n}$ be its finite subset.
	Then any set $U(\delta)$ of Lemma \ref{lem:covering_grid_inequalities} satisfies the following inequalities:
	$$\frac{\mu( \bar{B}(p,t - \delta \cdot n \cdot \rho))}{\delta^n} \leq |\bar{B}(p,t) \cap U(\delta)| \leq \frac{\mu( \bar{B}(p,t + \delta \cdot \gamma \cdot n \cdot \rho))}{\delta^n}, $$
	for all $p \in R$ and $t > n \cdot \delta \cdot \rho$.
\end{lem}
\begin{proof}

Let $p \in \R^{n}$ be an arbitrary point and let $t > n \cdot \delta \cdot \rho$ be an arbitrary real number.
By Lemma \ref{lem:covering_grid_inequalities} it follows:
$$|G(\delta) \cap \bar{B}(p , t + \frac{n \cdot \delta \cdot \rho}{2})| \leq |\bar{B}(p,t) \cap U(\delta)| \leq |G(\delta) \cap \bar{B}(p , t + \frac{n \cdot \delta \cdot \rho}{2})|.$$
Let $ \bar{W}(p,t + \frac{n \cdot \delta \cdot \rho}{2},\delta) = \cup_{q} \{ \bar{V}(q,\delta) \mid q \in \bar{B}(p, t + \frac{n \cdot \delta \cdot \rho}{2}) \}.$ By Lemma \ref{lem:covering_grid_inclusions} we have:  
$$
\bar{B}(p, t - n \cdot \delta \cdot \rho) \subseteq \bar{W}(p,t - \frac{n \cdot \delta \cdot \rho}{2},\delta) \text{ and }
\bar{W}(p,t + \frac{n \cdot \delta \cdot \rho}{2},\delta) \subseteq \bar{B}(p, t + n \cdot \delta \cdot \rho)
$$
By Lemma \ref{lem:basic_properties_lebesgue_measure} we have $\mu(\bar{W}(p,t + \frac{n \cdot \delta \cdot \rho}{2},\delta)) \leq \mu(\bar{B}(p, t + n \cdot \delta \cdot \rho))$. By Lemma \ref{lem:lebesgue_measure_w} we have:
$$\mu(\bar{W}(p,t + \frac{n \cdot \delta \cdot \rho}{2},\delta)) = \delta^n \cdot |\bar{B}(p,t + \frac{n \cdot \delta \cdot \rho}{2}) \cap G(\delta)|.$$
By combining the facts we obtain:
$$ |\bar{B}(p,t) \cap U(\delta)| \leq |G(\delta) \cap \bar{B}(p , t + \frac{n \cdot \delta \cdot \rho}{2})| \leq \frac{\mu(\bar{W}(p,t + \frac{n \cdot \delta \cdot \rho}{2},\delta))}{\delta^n} \leq \frac{\mu(\bar{B}(p, t + n \cdot \delta \cdot \rho))}{\delta^n} $$
$$ |\bar{B}(p,t) \cap U(\delta)| \geq |G(\delta) \cap \bar{B}(p , t - \frac{n \cdot \delta \cdot \rho}{2})| \geq \frac{\mu(\bar{W}(p,t - \frac{n \cdot \delta \cdot \rho}{2},\delta))}{\delta^n} \geq \frac{\mu(\bar{B}(p, t - n \cdot \delta \cdot \rho))}{\delta^n} $$

\end{proof}


\begin{lem}[Set $U(\delta)$ is locally finite]
	\label{lem:grid_locally_finite}
	Let $\R^{n}$ be a normed vector space.
	Let $R \subseteq \R^{n}$ be its finite subset
	Then any set $U(\delta)$ of Lemma \ref{lem:covering_grid_inequalities} is locally finite.
\end{lem}
\begin{proof}

With the exact same proof of Lemma \ref{lem:u_bounds} it can be shown that 
$$|\bar{B}(p,t) \cap U(\delta)| \leq \frac{\mu( \bar{B}(p,t + \delta \cdot n \cdot \rho))}{\delta^n}$$ 
is satisfied for all $p \in R$ and $t > 0$.
Therefore $|\bar{B}(p,t) \cap U(\delta)|$ is finite as well.

\end{proof}
\noindent
Recall that \emph{minimized expansion constant} of Definition \ref{dfn:expansion_constant} of a finite subset $R$ of a metric space $(X,d)$ was defined as  $c_m(R) = \inf\limits_{0 < \xi}\inf\limits_{R\subseteq A\subseteq X}\sup\limits_{p \in A,t > \xi}\dfrac{|\bar{B}(p,2t) \cap A|}{|\bar{B}(p,t) \cap A|}$ where $A$ is a locally finite set which covers $R$.

\begin{thmm}[The minimized expansion constant of a finite subset $R$ of $\R^{n}$ is at most $2^{n}$]
	\label{thm:normed_space_exp_constant}
	Let $R$ be a finite subset of a normed Euclidean space $\R^{n}$.
	Let $c_m(R)$ be the minimized expansion constant of Definition~\ref{dfn:expansion_constant}, 
	then $c_m(R) \leq 2^{n}$.
\end{thmm}
\begin{proof}
	Let $0 < \xi < \frac{d_{\min}(R)}{2}$ be an arbitrary real number.  	
	Let $0< \delta < \frac{\xi}{n \cdot \rho}$ be a real number.
	Since $\delta < \frac{d_{\min}(R)}{2 \cdot n \cdot \rho}$ by Lemma \ref{lem:u_bounds} we have: $$\frac{\mu( \bar{B}(p,t - \delta \cdot n \cdot \rho))}{\delta^n} \leq |\bar{B}(p,t) \cap U(\delta)| \leq \frac{\mu( \bar{B}(p,t + \delta \cdot \gamma \cdot n \cdot \rho))}{\delta^n}$$
Note that by Lemma \ref{lem:scale_property_lebesgue_measure} we have: $\mu(\bar{B}(q,y)) =  y^{n} \cdot \mu (\bar{B}(q,1) )$ for any $q \in \R^n$ and $y \in \R_{+}$. Therefore 
	$$ \frac{|\bar{B}(p, 2t) \cap U(\delta)|}{|\bar{B}(p, t) \cap U(\delta)|} \leq \frac{\mu(\bar{B}(p, 2t + n\delta\rho)) \cdot \delta^2}{\mu(\bar{B}(p, t - n\delta\rho)) \cdot \delta^2} = \frac{ (2t + n\delta\rho)^n \cdot \mu(\bar{B}(p,1)) }{ (t - n\delta\rho)^n \cdot \mu(\bar{B}(p,1)) } 
	= \frac{(2t + n\delta\rho)^n}{(t - n\delta\rho)^n},$$
	is satisfied for for all $t > \xi$.
	Since $0 < \xi < \frac{d_{\min}(R)}{2}$ was chosen arbitrarily, we conclude that:  $$c_m(R) =  \inf\limits_{0 < \xi}\inf\limits_{R\subseteq A\subseteq X}\sup\limits_{p \in A,t > \xi}\dfrac{|\bar{B}(p,2t)| \cap A}{|\bar{B}(p,t)| \cap A} \leq \lim_{\delta \rightarrow 0} \frac{|\bar{B}(p, 2t) \cap U(\delta)|}{|\bar{B}(p, t) \cap U(\delta)|} = \lim_{\delta \rightarrow 0} \frac{(2t + \delta \cdot n\rho)^n}{(t - \delta \cdot n\rho)^n} = \frac{2^n \cdot t^n}{t^n} = 2^{n}.$$

\end{proof}

\section{Counterexamples for ICML 2006 complexity analysis}
\label{app:challenge_trees}

This section show that both proofs of the main theorems \cite[Theorem~5]{beygelzimer2006cover} and \cite[Theorem~6]{beygelzimer2006cover} are incorrect. 
Since the same false arguments were later repeated in the papers \cite{march2010fast} and \cite{ram2009linear} 
we provide a detailed counterexamples that expose the contradiction within each of the proof of the theorems. 
\medskip

\noindent
Counterexample~\ref{cexa:construction_algorithm_of_original_cover_tree} shows that the proof of worst-case time complexity theorem of the Insert algorithm for an implicit cover tree \cite[Theorem~6]{beygelzimer2006cover} is incorrect.
\medskip

\noindent
Counterexample \ref{cexa:original_all_nearest_neighbors_algorithm} shows that the proof of  \cite[Theorem~5]{beygelzimer2006cover}, which gives an upper bound for the complexity of Algorithm \ref{alg:cover_tree_k-nearest_original} is incorrect as well. 
Both counterexamples are based on Example \ref{exa:tall_imbalanced_tree}, which extends Example \ref{exa:implicitexplicitexample}.

\medskip

\noindent
Although given proofs are incorrect, it is not known if claimed worst-case time complexity estimates are true for \cite[Theorem~5]{beygelzimer2006cover} and \cite[Theorem~6]{beygelzimer2006cover} which claim a parametrized worst-case near-linear time complexities for algorithms \cite[Algorithms~1]{beygelzimer2006cover} and \cite[Algorithms~2]{beygelzimer2006cover}, respectively. The issues are fixed by using similar techniques, but by taking more rigorous approach to the problem. New approach with new worst-case time complexity theorems for cover tree construction is explained in Section~\ref{sec:ConstructionCovertree} and for $k$-nearest neighbor search in Section~\ref{sec:better_approach_knn_problem}.

\begin{figure}
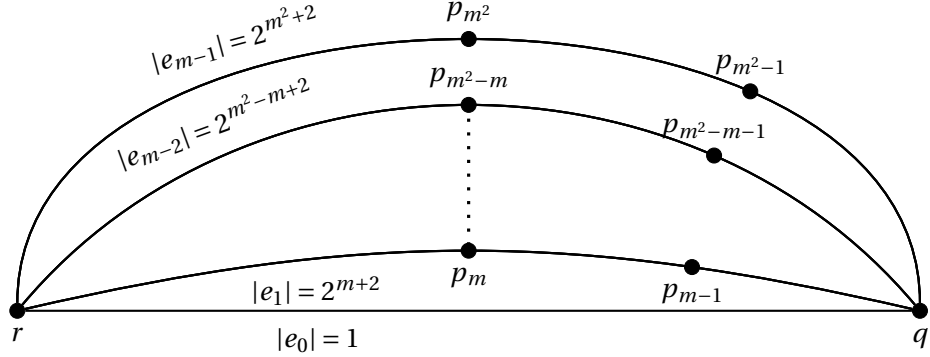

	\centering
	
	\input images/MultiGraphExampleExtended.tex
	
	\caption{Illustration of a graph $G$ and a point cloud $R$ defined in Example \ref{exa:tall_imbalanced_tree}}
	\label{fig:GraphConstructionOfExample}
\end{figure}

\begin{exa}[tall imbalanced tree]
	\label{exa:tall_imbalanced_tree}
	Let $(R,d)$ be a finite metric space. 
	Consider a tree in Figure \ref{fig:bad_cover_tree} containing $\sqrt{|R|}$ branches, each having a size of $\sqrt{|R|}$. 
	Simulation of Algorithms~\ref{alg:cover_tree_construction_original} and Algorithm~\ref{alg:cover_tree_k-nearest_original} will show that both algorithms iterate over all nodes of $R$. 
	Original  \cite[Theorem~5]{beygelzimer2006cover} and \cite[Theorem~6]{beygelzimer2006cover} claimed that the number of iterations is proportional to the length of a longest root-to-leaf path that has a size of $\sqrt{|R|}$, which will contradict our findings. 
	\medskip

    \noindent 	
	For any integer $m > 10$, let $G$ be a metric graph pictured in Figure \ref{fig:GraphConstructionOfExample} that has two vertices $r,q$ and $m+1$ edges $(e_i)$ for $i \in \{0, ..., m\}$, and the length of each edge $e_i$ is $|e_i| = 2^{m \cdot i +2}$ for $i \geq 1$.
	Finally, set $|e_0| = 1$. For every $i \in \{1, ..., m^2\}$ if $i$ is divisible by $m$ we set $p_{i}$ be the middle point of $e_{i / m}$ and for every other $i$ we define $p_i$ to be the middle point of segment $(p_{i+1}, q)$. 
	Let $d$ be the induced shortest path metric on the continuous graph $G$.
	Then $d(q,r) = 1$, $d(r, p_{i}) = 2^{i+1} + 1$, $d(q,p_{i}) = 2^{i} $.
	If $i > j $ and $\ceil{\frac{i}{m}} = \ceil{\frac{j}{m}}$, then $ d(p_{j}, p_{i})  = \sum\limits_{t = j+1}^{i} 2^{t}$. 
	We consider the reference set $R = \{r\} \cup \{p_{i} \mid i=1,2,3,...,m^2 \}$ with the metric $d$.
	\medskip

    \noindent 	
	Let us define a compressed cover tree $\T(R)$ by setting $r$ to be the root node and $l(p_{i}) = i$ for all $i$. 
	If $i$ is divisible by $m$, we set  $r$ to be the parent of $p_{i}$. 
	If $i$ is not divisible by $m$, we set $p_{i+1}$ to be the parent of $p_{i}$. 
	For every $i$ divisible by $m$, the point $p_i$ is in the middle of edge $e_{i / m}$, hence $d(p_{i}, r) \leq 2^{i + 1}$.
	For every $i$ not divisible by $m$, by definition, $p_i$ is the middle point of $(p_{i+1},q)$.
	Therefore, we have $d(p_i, p_{i+1}) \leq 2^{i+1}$. 
	Since for any point $p_i$ distance to its parent is at most $2^{i+1}$, the tree $\T(R)$ satisfies covering condition (\ref{dfn:cover_tree_compressed}b).
	For any integer $t$, the cover set is $C_t = \{r\} \cup \{ p_i \mid i \geq t\}$. 
	We will prove that $C_t$ satisfies (\ref{dfn:cover_tree_compressed}c). 
	Let $p_{i} \in C_t$.
	If $i$ is divisible by $m$, then 
	$d(r, p_i) = 2^{i+1} \geq 2^{t + 1} > 2^t.$ 
	If $i$ is not divisible by $m$, then
	$d(r, p_{i}) = d(r,q) + d(q,p_{i})  = 1 + 2^{i+1}  > 2^{t}$. 
	\medskip
	
	\noindent 
	Then the root $r$ is separated from the other points by the distance $2^t$. 
	Consider arbitrary points $p_{i}$ and $p_{j}$ with indices $i > j \geq t $ and $\ceil{\frac{i}{m}} = \ceil{\frac{j}{m}}$.
	Then
	$$d(p_{i}, p_{j}) = \sum^{i}_{s = j+1} 2^{s}  \geq 2^{j+1} \geq 2^{t+1} > 2^{t}.$$
	On the other hand, if $i > j \geq t $ and $\ceil{\frac{i}{m}} \neq \ceil{\frac{j}{m}}$, then
	$$d(p_{i} , p_{j}) = d(p_{i},q) + d(p_{j} ,q)  \geq  2^{i} + 2^{j} \geq 2^{j+1} \geq 2^{t+1} > 2^t.$$
	For any $t$, we have shown that all pairwise combinations of points of $C_{t}$ satisfy condition~(\ref{dfn:cover_tree_compressed}c).
	Hence this condition holds for the whole tree $\T(R)$.
	\bs
\end{exa}

\noindent 
Let us now define the concept of explicit depth, that corresponds to maximal root-to-node path of any cover tree. 
Recall that in \cite[Section~2]{beygelzimer2006cover} the explicit representation of cover tree was defined as 
"the explicit representation of the tree coalesces all nodes in which the only child is a self-child". Simplest way to interpret this is to consider cover sets $C_i$ and define $p \in C_i$ to be an explicit node, if $p$ has child at level $i-1$.  By  \cite[Lemma~4.3]{beygelzimer2006cover} depth of any node $p$ is "defined as the number of explicit grandparent nodes on the path from the root
to p in the lowest level in which $p$ is explicit". Explicit depth of a node $p$ in any compressed tree $\T$ will be defined in Definition \ref{dfn:explicit_depth_for_compressed_cover_tree} using the simplest interpretation of the aforementioned quotes.

\begin{lem}[Explicit depth for compressed cover tree]
	\label{dfn:explicit_depth_for_compressed_cover_tree}
	Let $R$ be a finite subset of a metric space with a metric $d$.
	Let $\T(R)$ be a compressed cover tree on $R$. 
	For any $p \in \T(R)$, let $s = (w_0, ... , w_m)$ be a node-to-root path of $p$. 
	Then the explicit depth $D(p)$ of node $p$ belonging to compressed cover tree can be interpreted as the sum 
	$$D(p) = \sum^{m-1}_{i = 0}| \{q \in \Child(w_{i+1}) \mid l(q) \in [l(w_i), l(w_{i+1}) - 1] \} |. $$
	\bs
\end{lem}
\begin{proof}
	Note that the node-to-root path of an implicit cover tree on $R$ has $l(w_{j+1}) - l(w_{j}) - 1$ extra copies of $w_{j+1}$ between every $w_{j}$ and $w_{j+1}$ for any index $j \in [0,m-1]$. Recall that a node is called explicit, if it has non-trivial children. 
	Therefore there will be exactly $$| \{q \in \Child(w_{i+1}) \mid l(q) \in [l(w_i), l(w_{i+1}) - 1] \} | $$ 
	explicit nodes between $w_{j}$ and $w_{j+1}$.
	It remains to take the total sum.
\end{proof}

\noindent 
Lemma \ref{lem:tall_imbalanced_tree_explicit_depth} shows that the cover tree of Example \ref{exa:tall_imbalanced_tree} the maximal node-to-root has $2 \cdot \sqrt{|R|}$ size, where $|R|$ is the size of dataset. 

\begin{lem}
	\label{lem:tall_imbalanced_tree_explicit_depth}
	Let $\T(R)$ be a compressed cover tree on the set $R$ from Example \ref{exa:tall_imbalanced_tree} for some $m \in \Z$. 
	For any $p \in R$, the explicit depth $D(p)$ of Definition~\ref{dfn:explicit_depth_for_compressed_cover_tree} has the upper bound $2m+1$. 
	\bs
\end{lem}
\begin{proof}
	For any $p_i$, if $i$ is divisible by $m$, then $r$ is the parent of $p_i$. By definition, the explicit depth is $D(p_i) = |\{p \in \Child(r) \mid l(p) \in [l(p_i), m^2] \} |$. 
	Since $r$ contains children on every level $j$, where $j$ is divisible by $m$, we have $D(p_i) =  m - \frac{i}{m} + 1$. 
	\medskip
	
	\noindent 
	Let us now consider an index $i$ that is not divisible by $m$. 
	Note that  $p_{j+1}$ is the parent of $p_j$ for all $j \in [i, m \cdot \ceil{i / m} - 1]$.  
	Then the path consisting of all ancestors of $p_i$ from $p_i$ to the root node $r$ has the form $(p_{i}, p_{i+1}, ..., p_{m \cdot \ceil{i / m}} , r)$. It follows that
	$$D(p_i) = \sum^{m \cdot \ceil{i / m} -1}_{j = i}| \{p \in \Child(p_j) \mid l(p) \in [l(p_{j}), l(p_{j+1}) - 1] \}| + D(p_{m \cdot \ceil{i / m}}).
	$$
	Since $i \geq m \cdot (\ceil{i / m} -1) + 1 $ and $\ceil{i / m} \geq 1$, we get the required upper bound: $$D(p_i) =  (m \cdot  \ceil{i / m} - i ) + (m - \ceil{i / m} + 1) \leq m + (m+1) = 2m + 1.$$ 
\end{proof}


\begin{figure}
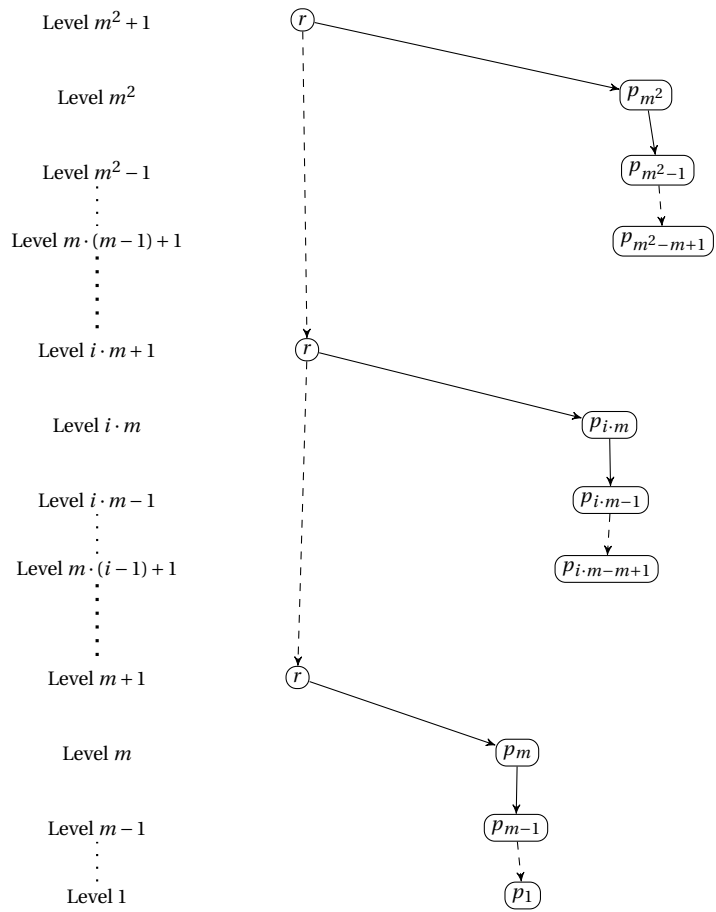

	\centering
	
	\input images/bad_tree_example_new_expanded.tex
	
	\caption{Illustration of explicit cover tree built on the dataset of Example \ref{exa:tall_imbalanced_tree}}
	\label{fig:bad_cover_tree}
	
\end{figure}

\begin{algorithm}
	\caption{Copy-pasted Insert() algorithm for inserting point $p$ into an implicit cover tree $T$ \cite[Algorithm~2]{beygelzimer2006cover}. This algorithm is launched with $i = l_{\max}$ and $Q_{i} = \{r\}$, where $r$ is the root node of $T$. }
	\label{alg:cover_tree_construction_original}
	\begin{algorithmic}[1]
		\STATE \textbf{Insert}(point $p$, cover set $Q_i$, level $i$)
		\STATE Set $Q = \{\Child(q) \mid q \in Q_i\}$
		\IF {$d(p,Q) > 2^{i}$}
		\STATE \textbf{return} "no parent found"
		\ELSE
		\STATE Set $Q_{i-1} = \{q \in Q \mid d(p,q) \leq 2^{i}\}$
		\IF{\textbf{Insert}$(p,Q_{i-1}, i-1)$ = "no parent found" and $d(p,Q_{i}) \leq 2^{i}$} 
		\STATE Pick $q \in Q_i$ satisfying $d(p,q) \leq 2^{i}$ and insert $p$ into $\Child(q)$, \textbf{return} "parent found"
		\ELSE
		\STATE \textbf{return} "no parent found"
		\ENDIF
		
		\ENDIF 
		
	\end{algorithmic}
\end{algorithm}

\begin{cexa}[Counterexample for a step in the proof of {\cite[Theorem~6]{beygelzimer2006cover}}]
	\label{cexa:construction_algorithm_of_original_cover_tree}
	
	Counterexample \ref{cexa:construction_algorithm_of_original_cover_tree} shows that there is a mistake in proof of \cite[Theorem~6]{beygelzimer2006cover}. The idea is based on adding a new point $q$ of Figure \ref{fig:GraphConstructionOfExample} to the tree $\T(R)$ of Example \ref{exa:tall_imbalanced_tree} that lures the Algorithm \ref{alg:cover_tree_construction_original} into using all branches of $\T(R)$. It follows that the Algorithm \ref{alg:cover_tree_construction_original} is launched $O(|R|)$ times. However, in the proof of \cite[Theorem~6]{beygelzimer2006cover} it was claimed that Algorithm \ref{alg:cover_tree_construction_original} is
	launched at most $4 \cdot D(\T(R))$, where $D(\T(R))$ is the explicit depth of explicit cover tree $\T(R)$.
	This is a contradiction, since $D(\T(R)) \leq 2\sqrt(|R|) + 1$.
	
	\medskip
    \noindent
	For more detailed exhibition let us first cite a part of the proof of \cite[Theorem~6]{beygelzimer2006cover}:
	
	\medskip
	
	\noindent
	"\textbf{Theorem 6} Any insertion or removal takes time at most $O(c^6\log(n))$"
	[In other words the run time of Algorithm \ref{alg:cover_tree_construction_original} is $O(c^6\log(n))$, where $n$ is the number points of original dataset $S$ on which tree $T$ was constructed.]
	
	\medskip
	
	\noindent 
	[\emph{Part of the proof} ]: \emph{" Let $k = c^2 \log(|S|)$ be the maximum explicit depth
	of any point, given by Lemma 4.3. Then the total
	number of cover sets with explicit nodes is at most
	$3k+k = 4 k$, where the first term follows from the fact
	that any node that is not removed must be explicit at
	least once every three iterations, and the additional
	$k$ accounts for a single point that may be implicit for
	many iterations.
	Thus the total amount of work in Steps 1 [Our line 2] and 2 [Our lines 3-5] is
	proportional to $O(k \cdot \max_i|Q_i|)$. Step 3 [Our lines 5-11] requires work
	no greater than step 1 [Our line 2]."}
	\medskip
	
	\noindent 
	\textbf{In our interpretation}: the above arguments says that the total number of times line $1$ [our line 2] was called during the algorithm has the upper bound $4 \cdot \max_{p \in R}D(p)$ , where $D(p)$ is the explicit depth of a point $p$, see Definition \ref{dfn:explicit_depth_for_compressed_cover_tree}. In this Counterexample we will show that $\T(R)$ from Example \ref{exa:tall_imbalanced_tree} does not satisfy the claimed inequality.

\medskip	
	
	\noindent
	Take the reference set $R$, the compressed cover tree $\T(R)$ and the point $q$ from Example \ref{exa:tall_imbalanced_tree} for any parameter $m > 200$. 
	Assume that we have already constructed tree $\T(R)$. 
	Let us show that $\T(R \cup q)$ constructed by Algorithm~\ref{alg:cover_tree_construction_original} from the input $q, i = m^2+1, Q_i = \{r\}$ runs at least $m^2-2$ self-recursions. 
	This will lead to a contradiction since by Lemma \ref{lem:tall_imbalanced_tree_explicit_depth} any node $ p\in \T(R)$ has $D(p) \leq 2m + 1$.

    \smallskip
	
	\noindent
	We show by induction on $m$ going down that, for every step $i \in [1, m^2]$, we have $Q_i = \{r,p_i\}$. 
	The proof for the base case $i = m^2$ is similar to the induction step and thus will be omitted. 
	Assume that $Q_{i}$ has the desired form for some $i$.
	Let us show that the claim holds for $i-1$. 
	For all levels $i-1$ divisible by $m$, the node $p_{i-1}$ is a child of the root $r$. 
	For all levels $i-1$ not divisible by $m$, the node $p_{i}$ is a child of $p_i$. 
	Since $\T(R)$ contains exactly one node at each level, in both cases we have $Q = \{r, p_{i}, p_{i-1}\}$. 
	Since $d(q,r) = 1$, $d(q,p_{i}) = 2^{i+1}$ and $d(q,p_{i-1}) = 2^{i}$  we have $$Q_{i-1} = \{p \in Q_i \mid d(p,q) \leq 2^i\} = \{r,p_i\}.$$
	
	\smallskip
	
	\noindent
	The actual implementation of Algorithm~\ref{alg:cover_tree_construction_original} iterates over all levels $i$ for which there exists a node in $Q_i$  that contains at least one non-trivial child on level $i-1$ and for which the condition in line $7$ is satisfied. Since for every index $i \in [2,m^2+1]$ we have $Q_i = \{r,p_i\}$ and since either $r$ or $p_{i}$ has a child at level $i-1$ and the condition in line $7$ is always satisfied, it follows that $m^2-2$ is a low bound for the number $\xi$ of self-recursions. 
	Therefore the contradiction follows from the inequality:
	$$m^2-2 \leq  \xi \leq 4 \cdot \max_{p \in R}D(p) \leq 8 \cdot (2m + 1) \leq 16 \cdot m + 8$$
	where $m > 20$. 
	\bs

\end{cexa}

\begin{algorithm}
	\caption{Copy-pasted \cite[Algorithm~1]{beygelzimer2006cover} based on an implicit cover tree $T$ \cite[Section ~ 2]{beygelzimer2006cover} for nearest neighborhood search, which is used in Counterexample \ref{cexa:original_all_nearest_neighbors_algorithm}. The children of a node $q$ of an implicit cover tree are defined as the nodes at one level below $q$ that have $q$ as their parent. In the actual implementation the loop in lines 3-6 runs only for the levels containing nodes with non-trivial children (not coinciding with their parents). 
		The level $+\infty$ can be replaced by $l_{\max}(T)$ and $-\infty$ can be replaced by $l_{\min}(T)$ in the code.
	}
	\label{alg:cover_tree_k-nearest_original}
	\begin{algorithmic}[1]
		\STATE \textbf{Input} : implicit cover tree $T$, a query point $p$
		\STATE Set $Q_{\infty} = C_{\infty}$ where $C_{\infty}$ is the root level of $T$
		\FOR{$i$ from $\infty$ down to $-\infty$}
		\STATE Set $Q = \{\text{Children}(q) \mid q \in Q_i\}.$
		\STATE Form cover set $Q_{i-1} = \{q \in Q \mid d(p,q) \leq d(p,Q) + 2^{i}\}$
		\ENDFOR
		\STATE \textbf{return} $\text{argmin}_{q \in Q_{-\infty}}d(p,q)$
	\end{algorithmic}
\end{algorithm}

\begin{cexa}[for a step in the proof of {\cite[Theorem~5]{beygelzimer2006cover}}]
	\label{cexa:original_all_nearest_neighbors_algorithm}
	
	Counterexample \ref{cexa:original_all_nearest_neighbors_algorithm} shows that there is a mistake in proof of \cite[Theorem~6]{beygelzimer2006cover}. The counterexample is obtained by running Algorithm \ref{alg:cover_tree_k-nearest_original} for node $q$ of Figure \ref{fig:GraphConstructionOfExample} and tree $\T(R)$ of Example \ref{exa:tall_imbalanced_tree}. 
    It it shown that Algorithm \ref{alg:cover_tree_k-nearest_original} iterates over all branches of $\T(R)$, therefore  lines 3-6 are considered exactly $|R|$ times. However, the proof of \cite[Theorem~6]{beygelzimer2006cover} claimed that the number of times lines 3-6 are considered is bounded by multiplication $\max_i|R_i| \cdot D(\T(R))$, where $D(\T(R))$ is the maximal path-to-root path that has an upper bound $2\sqrt{|R|}$. In this counterexample it will be also shown that $\max_i|R_i| \leq 3$ during the whole iteration of the algorithm, which will lead to contradiction $|R| \leq 3 \cdot 2\sqrt{|R|}$, when $|R|$ is sufficiently big. 
    
    \medskip
    \noindent 
   	For more detailed exhibition let us first cite a part of the proof of  \cite[Theorem~5]{beygelzimer2006cover}:
    
	\noindent
	\emph{"\textbf{Theorem 5}
	If the dataset $S \cup \{p\}$ has expansion constant $c$, the nearest neighbor of $p$ can be found in time $O(c^{12}\log(n))$."}
	
	\medskip
	
    \noindent	
	\emph{[\emph{Partial proof:}] "Let $Q^{*}$ be the last $Q$ considered by the Algorithm \ref{alg:cover_tree_k-nearest_original} (so $Q^{*}$ consists only of lead nodes with scale $-\infty$). Lemma 4.3 bounds the explicit depth of any node in the tree (and in particular any node in $Q^{*}$) by $k = O(c^2 \log (N))$. Consequently the number of iterations is at most $k|Q^{*}| \leq k \max_i|Q_i|$."}
	
	\medskip
	
	\noindent 
	\textbf{By our interpretation:} the above argument claims that the total number $\xi$ of times when Algorithm \ref{alg:cover_tree_k-nearest_original} runs lines 3-6 has an upper bound $\xi \leq \max_{p \in R}D(p) \cdot \max_i|Q_i|.$ Contradiction will be obtained by showing that $\T(R)$ from Example \ref{exa:tall_imbalanced_tree} does not satisfy this inequality.

\medskip	
	
	\noindent
	Take $R,\T(R)$ and $q$ from Example \ref{exa:tall_imbalanced_tree}. We will apply Algorithm \ref{alg:cover_tree_k-nearest_original} to the tree $\T(R)$ and query point $q$. 
	By Lemma \ref{lem:tall_imbalanced_tree_explicit_depth} the cover tree $\T(R)$ having parameter $m$ has $D(p) \leq 2m+1$ for all $p \in R$. A contradiction to the original argument will follow after showing that $\max|Q_i| \leq 2$ and $\xi \geq m^2 - 2$.
	
	\medskip

\noindent	
	Let us first estimate $\max_i |Q_i|$.  
	Similarly to Counterexample \ref{cexa:construction_algorithm_of_original_cover_tree} we will show that, for every iteration (lines 3-5) $i \in [1, m^2]$ of Algorithm \ref{alg:cover_tree_k-nearest_original}, we have 
	$Q_i = \{r,p_i\}$. 
	The proof for the basecase $i = m^2$ is similar to the induction step and thus will be omitted. 
	Assume that $Q_{i}$ has the desired form for some $i$.
	Let us show that the claim holds for $i-1$. 
	For all levels $i-1$ divisible by $m$, the node $p_{i-1}$ is a child of the root $r$. 
	For all levels $i-1$ not divisible by $m$, the node $p_{i-1}$ is a child of $p_{i}$. 
	Since $\T(R)$ contains exactly one node at each level, in both cases
	we have $Q = \{r, p_{i}. p_{i-1}\}$. 
	Since $d(q,r) = 1$, $d(q,p_{i}) = 2^{i+1}$ and $d(q,p_{i-1}) = 2^{i}$, we have 
	$$Q_{i-1} = \{p \in Q_t \mid d(p,q) \leq 2^i + 1\} = \{r,p_i\}$$
	Therefore it follows that $|Q_i| \leq 2$ for all $i \in [1, m^2]$.
	
	\medskip

    \noindent	
	The actual implementation of algorithm \ref{alg:cover_tree_k-nearest_original}iterates over all levels $i$ for which there exists a node in $Q_i$ containing at least one non-trivial child at level $i-1$. 
	Since $Q_i = \{r,p_i\}$ and for every index $i \in [2,m^2+1]$, either $r$ or $p_{i}$ has a child on level $i-1$, it follows that  $m^2-2$ is a low bound for the number $\xi$  of iterations. 
	A contradiction follows from  
	$$m^2 - 2 \leq \xi \leq  \max_{p \in R}D(p) \cdot  \max_i|Q_i| \leq (2m+1) \cdot 2 \leq 4m+2 \text{  for any }m > 20.$$ 
	
	\bs
\end{cexa}


\section{Distinctive descendant sets}
\label{sec:distinctive_descendant_set}

This section introduces auxiliary concepts for future proofs. 
The main concept is a distinctive descendant set in Definition \ref{dfn:distinctive_descendant_set}. 
The distinctive descendant set at a level $i$ of a node $p \in \T(R)$ in a compressed cover tree corresponds to the set of descendants of a copy of node $p$ at level $i$ in the original implicit cover tree $T(R)$.
Other important concepts are $\lambda$-point of Definition \ref{dfn:lambda-point} that is used in Algorithm \ref{alg:cover_tree_k-nearest} as an approximation for $k$-nearest neighboring point. The $\beta$-point property of $\lambda$-point defined in Lemma \ref{lem:beta_point} plays a major role in the proof of the main worst-case time complexity result Theorem \ref{thm:knn_KR_time}.

\begin{figure}
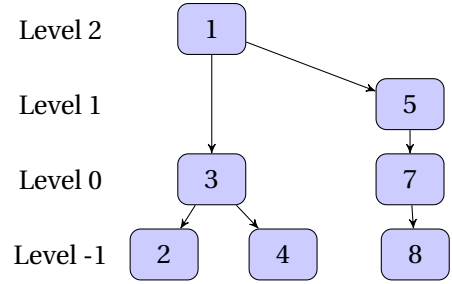

	\centering
	\input images/explicit_cover_tree_extension.tex 
	\caption{Consider a compressed cover tree $\T(R)$ that was built on set $R = \{1,2,3,4,5,7,8\}$. Let $\Sd_i(p, \T(R))$ be a distinctive descendant set of Definition \ref{dfn:distinctive_descendant_set}. Then $V_2(1) = \emptyset, V_{1}(1) = \{5\}$ and $V_{0}(1) = \{3,5,7\}$.
		And also $\Sd_2(1, \T(R)) = \{1, 2,3,4,5,7,8\}$, $\Sd_1(1, \T(R)) = \{1,2,3,4\} $ and $\Sd_{0}(1, \T(R)) = \{1\} $.}
	\label{fig:uniqueDescendant}
\end{figure}

\begin{dfn}[Distinctive descendant sets]
	\label{dfn:distinctive_descendant_set}
	Let $R\subseteq X$ be a finite reference set with a cover tree $\T(R)$. 
	For any node $p \in \T(R)$ in a compressed cove tree on a finite set $R$ and $i \leq l(p) - 1$, set
	$V_{i}(p) =  \{u \in \Desc(p) \mid i \leq l(u)\leq l(p) - 1\}.$
	If $i \geq l(p)$, then set $V_i(p) = \emptyset$. 
	For any level $i \leq l(p) $, the \emph{distinctive descendant set} is
	$\Sd_i(p, \T(R)) =   \Desc(p) \setminus \bigcup\limits_{u \in V_{i}(p)} \Desc(u)$ and has the size $|\Sd_i(p, \T(R)) |$. 
	\bs
\end{dfn}
\begin{lem}[Distinctive descendant set inclusion property]
	\label{lem:distinctive_descendant_inclusion}
	In conditions of Definition \ref{dfn:distinctive_descendant_set} let $p \in R$ and let $i,j$ be integers satisfying
	$l_{\min}(\T(R)) \leq i \leq j \leq l(p) - 1$.
	Then $\Sd_i(p, \T(R)) \subseteq \Sd_j(p, \T(R))$.
\end{lem}


Essential levels of a node $p \in \T(R)$ have 1-1 correspondence to the set consisting of all nodes containing $p$ in the explicit representation of cover tree in \cite{beygelzimer2006cover}, see Figure \ref{fig:tripleexample} middle. 

\begin{dfn}[Essential levels of a node]
	\label{dfn:essential_levels_node}
	Let $R\subseteq X$ be a finite reference set with a cover tree $\T(R)$. Let $q \in \T(R)$ be a node. Let $(t_i)$ for 
	$i \in \{0,1,...,n\}$ be a sequence of $H(\T(R))$ in such a way that $t_0 = l(q)$, $t_n = l_{\min}(\T(R))$ and for all $i$ we have $t_{i+1} = \nxt(q, t_{i}, \T(R))$. Define the set of essential indices $\Es(q,\T(R)) = \{ t_{i} \mid i \in \{0,...,n\} \}$.
	\bs
\end{dfn}

\begin{lem}[Number of essential levels]
	\label{lem:number_of_explicit_levels}
	Let $R\subseteq X$ be a finite reference set with a cover tree $\T(R)$. Then 
	$\sum_{p \in R}|\mathcal{E}(p,\T(R))| \leq 2 \cdot |R|,$
	where $\mathcal{E}(p,\T(R))$ appears in Definition \ref{dfn:essential_levels_node}. \bs
\end{lem}
\begin{proof}
	Let us prove this claim by induction on size $|R|$. In basecase $R = \{r\}$ and therefore $|\mathcal{E}(r,\T(R))| = 1$.
	Assume now that the claim holds for any tree $\T(R)$, where $|R| = m$ and let us prove that if we add any node $v \in X \setminus R$ to tree $\T(R)$, then $\sum_{p \in R}|\mathcal{E}(p,\T(R \cup \{v\}))| \leq 2 \cdot |R| + 2$. Assume that we have added $u$ to $\T(R)$, in such a way that $v$ is its new parent. Then $|\mathcal{E}(p,\T(R \cup \{v\}))| = |\mathcal{E}(p,\T(R))| + 1 $ and $|\mathcal{E}(v,\T(R \cup \{v\}))| = 1$. We have:
	$$ \sum_{p \in R \cup \{u\}}|\mathcal{E}(p,\T(R))| = \sum_{p \in R}|\mathcal{E}(p,\T(R))| + 1 + |\mathcal{E}(v,\T(R \cup \{v\}))| \leq 2\cdot |R| + 2 \leq 2(|R \cup \{v\}|) $$
	which completes the induction step. 
\end{proof}

\begin{algorithm}
	\caption{This algorithm returns sizes of distinctive descendant set $\Sd_i(p, \T(R))$ for all essential levels $i \in \Es(p,\T(R))$}
	\label{alg:cover_tree_distinctive_descendants}
	\begin{algorithmic}[1]
		\STATE \textbf{Function} : CountDistinctiveDescendants(Node $p$, a level $i$ of $\T(R)$)
		\STATE \textbf{Output} : an integer 
		
		\IF{$i > l_{\min}(\T(Q))$}
		\FOR {$q \in \text{Children}(p)$ having $l(p) = i-1$ or $q = p$}
		\STATE Set $s = 0$
		\STATE $j \leftarrow 1 + \nxt(q, i-1,\T(R))$
		\STATE $s \leftarrow s + $CountDistinctiveDescendants($q$, $j$)
		\ENDFOR
		\ELSE
		\STATE Set $s = 1$
		\ENDIF
		\STATE Set $|\Sd_i(p)| = s$ and \textbf{return} s
	\end{algorithmic}
\end{algorithm}

\begin{lem}
	\label{lem:distinctive_descendants_precompute}
	Let $R$ be a finite subset of a metric space.
	Let $\T(R)$ be a compressed cover tree on $R$. 
	Then, Algorithm~\ref{alg:cover_tree_distinctive_descendants} computes the sizes $|\Sd_{i}(p, \T(R))|$ for all $p \in R$ and essential levels $i \in \Es(p,\T(R))$ in time $O(|R|)$.
	\bs
\end{lem}
\begin{proof}
	By Lemma \ref{lem:number_of_explicit_levels} we have $\sum_{p \in R}|\mathcal{E}(p,\T(R))| \leq 2 \cdot |R|.$ Since CountDistinctiveDescendants is called once for every any combination $p \in R$ and $i \in \mathcal{E}(p,\T(R))$ it follows that the time complexity of Algorithm~\ref{alg:cover_tree_distinctive_descendants} is $O(R)$.
	

\end{proof}



Recall that the neighborhood $N(q;r) = \{p \in C \mid d(q,p) \leq d(q,r)\}$ was introduced in Definition~\ref{dfn:kNearestNeighbor}.
\begin{dfn}[$\la$-point]
	\label{dfn:lambda-point}
	Fix a query point $q$ in a metric space $(X,d)$ and fix any level $i \in \Z$. 
	Let $\T(R)$ be its compressed cover tree on a finite reference set $R \subseteq X$. 
	Let $C$ be a subset of a cover set $C_i$ from Definition~\ref{dfn:cover_tree_compressed} satisfying $\sum_{p \in C}|\Sd_i(p, \T(R))| \geq k$, where $\Sd_i(p, \T(R))$ is the distinctive descendant set from Definition \ref{dfn:distinctive_descendant_set}.
	For any $k\geq 1$, define $\la_k(q,C)$ as a point $\la\in C$ that minimizes $d(q,\la)$ subject to $\sum_{p \in N(q;\la)}|\Sd_i(p, \T(R)) |\geq k$. 
	\bs
\end{dfn}

\begin{algorithm}
	\caption{Finding $k$-lowest element of a finite subset $A \subseteq R$ with priority function $f: A \rightarrow \R$}
	\label{alg:k_smallest_elements}
	\begin{algorithmic}[1]
		\STATE \textbf{Input:} Ordered subset $A \subseteq R$, priority function $f:A \rightarrow \R$, an integer $k \in \Z$
		\STATE Initialize an empty max-binary heap $B$ and an empty array $D$ on points $A$.
		\FOR{$p \in A$}
		\STATE add $p$ to $B$ with priority $f(p)$
		\IF{$|H| \geq k$} 
		\STATE remove the point with a maximal value from $B$
		\ENDIF
		\ENDFOR
		\STATE Transfer points from the binary heap $B$ to the array $D$ in reverse order. 
		\STATE \textbf{return} $D$. 
	\end{algorithmic}
\end{algorithm}

\begin{algorithm}
	\caption{Computation of a $\lambda$-point of Definition \ref{dfn:lambda-point} in line \ref{line:knnu:dfnLambda} of Algorithm \ref{alg:cover_tree_k-nearest} }
	\label{alg:lambda}
	\begin{algorithmic}[1]
		\STATE \textbf{Input:} A point $q \in X$, a subset $C$ of a level set $C_i$ of a compressed cover tree $\T(R)$, an integer $k \in \Z$
		\STATE Define $f: C \rightarrow \R$ by setting $f(p) = d(p,q)$. 
		\STATE Run Algorithm \ref{alg:k_smallest_elements} on inputs $(C, f, k)$ and retrieve array $D$. 
		\STATE Find the smallest index $j$ such that $\sum^{j}_{t = 0}|\Sd_i(D[t] , \T(R))| \geq k$.
		\STATE \textbf{return} $\lambda = D[j]$. 
	\end{algorithmic}
\end{algorithm}

\begin{lem}
	\label{lem:time_k_smallest_elements}
	Let $A \subseteq R$ be a finite subset and let $f: A \rightarrow \R$ be a priority function and let $k \in \Z_{+}$. 
	Then Algorithm \ref{alg:k_smallest_elements} finds $k$-smallest elements of $A$ in time $|A| \cdot \log_2(k)$
\end{lem}
\begin{proof}
	Adding and removing element from binary heap data structure \cite[section~6.5]{Cormen1990} takes at most $O(\log(n))$ time, where $n$ is the size of binary heap.
	Since the size of our binary heap is capped at $k$ and we add/remove at most $|A|$ elements, the total time complexity is 
	$O(|A| \cdot \log_2(k))$.
\end{proof}

\begin{lem}[time complexity of a $\lambda$-point]
	\label{lem:time_lambdapoint}
	In the conditions of Definition~\ref{dfn:lambda-point}, the time complexity of Algorithm \ref{alg:lambda} is $O(|C| \cdot \log_2(k))$.
\end{lem}
\begin{proof}
	Note that in line $4$ we have $|\Sd_i(D[t] , \T(R))| \geq 1$ for all $t = 0 , ..., j$. 
	Therefore the time complexity of line $4$ is $O(k)$. 
	By Lemma \ref{lem:time_k_smallest_elements} The time complexity of line $3$ is $O(|C| \cdot \log_2(k))$, which proves the claim. 
\end{proof}

\begin{lem}[separation]
	\label{lem:separation}
	In the conditions of Definition~\ref{dfn:distinctive_descendant_set}, let $p\neq q$ be nodes of $\T(R)$ with $l(p) \geq i$, $l(q) \geq i$. Then $\Sd_i(p , \T(R)) \cap \Sd_{i}(q, \T(R)) = \emptyset$.  \bs
\end{lem}
\begin{proof}
	Without loss of generality assume that $l(p) \geq l(q)$. If $q$ is not a descendant of $p$, the lemma holds trivially due to $\Desc(q) \cap \Desc(p) = \emptyset$. 
	If $q$ is a descendant of $p$, then $l(q) \leq l(p) - 1$ and therefore $q \in V_i(p)$. 
	It follows that
	$\Sd_i(p , \T(R)) \cap \Desc(q) = \emptyset$ and therefore
	$$\Sd_{i}(p, \T(R)) \cap \Sd_{i}(q, \T(R))   \subseteq \Sd_{i}(p, \T(R)) \cap \Desc(q) = \emptyset.$$ 
\end{proof}

\begin{lem}[Sum lemma]
	\label{lem:sum}
	In the notations of Definition~\ref{dfn:distinctive_descendant_set} assume that $i$ is arbitrarily index and a subset $V \subseteq R$ satisfies $l(p) \geq i$ for all $p \in V$. Then 
	$$|\bigcup\limits_{p \in V}\Sd_i(p,\T(R))| = \sum\limits_{p \in V} |\Sd_i(p, \T(R))|.$$
\end{lem}
\begin{proof}
	Proof follows from Lemma \ref{lem:separation}.
\end{proof}

\noindent 
By Lemma \ref{lem:sum} in Definition \ref{dfn:lambda-point} one can assume that $|\bigcup_{p \in C}\Sd_i(p, \T(R)) | \geq k$.

\begin{lem}
	\label{lem:distinctive_descendant_child_level}
	In the notations of Definition~\ref{dfn:distinctive_descendant_set}, let $p \in \T(R)$ be any node. 
	If $w \in \Sd_i(p,\T(R))$ then either $w = p$ or there exists $a \in \Child(p) \setminus \{p\}$ such that $l(a) < i$ and $w \in \Desc(a)$.\bs
\end{lem}
\begin{proof}
	Let $w \in \Sd_i(p)$ be an arbitrary node satisfying $w \neq p$. 
	Let $s$ be the node-to-root path of $w$. 
	The inclusion $\Sd_i(p) \subseteq \Desc(p)$ implies that $w \in \Desc(p)$. 
	Let $a \in \Child(p) \setminus \{p\}$ be a child on the path $s$. 
	If $l(a) \geq i$ then $a \in V_i(p)$. 
	Note that $w \in \Desc(a)$.
	Therefore $w \notin \Sd_i(p)$, which is a contradiction.
	Hence $l(a) < i$. 
\end{proof}

\begin{lem}
	\label{lem:distinctive_descendant_distance}
	In the notations of Definition~\ref{dfn:distinctive_descendant_set}, let $p \in \T(R)$ be any node. 
	If $w \in \Sd_i(p,\T(R))$ then $d(w,p) \leq 2^{i+1}$.\bs
\end{lem}
\begin{proof}
	By Lemma \ref{lem:distinctive_descendant_child_level} either $w = \gamma$ or $w \in \Desc(a)$ for some $a \in \Child(\gamma) \setminus \{\gamma\}$ for which $l(a) < i$.
	If $w = \gamma$, then trivially $d(\gamma, w) \leq 2^{i}$. Else $w$ is a descendant of $a$, which is a child of node $\gamma$ on level $i-1$ or below, therefore by Lemma \ref{lem:compressed_cover_tree_descendant_bound} we have  $d(\gamma, w) \leq 2^{i}$ anyway.
\end{proof}

\begin{lem}
	\label{lem:child_set_equivalence}
	Let $R$ be a finite subset of a metric pace.
	Let $\T(R)$ be a compressed cover tree on $R$. 
	Let $R_j \subseteq C_j$, where $C_j$ is the $i$th cover set of $\T(R)$. Let $i = \max_{p \in R_j}\nxt(p,j,\T(R))$.
	Set 
	$\C_j(R_j) = R_{j} \cup \{a \in \Child(p) \text{ for some }p \in R_i \mid l(a) = i \}$.
	Then 
	$$\bigcup_{p \in \C_j(R_j)}\Sd_{i}(p, \T(R)) = \bigcup_{p \in R_j}\Sd_{j}(p, \T(R)).$$
\end{lem}
\begin{proof}
	Let $a \in \bigcup_{p \in \C_j(R_j)}\Sd_{i}(p, \T(R))$ be an arbitrary node. Then there exits $u \in \C_j(R_j)$ having 
	$ a \in \Sd_{i}(u, \T(R))$. 
	By definition of index $i$, either $u \in R_j$ or $u$ has a parent in $R_j$. If $u \in R_j$ then we note that $V_j(u) \subseteq V_i(u)$.
	Since $a \notin V_i(u)$, we also have $a \notin V_j(u)$. 
	
	\medskip
	
	\noindent 
	Otherwise let $w$ be a parent of $u$. Therefore there are no descendants of $w$ in having level in interval $[l(u) + 1 , l(p) - 1]$.
	Since $l(u) = i$ and $j > i$ it follows that $V_{j}(w) = \emptyset$. 
	Denote $w$ to be the lowest level ancestor of $u$ on level $j$. By cases above we have $a \notin V_{j}(w)$.
	Therefore it follows that
	$$a \in \Sd_j(w, \T(R)) \subseteq \bigcup_{p \in R_j}\Sd_j(p , \T(R)).$$
	To prove the converse inclusion assume now that $a \in \bigcup\limits_{p \in R_j}\Sd_{j}(p, \T(R))$. 
	Then $a \in \Sd_{j}(v, \T(R))$ for some $w \in R_j$.
	Assume that $w$ has no children at the level $i$. 
	Then $V_{j}(w) = V_{i}(w)$ and 
	$$a \in \Sd_{i}(w, \T(R)) \subseteq \bigcup_{p \in \C_j(R_j)}\Sd_{i}(p ,\T(R)).$$ 
	Assume now that $w$ has children at the level $i$.
	If there exists $b \in \Child(w)$ for which $a \in \Desc(b)$. 
	Since $V_{i}(b) = \emptyset$, we conclude that 
	$$a \in \Sd_{i}(b, \T(R)) \subseteq \bigcup_{p \in \C_j(R_j)}\Sd_{i}(p ,\T(R)).$$ 
	Assume that $a \notin \Desc(b)$ for all $b \in \Child(w)$ with $l(b) = i$.
	Then $a \in \Desc(w)$ and $a \notin \Desc(b')$ for any $b' \in V_j(w)$. Then $a \in \Sd_{i}(w, \T(R))$ and the proof finishes:
	$$\bigcup_{p \in R_j}\Sd_{j}(p, \T(R))   \subseteq   \bigcup_{p \in \C_j(R_j)}\Sd_{i}(p, \T(R)).  $$
\end{proof}

\begin{lem}[$\be$-point]
	\label{lem:beta_point}
	In the notations of Definition~\ref{dfn:lambda-point}, let $C\subseteq C_i$ so that $\cup_{p \in C}\Sd_i(p, \T(R))$ contains all $k$-nearest neighbors of $q$. 
	Set $\la = \la_k(q,C)$. 
	Then $R$ has a point $\beta$ among the first $k$ nearest neighbors of $q$ such that $d(q,\lambda) \leq  d(q,\beta) + 2^{i+1}$.\bs
\end{lem}
\begin{proof}
	We show that $R$ has a point $\beta$ among the first $k$ nearest neighbors of $q$ such that
	$$\beta \in \bigcup_{p \in C}\Sd_i(p, \T(R)) \setminus \bigcup_{p \in N(q, \lambda) \setminus \{\lambda\} }\Sd_i(p, \T(R)).$$
	Lemma~\ref{lem:sum} and Definition~\ref{dfn:lambda-point} imply that
	$$ | \bigcup_{p \in N(q, \lambda) \setminus \{\lambda\} }\Sd_i(p, \T(R)) |= \sum_{p \in N(q, \lambda) \setminus \{\lambda\} }| \Sd_i(p, \T(R)) | < k.$$ 
	Since $\cup_{p \in C}\Sd_i(p, \T(R))$ contains all $k$-nearest neighbors of $q$, a required point $\beta\in R$ exists. 
	\medskip
	
	\noindent 
	Let us now show that $\beta$ satisfies $d(q,\lambda) \leq  d(q,\beta) + 2^{i+1}$.
	Let $\gamma \in C \setminus N(q,\lambda) \cup \{\lambda \}$ be such that $\beta \in \Sd_i(\gamma, \T(R))$. 
	Since $\gamma \notin N(q,\lambda) \setminus \{\lambda\}$, we get $d(\gamma, q) \geq d(q, \lambda)$.  
	The triangle inequality says that $ d(q, \gamma) \leq d(q,\beta) + d(\gamma ,\beta) $.
	Finally Lemma~\ref{lem:distinctive_descendant_distance} implies that $d(\gamma, \beta) \leq 2^{i+1}$. 
	Then
	$$d(q,\lambda) \leq d(q, \gamma) \leq d(q,\beta) + d(\gamma ,\beta) \leq d(q,\beta) +2^{i+1}$$
	So $\beta$ is a desired $k$-nearest neighbor satisfying $d(q,\lambda) \leq  d(q,\beta) + 2^{i+1}$.
\end{proof}

\section{Construction of a compressed cover tree}
\label{sec:ConstructionCovertree}

This section introduces a new method Algorithm \ref{alg:cover_tree_k-nearest_construction_whole} for construction of a compressed cover tree, which is based on Insert() method \cite[Algorithm~2]{beygelzimer2006cover} that was specifically adapted for compressed cover tree. The proof of  \cite[Theorem~6]{beygelzimer2006cover}, which estimated the time complexity of \cite[Algorithm~2]{beygelzimer2006cover} was shown to be incorrect  Counterexample \ref{cexa:construction_algorithm_of_original_cover_tree}. The main contribution of this section are two new time complexity results that bound the time complexity of Algorithm \ref{alg:cover_tree_k-nearest_construction_whole}:
\begin{itemize}
    \item Theorem \ref{thm:construction_time} bounds the time complexity as 
    $O(c_m(R)^{10} \cdot \log_2(\Delta(R)) \cdot |R|)$ by using minimized expansion constant $c_m(R)$ and aspect ratio $\Delta(R)$ as parameters.
    \item  Theorem \ref{thm:construction_time_KR} bounds the time complexity as $O(c(R)^{12} \cdot \log_2|R| \cdot |R|)$ by using expansion constant $c(R)$ as parameter.
\end{itemize}
Definition~\ref{dfn:implementation_compressed_cover_tree} explains the concrete implementation of compressed cover tree.

\begin{dfn}[$\Child(p,i)$ and $\nxt(p,i,\T(R))$]
	\label{dfn:implementation_compressed_cover_tree}
	In a compressed cover tree $\T(R)$ on a set $R$, for any level $i$ and a node $p \in R$, set $\Child(p,i) = \{ a \in \Child(p) \mid l(a) = i \}$. 
	Let $\nxt(p,i,\T(R))$ be the maximal level $j$ satisfying $j < i$ and $\Child(p,i) \neq \emptyset$. 
	If such level does not exist, we set $j = l_{\min}(\T(R)) - 1$.
	For every node $p$, we store its set of children in a linked hash map so that
	
	\noindent
	(a)  any key $i$ gives access to $\Child(p,i)$, 
	
	\noindent
	(b) every $\Child(p,i)$ has access to $\Child(p,\nxt(p,i, \T(R)))$,
	
	\noindent
	(c) we can directly access $\max \{j \mid \Child(p,j) \neq \emptyset\}$.
	\bs
\end{dfn}

\begin{dfn}[construction iteration set $L(\T(W),p)$]
	\label{dfn:cover_tree_construction_iteration_set}
	Let $W$ be a finite subset of a metric space $(X,d)$. 
	Let $\T(W)$ be a cover tree of Definition \ref{dfn:cover_tree_compressed} built on $W$ and let $p \in X \setminus W$ be an arbitrary point.
	Let $L(\T(W),p) \subseteq H(\T(R))$ be the set of all levels $i$ during iterations \ref{line:cof:loop_start}-\ref{line:cof:loop_end} of Algorithm \ref{alg:cover_tree_k-nearest_construction} launched with inputs 
	$\T(W),p$. 
	Set $\eta(i) = \min_{t} \{ t \in L(\T(W),p) \mid t > i\}$. 
	\bs
\end{dfn}
\begin{algorithm}
	\caption{Building a compressed cover tree $\T(R)$ from Definition \ref{dfn:cover_tree_compressed}.
	}
	\label{alg:cover_tree_k-nearest_construction_whole}
	\begin{algorithmic}[1]
		\STATE \textbf{Input} : a finite subset $R$ of a metric space $(X,d)$
		\STATE \textbf{Output} : a compressed cover tree $\T(R)$. 
		\STATE Choose a random point $r \in R$ to be a root of $\T(R)$ 
		\STATE Build the initial compressed cover tree $\T = \T(\{r\})$ by making $l(r) = +\infty$. 
		\FOR{$p \in R \setminus \{r\}$} \label{line:con:for:begin}
		\STATE $\T \leftarrow $ run AddPoint$(\T , p )$ described in Algorithm \ref{alg:cover_tree_k-nearest_construction}.
		\ENDFOR  \label{line:con:for:end}
		\STATE For root $r$ of $\T$ set $l(r) = 1 +  \max_{p \in R \setminus \{r\}}l(p)$
	\end{algorithmic}
\end{algorithm}
\begin{algorithm}
	\caption{Building $\T(W \cup \{p\})$ in 
		lines \ref{line:con:for:begin}-\ref{line:con:for:end} of Algorithm \ref{alg:cover_tree_k-nearest_construction_whole}.}
	\label{alg:cover_tree_k-nearest_construction}
	\begin{algorithmic}[1]
		\STATE \textbf{Function} AddPoint(a compressed cover tree $\T(W)$ with a root $r$, a point $p\in X$)
		\STATE \textbf{Output} : compressed cover tree $\T(W \cup \{p\})$. 
		\STATE Set $i \leftarrow l_{\max}(\T(W)) - 1$ and $\eta(l_{\max} - 1) = l_{\max}$
		 \COMMENT{If the root $r$ has no children then $ i \leftarrow -\infty$}
		\STATE Set $R_{l_{\max}} \leftarrow \{r\}$.
		\WHILE{$i \geq l_{\min}$} \label{line:cof:loop_start}
		\STATE Assign $\mathcal{C}_i(R_{\eta(i)}) \leftarrow  R_{\eta(i)} \cup \{a \in \Child(p) \text{ for some }p \in R_{\eta(i)} \mid l(a) = i \} $
		\label{line:cof:dfn_C}
		\STATE Set $R_{i} = \{a \in \C_i(R_{\eta(i)}) \mid d(p,a) \leq 2^{i+1} \}$
		\label{line:cof:defRim1}
		\IF {$R_i$ is empty}\label{line:cof:inner_loop:begin}
		\STATE  Launch Algorithm \ref{alg:construction_parent_assign} with parameters $(p, R_{\eta(i)})$. \label{line:cof:inner_loop:mid}
		\ENDIF \label{line:cof:inner_loop:end}
		\STATE $t = \max_{ a \in R_{i}} \nxt(a,i,\T(W)) $  \label{line:cof:dfn_t}
		\COMMENT{If $R_{i}$ has no children we set $t = l_{\min} - 1$}
		\STATE $\eta(i) \leftarrow i$ and $i \leftarrow t$ 
		\ENDWHILE \label{line:cof:loop_end}
		\STATE Launch Algorithm \ref{alg:construction_parent_assign} with parameters $(p, R_{\eta(i)})$. \label{line:cof:selectParentEnd}
	\end{algorithmic}
\end{algorithm}
\begin{algorithm}
\caption{Assign node subprocedure}
\label{alg:construction_parent_assign}
\begin{algorithmic}[1]
\STATE \textbf{Function} AssignParent(Point $p$, subset of nodes $U \subseteq \T(W)$)
\STATE \textbf{Output:} Compressed cover tree $\T(W \cup \{p\})$
\STATE Pick $v \in U$ minimizing $d(v,p)$.
\STATE Set $l(p) = \floor{\log_2(d(p,v)}-1$ and let $v$ be a parent of $p$. \label{line:construction_assign_important}
\end{algorithmic}
\end{algorithm}

\noindent
Let $R$ be a finite subset of a metric space $(X,d)$. 
A compressed cover tree $\T(R)$ will be incrementally constructed by adding points one by one as summarized in Algorithm \ref{alg:cover_tree_k-nearest_construction_whole}. 
First we select a root node $r \in R$ and form a tree $\T(\{r\})$ of a single node $r$ at the level $l_{\max} = l_{\min} = +\infty$. 
Assume that we have a compressed cover tree $\T(W)$ for a subset $W \subset R$. 
For any point $p \in R \setminus W$, Algorithm \ref{alg:cover_tree_k-nearest_construction} builds a larger compressed cover tree $\T(W \cup \{p\})$ from $\T(W)$. 

\medskip 
\noindent
Note that during the construction of the compressed cover tree in Algorithm \ref{alg:cover_tree_k-nearest_construction} we write down additional information for every node $p$ , which includes the number of descendants of node $p$ and the maximal level of nodes in set $\Child(p)$. 

\begin{lem}
\label{lem:separation_of_descendants}
Let $\T(R)$ be a cover tree and let $p \in X$ be a point and let $i \in \Z$. Assume that for some $q \in \T(R)$ we have $d(p,q) > 2^{i+1}$. 
Let $\Sd_i(q,\T(R))$ be as defined in Definition~\ref{dfn:distinctive_descendant_set}.
Then for any $\theta \in \Sd_i(q, \T(R)) \setminus \{q\}$ we have $d(\theta,p) > 2^{l(\theta)}$.
\end{lem}
\begin{proof}
Let $S= (\theta = a_0, ...,a_m = q)$ be a node to node path. Since $\theta \in \Sd_i(q, \T(R)) \setminus \{q\}$ by Lemma~\ref{lem:distinctive_descendant_child_level}
we have $l(a_{m-1}) \leq i - 1$. Therefore $l(\theta) = l(a_0) \leq ... \leq l(a_{m-1}) \leq i-1$. We have the following inequality:
 $$d(q,\theta) \leq  \sum\limits^{h-1}_{z = 0}d(a_z, a_{z+1})  \leq \sum\limits^{j}_{x = l(\theta)+1} 2^{x} = (2^{j+1} - 2^{l(\theta)+1}).$$
  By triangle inequality we have: $d(p,\theta) \geq d(p,\gamma) - d(\gamma, \theta) > 2^{j+1} - (2^{j+1} - 2^{l(\theta)+1}) > 2^{l(\theta)}$.
Therefore $d(p,\theta) > 2^{l(\theta)}$ which proves the claim. 
\end{proof}



\begin{thmm}[correctness of Algorithm \ref{alg:cover_tree_k-nearest_construction_whole}]
	\label{thm:construction_correctness}
	Algorithm \ref{alg:cover_tree_k-nearest_construction_whole} builds a compressed cover tree satisfying Definition~\ref{dfn:cover_tree_compressed}.
	\bs
\end{thmm} 
\begin{proof}
	It suffices to prove that Algorithm~\ref{alg:cover_tree_k-nearest_construction} correctly extends a compressed cover tree $\T(W)$ for any finite subset $W\subseteq X$ by adding a point $p$. Let us prove that $\T(W \cup \{p\})$ satisfies Definition~\ref{dfn:cover_tree_compressed}. 
	
	\medskip
    \noindent
    We first note that the parent $v$ of $p$ is always assigned in Algorithm~\ref{alg:construction_parent_assign} by setting $l(p) =\floor{\log_2(d(p,v)}-1$.
    Note that the set $U$ is never empty, when Algorithm~\ref{alg:construction_parent_assign} is launched.
    The covering condition (\ref{dfn:cover_tree_compressed}b) after adding point $p$ to $\T(W)$ follows from the following inequality:
    $$d(p,w) \leq 2^{\floor{\log_2(d(p,v)}} \leq 2^{l(p) + 1}.$$
    
	
	\smallskip
	\noindent 
	To check (\ref{dfn:cover_tree_compressed}c) Consider arbitrary cover set $C_{h} = \{q \in \T(W \cup \{p\}) \mid l(q) \geq h\}$. Since we have assumed that $\T(W)$ is a valid cover tree, all the cover sets $C_h$ for $h > l(p)$ satisfy the condition. Let us consider cover sets having $h \leq l(p)$. Let $\theta \in C_h$ be an arbitrary node. 
 Consider a sequence of iterations $l_{\min}(\T(W)) \leq a(0) < a(1) < ... < a(t) = l_{\max}(\T(W))$ that were considered during run-time of the algorithm. Note that the parent of $p$ was assigned at $i = a(0)$. Since $\theta \in W = \Sd_{l_{\max}(r,\T(W))}$, 
 either (\textbf{a}) $\theta \in \bigcup_{q \in R_{a(0)}}\Sd_{a(0)}(q,R_{a(0)})$ or (\textbf{b}) 
 there exists index $j$ satisfying 
 $$\theta \in \bigcup_{q \in R_{a(j+1)}}\Sd_{a(j+1)}(q, \T(W)) \setminus \bigcup_{q \in R_{i}}\Sd_{a(j)}(q, \T(W)).$$
 Let us first consider case $\textbf{(a)}$. Let $v$ be a parent of $p$ in $\T(W \cup \{p\})$. Recall that the parent $v$ of $p$ was assigned in line~\ref{line:construction_assign_important} of Algorithm \ref{alg:construction_parent_assign}. Therefore we have $d(v,p) \leq d(p,\theta)$ and by line~ \ref{line:construction_assign_important} we have:
  $$d(p,\theta) \geq d(p,v) \geq 2^{l(p) + 1} > 2^{l(p)} \geq 2^{h},$$
  which proves the claim.

  \smallskip
	
 \noindent 
 Assume now (\textbf{b}) holds. Denote $i = a(j+1)$, since $a(j)$ was previous level, it follows $\eta(i) = a(j)$. By Lemma~\ref{lem:child_set_equivalence} we have: 
 $$\bigcup_{q \in \C_{i}(R_{\eta(i)})}\Sd_{i}(q, \T(W)) = \bigcup_{q \in R_{\eta(i)}}\Sd_{\eta(i)}(q, \T(W)).$$
 Therefore there exists a node $u \in \C_i(R_{\eta(i)}) \setminus R_i$ for which $\theta \in \Sd_i(u, \T(W))$.
By line~\ref{line:cof:defRim1} of Algorithm~\ref{alg:cover_tree_k-nearest_construction} we have $d(u, p) > 2^{i+1}$.
If $u = \theta$, then the parent of $p$ was selected from set $R_{\eta(i)}$ and the proof is similar to (\textbf{a}).
Else by Lemma~\ref{lem:separation_of_descendants} it follows that $d(p,\theta) > 2^{l(\theta)} \geq 2^{h}$ which proves the claim.

 \end{proof}

\begin{lem}[time complexity of a key step for $\T(R)$]
	\label{lem:general_construction_time}
	Arbitrarily order all points of a finite reference set $R$ in a metric space $(X,d)$ starting from the root: $r=p_1$, $p_2,\dots,p_{|R|}$.
	Set $W_1=\{r\}$ and $W_{y+1}=W_y\cup\{p_y\}$ for  
	$y=1,...,|R|-1$. 
	Then Algorithm~\ref{alg:cover_tree_k-nearest_construction_whole} builds a compressed cover tree $\T(R)$ in time $$O\Big((c_m(R))^8 \cdot \max\limits_{y=1,...,|R|-1}L(\T(W_{y}),p_{y}) \cdot |R|\Big),$$
	where $c_m(R)$ is the minimized expansion constant from Definition \ref{dfn:expansion_constant}.
	\bs
\end{lem}

\begin{proof}
	The worst-case time complexity of Algorithm \ref{alg:cover_tree_k-nearest_construction_whole} is dominated by lines \ref{line:con:for:begin}-\ref{line:con:for:end} which call Algorithm \ref{alg:cover_tree_k-nearest_construction} $O(|R|)$ times in total. 
	\medskip
	
	\noindent
	Assume that we have already constructed a cover tree on set $\T(W_y)$, the goal Algorithm \ref{alg:cover_tree_k-nearest_construction} is to construct tree $\T(W_y \cup \{p_{y+1}\})$.
	By Definition \ref{dfn:cover_tree_construction_iteration_set} loop on lines
	\ref{line:cof:loop_start}-\ref{line:cof:loop_end} is performed $L(\T(W_y),p_{y+1})$ times. 
	Let $R_{*}$ be the maximal size of set $R_i$ during all iterations $i \in L(\T(W_y),p_{y+1})$.
	By Lemma~\ref{lem:compressed_cover_tree_width_bound} since $W_{y+1} \subseteq R \subseteq X$ we have 
	$$|\C_{i}(R_{\eta(i)})| \leq c_m(W_{y+1})^4 \cdot |R_{*}| \leq c_m(R)^4 \cdot |R_{*}|$$ nodes, where $\C_{\eta(i)}(R_{\eta(i)})$ is defined in 
 line~\ref{line:cof:dfn_C}.  Therefore both, lines \ref{line:cof:defRim1}  and \ref{line:cof:dfn_C} take at most $c_m(R)^4|R_{*}|$ time. In line \ref{line:cof:dfn_t} we handle $|R_{*}|$ elements, for each of them we can retrieve index $\nxt(a,i, \T(W_y))$ in $O(1)$ time, since for every $a \in \T(R)$ we can update the last index $j$, when $a$ had children on level $j$ in line \ref{line:cof:dfn_C}. Therefore line  \ref{line:cof:dfn_t} takes at most $O(|R_{*}|)$ time.  
    Algorithm~\ref{alg:construction_parent_assign} takes at most $O(|R_{*}|)$ time. 
    Therefore line~\ref{line:cof:inner_loop:mid} and line~\ref{line:cof:selectParentEnd} take at most $O(|R_{*}|)$ time.
	Let us now bound $|R_{*}|$ during the whole run-time of the algorithm.
	
	\medskip
	
	\noindent
	Let $i$ be an arbitrary level. 
	Note that $R_{i} \subseteq B(p,2^{i+1}) \cap C_{i}$ where $C_{i}$ is a $i$th cover set of $\T(R)$. Since $C_{i}$ is $2^{i}$-spares subset of $R$ we can apply packing Lemma \ref{lem:packing} with $r = 2^{i+1}$ and $\delta = 2^{i}$ to obtain 
	$|B(p,2^{i+1}) \cap C_{i} | \leq (c_m(W))^4 $.  
	Lemma \ref{lem:expansion_constant_property} implies $(c_m(W))^4  \leq (c_m(R))^4 $, therefore $|B(p,2^{i}) \cap C_{i} | \leq (c_m(R))^4$. 
	\smallskip
	
	\noindent
	The time complexity of loop \ref{line:cof:loop_start} - \ref{line:cof:loop_end} in Algorithm \ref{alg:cover_tree_k-nearest_construction} is dominated by line \ref{line:cof:dfn_C} that has time $O(|C(R_i)|) \leq O((c_m(R))^4 \cdot |R_{*}|) \leq O((c_m(R))^8)$. 
	Then the whole  Algorithm \ref{alg:cover_tree_k-nearest_construction_whole} has time
	$$O((c_m(R))^8 \cdot  \max\limits_{y=2,...,|R|}L(\T(W_{y-1}),p_{y}) \cdot |R|)$$ as desired. 
\end{proof}
\begin{thmm}[time complexity of $\T(R)$ via aspect ratio]
	\label{thm:construction_time}
	Let $R$ be a finite subset of a metric space $(X,d)$ having aspect ratio $\Delta(R)$.
	Algorithm~\ref{alg:cover_tree_k-nearest_construction_whole} builds
	a compressed cover tree $\T(R)$ in time $O((c_m(R))^8 \cdot \log_2(\Delta(R)) \cdot |R|),$
	where $c_m(R)$ is the minimized expansion constant from Definition~\ref{dfn:expansion_constant}.
\end{thmm}
\begin{proof}
	In Lemma \ref{lem:general_construction_time} use the upper bounds  due to Lemma \ref{lem:depth_bound} as follows: \\
	$\max\limits_{y \in 2,...,|R|}|L(\T(W_{y-1}),p_{y})| \leq H(\T(R))\leq 1 + \log_2(\Delta(R))$.
\end{proof}

\begin{lem}
	\label{lem:knn_next_level_finder_for_log_depth}
	Let $(X,d)$ be a metric space and let $W \subseteq X$ be its finite subset. Let $q \in X \setminus W$ be an arbitrary point.
	Let $i \in L(\T(W),q)$ be arbitrarily iteration of Definition \ref{dfn:cover_tree_construction_iteration_set}. Assume that $t = \eta(\eta(i+1))$ is defined. Then there exists $p \in R$ satisfying $2^{i+1} < d(p,q) \leq 2^{t+1}$.
\end{lem}
\begin{proof}

	Note first that since $\eta(i+3) \in L(\T(R),q)$, there exists distinct 
	$u \in R_{\eta(\eta(i+3))}$ and $v \in \C_{\eta(i+1)}(R_{\eta(\eta(i+1)}))$, in such a way that $u$ is the parent of $v$. 
	Let us show that both of $u,v$ cant belong to set $R_i$. Assume contrary that both $u,v \in R_i$. Then by line  \ref{line:cof:defRim1} of Algorithm \ref{alg:cover_tree_k-nearest_construction}
	we have $d(v,q) \leq 2^{i+1}$ and $d(u,q) \leq 2^{i+1}$. By triangle inequality $d(u,v) \leq d(u,q) + d(q,v) \leq 2^{i+2} \leq 2^{\eta(i+1)}$.
	Recall that we denote a level of a node by $l$.
	On the other hand we have $l(u) \geq \eta(i+1)$ and $l(v) \geq \eta(i+1)$, by separation condition of Definition \ref{dfn:cover_tree_compressed} we have $d(u,v) > 2^{\eta(i+1)}$, which is a contradiction. Therefore only one of $\{u,v\}$ 
	can belong to $R_i$. It sufficies two consider the two cases below: 
	
	\medskip
	
	\noindent
	\textbf{Assume that }$v \notin R_i$. Since $v$ is children of $u$ we have $d(u,v) \leq 2^{\eta(i+1) + 1}$.
	By line \ref{line:cof:defRim1} of Algorithm \ref{alg:cover_tree_k-nearest_construction} we have $d(u,q) \leq 2^{\eta(i+1) + 1}$.
	By triangle inequality 
	$$d(v,q) \leq d(v,u) + d(u,q) \leq 2^{ \eta(i+1) + 1 } + 2^{\eta(i+1) + 1} \leq 2^{\eta(i+1) + 2} \leq 2^{\eta(\eta(i+1)) + 1} $$
	Since $v \notin R_i$ there exists level $t$ having $\eta(i+1) \geq t \geq i$ and $v \in \C_{t}(R_{\eta(t)}) \setminus R_t$.
	Therefore by line \ref{line:cof:defRim1} of Algorithm \ref{alg:cover_tree_k-nearest_construction}  we have $d(q,v) > 2^{t+1} \geq 2^{i+1}$.
	It follows that we have found point $v \in R$ satisfying $2^{i+1} < v \leq 2^{\eta(\eta(i+1)) + 1}$. Therefore $p = v$, is the desired point.
	
	\medskip
	
	\noindent
	\textbf{Assume that }$u \notin R_i$. Since $u \in R_{\eta(\eta(i+1))}$, by line \ref{line:cof:defRim1} of Algorithm \ref{alg:cover_tree_k-nearest_construction} we have $d(u,q) \leq 2^{\eta(\eta(i+1)) + 1}$.
	On the other hand since $u \notin R_i$, there exists level $t$ having $\eta(i+3) \geq t \geq i$ and $u \in \C_{t}(R_{\eta(t)}) \setminus R_t$. Therefore by line \ref{line:cof:defRim1} of Algorithm \ref{alg:cover_tree_k-nearest_construction} we have $d(q,u) > 2^{t+2} \geq 2^{i+2}$.
	It follows that we have found point $u \in R$ satisfying $2^{i+1} < u \leq 2^{\eta(\eta(i+1)) + 1}$. Therefore $p = u$, is the desired point.
	
\end{proof}

\begin{lem}
	\label{lem:construction_depth_bound}
	Let $A, W$ be finite subsets of a metric space $X$ satisfying $W \subseteq A \subseteq X$, take $q \in A \setminus W$.
	Given a compressed cover tree $\T(W)$ on $W$, Algorithm~\ref{alg:cover_tree_k-nearest_construction} runs lines \ref{line:cof:loop_start}-\ref{line:cof:loop_end} this number of times: $|L(\T(W),q)| = O\big(c(A)^2 \cdot \log_2(|A|)\big)$.
	\bs
\end{lem}
\begin{proof}
	Let $x \in L(\T(R),q)$ be the lowest level of $L(\T(R),q)$.
	Define $s_1 = \eta(\eta(x)+1)$ and let $s_i = \eta(\eta(\eta(s_{i-1}+1))+1)$, if it exists. Assume that $s_{n+1}$ is the last sequence element for which $\eta(\eta(\eta(s_{n-1}+1))+1)$ is defined. Define $S = \{s_1,...,s_{n}\}$. For every $i \in \{1,...,n\}$ let $p_i$ be the point provided by Lemma \ref{lem:knn_next_level_finder_for_log_depth} that satisfies $$ 2^{s_i+1} < d(p_i,q) \leq 2^{\eta(\eta(s_{i}+1)) + 1}.$$
	Let $P$ be the sequence of points $p_i$.  Denote $n = |P| = |S|$. 
	Let us show that $S$ satisfies the conditions of Lemma \ref{lem:growth_bound_extension}. Note that:
	$$4 \cdot d(p_i,q)\leq 4 \cdot 2^{\eta(\eta(s_{i}+1)) + 1} \leq 2^{\eta(\eta(s_{i}+1)) + 3} \leq 2^{\eta(\eta(\eta(s_{i}+1))+1) + 1} \leq 2^{s_{i+1}+1} \leq d(p_{i+1},q)$$
	By Lemma \ref{lem:growth_bound_extension} applied for set $A$ and sequence $P$ we get:
	$$|\bar{B}(q,\frac{4}{3}  d(q,p_n))| \geq (1+\frac{1}{c(R)^2})^{n} \cdot |\bar{B}(q,\frac{1}{3} d(q,p_1))|$$
	Since $\eta(x) \in L(\T(R),q)$ , there exists some point $u \in R_{\eta(x)}$. 
	By definition of $R_i$ we have $d(u,q) \leq 2^{\eta(x) + 1}$. 
	Also $$2^{\eta(\eta(x) + 1)-1} \leq \frac{2^{\eta(\eta(x) + 1)+1}}{3} < \frac{d(q,p_1)}{3}$$
	It follows that:
	$$1 \leq  |\bar{B}(q, 2^{\eta(x) + 1})| \leq |\bar{B}(q, 2^{\eta(\eta(x)+1) - 1}| \leq |\bar{B}(q, \frac{d(q,p_1)}{3})|$$
	Therefore we have
	$$|A| \geq \frac{|\bar{B}(q,\frac{4}{3} \cdot d(q,p_n))|}{|\bar{B}(q,\frac{1}{3} \cdot d(q,p_1))|} \geq (1+\frac{1}{c(A)^2})^{n}$$
	Note that $c(A) \geq 2$ by definition of expansion constant. Then by applying $\log$ and by using Lemma \ref{lem:hard_function_bound} we obtain: $c(A)^2\log(A) \geq n = |S|$. 
	Let $x$ be minimal level of $L(\T(W),q)$ and let $y$ be the maximal level of $L(\T(W),q)$ 
	Note that $S$ is a sub sequence of $L$ in such a way that:
	\begin{itemize}
		\item $[x,s_1] \cap L(\T(R),q) \leq 3$, 
		\item for all $i \in 1,..., n$ we have $[s_i, s_{i+1}] \cap L(\T(R),q) \leq 6 $
		\item $[s_n, y] \cap L(\T(R),q) < 12$
	\end{itemize}
	Since segments $[x,s_1],[s_1,s_2], ...,  [s_2,s_n], [s_n,y]$ cover $|L(\T(R),q)|$,
	it follows that $|S| \geq \frac{|L(\T(R),q)|}{12}$. We obtain that $$|L(\T(R),q)| \leq 12 \cdot c(A)^2 \cdot \log_2(|A|),$$ which proves the claim.
\end{proof}

\begin{thmm}[time complexity of $\T(R)$ via expansion constants]
	\label{thm:construction_time_KR}
	Let $R$ be a finite subset of a metric space $(X,d)$. Let $A$ be a finite subset of $X$ satisfying $R \subseteq A \subseteq X$. 
	Then Algorithm~\ref{alg:cover_tree_k-nearest_construction_whole} builds a compressed cover tree $\T(R)$ in time $O((c_m(R))^8 \cdot c(A)^2 \cdot \log_2(|A|) \cdot |R|)$, see $c(A),c_m(R)$ in Definition \ref{dfn:expansion_constant}.
	\bs
\end{thmm}
\begin{proof}
	It follows from Lemmas~\ref{lem:construction_depth_bound} and~\ref{lem:general_construction_time}.
\end{proof}


\begin{cor}
	\label{cor:construction_time_KR}
	Let $R$ be a finite subset of a metric space $(X,d)$. 
	Then Algorithm~\ref{alg:cover_tree_k-nearest_construction_whole} builds
	a compressed cover tree $\T(R)$ in time $O((c_m(R))^8 \cdot c(R)^2 \cdot \log_2(|R|)) \cdot |R|),$
	where $c(R),c_m(R)$ appeared in Definition \ref{dfn:expansion_constant}.
	\bs
\end{cor}
\begin{proof}
	The proof follows from Theorem~\ref{thm:construction_time_KR} by setting $A = R$.
\end{proof}

\section{$k$-nearest neighbor search algorithm}
\label{sec:better_approach_knn_problem}


This section is motivated by
Counterexample~\ref{cexa:original_all_nearest_neighbors_algorithm} which showed that the proof of past time complexity result  \cite[Theorem~5]{beygelzimer2006cover} of nearest neighborhood search algorithm contained mistakes.  The two main results of this sections are Corollary~\ref{cor:cover_tree_knn_miniziminzed_constant_time}
and Theorem \ref{thm:knn_KR_time} which provide new time complexity results for $k$-nearest neighborhood problem, assuming that a compressed cover tree was already constructed for the reference set $R$. For the construction algorithm of compressed cover tree and its time complexity, we refer to Section~\ref{sec:ConstructionCovertree}.

 \medskip
 \noindent
 The past mistakes are resolved by introducing a new Algorithm \ref{alg:cover_tree_k-nearest} for finding $k$-nearest neighbors that generalize and improves the original method for finding nearest neighbors using an implicit cover. \cite[Algorithm~1]{beygelzimer2006cover}. The first improvement is $\lambda$-point of line~\ref{line:knnu:dfnLambda} which allows us to search for all $k$-nearest neighbors of a given query point for any $k \geq 1$. The second improvement is a new loop break condition on line \ref{line:knnu:qtoofar:condition}. The new loop break condition is utilized in the proof of Lemma \ref{lem:knn_depth_bound} to conclude that the total number of performed iterations is bounded by  $O(c(R)^2\log(|R|))$  during whole run-time of the algorithm. The latter improvement closes the past gap in proof of \cite[Theorem~5]{beygelzimer2006cover} by bounding the number of iterations independently from the explicit depth  $D(\T(R))$ of Definition~\ref{dfn:explicit_depth_for_compressed_cover_tree}, that generated the past confusion. 
 
  \medskip
 \noindent 
 Recall from Definition~\ref{dfn:essential_levels_node} that an essential set $\Es(p,\T(R)) \subseteq H(\T(R)$ consists of all levels $i \in H(\T(R))$ for which $p$ has non-trivial children in $\T(R)$ at level $i$. 
 By Lemma~\ref{lem:distinctive_descendants_precompute} the sizes of distinctive descendants $|\Sd_i(p, \T(R))|$ can be precomputed in a linear time $O(|R|)$ for all $p \in R$ and $i \in \Es(p,\T(R))$.
 Since the size of distinctive descendant set $|\Sd_i(p, \T(R))|$ can only change at indices $i \in \Es(p,\T(R))$, we assume that the sizes of $|\Sd_i(p, \T(R))|$ can be retrieved in a constant time $O(1)$ for any $p \in R$ and $i \in H(\T(R))$ during the run-time of Algorithm~\ref{alg:cover_tree_k-nearest}.



\begin{dfn}
	\label{dfn:knn_iteration_set}
	Let $R$ be a finite subset of a metric space $(X,d)$. 
	Let $\T(R)$ be a cover tree of Definition \ref{dfn:cover_tree_compressed} built on $R$ and let $q \in X$ be arbitrary point.
	Let $L(\T(R),q) \subseteq H(\T(R))$ be the set of all levels $i$ during iterations of lines~\ref{line:knnu:loop_begin}-\ref{line:knnu:loop_end} of Algorithm~\ref{alg:cover_tree_k-nearest} launched with inputs 
	$\T(R),q$. 
	If Algorithm~\ref{alg:cover_tree_k-nearest} reaches line \ref{line:knnu:qtoofar} at level 
	$\varrho \in L(\T(R),q)$, then we say that is \emph{special}. 
	We denote $\eta(i) = \min_{t} \{ t \in L(\T(R),q) \mid t > i\}$. 
	\bs
\end{dfn}

\noindent
Note that $\eta(i)$ of Definition \ref{dfn:knn_iteration_set} may be undefined. If $\eta(i)$ is defined, then by definition we have $\eta(i) \geq i + 1$.
Let $d_k(q,R)$ be the distance of $q$ to its $k$th nearest neighbor in $R$. 

\begin{algorithm}
	\caption{$k$-nearest neighbor search by a compressed cover tree}
	\label{alg:cover_tree_k-nearest}
	\begin{algorithmic}[1]
		\STATE \textbf{Input} : compressed cover tree $\T(R)$, a query point $q\in X$, an integer $ k \in \Z_{+} $
		\STATE Set $i \leftarrow l_{\max}(\T(R)) - 1$ and $\eta(l_{\max}-1) = l_{\max}$
		\STATE  Let $r$ be the root node of $\T(R)$. Set $R_{l_{\max}}=\{r\}$.
		\WHILE{$i \geq l_{\min}$} \label{line:knnu:loop_begin}
		\STATE Assign $\mathcal{C}_i(R_{\eta(i)}) \leftarrow  R_{\eta(i)} \cup \{a \in \Child(p) \text{ for some }p \in R_{\eta(i)} \mid l(a) = i \} $ \\ \COMMENT{Recall that $\Child(p)$ contains node $p$ } \label{line:knnu:dfn_C}
		\STATE Compute $\lambda = \lambda_k(q,\C_{i}(R_{\eta(i)}))$ from Definition~\ref{dfn:lambda-point} \label{line:knnu:dfnLambda} by 
		Algorithm \ref{alg:lambda}.
		\STATE Find $R_{i} = \{p \in \C_i(R_{\eta(i)}) \mid d(q,p) \leq d(q,\lambda) + 2^{i+2}\}$ \label{line:knnu:dfnRi}
		\IF {$d(q,\lambda) > 2^{i+2}$} \label{line:knnu:qtoofar:condition}
		\STATE Define list $S = \emptyset$
		\FOR{$p \in R_i$} \label{line:knnu:qtoofar:loop:start}
		\STATE Update $S$ by running Algorithm \ref{alg:cover_tree_k-nearest_final_collection} on $(p,i)$ \label{line:knnu:qtoofar:launch} 
		\ENDFOR  \label{line:knnu:qtoofar:loop:end}
		\STATE Compute and \textbf{output} $k$-nearest neighbors of the query point $q$ from set $S$.
		\label{line:knnu:qtoofar}
		\ENDIF \label{line:knnu:qtoofar:condition:endif}
		\STATE Set $j \leftarrow \max_{ a \in R_{i}} \nxt(a,i,\T(R))$
		\COMMENT{If such $j$ is undefined, we set $j = l_{\min}-1$} \label{line:knnu:dfnindexj}
		\STATE Set $\eta(j) \leftarrow i$ and $i \leftarrow j$.
		\ENDWHILE \label{line:knnu:loop_end}
		\STATE Compute and \textbf{output} $k$-nearest neighbors of query point $q$ from the set $R_{l_{\min}}$.
		\label{line:knnu:final_line}
	\end{algorithmic}
\end{algorithm}

\begin{algorithm}
	\caption{The node collector called in line~\ref{line:knnu:qtoofar:launch} of Algorithm~\ref{alg:cover_tree_k-nearest}.}
	\label{alg:cover_tree_k-nearest_final_collection} 
	\begin{algorithmic}[1]
		\STATE \textbf{Input: }$p \in R$, index $i$.
		\STATE \textbf{Output: } a list $S \subseteq R$ containing all nodes of $\Sd_i(p,\T(R))$.
		\STATE Add $p$ to list $S$. 
		\IF {$i > l_{\min}(\T(R))$}
		\STATE Set $j = \nxt(p,i,\T(R))$ 
		\STATE Set $C =\{a \in \Child(p) \mid l(a) = j \}$
		\FOR{$u \in C$}
		\STATE Call Algorithm \ref{alg:cover_tree_k-nearest_final_collection} with $(u,j)$.
		\ENDFOR
		\ENDIF
	\end{algorithmic}
\end{algorithm}

\begin{exa}[Simulated run of Algorithm \ref{alg:cover_tree_k-nearest}]
	\label{exa:simulatedRun}
	Let $R$ and $\T(R)$ be as in Example \ref{exa:cover_tree_big}. Let $q = 0$ and $k = 5$. Figures \ref{fig:iteration3goodexample}, \ref{fig:iteration2goodexample},  \ref{fig:iteration1goodexample} and \ref{fig:iteration0goodexample} illustrate simulated run of Algorithm \ref{alg:cover_tree_k-nearest} on input $(\T(R), q, k)$. Recall that $l_{\max} = 2$ and $l_{\min} = -1$. During the iteration $i$ of Algorithm \ref{alg:cover_tree_k-nearest} we maintain the following coloring: Points in $R_i$ are colored orange. Points $\C_{\eta(i)}(R_{\eta(i)})$ (of line 5) that are not contained in $R_i$ are colored yellow. The $\lambda$-point of line \ref{line:knnu:dfnLambda} is denoted by using purple color. All the nodes that were present in $R_{\eta(i)}$ , but are no longer included in $R_i$ will be colored red. Finally all the points that are selected as $k$-nearest neighbors of $q$ are colored green in the final iteration. Nodes that haven't been yet visited or that will never be visited are colored white. Let $R_{2} = \{8\}$. Consider the following steps:

	\smallskip
	
	\noindent
	\textbf{Iteration} $i = 1$:  Figure \ref{fig:iteration3goodexample} illustrates iteration $i = 1$ of the Algorithm \ref{alg:cover_tree_k-nearest}. In line \ref{line:knnu:dfn_C} we find
	$\C_2(R_2) = \{4,8,12\}$. Since node $4$ minimizes distance $d(\C_2(R_2),0)$ and distinctive descendant set $\Sd_2(4, \T(R))$ consists of 7 elements we get $\lambda = 4$ and therefore $d(q,\lambda) = 4 \leq 2^{i+2} = 8$.
	In line \ref{line:knnu:dfnRi} we find $R_{1} = \{r \in C \mid d(0,r) \leq d(q,\lambda) + 2^{3} = 12\} = \{4,8,12\}$.
	
	\smallskip
	
	\noindent
	\textbf{Iteration} $i = 0$:   Figure \ref{fig:iteration2goodexample} illustrates iteration $i = 0$ of the Algorithm \ref{alg:cover_tree_k-nearest}. In line \ref{line:knnu:dfn_C}  we find 
	$\C_1(R_1) = \{2,4,6,8,10,12,14\}.$ Since $|\Sd_1(2, \T(R))| = 3$, $|\Sd_1(4, \T(R))| = 1$ and $|\T_1(6)| = 3$ and $6$ is the node with smallest to distance $0$ satisfying $\sum_{p \in N(0, 6) = \{2,4,6\}} | \Sd_1(p, \T(R))| \geq 5 = k.$ It follows that $\lambda = 6$.  In line \ref{line:knnu:dfnRi} we find $R_{0} = \{r \in \C(R_1) \mid d(0,r) \leq d(q,\lambda) + 2^{2} = 10\} = \{2,4,6,8,10\}$.
	Since $d(q,\lambda) > 2^{i+2} = 4$. We proceed into lines \ref{line:knnu:qtoofar:condition} - \ref{line:knnu:qtoofar:condition:endif}
	
	\smallskip
	
	\noindent
	\textbf{Final block} lines \ref{line:knnu:qtoofar:condition} - \ref{line:knnu:qtoofar:condition:endif} for  
	$i = 0$:  Figure \ref{fig:iteration1goodexample} marks all the 
	points $S$ discovered by line \ref{line:knnu:qtoofar:launch} as orange. Figure \ref{fig:iteration0goodexample} illustrates the final selection of $k$ points from set $S$ that are selected as the final output $\{1,2,3,4,5\}$.


\end{exa}

\begin{figure}[H]
	\centering
	\input images/CoverTreeLongExample/good_tree_example_1.tex
	\caption{Iteration $i = 1$ of simulation in Example \ref{exa:simulatedRun} of Algorithm \ref{alg:cover_tree_k-nearest} }
	\label{fig:iteration3goodexample}
\end{figure}
\begin{figure}[H]
	\centering
	\input images/CoverTreeLongExample/good_tree_example_2.tex
	\caption{Iteration $i = 0$ of simulation in Example \ref{exa:simulatedRun} of Algorithm \ref{alg:cover_tree_k-nearest} }
	\label{fig:iteration2goodexample}
\end{figure}

\begin{figure}[H]
	\centering
	\input images/CoverTreeLongExample/good_tree_example_33.tex
	\caption{Line \ref{line:knnu:qtoofar:launch} of Iteration $i = 0$ of simulation in Example \ref{exa:simulatedRun} of Algorithm \ref{alg:cover_tree_k-nearest} }
	\label{fig:iteration1goodexample}
\end{figure}

\begin{figure}[H]
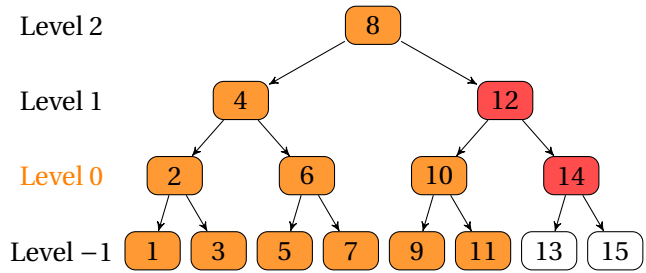

	\centering
	\input images/CoverTreeLongExample/good_tree_example_4.tex
	\caption{Line \ref{line:knnu:qtoofar} of iteration $i = 0$ of simulation in Example \ref{exa:simulatedRun} of Algorithm \ref{alg:cover_tree_k-nearest} }
	\label{fig:iteration0goodexample}
\end{figure}

Note that $\bigcup_{p \in R_i}\Sd_i(p, \T(R))$ is decreasing set for which $\bigcup_{p \in R_{l_{\max}}}\Sd_{l_{\max}}(p, \T(R)) = R$ and $$\bigcup_{p \in R_{l_{\min}}}\Sd_{l_{\min}}(p, \T(R)) = R_{l_{\min}}.$$

\begin{lem}[$k$-nearest neighbors in the candidate set for all $i$]
	\label{lem:cover_tree_knn_correct_lem}
	Let $R$ be a finite subset of an ambient metric space $(X,d)$, let $q \in X$ be a query point and let $k \in \Z \cap [1,\infty)$ be a parameter. Let $\T(R)$ be a compressed cover tree of $R$. Assume that $|R| \geq k$. Then for any iteration $i \in L(\T(R),q)$ of Definition \ref{dfn:knn_iteration_set} the candidate set $\bigcup_{p \in R_i}\Sd_i(p, \T(R))$ contains all $k$-nearest neighbors of $q$. \bs
	
\end{lem}
\begin{proof}
	
	Since $R_{l_{\max}} = \{r\}$, where $r$ is the root $\T(R)$ we have $S_{l_{\max}}(r,\T(R)) = R$ and therefore any point among $k$-nearest neighbor of $q$ is contained in $R_{l_{\max}}$. Let $i$ be the largest index for which there exists a point among $k$-nearest neighbor of $q$ that doesn't belong to $\bigcup_{p \in R_{i}}\Sd_i(p, \T(R))$. Let us denote such point by $\beta$, then:
	$$\beta \in \bigcup_{p \in R_{\eta(i)}}\Sd_{\eta(i)}(p, \T(R)) \setminus \bigcup_{p \in R_{i}}\Sd_{i}(p, \T(R)).$$ 
	By Lemma \ref{lem:child_set_equivalence} we have 
	\begin{ceqn}
		
		\begin{equation}
			\label{eqa:neighborsContained}
			\bigcup_{p \in \C_{\eta(i)}(R_{\eta(i)})}\Sd_{i}(p, \T(R)) = \bigcup_{p \in R_{\eta(i)}}\Sd_{\eta(i)}(p, \T(R))
		\end{equation}
		
	\end{ceqn}
	\noindent 
	Let $\lambda$ be as in line \ref{line:knnu:dfnLambda} of Algorithm \ref{alg:cover_tree_k-nearest}. By Equation (\ref{eqa:neighborsContained}) we have $$|\bigcup_{p \in \C_{\eta(i)}(R_{\eta(i)})}\Sd_{i}(p, \T(R))| \geq k,$$ therefore by Definition \ref{dfn:lambda-point} such $\lambda$ exists. Since $\beta \in \bigcup_{p \in \C_{\eta(i)}(R_{\eta(i)})}\Sd_{i}(p, \T(R))$, there exists $\alpha \in \C_{\eta(i)}(R_{\eta(i)})$ satisfying $\beta \in \Sd_{i}(\alpha, \T(R))$. By assumption it follows $\alpha \notin R_{i}$. By line \ref{line:knnu:dfnRi} of the algorithm we have
	\begin{ceqn}
		\begin{equation}
			\label{eqa:neighborsContained2}
			d(\alpha, q) > d(q, \lambda) + 2^{i+2}.
		\end{equation}
	\end{ceqn}
	Let $w$ be arbitrary point in set $\bigcup_{p \in N(q;\la)}\Sd_{i}(p, \T(R))$. Therefore $w \in \Sd_{i}(\gamma, \T(R))$ for some $\gamma \in  N(q;\la)$. By Lemma \ref{lem:distinctive_descendant_distance} applied on $i$ we have  $d(\gamma, w) \leq 2^{i+1}$. By Definition \ref{dfn:lambda-point}  since $\gamma \in  N(q;\la)$ we have $d(q,\gamma) \leq d(q,\lambda)$. By (\ref{eqa:neighborsContained2}) and the triangle inequality we obtain:
	\begin{ceqn}
		\begin{equation}
			\label{eqa:neighborsContained3}
			d(q,w) \leq d(q, \gamma) + d(\gamma,w) \leq d(q,\lambda) + 2^{i+1} < d(\alpha,q) - 2^{i+1} 
		\end{equation}
	\end{ceqn}
	On the other hand $\beta$ is a descendant of $\alpha$ thus we can estimate:
	\begin{ceqn}
		\begin{equation}
			\label{eqa:neighborsContained4}
			d(q,\beta) \geq d(q,\alpha) - d(\alpha,\beta) \geq d(\alpha,q) - 2^{i+1} 
		\end{equation}
	\end{ceqn}
	By combining Inequality (\ref{eqa:neighborsContained3}) with Inequality (\ref{eqa:neighborsContained4}) we obtain $d(q,w) < d(q,\beta)$. Since $w$ was arbitrary point from $\bigcup_{p \in N(q;\la)}\Sd_{i}(p, \T(R))$, that contains at least  $k$ points, $\beta$ cannot be any $k$-nearest neighbor of $q$, which is a contradiction. 
\end{proof}

\begin{thmm}[correctness of Algorithm~\ref{alg:cover_tree_k-nearest}]
	\label{thm:cover_tree_knn_correct}
	Algorithm~\ref{alg:cover_tree_k-nearest} correctly finds all $k$ nearest neighbors of query point $q$ within the reference set $R$. \bs
\end{thmm}
\begin{proof}
	
	Note that Algorithm~\ref{alg:cover_tree_k-nearest} is terminated by either reaching line \ref{line:knnu:final_line} or by going inside 
	block \ref{line:knnu:qtoofar:loop:start} - \ref{line:knnu:qtoofar:loop:end}.
	
	\medskip
	\noindent
	Assume first that Algorithm~\ref{alg:cover_tree_k-nearest} is terminated by reaching line \ref{line:knnu:final_line}.
	Claim follows directly from Lemma \ref{lem:cover_tree_knn_correct_lem} by noting that since 
	$i = l_{\min}$ all the nodes $p \in R_{l_{\min}}$ do not have any children. Therefore it follows $\bigcup_{p \in R_{l_{\min}}}\Sd_i(p, \T(R)) = R_{l_{\min}}$. Thus all the $k$-nearest neighbors of $q$ are contained in the set $R_{l_{\min}}$.
	
	\medskip
	\noindent
	Assume then that block \ref{line:knnu:qtoofar:loop:start} - \ref{line:knnu:qtoofar:loop:end} is reached during some iteration 
	$i \in L(\T(R),q)$. By Lemma \ref{lem:cover_tree_knn_correct_lem} set $\bigcup_{p \in R_i}\Sd_i(p, \T(R))$ contains all $k$-nearest neighbors of $q$. Note that in line \ref{line:knnu:qtoofar:launch} we collect all nodes of $\bigcup_{p \in R_i}\Sd_i(p, \T(R))$ into single array $S$. Therefore in line \ref{line:knnu:qtoofar} we correctly select $k$ nearest neighbors of $q$ from array $S$, which proves the claim. 
\end{proof}

\begin{lem}
	\label{lem:knn:time}
	Algorithm~\ref{alg:cover_tree_k-nearest} has
	the following time complexities of its lines
	\smallskip
	
	\noindent
	(a) 
	$\max\{\li{\ref{line:knnu:loop_begin}-\ref{line:knnu:qtoofar:condition}}, \li{\ref{line:knnu:qtoofar:condition:endif}-\ref{line:knnu:loop_end}} , \li{\ref{line:knnu:final_line}}\} = O\big(c_m(R)^{10} \cdot \log_2(k)\big)$;
	\smallskip
	
	\noindent
	(b) 
	$\li{\ref{line:knnu:qtoofar:condition}-\ref{line:knnu:qtoofar:condition:endif} } = O\big(|\bar{B}(q, 5 d_k(q,R))| \cdot \log_2(k)\big).$
	\bs
\end{lem}
\begin{proof}[\textbf{Proof of Lemma}~\ref{lem:knn:time}]
	\textbf{(a)}
	Let $\varrho \in L(\T(R),q)$ be as in Definition \ref{dfn:knn_iteration_set}. 
	Note that if iteration $\varrho$ is encountered, it becomes the last iteration of $L(\T(R),q)$.
	The total number of children encountered in line \ref{line:knnu:dfn_C} during single iteration (\ref{line:knnu:loop_begin}-\ref{line:knnu:loop_end}) is at most is at most 
	$(c_m(R))^4 \cdot \max\limits_{i \in L(\T(R),q) \setminus \varrho}|R_i|$
	by Lemma \ref{lem:compressed_cover_tree_width_bound}. From Lemma \ref{lem:time_lambdapoint} we obtain that line \ref{line:knnu:dfnLambda}, which launches Algorithm \ref{alg:lambda} takes at most
	$$|\C(R_i)| \cdot \log_2(k) = (c_m(R))^4 \cdot \max\limits_{L(q,\T(R)) \setminus \varrho} |R_i| \cdot \log_2(k) $$ time. 
	Line \ref{line:knnu:dfnRi} never does more work than line \ref{line:knnu:dfn_C}, since in the worst case scenario $R_{\eta(i)}$ is copied to $R_{i}$ in its current form. Line \ref{line:knnu:dfnindexj} handles $|R_{i}|$ nodes, since we can keep track of value of $\nxt(a,i,\T(R))$ of Definition \ref{dfn:implementation_compressed_cover_tree} by updating it when necessary in line \ref{line:knnu:dfn_C} we can retrieve its value in $O(1)$ time. Therefore maximal run-time of line \ref{line:knnu:dfnindexj} is $\max\limits_{i \in L(q,\T(R)) \setminus \varrho}|R_i|$.  Final line \ref{line:knnu:final_line}  picks lowest $k$-elements $R_{\eta(i)}$ ranked by function 
	$f(p) = d(p,q)$. By Lemma \ref{lem:time_k_smallest_elements} it can be computed in time 
	$O(\log_2(k) \cdot \max\limits_{L(q,\T(R)) \setminus \varrho}|R_i|)$.  
	It follows that 
	\begin{ceqn}
		\begin{equation} \label{eqa:thmimpeqadx}
			\centering
			\max(\li{\ref{line:knnu:loop_begin},\ref{line:knnu:qtoofar:condition}}, \li{\ref{line:knnu:qtoofar:condition:endif},\ref{line:knnu:loop_end}} , \li{\ref{line:knnu:final_line}}) = O\Big(c_m(R)^4 \cdot \max_{i \in L(q,\T(R)) \setminus \varrho}|R_i| \cdot \log_2(k) \Big )
		\end{equation}
	\end{ceqn}
	Let us now bound $\max_{i \in L(q,\T(R)) \setminus \varrho}|R_i|$, by showing $|R_i| \leq c_m(R)^6$. 
	Let $C_i$ be the $i$th level of $\T(R)$ as in Definition \ref{dfn:cover_tree_compressed}.
	For all $i \in L(\T(R),q) \setminus \varrho$ we have:
	\begin{ceqn}
		\begin{align}
			\label{eqa:ModifiedQBoundOne}
			R_{i} &= \{r \in \C_i(R_{\eta(i)}) \mid d(p,q) \leq d(q,\lambda) + 2^{i+2}\} \\
			&= B(q,d(q,\lambda)+2^{i+2}) \cap \C_i(R_i) \\
			&\subseteq B(q,2^{i+3}) \cap C_{i} 
			\label{eqa:ModifiedQboundTwo}
		\end{align}
	\end{ceqn}
	From cover-tree condition we know that all the points in $C_{i}$ are separated by $2^{i}$.
	We will now apply Lemma \ref{lem:packing} with $t = 2^{i+3}$ and $\delta = 2^{i}$.
	Since $4\frac{t}{\delta} + 1 = 2^5 + 1 \leq 2^6$ we obtain 
	$\max\limits_{i \in L(q,\T(R)) \setminus \varrho}|R_{i}| \leq |B(q,2^{i+2} ) \cap C_{i}| \leq c_m(R)^6$. The claim follows by replacing $\max\limits_{i \in L(q,\T(R)) \setminus \varrho}|R_{i}|$ with $c_m(R)^6$ in (\ref{eqa:thmimpeqadx}).

	
	
	\medskip
	
	\noindent
	\textbf{(b)}
	Let us now bound the run-time of $\li{\ref{line:knnu:qtoofar:condition}, \ref{line:knnu:loop_end}}$.
	which runs Algorithm \ref{alg:cover_tree_k-nearest_final_collection} for all $(p,i)$, where $p \in R_i$.
	Let $\Sd$ be a distinctive descendant set from Definition \ref{dfn:distinctive_descendant_set}. 
	Algorithm \ref{alg:cover_tree_k-nearest_final_collection} visits every node $u \in \cup_{p \in R_{i}}\Sd_{i}(p, \T(R))$ once, therefore its running time is $O(\cup_{p \in R_{i}}|\Sd_{i}(p, \T(R))|)$. Let us now show that 
	$$\cup_{p \in R_{i}}\Sd_{i}(p, \T(R)) \subseteq \bar{B}(q, 5d_k(q,R))$$
	Note first that by Lemma \ref{lem:cover_tree_knn_correct_lem} set $\cup_{p \in R_{i}}\Sd_{i}(p, \T(R))$ contains all $k$-nearest neighbors of $q$. Using Lemma \ref{lem:beta_point} we find $\beta$ among $k$-nearest neighbors of $q$ satisfying $d(q,\lambda) \leq d(q,\beta) + 2^{i+1}$. From assumption It follows $2^{i+1} \leq  d(q,\beta)$ .
	\medskip
	
	\noindent
	By line \ref{line:knnu:qtoofar:condition}
	we have $d(q, \lambda) \leq 2^{i+1}$.
	By line \ref{line:knnu:qtoofar} we perform depth-first traversal on $$A = \cup_{p \in R_i}\Sd_i(p, \T(R)).$$ 
	Let $u \in \cup_{p \in R_i}\Sd_i(p, \T(R))$ be arbitrary node and let $v \in R_i$ be such that $u \in \Sd_i(v,\T(R))$.
	By Lemma \ref{lem:distinctive_descendant_distance} we have $d(u,v) \leq 2^{i+1}$. Since $v \in R_i$ we have 
	$d(q,v) \leq d(\lambda,q) + 2^{i+2}$. By triangle inequality
	$$d(u,q) \leq d(u,v) + d(v,q) \leq  2^{i+1} + d(\lambda,v) + 2^{i+2}  \leq 2^{i+1} + 2^{i+1} + d(q,\beta) + 2^{i+2} \leq 5 \cdot d(q,\beta)$$
	It follows that $\cup_{p \in R_i}\Sd_i(p, \T(R)) \subseteq \bar{B}(q,5 \cdot d(q,\beta))$.
	Let us now bound the time complexity of line \ref{line:knnu:qtoofar}. 
	By Lemma \ref{lem:time_k_smallest_elements} for any set $A$ is takes $\log(k) \cdot |A|$ time to select $k$-lowest elements. We have:
	$$\li{\ref{line:knnu:qtoofar:condition}, \ref{line:knnu:loop_end}} = O(|\bar{B}(q,5 \cdot d_k(q,R))| \cdot \log(k)). $$
\end{proof}

\begin{thmm}
	\label{thm:cover_tree_knn_general_time}
	Let $R$ be a finite set in a metric space $(X,d)$, $c_m(R)$ be the minimized constant from Definition \ref{dfn:expansion_constant}.
	Given a compressed cover tree $\T(R)$, Algorithm~\ref{alg:cover_tree_k-nearest} finds all $k$ nearest neighbors of a query point $q\in X$ in time 
	$$O\Big ( \log_2(k) \cdot ((c_m(R))^{10}  \cdot |L(q,\T(R))| + |\bar{B}(q, 5 d_k(q,R)) |)\Big ),\text{ where }$$
	$L(\T(R),q) \subseteq H(\T(R))$ is the set of all levels $i$ in lines~\ref{line:knnu:loop_begin}-\ref{line:knnu:loop_end} of Algorithm~\ref{alg:cover_tree_k-nearest}. 
\end{thmm}
\begin{proof}
	Apply Lemma~\ref{lem:knn:time} to estimate the time complexity of Algorithm~\ref{alg:cover_tree_k-nearest}: \\ 
	$O\big( |L(\T(R),q)| \cdot 
	(\li{\ref{line:knnu:loop_begin}-\ref{line:knnu:qtoofar:condition}}
	+ \li{\ref{line:knnu:qtoofar:condition:endif}-\ref{line:knnu:loop_end}}  
	+ \li{\ref{line:knnu:final_line}}) 
	+\li{\ref{line:knnu:qtoofar:condition}-\ref{line:knnu:qtoofar:condition:endif} }\big)$.
\end{proof}

\noindent 
Corollary \ref{cor:cover_tree_knn_miniziminzed_constant_time} gives a run-time bound using only minimized expansion constant $c_m(R)$,
where if $R \subset \R^{m}$, then $c_m(R) \leq 2^{m}$. Recall that $\Delta(R)$ is aspect ratio of $R$ introduced in 
Definition \ref{dfn:radius+d_min}.

\begin{cor}
	\label{cor:cover_tree_knn_miniziminzed_constant_time}
	Let $R$ be a finite set in a metric space $(X,d)$. 
	Given a compressed cover tree $\T(R)$, Algorithm~\ref{alg:cover_tree_k-nearest} finds all $k$ nearest neighbors of $q$ in time $$O\Big ((c_m(R))^{10} \cdot \log_2(k) \cdot \log_2(\Delta(R)) + |\bar{B}(q, 5d_k(q,R))| \cdot \log_2(k) \Big ).$$
\end{cor}
\begin{proof}
	Replace $|L(q,\T(R))|$ in the time complexity of Theorem \ref{thm:cover_tree_knn_general_time} by its upper bound from
	Lemma~\ref{lem:depth_bound}: $|L(q,\T(R))| \leq |H(\T(R))| \leq \log_2(\Delta(R))$.
\end{proof}

\noindent
If we are allowed to use the standard expansion constant, that corresponds to KR-dimension of \cite{krauthgamer2004navigating}, then we obtain a stronger result, Theorem \ref{thm:knn_KR_time}.




\begin{lem}
	\label{lem:upper_bound_to_points_beloning_to_reference_set}
	Let $R$ be a finite reference set in a metric space $(X,d)$ and let $q \in X$ be a query point.
	Let $\varrho$ be the \emph{special} level of $L(\T(R),q)$. Let $i \in L(\T(R),q) \setminus \varrho$ be any level. 
	Then if $p \in R_i$ we have $d(p,q) \leq 2^{i+3}$.
\end{lem}
\begin{proof}
	By assumption in this part of the algorithm we have $d(q, \lambda) \leq 2^{i+2}$. By line \ref{line:knnu:dfnRi} of Algorithm \ref{alg:cover_tree_k-nearest}, since $p \in R_i$ we have $d(p,q) \leq d(q,\lambda) + 2^{i+2} \leq 2^{i+2} + 2^{i+2} \leq 2^{i+3}$, which proves the claim.
\end{proof}

\begin{lem}
	\label{lem:lower_bound_to_points_not_belonging_to_reference_set}
	Let $R$ be a finite reference set in a metric space $(X,d)$ and let $q \in X$ be a query point.
	Let $\varrho$ be the \emph{special} level of $L(\T(R),q)$. Let $i \in L(\T(R),q) \setminus \varrho$ be any level. 
	Then if $p \in \C_{i}(R_{\eta(i)}) \setminus R_i$, we have $d(p,q) > 2^{i+2}$.
\end{lem}
\begin{proof}
	By assumption $p \in \C_{i}(R_{\eta(i)}) \setminus R_i$.
	By line \ref{line:knnu:dfnRi} of Algorithm \ref{alg:cover_tree_k-nearest}
	it follows that $d(q,p) > 2^{i+2} + d(q,\lambda) \geq 2^{i+2}$.
	Therefore $d(q,p) > 2^{i+2}$, which proves the claim. 
\end{proof}

\begin{lem}
	\label{lem:knn_next_level_finder_for_log_depth_two}
	Let $i$ be a non-minimal level of $L(\T(R),q)$ of Definition \ref{dfn:knn_iteration_set}. Assume that $t = \eta(\eta(i+3))$ is defined.
	Then there exists $p \in R$ satisfying $2^{i+2} < d(p,q) \leq 2^{t+4}$.
\end{lem}
\begin{proof}
	
	
	
	
	Note first that since $\eta(i+3) \in L(\T(R),q)$, there exists distinct 
	$u \in R_{\eta(\eta(i+3))}$ and $v \in \C_{\eta(i+3)}(R_{\eta(\eta(i+3)}))$, in such a way that $u$ is the parent of $v$. 
	Let us show that both of $u,v$ cant belong to set $R_i$. Assume contrary that both $u,v \in R_i$. Then by Lemma \ref{lem:upper_bound_to_points_beloning_to_reference_set} we have
	$d(v,q) \leq 2^{i+3}$ and $d(u,q) \leq 2^{i+3}$. By triangle inequality $d(u,v) \leq d(u,q) + d(q,v) \leq 2^{i+4} \leq 2^{\eta(i+3)}$.
	Recall that we denote a level of a node by $l$.
	On the other hand we have $l(u) \geq \eta(i+3)$ and $l(v) \geq \eta(i+3)$, by separation condition of Definition \ref{dfn:cover_tree_compressed} we have $d(u,v) > 2^{\eta(i+3)}$, which is a contradiction. Therefore only one of $\{u,v\}$ 
	can belong to $R_i$. It sufficies two consider the two cases below: 
	
	\medskip
	
	\noindent
	\textbf{Assume that }$v \notin R_i$. Since $v$ is children of $u$ we have $d(u,v) \leq 2^{\eta(i+3) + 1}$.
	By Lemma \ref{lem:upper_bound_to_points_beloning_to_reference_set} we have $d(u,q) \leq 2^{\eta(\eta(i+3)) + 3}$.
	By triangle inequality 
	$$d(v,q) \leq d(v,u) + d(u,q) \leq 2^{\eta(\eta(i+3)) + 3} + 2^{\eta(i+3) + 1} \leq 2^{\eta(\eta(i+3)) + 4}$$
	Since $v \notin R_i$ there exists level $t$ having $\eta(i+3) \geq t \geq i$ and $v \in \C_{t}(R_{\eta(t)}) \setminus R_t$.
	Therefore by Lemma \ref{lem:lower_bound_to_points_not_belonging_to_reference_set} we have $d(q,v) > 2^{t+2} \geq 2^{i+2}$.
	It follows that we have found point $v \in R$ satisfying $2^{i+2} < v \leq 2^{\eta(\eta(i+3)) + 4}$. Therefore $p = v$, is the desired point.
	
	\medskip
	
	\noindent
	\textbf{Assume that }$u \notin R_i$. Since $u \in R_{\eta(\eta(i+3))}$, by Lemma \ref{lem:upper_bound_to_points_beloning_to_reference_set} we have $d(u,q) \leq 2^{\eta(\eta(i+3)) + 3}$.
	On the other hand since $u \notin R_i$, there exists level $t$ having $\eta(i+3) \geq t \geq i$ and $u \in \C_{t}(R_{\eta(t)}) \setminus R_t$. Therefore by Lemma \ref{lem:lower_bound_to_points_not_belonging_to_reference_set} we have $d(q,u) > 2^{t+2} \geq 2^{i+2}$.
	It follows that we have found point $u \in R$ satisfying $2^{i+2} < u \leq 2^{\eta(\eta(i+3)) + 4}$. Therefore $p = u$, is the desired point.

\end{proof}

\begin{lem}
	\label{lem:knn_depth_bound}
	Algorithm \ref{alg:cover_tree_k-nearest} executes lines \ref{line:knnu:loop_begin}-\ref{line:knnu:loop_end} the following number of times: $|L(\T(R),q)| = O(c(R \cup \{q\})^2 \cdot \log_2(|R|))$.
	\bs
\end{lem}
\begin{proof}[\textbf{Proof of Lemma}~\ref{lem:knn_depth_bound}]
	Let $x \in L(\T(R),q)$ be the lowest level of $L(\T(R),q)$.
	Define $s_1 = \eta(\eta(x)+1)$ and let $s_i = \eta(\eta(\eta(s_{i-1}+3))+3)$, if it exists. Assume that $s_{n+1}$ is the last sequence element for which $\eta(\eta(\eta(s_{n-1}+3))+3)$ is defined. Define $S = \{s_1,...,s_{n}\}$. For every $i \in \{1,...,n\}$ let $p_i$ be the point provided by Lemma \ref{lem:knn_next_level_finder_for_log_depth_two} that satisfies $$ 2^{s_i+2} < d(p_i,q) \leq 2^{\eta(\eta(s_{i}+3)) + 4}.$$
	Let $P$ be the sequence of points $p_i$.  Denote $n = |P| = |S|$. 
	Let us show that $S$ satisfies the conditions of Lemma \ref{lem:growth_bound_extension}. Note that:
	$$4 \cdot d(p_i,q)\leq 4 \cdot 2^{\eta(\eta(s_{i}+3)) + 4} \leq 2^{\eta(\eta(s_{i}+3)) + 6} \leq 2^{\eta(\eta(\eta(s_{i}+3))+3) + 2} \leq 2^{s_{i+1}+2} \leq d(p_{i+1},q)$$
	By Lemma \ref{lem:growth_bound_extension} applied for $A = R \cup q$ and sequence $P$ we get:
	$$|\bar{B}(q,\frac{4}{3}  d(q,p_n))| \geq (1+\frac{1}{c(R)^2})^{n} \cdot |\bar{B}(q,\frac{1}{3} d(q,p_1))|$$
	Since $\eta(x) \in L(\T(R),q)$ , there exists some point $u \in R_{\eta(x)}$. 
	By Lemma \ref{lem:upper_bound_to_points_beloning_to_reference_set} we have $d(u,q) \leq 2^{\eta(x) + 3}$. 
	Also $2^{\eta(\eta(x) + 1)+1} \leq \frac{2^{\eta(\eta(x) + 1)+2}}{3} < \frac{d(q,p_1)}{3}$
	It follows that:
	$$1 \leq  |\bar{B}(q, 2^{\eta(x) + 3})| \leq |\bar{B}(q, 2^{\eta(\eta(x) + 1)} + 1)| \leq |\bar{B}(q, \frac{d(q,p_1)}{3})|$$
	Therefore we have
	$$|R| \geq \frac{|\bar{B}(q,\frac{4}{3} \cdot d(q,p_n))|}{|\bar{B}(q,\frac{1}{3} \cdot d(q,p_1))|} \geq (1+\frac{1}{c(R \cup \{q\})^2})^{n}$$
	Note that $c(R \cup \{q\}) \geq 2$ by definition of expansion constant. Then by applying $\log$ and by using Lemma \ref{lem:hard_function_bound} we obtain: $c(R \cup \{q\})^2\log(|R|) \geq n = |S|$. 
	Let $x$ be minimal level of $L(\T(R),q)$ and let $y$ be the maximal level of $L(\T(R),q)$ 
	Note that $S$ is a sub sequence of $L$ in such a way that:
	\begin{itemize}
		\item $[x,s_1] \cap L(\T(R),q) \leq 3$, 
		\item for all $i \in 1,..., n$ we have $[s_i, s_{i+1}] \cap L(\T(R),q) \leq 10 $
		\item $[s_n, y] \cap L(\T(R),q) < 20$
	\end{itemize}
	Since segments $[x,s_1],[s_1,s_2], ...,  [s_2,s_n], [s_n,y]$ cover $|L(\T(R),q)|$,
	it follows that $|S| \geq \frac{|L(\T(R),q)|}{20}$. We obtain that $$|L(\T(R),q)| \leq 20 \cdot c(R \cup \{q\})^2 \cdot \log_2(|R|),$$ which proves the claim.
\end{proof}

\begin{thmm}
	\label{thm:knn_KR_time}
	Let $R$ be a finite reference set in a metric space $(X,d)$. Let $q\in X$ be a query point, $c(R \cup \{q\})$ be the expansion constant of $R \cup \{q\}$ and $c_m(R)$ be the minimized expansion constant from Definition \ref{dfn:expansion_constant}. Given a compressed cover tree $\T(R)$, Algorithm~\ref{alg:cover_tree_k-nearest} finds all $k$ nearest neighbors of $q$ in time 
	$$O\Big ( c(R \cup \{q\})^2 \cdot \log_2(k) \cdot \big((c_m(R))^{10}  \cdot \log_2(|R|) + c(R \cup \{q\}) \cdot k\big) \Big).$$
\end{thmm}
\begin{proof}
	By Theorem~\ref{thm:cover_tree_knn_general_time} the required time complexity is
	$$O\Big ((c_m(R))^{10} \cdot \log_2 (k) \cdot |L(q,\T(R))| + |\bar{B}(q, 5d(q,\beta)) | \cdot \log_2(k) \Big )$$
	for some point $\beta$ among the first $k$-nearest neighbors of $q$.
	Apply Definition \ref{dfn:expansion_constant}:
	\begin{ceqn}
		\begin{align}
			|B(q,5d(q,\beta))| \leq (c(R \cup \{q\}))^3 \cdot |B(q,\frac{5}{8}d(q,\beta))|
		\end{align}
	\end{ceqn}
	Since $|B(q,\frac{5}{8}d(q,\beta))| \leq k$, we have $|B(q,5d(q,\beta))| \leq (c(R \cup \{q\}))^3  \cdot k$.
	It remains to apply Lemma \ref{lem:knn_depth_bound}: $|L(q,\T(R))| = O(c(R \cup \{q\})^2 \cdot \log_2(|R|))$.
\end{proof}

\noindent
Corollary~\ref{cor:cover_tree_knn_time} combines Theorem~\ref{thm:construction_time_KR} with Theorem~\ref{thm:knn_KR_time}, to show that
Problem~\ref{pro:knn} can be resolved in $O(c^{O(1)} \cdot \log(k) \cdot \max\{|Q|, |R|\} \cdot (\log(|R|)) + k)$ time. 

\begin{cor}[solution to Problem \ref{pro:knn}]
	\label{cor:cover_tree_knn_time}
	In the notations of Theorem~\ref{thm:knn_KR_time}, 
	set $c = \max\limits_{q \in Q}c(R \cup \{q\})$.
	Algorithms~\ref{alg:cover_tree_k-nearest_construction_whole} and \ref{alg:cover_tree_k-nearest} solve Problem \ref{pro:knn} in time 
	$$O\Big( \max(|Q|,|R|) \cdot  c^2 \cdot \log_2(k) \cdot \big ((c_m(R))^{10}  \cdot \log_2(|R|) + c \cdot k \big ) \Big).$$
\end{cor}
\begin{proof}
	For any $q \in Q$, since $\log_2(R \cup \{q\}) \leq 2\log_2(R)$,
	a tree $\T(R)$ can be built in time $$O(c^2 \cdot c_m(R)^8 \cdot \log(|R|))$$ by Theorem~\ref{thm:construction_time_KR}.
	Therefore the time complexity is dominated by running Algorithm~\ref{alg:cover_tree_k-nearest} on all points $q \in Q$.
	The final complexity is obtained by multiplying the time from Theorem~\ref{thm:knn_KR_time} by $|Q|$.
\end{proof}


\section{Approximate $k$-nearest neighbor search}
\label{sec:approxknearestneighbor}




The original navigating nets and cover trees were used in \cite[Theorem~2.2]{krauthgamer2004navigating} and \cite[Section~3.2]{beygelzimer2006cover} to solve the $(1+\epsilon)$-approximate nearest neighbor problem for $k=1$.  
The main result, Theorem~\ref{thm:approximate_k_nearestneighbors} justifies a near linear parameterized complexity to find approximate a $k$-nearest neighbor set $\mathcal{P}$ formalized in Definition \ref{dfn:ApproxKNearestNeighbor}. 

\begin{dfn}[approximate $k$-nearest neighbor set $\AP$]
	\label{dfn:ApproxKNearestNeighbor}
	Let $R$ be a finite reference set and let $Q$ be a finite query set of a metric space $(X,d)$.
	Let $q \in Q \subseteq X$ be a query point, $k \geq 1$ be an integer and $\epsilon > 0$ be a real number. 
	Let $\mathcal{N}_k = \cup_{i=1}^k \NN_i(q)$ be the union of neighbor sets from Definition~\ref{dfn:kNearestNeighbor}. 
	A set $\mathcal{P} \subseteq R$ is called an \emph{approximate $k$-nearest neighbors set}, if $|\mathcal{P}| = k$ and there is an injection $f: \mathcal{P} \rightarrow \mathcal{N}_k$ satisfying $d(q, p) \leq (1+\epsilon) \cdot d(q,f(p)) $ for all $p \in \mathcal{P}$.  
	\bs
\end{dfn}


\begin{algorithm}[ht]
	\caption{This algorithm finds approximate $k$-nearest neighbor of Definition \ref{dfn:ApproxKNearestNeighbor}.}
	\label{alg:cover_tree_k-nearest_approximate}
	\begin{algorithmic}[1]
		\STATE \textbf{Input} : compressed cover tree $\T(R)$, a query point $q\in X$, an integer $ k \in \Z_{+}$, real $\epsilon \in \R_{+}$.
		\STATE Set $i \leftarrow l_{\max}(\T(R)) - 1$ and $\eta(l_{\max}-1) = l_{\max}$. Set $R_{l_{\max}}=\{\text{root}(\T(R))\}$.
		\WHILE{$i \geq l_{\min}$} \label{line:aknn:loop_begin}
		\STATE Assign $\mathcal{C}_i(R_{\eta(i)}) \leftarrow  R_{\eta(i)} \cup \{a \in \Child(p) \text{ for some }p \in R_{\eta(i)} \mid l(a) = i \}$. 
		\STATE Compute $\lambda = \lambda_k(q,\C_{i}(R_{\eta(i)}))$ from Definition~\ref{dfn:lambda-point} \label{line:aknn:dfnLambda} by 
		Algorithm \ref{alg:lambda}.
		\STATE Find $R_{i} = \{p \in \C_i(R_{\eta(i)}) \mid d(q,p) \leq d(q,\lambda) + 2^{i+2}\}$ \label{line:aknn:dfnRi}.
		\IF {$\frac{2^{i+2}}{\epsilon} + 2^{i+1}  \leq d(q, \lambda)$}\label{line:aknn:ifcondition}
		\STATE Let $\mathcal{P} = \emptyset$. 
		\FOR {$p \in \C_i(R_{\eta(i)})$}
		\IF {$d(p,q) < d(q,\lambda)$}
		\STATE $\mathcal{P} = \mathcal{P} \cup \Sd_{i}(p,\T(R))$
		\ENDIF
		\ENDFOR
		\STATE Fill $\mathcal{P}$ until it has $k$ points by adding points from sets $\Sd_{i}(p,\T(R))$, where $d(p,q) = d(q, \lambda)$.
		\STATE \textbf{return} $\mathcal{P}$. 
		\ENDIF \label{line:aknn:endifcondition}

		\STATE Set $j \leftarrow \max_{ a \in R_{i}} \nxt(a,i,\T(R))$.
		\COMMENT{If such $j$ is undefined, we set $j = l_{\min}-1$} \label{line:aknn:dfnindexj}
		\STATE Set $\eta(j) \leftarrow i$ and $i \leftarrow j$.
		\ENDWHILE \label{line:aknn:loop_end}
		\STATE Compute and \textbf{output} $k$-nearest neighbors of query point $q$ from the set $R_{l_{\min}}$.
		\label{line:aknn:final_line}
	\end{algorithmic}
\end{algorithm}




Definition \ref{dfn:knn_iteration_set_approx} is analog of Definition \ref{dfn:knn_iteration_set} for $(1+\epsilon)$-approximate $k$-nearest neighbor search. 
\begin{dfn}[Iteration set of approximate $k$-nearest neighborhood search]
	\label{dfn:knn_iteration_set_approx}
	Let $R$ be a finite subset of a metric space $(X,d)$. 
	Let $\T(R)$ be a cover tree of Definition \ref{dfn:cover_tree_compressed} built on $R$ and let $q \in X$ be an arbitrary point.
	Let $L(\T(R),q) \subseteq H(\T(R))$ be the set of all levels $i$ during iterations of lines~\ref{line:aknn:loop_begin}-\ref{line:aknn:loop_end} of Algorithm~\ref{alg:cover_tree_k-nearest_approximate} launched with inputs $(\T(R),q)$. 
	We denote $\eta(i) = \min_{t} \{ t \in L(\T(R),q) \mid t > i\}$. 
	\bs
	
\end{dfn}

\begin{lem}[$k$-nearest neighbors in the candidate set for all $i$]
	\label{lem:cover_tree_knn_correct_lem_approx}
	Let $R$ be a finite subset of an ambient metric space $(X,d)$, let $q \in X$ be a query point , let $k \in \Z \cap [1,\infty)$ and $\epsilon \in \R_{+}$ be parameters. Let $\T(R)$ be a compressed cover tree of $R$. Assume that $|R| \geq k$. Then for any iteration $i \in L(\T(R),q)$ of Algorithm \ref{alg:cover_tree_k-nearest_approximate} the candidate set $\bigcup_{p \in R_i}\Sd_i(p, \T(R))$ contains all $k$-nearest neighbors of $q$. \bs
\end{lem}
\begin{proof}
	Proof of this lemma is similar to Lemma \ref{lem:cover_tree_knn_correct_lem_approx} and is therefore omitted.
\end{proof}

\noindent 
Lemma \ref{lem:approximate_knn_correctness} shows that Algorithm \ref{alg:cover_tree_k-nearest_approximate} correctly returns an Approximate $k$-nearest neighbor set of Definition \ref{dfn:ApproxKNearestNeighbor}.

\begin{lem}[Correctness of Algorithm \ref{alg:cover_tree_k-nearest_approximate}]
	\label{lem:approximate_knn_correctness}
	
	Algorithm \ref{alg:cover_tree_k-nearest_approximate}
	finds an approximate $k$-nearest neighbors set of any query point $q \in X$. 
	\bs
\end{lem}
\begin{proof}
	
	Assume first that condition on line \ref{line:aknn:ifcondition} of Algorithm \ref{alg:cover_tree_k-nearest_approximate}
	is satisfied during some iteration $i \in H(\T(R))$ of Algorithm \ref{alg:cover_tree_k-nearest_approximate}. Let us denote
	$$\mathcal{A} = \bigcup_{p \in \C_i(R_{\eta(i)})} \{\Sd_{i}(p,\T(R)) \mid d(p,q) < d(q,\lambda) \}, 
	\mathcal{B} = \bigcup_{p \in \C_i(R_{\eta(i)})} \{\Sd_{i}(p,\T(R)) \mid d(p,q) = d(q,\lambda) \}.$$
	By Algorithm \ref{alg:cover_tree_k-nearest_approximate} set $\mathcal{P}$ contains all points of $\mathcal{A}$ and rest of the points are filled form $\mathcal{B}$.
	We will now form $f: \mathcal{P} \rightarrow \mathcal{N}_k$ by mapping every point $p \in \mathcal{A} \cap \mathcal{P}$ into itself and then by extending $f$ to be injective map on whole set $\mathcal{P}$ . The claim holds trivially for all points $p \in \mathcal{A} \cap \mathcal{P}$. Let us now consider points  $p \in \mathcal{P} \setminus \mathcal{A}$. Let $\gamma \in \C_i(R_{\eta(i)})$ be such that $p \in \Sd_i(\gamma, \T(R))$ and let $\psi \in \C_i(R_{\eta(i)})$ be such that $f(p) \in  \Sd_i(\psi, \T(R)) $. By using triangle inequality, Lemma \ref{lem:compressed_cover_tree_descendant_bound} and the fact that
	$p \in \mathcal{A} \cup \mathcal{B}$ we obtain:
	\begin{ceqn}
		\begin{equation}
			\label{eqa:ANNCorrectness1}
			d(q, p) \leq  d(q, \gamma) + d(\gamma,p) \leq d(q, \lambda) + 2^{i+1}
		\end{equation}
	\end{ceqn} 
	On the other hand since $f(p) \notin \mathcal{A}$ we have
	\begin{ceqn}
		\begin{equation}
			\label{eqa:ANNCorrectness2}
			(1+\epsilon) \cdot d(q, f(p)) \geq (1+\epsilon) \cdot ( d(q, \psi) - d(\psi,f(p))) \geq (1+\epsilon) \cdot (d(q, \lambda) - 2^{i+1})
		\end{equation}
	\end{ceqn}
	Note that by line \ref{line:aknn:ifcondition} we have $\frac{2^{i+2}}{\epsilon} + 2^{i+1} \leq d(q, \lambda)$. It follows that 
	$2^{i+2} \leq \epsilon \cdot d(q, \lambda) - \epsilon \cdot 2^{i+1}$.
	Therefore we have:
	\begin{ceqn}
		\begin{equation}
			\label{eqa:ANNCorrectness3}
			d(q, \lambda) + 2^{i+1}  \leq d(q,\lambda) + 2^{i+2} - 2^{i+1} \leq (1+\epsilon) \cdot (d(q,\lambda) - 2^{i+1})
		\end{equation}
	\end{ceqn}
	By combining Equations (\ref{eqa:ANNCorrectness1}) - (\ref{eqa:ANNCorrectness3}) we obtain $d(q, p) \leq (1+\epsilon) \cdot d(q,f(p)) $.
	If the condition on line \ref{line:aknn:ifcondition} of Algorithm \ref{alg:cover_tree_k-nearest_approximate} is never satisfied, then the Algorithm finds real $k$-nearest neighbors of point $q$ in the end of the algorithm and therefore the claim holds. 
	
\end{proof}

\begin{thmm}[correctness of modified Algorithm~\ref{alg:cover_tree_k-nearest} ]
	\label{thm:approximate_k_nearestneighbors}
	In the notations of Definition \ref{dfn:ApproxKNearestNeighbor}, the complexity of Algorithm \ref{alg:cover_tree_k-nearest_approximate} is
	$O\Big(  (c_m(R))^{8 + \lceil \log(2 + \frac{1}{\epsilon}) \rceil}  \cdot \log_2(k) \cdot \log_2(\Delta(R)) + k \Big ).$
	\bs
\end{thmm}
\begin{proof}
	Similarly to Lemma \ref{lem:knn:time} it can be shown that Algorithm  \ref{alg:cover_tree_k-nearest_approximate}  is bounded by: 
	\begin{ceqn}
		\begin{equation}
			\label{eqa:boundapproximate}
			O((c_m(R))^4 \cdot \log_2(k) \cdot \max_i|R_i| \cdot |H(\T(R))| + \li{\ref{line:aknn:ifcondition} - \ref{line:aknn:endifcondition}})
		\end{equation}
	\end{ceqn}
	
    \noindent
	Note first that in lines \ref{line:aknn:ifcondition} - \ref{line:aknn:endifcondition} we loop over set $\C_i(R_{\eta(i)})$ and select $k$ points from it. Therefore $\li{\ref{line:aknn:ifcondition} - \ref{line:aknn:endifcondition}} = k + |\C_i(R_{\eta(i)})|$.

	\medskip
	
	\noindent
	Let us now bound the size of $R_i$. By line \ref{line:aknn:ifcondition} of Algorithm \ref{alg:cover_tree_k-nearest_approximate} either Algorithm \ref{alg:cover_tree_k-nearest_approximate} is launched that terminates the program or $\frac{2^{i+2}}{\epsilon} + 2^{i+1}  > d(q, \lambda)$. Let $C_i$ be the $i$th cover set of $\T(R)$.
	To bound $|R_i|$ we can assume the latter. Similarly to Theorem \ref{thm:knn_KR_time} we have:
	\begin{ceqn}
		\begin{align}
			\label{eqa:QBoundOne1}
			R_{i} &= \{r \in \C_i(R_{\eta(i)}) \mid d(p,q) \leq d(q,\lambda) + 2^{i+2}\} \\
			&= \bar{B}(q,d(q,\lambda)+2^{i+2}) \cap \C_i(R_{\eta(i)}) \\
			&\subseteq \bar{B}(q,d(q,\lambda)+2^{i+2}) \cap C_{i} \\
			&\subseteq \bar{B}(q,2^{i+2}(\frac{3}{2} + \frac{1}{\epsilon})) \cap C_{i} 
			\label{eqa:QboundTwo1}
		\end{align}
	\end{ceqn}
	Since the cover set $C_{i}$ is a $2^{i}$-sparse subset of the ambient metric space $X$, we can apply Lemma~\ref{lem:packing} with $t = 2^{i+2}(\frac{3}{2} + \frac{1}{\epsilon})$ and $\delta = 2^{i}$. 
	Since $4\frac{t}{\delta} + 1 = 2^4(\frac{3}{2} + \frac{1}{\epsilon}) + 1 \leq 2^4(2 + \frac{1}{\epsilon})$, we get $\max |R_i| \leq (c_m(R))^{4 + \lceil \log_2(2 + \frac{1}{\epsilon}) \rceil}$. 
	The final complexity is obtained by plugging the upper bound of $|R_i|$ above into (\ref{eqa:boundapproximate}).
\end{proof}

\begin{cor}[complexity for approximate $k$-nearest neighbors set $\AP$]
	\label{cor:approximate_k_nearestneighbors}
	In the notations of Definition \ref{dfn:ApproxKNearestNeighbor}, an approximate $k$-nearest neighbors set is found for all $q \in Q$ in time
	$O\Big( |Q| \cdot  (c_m(R))^{8 + \lceil \log(2 + \frac{1}{\epsilon}) \rceil} \cdot \log(k) \cdot \log_2(\Delta(R)) + |Q| \cdot k \Big ).$
	\bs
\end{cor}
\begin{proof}
	This corollary follows directly from Theorem \ref{thm:approximate_k_nearestneighbors} .
\end{proof}

\section{Counterexamples for NIPS 2009 dual-tree complexity analysis}
\label{sec:challenges_paired_tree}

This section contains counterexamples for \cite{ram2009linear}.
In 2009 \cite[Theorem~3.1]{ram2009linear} revisited the time complexity for all 1st nearest neighbors and claimed the upper bound $O(c(R)^{12}c(Q)^{4\kappa}\max\{|Q|,|R|\})$, where $c(Q),c(R)$ are expansion constants of the query set $Q$ and reference set $R$.
The degree of bichromaticity $\kappa$ is a parameter of both sets $Q,R$, see \cite[Definition~3.1]{ram2009linear}.
We have found the following issues.
\medskip

\noindent 
First, Counterexample~\ref{cexa:dualtreecode} shows that \cite[Algorithm~1]{ram2009linear} when $Q=R$ returns the same point for all $q\in Q$, which shows that \cite[Algorithm~1]{ram2009linear} cannot be used to find non-trivial neighbors in this case.
Second, Remark~\ref{rem:kappa} explains several possible interpretations of \cite[Definition~3.1]{ram2009linear} for the parameter $\kappa$. 
Third, \cite[Theorem~3.1]{ram2009linear} similarly to \cite[Theorem~5]{beygelzimer2006cover} relied on the same estimate of recursions in the proof of \cite[Lemma~4.3]{beygelzimer2006cover}.
Counterexample~\ref{cexa:dualtreeproof} explains step-by-step why the proof of the time complexity result of \cite[Algorithm~1]{ram2009linear} is incorrect and requires a clearer definition of $\kappa$.

\medskip

\noindent
In 2015 Curtin \cite{curtin2015plug} made an attempt to fix the issues of \cite{ram2009linear} by introducing other parameters:
the imbalance $I_t$ in \cite[Definition~3]{curtin2015plug} and $\theta$ in \cite[Definition~4]{curtin2015plug}.
These parameters measured extra recursions that occurred due to possible imbalances in trees built on $Q,R$, which was missed in the past. 
\cite[Theorem~2]{curtin2015plug} shows that, for constructed cover trees on a query set $Q$ and a reference set $R$, Problem~\ref{pro:knn} for $k=1$ (only 1st nearest neighbors) can be solved in time 
$$O\Big(c^{O(1)}\big(|R| + |Q| + I_t + \theta\big)\Big).\eqno{(*)},$$
where $c$ is expansion constant that depends on $Q$ and $R$.
The problem with this approach is that in worst case $I_t$ is quadratic $O(|R|^2)$. To make the time complexity linear, we would have to show
$I_t = O(c^{O(1)} \cdot \max\{|R| , |Q|\})$. However, no such result exist at the moment. 

\begin{algorithm}
	\caption{
		Original \cite[Algorithm~1]{ram2009linear} is analyzed in Counterexamples~\ref{cexa:dualtreecode} and~\ref{cexa:dualtreeproof}.
	}
	\label{alg:cover_tree_k-nearest_dt_original}
	\begin{algorithmic}[1]
		\STATE \textbf{Function} FindAllNN(a node $q_j\in T(Q)$, a subset $R_i$ of a cover set $C_i$ of $T(R)$).
		\IF {$i = -\infty$}
		\STATE for each $q_j \in L(q_j)$ \textbf{return} $\text{argmin}_{r \in R_{-\infty}} d(q,r)$
		\STATE \COMMENT{here $L(q_j)$ is the set of all descendants of the node $q_j$}
		\ELSIF{$j < i$}
		\STATE $\C(R_i) = \{\text{Children}(r) \mid r \in R_i\} $ \COMMENT{in original pseudo-code the notation is $R = \C(R_i)$}
		\STATE $R_{i-1} = \{r \in R \mid d(q_j,r) \leq d(q_j, R) + 2^{i} + 2^{j+2} \}$
		\STATE FindAllNN($q_{j-1}, R_i$) \COMMENT{ $q_{j-1}$ is the same point as $q_j$ on one level below}
		\ELSE{}
		\STATE for each $p_{j-1} \in \text{Children}(q_j)$ FindAllNN($p_{j-1},R_{i}$)
		\ENDIF
	\end{algorithmic}
\end{algorithm}

\medskip

\noindent
The step-by-step execution of Algorithm~\ref{alg:cover_tree_k-nearest_dt_original} will show that the number of reference expansions has a lower bound $O(\max\{|Q|,|R|\}^2)$.  Recall that \cite[End of Section~1]{ram2009linear} defined the all-nearest-neighbor problem as follows. 
"\textbf{All Nearest-neighbors:} For all queries $q \in Q$ find $r^{*}(q) \in R$ such that $r^{*}(q) = \mathrm{argmin}_{r \in R}d(q,r)"$. 
For $Q=R$, the last formula produces trivial self-neighbors.
\medskip

\noindent
In original Algorithm~\ref{alg:cover_tree_k-nearest_dt_original}, the node $q_j$ has a level $j$, a reference subset $R_i\subset R$ is a subset of $C_i$ for an implicit cover tree $T(R)$. 
The algorithm is called for a pair $q_j, R_{i} = \{r\}$, where $q_j$ is the root of the query tree at the maximal level $j = l_{\max}(T(Q))$,
and $r$ is the root of the reference tree at the maximal level $i = l_{\max}(T(R))$. 
Split Algorithm~\ref{alg:cover_tree_k-nearest_dt_original} into these blocks:
\smallskip

\noindent
lines 2-4 : FinalCandidates,
\smallskip

\noindent
lines 5-9 : reference expansion,
\smallskip

\noindent
lines 9-11 : query expansion.

\begin{cexa}
	\label{cexa:dualtreecode}
	In the notations of  Example \ref{exa:tall_imbalanced_tree}, $m$ is a parameter of $R$.
	Build a compressed cover tree $\T(R)$ as in Figure \ref{fig:bad_cover_tree}. 
	Set $Q = R$. 
	First we show that Algorithm \ref{alg:cover_tree_k-nearest_dt_original} returns the trivial neighbor for every point for $\T(Q)=\T(R)$. 
	\medskip
	
	\noindent
	We start the simulation with the query node $r$ on the level $m^2+1$, which has the reference subset $R_{m^2+1} = \{r\}$. 
	The query node and the reference set are at the same levels, so we run the query expansions (lines 9-11). 
	The node $r$ has $p_{m^2}$ and $r$ as its children.
	Hence the algorithm goes into the branches that have $p_{m^2}$ as the query node and into the branches that have $r$ as the query node.
	\medskip
	
	\noindent
	Let us focus on all recursions having $p_{m^2}$ as the query node. In the first recursion involving the node $p_{m^2}$, we have $i = m^2+1, j = m^2$. Thus $j < i$ and we run reference expansions (lines 5-9). 
	The node $r$ has two children at the level $m^2$, so $\C(R_i) = \{p_{m^2}, r\}$ . Since
	$d(p_{m^2},p_{m^2}) = 0$ and $d(p_{m^2}, r) = 2^{m^2+1}$ on line 7, we have:
	$$R_{m^2}  = \{r \in \C(R_i) \mid d(q_j,r) \leq 2^{m^2 + 1} + 2^{m^2+2} \} = \{p_{m^2}, r\}.$$
	Similarly, for $i = m^2, j = m^2-1, q_j = p_{m^2}$, we have $\C(R_i) =  \{p_{m^2}, p_{m^2-1}, r\}$ and since $d(p_{m^2 - 1}, p_{m^2}) = 2^{m^2}$ and $d(r,p_{m^2}) = 2^{m^2 + 1}$ we have:
	$$R_{m^2-1} =  \{r \in \C(R_i) \mid d(q_j,r) \leq 2^{m^2} + 2^{m^{2} + 1} \} = \{p_{m^2}, p_{m^2-1}\}. $$
	For $i = m^2-1, j = m^2-2, q_j = p_{m^2}$, we have $\C(R_i) =  \{p_{m^2}, p_{m^2-1}, p_{m^2-2}\}$.
	Since $d(p_{m^2 - 1}, p_{m^2}) = 2^{m^2}$ and $d(p_{m^2-2},p_{m^2}) = 2^{m^2} + 2^{m^2-1}$, we have:
	$$R_{m^2-2} =  \{r \in \C(R_i) \mid d(q_j,r) \leq 2^{m^2 - 1} + 2^{m^{2}} \} = \{p_{m^2}, p_{m^2-1}, p_{m^2-2}\}.$$
	Finally, for $i = m^2-2, j = m^2-3, q_j = p_{m^2}$, we have $\C(R_i) =  \{p_{m^2}, p_{m^2-1},p_{m^2-2}, p_{m^3-3}\}$ and 
	$d(p_{m^2}, p_{m^3-3}) = 2^{m^2} + 2^{m^2-1} + 2^{m^2-2}$. The previous inequalities imply that
	$$R_{m^2-3} =  \{r \in \C(R_i) \mid d(q_j,r) \leq 2^{m^2-2} + 2^{m^{2}-1} \} = \{p_{m^2}\}.$$
	Since  $R_{m^2 -3} = \{p_{m^2}\}$, the nearest neighbor of $p_{m^2}$ will be chosen to be $p_{m^2}$. 
	The same argument can be repeated for all $p_t \in R$. 
	It follows that Algorithm \ref{alg:cover_tree_k-nearest_dt_original} finds trivial nearest neighbor for every point $p_t \in R$. 
	\bs
\end{cexa}







\begin{exa}
	\label{exa:dualtreeprooffaultyconst}
	To avoid the issue of finding trivial nearest neighbors in Counterexample~\ref{cexa:dualtreecode}, we will modify Example~\ref{exa:tall_imbalanced_tree}. 
	For any integer $m > 100$, let $G$ be a metric graph that has $2$ vertices $r$ and $q$ and $2m-1$ edges $\{e_{0}\} \cup \{e_{1}, ..., e_{m-1},h_{1},..., h_{m-1}\}$.
	The edge-lengths are $|e_i| = 2^{i \cdot m+2}$ and $|h_i| = 2^{i \cdot m+2}$  for all $i \in [1,m]$, finally $|e_{0}| = 1$.
	\medskip
	
	\noindent
	For every $i \in \{1, ..., m^2\}$, if $i$ is divisible by $m$, we set $q_{i}$ to be the middle point of $e_{i / m}$ and $r_{i}$ to be the middle point of $h_{i / m}$. 
	For every other $i$ not divisible by $m$, we define $q_i$ to be the middle point of segment $(q_{i+1}, q)$ and $r_i$ to be the middle point of segment $(r_{i+1}, q)$.
	\medskip
	
	\noindent
	Let $d$ be the shortest path metric on the graph $G$. 
	Then $d(q_i,r) = d(q_i, q) + 1 = 2^{i+1} + 1$ , $d(q_i, r_j) = 2^{i+1} + 2^{j+1}$ and $d(q,r) = 1$. 
	Let $R = \{r, r_{m^2}, r_{m^2-1}, ..., r_1\}$ and let $Q = \{r,  q_{m^2}, q_{m^2-1}, ..., q_1 \}$. 
	Let compressed cover trees $\T(Q),\T(R)$ have the same structure as the compressed cover tree $\T(R)$ in Example~\ref{exa:tall_imbalanced_tree}.
	\bs
\end{exa}

\begin{rem}
	\label{rem:kappa}
	\cite[Definition~3.1]{ram2009linear} introduced the degree of bichromaticity $\kappa$ as follows:
	\medskip
	
	\noindent
	"\textbf{Definition 3.1} Let $S$ and $T$ be cover trees built on query set $Q$ and reference set $R$ respectively.
	Consider a dual-tree algorithm with the property that the scales of $S$ and $T$ are kept as close as
	possible – i.e. the tree with the larger scale is always descended. Then, the degree of bichromaticity
	$\kappa$ of the query-reference pair $(Q, R)$ is the maximum number of descends in $S$ between any two
	descends in $T$".
	\medskip
	
	\noindent
	There are at least two different interpretations of this definition. Our best interpretation is that $\kappa$ is the maximal number of levels in $T$ containing at least one node between any two consecutive levels of $S$. However, if $q$ is a leaf node of $S$, but there are still many levels between level of $q$ and $l_{\min}(T)$, it is not clear from the definition if $\kappa$ includes these levels.
	\medskip
	
	\noindent
	\cite[page~3284]{curtin2015plug} pointed out that "
	Our results are similar to that of Ram et al. (2009a), but those results depend on a
	quantity called the constant of bichromaticity, denoted $\kappa$, which has unclear relation to
	cover tree imbalance. The dependence on $\kappa$ is given as $c_q^{4\kappa}$ , which is not a good bound,
	especially because $\kappa$ may be much greater than 1 in the bichromatic case (where $S_q = S_r$)".
\end{rem}

\noindent 
To keep track of the indices $i,j$ the function call FindAllNN($q_j$, $R_i$) will be expressed as\linebreak
FindAllNN($i,j,q_j, R_i$) in Counterexample \ref{cexa:dualtreecode}.

\begin{cexa}[Counterexample to {\cite[Theorem~3.1]{ram2009linear}} ]
	\label{cexa:dualtreeproof}
	We will now show that in addition to the problems in the pseudocode the proof of \cite[Theorem~3.1]{ram2009linear} is incorrect.  Let us consider the following quote from its proof. 
	\medskip
	
	\noindent
	\emph{" \textbf{Theorem 3.1} Given a reference set $R$ of size $N$ and expansion constant $c_R$, a query set $Q$ of size
	$O(N)$ and expansion constant $c_Q$, and bounded degree of bichromaticity $\kappa$ of the $(Q,R)$ pair, the
	FindAllNN subroutine of Algorithm 1 computes the nearest neighbor in $R$ of each point in $Q$ in
	$O(c^{12}_Rc^{4\kappa}_QN)$ time.}
	\medskip
	
	\noindent
	\emph{[\emph{Part of the proof}]
	Since at any level of recursion, the size of $R$ [Corresponding to $\C(R_i)$ in Algorithm \ref{alg:cover_tree_k-nearest_dt_original} ] is bounded by $c_R^4\max_i{R_i}$ (width bound), and the
	maximum depth of any point in the explicit tree is $O(c^2_R \log(N))$ (depth bound), the number of nodes
	encountered in Line 6 is $O(c_R^{6} \max_i |R_i|\log(N))$. Since the traversal down the query tree causes
	duplication, and the duplication of any reference node is upper bounded by $c_Q^{4\kappa}$ , Line 6 [corresponds to line 8 in Algorithm \ref{alg:cover_tree_k-nearest_dt_original}] takes at most
	$c^{4\kappa}_Qc^6_R\max_i|R_i|\log(N)$ in the whole algorithm. }
	"
	\medskip
	
	\noindent
	\textbf{The above arguments claimed the algorithm runs Line 8 at most this number of times:}
	\begin{equation}
		\label{eqa:AuthorClaimingDualTreeKNN}
		\#(\text{Line 8}) \leq \max_{p \in R}D(p) \cdot \max_i \C(R_i)  \cdot (\text{number of duplications}).
	\end{equation}
	\textbf{It will be shown that cover tree $\T(R)$ from Example \ref{exa:dualtreeprooffaultyconst} does not satisfy the inequality above. }

\medskip	
	
	Let $X, \T(R), \T(Q), R, Q$ be as in Example \ref{exa:dualtreeprooffaultyconst} for some parameter $m$. We will consider the simulation of Algorithm \ref{alg:cover_tree_k-nearest_dt_original} on pair $(\T(Q),\T(R))$. We note first Lemma \ref{lem:tall_imbalanced_tree_explicit_depth} applied on $\T(R)$ provides $\max_{p \in R}D(p) \leq 2m+1$
	Similarly, as in Counterexample \ref{cexa:original_all_nearest_neighbors_algorithm}, a contradiction will be achieved by showing that $R_i$ and a set of its children $\C(R_i)$ will have a constant size bound on any recursion $(i,j)$ of Algorithm \ref{alg:cover_tree_k-nearest_dt_original}.
	\medskip
	
	\noindent
	Since $\T(R)$ contains at most one children on every level $i$ we have $|\C(R_i)| \leq |R_i| + 1$ for any recursion of FindAllNN algorithm. For any $i > m^2$ denote $r_i$ and $q_i$ to be $r$.
	Note first that since $l(q_t) = t$ for any $t \in [1,m^2]$, then $q_t$ is recursed into from FindAllNN($t+1,t+1,p,R_i$), where $p$ is parent node of $q_t$. Therefore it follows that $t \geq i + 1 $ in any stage of the recursion. 
	Let us prove that for any $i \in [1,m^2+1]$ following two claims hold: (1) Function FindAllNN($i$ , $j = i-1$, $q_t$, $R_i$) is called for all $t \geq i + 1$ and (2) We have  $R_i = \{r_{i+1}, r_{i}, r\}$ in this stage of the algorithm. The claim will be proved by induction on $i$. Let us first prove case $i = 2m+1$.  Note that  Algorithm \ref{alg:cover_tree_k-nearest_dt_original} is originally launched from FindAllNN($2m+1,2m+1,r, \{r\}$), 
	therefore the first claim holds. Second claim holds trivially since $r_{2m+2} = r$ and $r_{2m+1} = r$.
	
	\medskip
	
	\noindent
	Assume now that the claim holds for some $i$, let us show that the claim will always hold for $i-1$. Assume that FindAllNN($i , j = i-1, q_{t}, R_i)$ was called for some $t \geq i+1$. Since $j = i-1$, we perform a reference expansion (lines 5-9). By line $6$ and induction assumption we have $\C(R_i) = \{r, r_{i+1}, r_{i}, r_{i-1}\}$. Assume first that $q_t = r$.
	Recall that for any $u \in [1,m^2]$ we have $d(r,r_{u}) \geq 2^{u+1} $. It follows that
	$$R_{i-1} = \{ r' \in \C(R_i) \mid d(r, r') \leq  2^{i} + 2^{i+1}\} = \{r,r_{i},r_{i-1}\}.$$
	Let us now consider case $q_t \neq r$. We have $d(r,q_t) = 2^{t+1}$ and $d(q_t, r_{u+1}) = 2^{t+1} + 2^{u+2}$ for any $u \in [1,m^2+1]$. Therefore
	$$R_{i-1} = \{ r' \in C_{i-1} \mid d(q_t, r') \leq d(q_t,r) + 2^i + 2^{i+1} \leq 2^{t+1} + 2^{i} + 2^{i+1}\}.$$
	It follows that $R_{i-1} = \{r, r_{i}, r_{i-1}\}$. In both cases we proceed to line $8$ where we launch FindAllNN$(i-1 , i-1, q_{t}, R_{i-1})$. After proceeding into the recursion we have $j = i$ and therefore query-expansion (lines 9-11) will be performed.  Note that $q_{t}$ was chosen so that $t \geq i+1$. Since every $q_{t-1}$ is either a child of $r$ or $q_{t}$ it follows that FindAllNN$(i-1 , i-2, q_{t'}, R_{i-1})$ will be called for all $t' \geq t-1 \geq i$. Then condition (2) of the induction claim holds as well. 

	\medskip
	
	\noindent
	It remains to show that Algorithm \ref{alg:cover_tree_k-nearest_dt_original} $(q, R_i = \{r\})$ has $O(m^4)$ low bound on the number of times reference expansions (lines 5-9) are performed.
	Let $\xi$ be the number of times Algorithm \ref{alg:cover_tree_k-nearest_dt_original} performs 
	reference expansions.  For every $q' \in Q$ denote $\xi(q')$ to be the total number of reference expansions performed for $q'$. Recall that any query node $q'$ is introduced in the query expansion (lines 9-11) for parameters $(i = u+1, j = u+1, p, R_i)$, where $p$ is the parent node of $q'$. Since $R_i$ is non empty for all levels $[1,u]$ we have $\xi(q_u) \geq u - 1$ for all $u$. Thus 
	$$\xi = \sum_{q' \in Q} \xi(q') \geq \sum^{m^2+1}_{u = 2} u-2 = O(m^4).$$
	There are different interpretations for the number of duplications. Note that the query tree $\T(Q)$ has exactly one new child on every level and that trees $\T(Q)$ and $\T(R)$ contain exactly the same levels. By using the definitions the number of duplications should be $2$. However, since there can be other interpretations for the number of duplications, we make a rough estimate that the number of duplications is upper bounded by the number of nodes in query tree $O(m^2)$. By using Inequality (\ref{eqa:AuthorClaimingDualTreeKNN}), we obtain the following contradiction: 
	$$O(m^4) = \xi \leq  \max_{p \in R}D(p) \cdot (\max_i \C(R_i)) \cdot (\text{number of duplications}) \leq (2m+1) \cdot 4 \cdot m^2 \leq O(m^3).$$
\end{cexa}

\section{Discussions: current contributions and future steps}
\label{sec:Conclusions}

This chapter rigorously proved the time complexity of the exact $k$-nearest neighbor search.
The motivations were the past challenges in the proofs of time complexities in \cite[Theorem~2.7]{krauthgamer2004navigating}, \cite[Theorem~5]{beygelzimer2006cover}, \cite[Theorem~3.1]{ram2009linear}, \cite[Theorem~5.1]{march2010fast}.
Though \cite[Section~5.3]{curtin2015improving} pointed out some difficulties, no corrections were published.
Main Theorem~\ref{thm:knn_KR_time} and Corollary~\ref{cor:cover_tree_knn_time} fill the above gaps.
\medskip

\noindent
First, Definition~\ref{dfn:kNearestNeighbor} and Problem~\ref{pro:knn} rigorously deal with a potential ambiguity of $k$-nearest neighbors at equal distances, which wasn't discussed in the past work.
The main new data structure of a compressed cover tree in Definition~\ref{dfn:cover_tree_compressed} substantially simplifies the navigating nets \cite{krauthgamer2004navigating} and original cover trees  \cite{beygelzimer2006cover} by avoiding any repetitions of given data points.
This compression has substantially clarified the construction and search Algorithms~\ref{alg:cover_tree_k-nearest_construction_whole} and~\ref{alg:cover_tree_k-nearest}. 
\medskip

\noindent
Second, In Section \ref{sec:minimized_exp_constant} we showed that the newly defined minimized expansion constant $c_m$ of any finite subset $R$ of normed vector space on $\R^{n}$ has upper-bound $2^{m}$. In the future, it can be similarly shown that if $R$ is uniformly distributed then classical expansion constant $c(R)$ is $2^{m}$ as well. 

\medskip

\noindent
Third, the claims of original work \cite{beygelzimer2006cover} regarding near-linear time complexities of cover tree construction algorithm, as well as nearest neighbor search had incorrect proofs.
Counterexamples~\ref{cexa:construction_algorithm_of_original_cover_tree} and~\ref{cexa:original_all_nearest_neighbors_algorithm}.
provide concrete examples of datasets $R$ that break claimed proofs. 
Additionally, in 2009 \cite[Theorem~3.1]{ram2009linear}, there was an attempt to show that after the construction of all the relevant data structures, such as the implicit cover tree, Problem \ref{pro:knn} can be solved in a linear time $c(R)^{O(1)} \cdot |R|$. 
In Section \ref{sec:challenges_paired_tree} 
Counterexample \ref{cexa:dualtreeproof} shows that the number of total reference recursions was estimated similarly wrong as in the proof of \cite[Theorem~5]{beygelzimer2006cover}.

\medskip

\noindent
Fourth, the approach of \cite{beygelzimer2006cover} is corrected in Section \ref{sec:ConstructionCovertree} and \ref{cor:construction_time_KR}.
Assuming that expansion constants and aspect ratio of a reference set $R$ are fixed, Corollary \ref{cor:construction_time_KR} and Corollary \ref{cor:cover_tree_knn_time} rigorously show that the time complexities are linear in the maximum size of $R$ and a query set $Q$ and near-linear $O(k\log k)$ in the number $k$ of neighbors. 

\medskip

\noindent
Finally, In future, we wish to improve the time complexity result of $k$-nearest neighbor search to a purely linear-time 
$O(c(R)^{O(1)} \cdot |R|$ by constructing cover trees on both, query and reference sets. Since a similar approach of \cite{ram2009linear} was shown to have incorrect proof and \cite{curtin2015plug} contain additional parameters $I, \theta$, it requires significantly more effort to understand if $O(c(R)^{O(1)} \cdot |R|)$ is archivable by means of a compressed cover tree data structure.



%% file: images/cover-tree.txt
\begin{tikzpicture}[align=center, node distance = 1.0cm, scale = 0.45]
    
	\node (scaleminftext) {$C_{-\infty}$};
	\node [blockz, right of =scaleminftext] (scaleminf1) {1};
	\node [blockz, right of =scaleminf1] (scaleminf2) {2};
	\node [blockz, right of =scaleminf2] (scaleminf3) {3};
	\node [blockz, right of =scaleminf3] (scaleminf4) {4};
	\node [blockz, right of =scaleminf4] (scaleminf5) {5};

	\node [above of =scaleminftext](scalem1text) {$C_{-1}$};
	\node [blockz, above of =scaleminf1] (scalem11) {1};
	\node [blockz, above of =scaleminf2] (scalem12) {2};
	\node [blockz, above of =scaleminf3] (scalem13) {3};
	\node [blockz, above of =scaleminf4] (scalem14) {4};
	\node [blockz, above of =scaleminf5] (scalem15) {5};
	
	\node [above of =scalem1text](scale0text) {$C_{0}$};
	\node [blockz, above of =scalem11] (scale01) {1};
	\node [blockz, above of =scalem13] (scale03) {3};
	\node [blockz, above of =scalem15] (scale05) {5};	
	
	\node [above of =scale0text](scale1text) {$C_{1}$};
	\node [blockz, above of =scale01] (scale11) {1};
	\node [blockz, above of =scale05] (scale15) {5};
	
	\node [above of =scale1text](scale2text) {$C_{2}$};
	\node [blockz, above of =scale11] (scale21) {1};

	\node [above of =scale2text](scaleinftext) {$C_{\infty}$};
	\node [blockz, above of =scale21] (scaleinf1) {1};
    
      \draw[dashed,->] (scalem11) -> (scaleminf1);
      \draw[dashed,->] (scalem12) -> (scaleminf2);
      \draw[dashed,->] (scalem13) -> (scaleminf3);
      \draw[dashed,->] (scalem14) -> (scaleminf4);
      \draw[dashed,->] (scalem15) -> (scaleminf5);

	  \draw[->] (scale01) -> (scalem11);

	  \draw[->] (scale03) -> (scalem12); 
     \draw[dashed, ->] (scale03) -> (scalem13);

	  \draw[->] (scale05) -> (scalem15);

	  \draw[->] (scale03) -> (scalem14);
	  
	  \draw[dashed,->] (scaleinf1) -> (scale21);
	  \draw[->] (scale21) -> (scale11);
	  \draw[->] (scale15) -> (scale05);
	  
	  \draw[->] (scale11) -> (scale01);
	  \draw[->] (scale11) -> (scale03);
	  \draw[->] (scale21) -> (scale15);
      
\end{tikzpicture}

%% file: images/explicit_cover_tree.tex
\begin{tikzpicture}[align=center, node distance = 1.0cm, scale = 0.45]
    

	\node (scalem2text) {Level 2};
	\node[below of =scalem2text] (scalem1text) {Level 1};
	\node[below of =scalem1text] (scalem0text) {Level 0};
	\node[below of =scalem0text] (scalemm1text) {Level -1};
	\node [blockz,  right=25pt of scalem2text ] (node1) {1};
	\node [blockz,  right=100pt of scalem1text] (node5) {5};
	\node [blockz,  right=25pt of scalem0text] (node3) {3};
	\node [blockz,  right=5pt of scalemm1text] (node2) {2};
	\node [blockz,  right=50pt of scalemm1text] (node4) {4};

	  \draw[->] (node1) -> (node5);
	  \draw[->] (node1) -> (node3);
	  \draw[->] (node3) -> (node2);
	  \draw[->] (node3) -> (node4);

	
	
	

    

	  
	  
	  
      
\end{tikzpicture}

%% file: images/easy_tree_one.tex
\begin{tikzpicture}[align=center, node distance = 1.0cm, scale = 0.45]

	\node (scalem2text) {Level $i$};
	\node[below of =scalem2text] (scalem1text) {Level $i-1$};
	\node[below of =scalem1text] (scalem0text) {Level -1};
	\node [blockz,  right=30pt of scalem2text ] (node1) {$2^{i}$};
    \node [blockz,  right=25pt of scalem1text ] (node2) {1};
	\node [blockz,  right=32pt of scalem0text ] (node3) {0};
	

	  \draw[->] (node1) -> (node2);
	   \draw[->] (node2) -> (node3);

\end{tikzpicture}

%% file: images/easy_tree_two.tex
\begin{tikzpicture}[align=center, node distance = 1.0cm, scale = 0.45]

	\node (scalem2text) {Level $i$};
	\node[below of =scalem2text] (scalem1text) {Level $i-1$};
	\node[below of =scalem1text] (scalem0text) {Level -1};
	\node [blockz,  right=30pt of scalem2text ] (node1) {$2^{i}$};
    \node [blockz,  right=25pt of scalem1text ] (node2) {0};
	\node [blockz,  right=32pt of scalem0text ] (node3) {1};
	

	  \draw[->] (node1) -> (node2);
	   \draw[->] (node2) -> (node3);

\end{tikzpicture}

%% file: images/easy_tree_three.tex
\begin{tikzpicture}[align=center, node distance = 1.0cm, scale = 0.45]

	\node (scalem2text) {Level $i$};
	\node[below of =scalem2text] (scalem1text) {Level $i-1$};
	\node[below of =scalem1text] (scalem0text) {Level -1};
	\node [blockz,  right=30pt of scalem2text ] (node1) {0};
    \node [blockz,  right=50pt of scalem1text ] (node2) {$2^{i}$};
	\node [blockz,  right=32pt of scalem0text ] (node3) {1};
	

	  \draw[->] (node1) -> (node2);
	   \draw[->] (node1) -> (node3);

\end{tikzpicture}

%% file: images/CoverTreeLongExample/good_tree_example_0.tex
\begin{tikzpicture}[align=center, node distance = 1.0cm, scale = 0.45]

	\node (scale3) {Level $2$};
	\node[below of =scale3] (scale2) {Level 1};
	\node[below of =scale2] (scale1) {Level 0};
	\node[below of =scale1] (scale0) {Level $-1$};
	\node [blockzm1,  right=1pt of scale0 ] (node1) {1};
	\node [blockzm1,  right=26pt of scale0 ] (node3) {3};
	\node [blockzm1,  right=51pt of scale0 ] (node5) {5};
	\node [blockzm1,  right=76pt of scale0 ] (node7) {7};
	\node [blockzm1,  right=101pt of scale0 ] (node9) {9};
	\node [blockzm1,  right=126pt of scale0 ] (node11) {11};
	\node [blockzm1,  right=151pt of scale0 ] (node13) {13};
	\node [blockzm1,  right=176pt of scale0 ] (node15) {15};
		\node [blockzm1,  right=13pt of scale1 ] (node2) {2};
		\node [blockzm1,  right=63pt of scale1 ] (node6) {6};
		\node [blockzm1,  right=113pt of scale1 ] (node10) {10};
		\node [blockzm1,  right=163pt of scale1 ] (node14) {14};

		\node [blockzm1,  right=38pt of scale2 ] (node4) {4};
		\node [blockzm1,  right=138pt of scale2 ] (node12) {12};
		\node [blockzm1,  right=88pt of scale3 ] (node8) {8};
	

	
    \draw[->] (node2) -> (node1);
    \draw[->] (node2) -> (node3);
     \draw[->] (node6) -> (node5);
    \draw[->] (node6) -> (node7);
     \draw[->] (node10) -> (node9);
    \draw[->] (node10) -> (node11);
     \draw[->] (node14) -> (node13);
    \draw[->] (node14) -> (node15);

    \draw[->] (node8) -> (node4);
    \draw[->] (node8) -> (node12);

      \draw[->] (node4) -> (node2);
     \draw[->] (node4) -> (node6);
         
         \draw[->] (node12) -> (node10);
     \draw[->] (node12) -> (node14);

\end{tikzpicture}

%% file: images/tripleExample.tex
\begin{tikzpicture}[align=center, node distance = 1.0cm, scale = 0.45]

    \node (scaleinf) {$l = \infty$};
    \node [below = 12.5pt of scaleinf](scalemidinf) {};
     \node[below=25pt of scaleinf] (scale6) {$l = 5$};
    \node[below of =scale6] (scale5) {$l = 4$};
    \node[below of =scale5] (scale4) {$l = 3$};
	\node[below of =scale4] (scale3) {$l = 2$};
	\node[below of =scale3] (scale2) {$l = 1$};
	\node[below=25pt of scale2] (scale1) {$l = -\infty$};

 	\node [blockzm2,  right=70pt of scaleinf ] (node1x) {$r$};
 	\node [blockzm2, below=25pt of node1x ] (node1a) {$r$};
 	\node [blockzm2, below of = node1a ] (node1b) {$r$};
 	\node [blockzm2,  below of = node1b ] (node1c) {$r$};
 	\node [blockzm2,   below of = node1c ] (node1d) {$r$};
 	\node [blockzm2,  below of = node1d ] (node1e) {$r$};
 	\node [blockzm2,   below =25pt of node1e ] (node1f) {$r$};
 	
 	\draw[dashed, ->] (node1x) -> (node1a);
 	    \draw[->] (node1a) -> (node1b);
 	     \draw[->] (node1b) -> (node1c);
 	      \draw[->] (node1c) -> (node1d);
 	       \draw[->] (node1d) -> (node1e);
 	      \draw[dashed, ->] (node1e) -> (node1f);

 		\node [blockzm2, right=10pt of node1b ] (node2b) {$p_{4}$};
 	\node [blockzm2,  below of = node2b ] (node2c) {$p_{4}$};
 	\node [blockzm2,   below of = node2c ] (node2d) {$p_{4}$};
 	\node [blockzm2,  below of = node2d ] (node2e) {$p_{4}$};
 	\node [blockzm2,   below =25pt of node2e ] (node2f) {$p_{4}$};
 	
 	    \draw[->] (node1a) -> (node2b);
 	    
 	    \draw[->] (node2b) -> (node2c);
 	      \draw[->] (node2c) -> (node2d);
 	       \draw[->] (node2d) -> (node2e);
 	      \draw[dashed, ->] (node2e) -> (node2f);

 		\node [blockzm2, right = 10pt of node2c ] (node3c) {$p_{3}$};
 	\node [blockzm2,   below of = node3c ] (node3d) {$p_{3}$};
 	\node [blockzm2,  below of = node3d ] (node3e) {$p_{3}$};
 	\node [blockzm2,   below =25pt of node3e ] (node3f) {$p_{3}$};
 	
 	        \draw[->] (node2b) -> (node3c);
 	      \draw[->] (node3c) -> (node3d);
 	       \draw[->] (node3d) -> (node3e);
 	      \draw[dashed, ->] (node3e) -> (node3f);

 		\node [blockzm2,   left = 10 pt of  node1d ] (node4d) {$p_{2}$};
 	\node [blockzm2,  below of = node4d ] (node4e) {$p_{2}$};
 	\node [blockzm2,   below =25pt of node4e ] (node4f) {$p_{2}$};

            \draw[->] (node1c)   -> (node4d);
 	       \draw[->] (node4d) -> (node4e);
 	      \draw[dashed, ->] (node4e) -> (node4f);

 	\node [blockzm2,  left = 10 pt of node4e ] (node5e) {$p_{1}$};
 	\node [blockzm2,   below =25pt of node5e ] (node5f) {$p_{1}$};
 	
 	\draw[->] (node4d) -> (node5e);
 	 \draw[dashed, ->] (node5e) -> (node5f);
 	
 	\node[rectangle, draw, fill=white, 
    text width=1.0em, text centered, rounded corners, minimum height=50pt, minimum width = 1.0em, right = 220pt of scalemidinf] (ynode1a) {$r$};
    
    	\node[rectangle, draw, fill=white, 
    text width=1.0em, text centered, rounded corners, minimum height=50pt, minimum width = 1.0em, below = 10pt of ynode1a] (ynode1b) {$r$};
    
    	\node[rectangle, draw, fill=white, 
    text width=1.0em, text centered, rounded corners, minimum height=80pt, minimum width = 1.0em, below = 10pt of ynode1b] (ynode1c) {$r$};

    \node[blockzm2, above right = -18 pt and 10 pt of ynode1b] (ynode2a) {$p_4$};
    
    	\node[rectangle, draw, fill=white, 
    text width=1.0em, text centered, rounded corners, minimum height=110pt, minimum width = 1.0em, below = 10pt of ynode2a] (ynode2b) {$p_4$};
    
    \node[rectangle, draw, fill=white, 
    text width=1.0em, text centered, rounded corners, minimum height=110pt, minimum width = 1.0em, right = 10pt of ynode2b] (ynode3a) {$p_3$};

        \node[blockzm2, above left = -15 pt and 10 pt of ynode1c] (ynode4a) {$p_2$};
    
    	\node[rectangle, draw, fill=white, 
    text width=1.0em, text centered, rounded corners, minimum height=55pt, minimum width = 1.0em, below = 10pt of ynode4a] (ynode4b) {$p_2$};
    
    \node[rectangle, draw, fill=white, 
    text width=1.0em, text centered, rounded corners, minimum height=55pt, minimum width = 1.0em, left = 10pt of ynode4b] (ynode5a) {$p_1$};

    \draw[->] (ynode1a) -> (ynode1b);
    \draw[->] (ynode1b) -> (ynode1c);
    \draw[->] (ynode1a) -> (ynode2a);
    \draw[->] (ynode2a) -| (ynode3a);
    \draw[->] (ynode2a) -> (ynode2b);
    
      \draw[->] (ynode1b) -> (ynode4a);
     \draw[->] (ynode4a) -| (ynode5a);
     \draw[->] (ynode4a) -> (ynode4b);
    

 	\node[blockzm2, right = 285pt of scale6] (xnode1) {$r$};
	\node[blockzm2, below right = 10 pt and 10 pt of xnode1] (xnode2) {$p_4$};
 	\node[blockzm2, below right = 70 pt and 10 pt of xnode1] (xnode4) {$p_2$};
 	\node[blockzm2, below right = 10 pt and 8 pt of xnode4] (xnode5) {$p_1$};
 	\node[blockzm2, below right = 12 pt and 8 pt of xnode2] (xnode3) {$p_3$};
 	
 	\draw[->]	(xnode1) -> (xnode2);
 	\draw[->]	(xnode2) -> (xnode3);
 	\draw[->]	(xnode1) -> (xnode4);
 	\draw[->]	(xnode4) -> (xnode5);

\end{tikzpicture}

%% file: images/outlierconstruction.tex
 \begin{tikzpicture}[scale = 1.0, node distance = 1cm]


\foreach \Point in {(-2,1), (-1,1), (-2,0), (-1,0), (-2,-1) , (-1,-1) , (0,-1), (0,0), (0,1)}{
    \node[circle,fill=black,inner sep=1.5pt] at \Point {};
}


\node [circle,fill=red,inner sep=1.5pt, red, label = {$p$}] at (2.5,0) {};
\draw (2.5,0) circle[radius = 1.9];

\node at (-1,-1.5) {$R \setminus \{p\}$};

\end{tikzpicture}

%% file: images/packingLemmaIllustration.tex
 \begin{tikzpicture}[scale = 1.0]




\node [circle,fill=red, red, inner sep=1pt, label = {$p$}] at (0,0) {};
\draw (0,0) circle[radius = 2];

\node [circle,fill=black, black, inner sep=1pt] at (-1,-1) {};
\draw[line width = .5pt, dash pattern=on 1pt off 2pt] (-1,-1) circle[ radius = 0.5];

\node [circle,fill=black, black, inner sep=1pt] at (1,0.75) {};
\draw[line width = .5pt, dash pattern=on 1pt off 2pt] (1,0.75) circle[ radius = 0.5];

 \draw[dashed] (1,0.75) -- (1.5,0.75);
 \node at (1.25,0.6) {\tiny $\delta / 2$}; 

\node [circle,fill=black, black, inner sep=1pt] at (-0.1,-0.4) {};
\draw[line width = .5pt, dash pattern=on 1pt off 2pt] (-0.1,-0.4) circle[ radius = 0.5];

\node [circle,fill=black, black, inner sep=1pt] at (0.75,-1.25) {};
\draw[line width = .5pt, dash pattern=on 1pt off 2pt] (0.75,-1.25) circle[ radius = 0.5];

\node [circle,fill=black, black, inner sep=1pt] at (-0.2,1.2) {};
\draw[line width = .5pt, dash pattern=on 1pt off 2pt] (-0.2,1.2) circle[ radius = 0.5];

\node [circle,fill=black, black, inner sep=1pt] at (-1.1,0.65) {};
\draw[line width = .5pt, dash pattern=on 1pt off 2pt] (-1.1,0.65) circle[ radius = 0.5];

 \draw[dashed] (0,0) -- node[below] {$t$} node[above] {} ++ (2,0);


\node at (-1.5,-0.25) {$S$};

\end{tikzpicture}

%% file: images/MultiGraphExampleExtended.tex
\tikzset{markpos/.style args={#1 at #2}{decoration={
  markings,
  mark=at position #2 with {\coordinate(#1);}},postaction={decorate}}}
  
\begin{tikzpicture}
\node[circle,fill=black,inner sep=2pt,draw,label = below:{$r$}] (a) at (180:6cm) {};
\node[circle,fill=black,inner sep=2pt,draw,label = below:{$q$}] (b) at (0:6cm) {};
\draw[thick] (a) edge[bend left=90, markpos=mymark1 at 0.5] (b);
\draw[thick] (a) edge[bend left=90, markpos=edgemark1 at 0.30] (b);
\draw[thick] (a) edge[bend left=90, markpos=anothermark1 at 0.75] (b);

\draw[thick] (a) edge[bend left=50, markpos=mymark2 at 0.5] (b);
\draw[thick] (a) edge[bend left=50, markpos=edgemark2 at 0.25] (b);
\draw[thick] (a) edge[bend left=50, markpos=anothermark2 at 0.75] (b);

\draw[thick] (a) edge[bend left=13, markpos=mymark3 at 0.5] (b);
\draw[thick] (a) edge[bend left=13, markpos=edgemark3 at 0.33] (b);
\draw[thick] (a) edge[bend left=13, markpos=anothermark3 at 0.75] (b);

\draw[thick] (a) edge[markpos=edgemark4 at 0.33] (b);

\node[circle,fill=black,inner sep=2pt,draw,label = {$p_{m^2}$}] (pm1) at (mymark1) {};
\node[circle,fill=black,inner sep=2pt,draw,label = {$p_{m^2-m}$}] (pm2) at (mymark2) {};
\node[circle,fill=black,inner sep=2pt,draw,label = below:{$p_{m}$}] (p2) at (mymark3) {};

\node[circle,fill=black,inner sep=2pt,draw,label = {$p_{m^2-1}$}] (pm1z) at (anothermark1) {};
\node[circle,fill=black,inner sep=2pt,draw,label = {$p_{m^2-m-1}$}] (pm2z) at (anothermark2) {};
\node[circle,fill=black,inner sep=2pt,draw,label = below:{$p_{m-1}$}] (p2z) at (anothermark3) {};

\node[label = below:{$|e_{0}| = 1$}] (l1) at (edgemark4){};

\node[label = below:{$|e_{1}| = 2^{m + 2} $}] (l2) at (edgemark3){};

\node[label = {[rotate=20]above:$|e_{m-2}| = 2^{m^2 -m + 2} $}] (l2) at (edgemark2){};
\node[label = {[rotate=15]above:$|e_{m-1}| = 2^{m^2+2} $}] (l3) at (edgemark1){};

  \draw[loosely dotted, very thick] (mymark2) -- (mymark3);

\end{tikzpicture}

%% file: images/bad_tree_example_new_expanded.tex
\begin{tikzpicture}[align=center, node distance = 1.0cm, scale = 0.25]

\scriptsize

	\node (scalem2text) {Level $m^2 + 1$};

	\node[below of =scalem2text] (scalem1text) {Level $m^2 $};
	\node[below of =scalem1text] (scalemh01text) {Level $m^2 - 1$};
	\node[below = 15 pt of scalemh01text] (scalem11text) {Level $m \cdot (m-1) + 1$};

		\node[below = 30 pt of  scalem11text] (scalem1text2) {Level $i \cdot m + 1 $};
		
			\node[below of = scalem1text2] (scalem1text22) {Level $i \cdot m$};
		
	\node[below of =scalem1text22] (scalemh01text2) {Level $i \cdot m - 1$};
	\node[below = 15 pt of scalemh01text2] (scalem11text22) {Level $m \cdot (i-1) + 1$};

		\node [blockzflexi,  right=50pt of scalem2text ] (noderoot) {$r$};

	\node [blockzflexi,  right=180pt of scalem1text] (node1) {$p_{m^2}$};
	\node [blockzflexi, right=175pt of scalemh01text] (node2) {$p_{m^2-1}$};

			\node [blockzflexi, right=160pt of scalem11text] (node4) {$p_{m^2 - m +1}$};

				\node [blockzflexi,  right=164pt of scalem1text22] (node1d) {$p_{i \cdot m}$};
				
					\node [blockzflexi,  right=50pt of scalem1text2] (node1dcopyr) {$r$};

	\node [blockzflexi, right=155pt of scalemh01text2] (node2d) {$p_{i \cdot m-1}$};
			\node [blockzflexi, right=140pt of scalem11text22] (node4d) {$p_{i \cdot m - m +1}$};

	\node[below = 30 pt of scalem11text22] (scalem1textla) {Level  $m+1$};

		\node[below  of =  scalem1textla] (scalem11text2) {Level  $m$};
		
			\node [blockzflexi,  right=50pt of scalem1textla] (noderootcopy3) {$r$};

	\node[below of =scalem11text2] (scalemm1text) {Level $m-1$};
	\node[below = 15 pt of scalemm1text] (level1) {Level $1$};

			\node [blockzflexi, right=135pt of scalem11text2] (node5) {$p_{m}$};
				
			\node [blockzflexi, right=125pt of scalemm1text] (node6) {$p_{m-1}$};
			
			\node [blockzflexi, right=140pt of level1] (node7) {$p_{1}$};


	  \draw[->] (noderoot) -> (node1);
	  \draw[->,dashed] (noderoot) -> (node1dcopyr);
	  \draw[->] (node1dcopyr) -> (node1d);
	  \draw[->] (node1) -> (node2);
	  \draw[->, dashed] (node2) -> (node4);
	  
	   \draw[->] (node1d) -> (node2d);
	  \draw[->, dashed] (node2d) -> (node4d);
	 
	  \draw[->, dashed] (node1dcopyr) -> (noderootcopy3);
	  
	  \draw[->] (noderootcopy3) -> (node5);
	  \draw[->] (node5) -> (node6);
	  \draw[->, dashed] (node6) -> (node7);
	 
	 \draw[loosely dotted, thick] (scalem11text22) -- (scalemh01text2);
	 \draw[loosely dotted, very thick] (scalem11text) -- (scalem1text2);
	  \draw[loosely dotted, thick] (scalemh01text) -- (scalem11text);
	  \draw[loosely dotted, very thick] (scalem11text22) -- (scalem1textla);
	  \draw[loosely dotted, thick] (level1) -- (scalemm1text);

\end{tikzpicture}

%% file: images/explicit_cover_tree_extension.tex
\begin{tikzpicture}[align=center, node distance = 1.0cm, scale = 0.45]

	\node (scalem2text) {Level $2$};
	\node[below of =scalem2text] (scalem1text) {Level 1};
	\node[below of =scalem1text] (scalem0text) {Level 0};
	\node[below of =scalem0text] (scalemm1text) {Level -1};

	\node [blockz,  right=25pt of scalem2text ] (node1) {1};
	\node [blockz,  right=100pt of scalem1text] (node5) {5};
	\node [blockz,  right=25pt of scalem0text] (node3) {3};
	\node [blockz,  right=5pt of scalemm1text] (node2) {2};
	\node [blockz,  right=50pt of scalemm1text] (node4) {4};
	
	\node [blockz,  right=100pt of scalem0text] (node7) {7};
	\node [blockz,  right=100pt of scalemm1text] (node8) {8};

	  \draw[->] (node1) -> (node5);
	  \draw[->] (node1) -> (node3);
	  \draw[->] (node3) -> (node2);
	  \draw[->] (node3) -> (node4);
	  \draw[->] (node5) -> (node7);
	  \draw[->] (node7) -> (node8);

\end{tikzpicture}

%% file: images/CoverTreeLongExample/good_tree_example_1.tex
\begin{tikzpicture}[align=center, node distance = 1.0cm, scale = 0.45]

	\node (scale3) {Level 2};
	\node[below of =scale3] (scale2) {\color{orange} Level 1};
	\node[below of =scale2] (scale1) {Level 0};
	\node[below of =scale1] (scale0) {Level $-1$};
	\node [blockzm1,  right=1pt of scale0 ] (node1) {1};
	\node [blockzm1,  right=26pt of scale0 ] (node3) {3};
	\node [blockzm1,  right=51pt of scale0 ] (node5) {5};
	\node [blockzm1,  right=76pt of scale0 ] (node7) {7};
	\node [blockzm1,  right=101pt of scale0 ] (node9) {9};
	\node [blockzm1,  right=126pt of scale0 ] (node11) {11};
	\node [blockzm1,  right=151pt of scale0 ] (node13) {13};
	\node [blockzm1,  right=176pt of scale0 ] (node15) {15};

	\node [blockzm1,  right=13pt of scale1 ] (node2) {2};
	\node [blockzm1,  right=63pt of scale1 ] (node6) {6};
	\node [blockzm1,  right=113pt of scale1 ] (node10) {10};
	\node [blockzm1,  right=163pt of scale1 ] (node14) {14};

	\node [blockzm1p,  right=38pt of scale2 ] (node4) {4};
	\node [blockzm1y,  right=138pt of scale2 ] (node12) {12};
	\node [blockzm1o,  right=88pt of scale3 ] (node8) {8};

            \draw[->] (node2) -> (node1);
    \draw[->] (node2) -> (node3);
     \draw[->] (node6) -> (node5);
    \draw[->] (node6) -> (node7);
     \draw[->] (node10) -> (node9);
    \draw[->] (node10) -> (node11);
     \draw[->] (node14) -> (node13);
    \draw[->] (node14) -> (node15);

    \draw[->] (node8) -> (node4);
    \draw[->] (node8) -> (node12);

      \draw[->] (node4) -> (node2);
     \draw[->] (node4) -> (node6);
         
         \draw[->] (node12) -> (node10);
     \draw[->] (node12) -> (node14);


\end{tikzpicture}

%% file: images/CoverTreeLongExample/good_tree_example_2.tex
\begin{tikzpicture}[align=center, node distance = 1.0cm, scale = 0.45]

	\node (scale3) {Level 2};
	\node[below of =scale3] (scale2) {Level 1};
	\node[below of =scale2] (scale1) {\color{orange} Level 0};
	\node[below of =scale1] (scale0) {Level $-1$};
	\node [blockzm1,  right=1pt of scale0 ] (node1) {1};
	\node [blockzm1,  right=26pt of scale0 ] (node3) {3};
	\node [blockzm1,  right=51pt of scale0 ] (node5) {5};
	\node [blockzm1,  right=76pt of scale0 ] (node7) {7};
	\node [blockzm1,  right=101pt of scale0 ] (node9) {9};
	\node [blockzm1,  right=126pt of scale0 ] (node11) {11};
	\node [blockzm1,  right=151pt of scale0 ] (node13) {13};
	\node [blockzm1,  right=176pt of scale0 ] (node15) {15};

	\node [blockzm1y,  right=13pt of scale1 ] (node2) {2};
	\node [blockzm1p,  right=63pt of scale1 ] (node6) {6};
	\node [blockzm1y,  right=113pt of scale1 ] (node10) {10};
	\node [blockzm1y,  right=163pt of scale1 ] (node14) {14};

	\node [blockzm1o,  right=38pt of scale2 ] (node4) {4};
	\node [blockzm1o,  right=138pt of scale2 ] (node12) {12};
	\node [blockzm1o,  right=88pt of scale3 ] (node8) {8};

    \draw[->] (node2) -> (node1);
    \draw[->] (node2) -> (node3);
     \draw[->] (node6) -> (node5);
    \draw[->] (node6) -> (node7);
     \draw[->] (node10) -> (node9);
    \draw[->] (node10) -> (node11);
     \draw[->] (node14) -> (node13);
    \draw[->] (node14) -> (node15);

    \draw[->] (node8) -> (node4);
    \draw[->] (node8) -> (node12);

      \draw[->] (node4) -> (node2);
     \draw[->] (node4) -> (node6);
         
         \draw[->] (node12) -> (node10);
     \draw[->] (node12) -> (node14);


\end{tikzpicture}

%% file: images/CoverTreeLongExample/good_tree_example_33.tex
\begin{tikzpicture}[align=center, node distance = 1.0cm, scale = 0.45]

	\node (scale3) {Level 2};
	\node[below of =scale3] (scale2) {Level 1};
	\node[below of =scale2] (scale1) {\color{orange} Level 0};
	\node[below of =scale1] (scale0) {Level $-1$};
	\node [blockzm1o,  right=1pt of scale0 ] (node1) {1};
	\node [blockzm1o,  right=26pt of scale0 ] (node3) {3};
	\node [blockzm1o,  right=51pt of scale0 ] (node5) {5};
	\node [blockzm1o,  right=76pt of scale0 ] (node7) {7};
	\node [blockzm1o,  right=101pt of scale0 ] (node9) {9};
	\node [blockzm1o,  right=126pt of scale0 ] (node11) {11};
	\node [blockzm1,  right=151pt of scale0 ] (node13) {13};
	\node [blockzm1,  right=176pt of scale0 ] (node15) {15};

	\node [blockzm1o,  right=13pt of scale1 ] (node2) {2};
	\node [blockzm1o,  right=63pt of scale1 ] (node6) {6};
	\node [blockzm1o,  right=113pt of scale1 ] (node10) {10};
	\node [blockzm1r,  right=163pt of scale1 ] (node14) {14};

	\node [blockzm1o,  right=38pt of scale2 ] (node4) {4};
	\node [blockzm1r,  right=138pt of scale2 ] (node12) {12};
	\node [blockzm1o,  right=88pt of scale3 ] (node8) {8};

    \draw[->] (node2) -> (node1);
    \draw[->] (node2) -> (node3);
     \draw[->] (node6) -> (node5);
    \draw[->] (node6) -> (node7);
     \draw[->] (node10) -> (node9);
    \draw[->] (node10) -> (node11);
     \draw[->] (node14) -> (node13);
    \draw[->] (node14) -> (node15);

    \draw[->] (node8) -> (node4);
    \draw[->] (node8) -> (node12);

      \draw[->] (node4) -> (node2);
     \draw[->] (node4) -> (node6);
         
         \draw[->] (node12) -> (node10);
     \draw[->] (node12) -> (node14);

\end{tikzpicture}

%% file: images/CoverTreeLongExample/good_tree_example_4.tex
\begin{tikzpicture}[align=center, node distance = 1.0cm, scale = 0.45]

	\node (scale3) {Level 2};
	\node[below of =scale3] (scale2) {Level 1};
	\node[below of =scale2] (scale1) {\color{orange} Level 0};
	\node[below of =scale1] (scale0) {Level $-1$};
	\node [blockzm1g,  right=1pt of scale0 ] (node1) {1};
	\node [blockzm1g,  right=26pt of scale0 ] (node3) {3};
	\node [blockzm1g,  right=51pt of scale0 ] (node5) {5};
	\node [blockzm1r,  right=76pt of scale0 ] (node7) {7};
	\node [blockzm1r,  right=101pt of scale0 ] (node9) {9};
	\node [blockzm1r,  right=126pt of scale0 ] (node11) {11};
	\node [blockzm1,  right=151pt of scale0 ] (node13) {13};
	\node [blockzm1,  right=176pt of scale0 ] (node15) {15};

	\node [blockzm1g,  right=13pt of scale1 ] (node2) {2};
	\node [blockzm1r,  right=63pt of scale1 ] (node6) {6};
	\node [blockzm1r,  right=113pt of scale1 ] (node10) {10};
	\node [blockzm1r,  right=163pt of scale1 ] (node14) {14};

	\node [blockzm1g,  right=38pt of scale2 ] (node4) {4};
	\node [blockzm1r,  right=138pt of scale2 ] (node12) {12};
	\node [blockzm1r,  right=88pt of scale3 ] (node8) {8};

    \draw[->] (node2) -> (node1);
    \draw[->] (node2) -> (node3);
     \draw[->] (node6) -> (node5);
    \draw[->] (node6) -> (node7);
     \draw[->] (node10) -> (node9);
    \draw[->] (node10) -> (node11);
     \draw[->] (node14) -> (node13);
    \draw[->] (node14) -> (node15);

    \draw[->] (node8) -> (node4);
    \draw[->] (node8) -> (node12);

      \draw[->] (node4) -> (node2);
     \draw[->] (node4) -> (node6);
         
         \draw[->] (node12) -> (node10);
     \draw[->] (node12) -> (node14);


\end{tikzpicture}

%% file: chapters/MSTUpdated.tex
\chapter{Fast algorithm for minimum spanning tree based on compressed cover tree.}
\label{ch:mst}

\section{Minimum spanning tree} 
\label{sec:emst_intro}

Recall that a tree $T$ is any connected graph without cycles.
A minimum spanning tree on a finite metric space $(R,d)$ is a tree $\MST(R)$ with vertex $R$ that has a minimal total length of edges. A minimum spanning tree, $\MST(R)$ and is a fundamental object of computational geometry that finds its applications in various data structures, such as topomap \cite{doraiswamy2020topomap} and accelerated hierarchical density based clustering \cite{mcinnes2017accelerated},  mergegram \cite{elkin2021isometry}. This chapter corrects the work of \cite{march2010fast} and introduces a new algorithm that finds $\MST(R)$ in a parametrized near-linear time in the size of $|R|$. 

\begin{dfn}[Minimum spanning tree of a metric space]
	\label{dfn:mst}
	
	Let $(R,d)$ be a finite metric space. We say that a graph $T$ is a spanning tree of $X$, 
	if $T$ has the points of $X$ as its vertex set and if $T$ is connected. A graph $T$ is a minimum spanning tree of graph $G$, if the total length of all the edges $w(T)$ in the graph $T$ is minimal among all the spanning trees of $X$.
	A minimum spanning tree of $(R,d)$ will be denoted as $\MST(R)$. \bs
\end{dfn}

\medskip 

\noindent
The \emph{exact} metric spanning tree problem asks for exact $\MST(R)$ \cite{march2010fast,chatterjee2010geometric}. For any graph $G$ let us denote $\text{wt}$ to be the sum of all edge lengths in $G$.
In \emph{approximate} the goal is to find a spanning tree $T$ on $R$, in such a way that $wt(T) \leq (1+\epsilon) \cdot wt(\MST(R))$, where $\epsilon > 0$ is fixed and $\MST(R)$ is the exact minimum spanning tree  \cite{arya2016fast,wang2018fast}.
The most recent approaches \cite{wang2021fast} focus on developing parallel algorithms for the construction of $\MST(R)$.  

\begin{prob}
	\label{pro:MST}
	Given a finite subset $R$ of a metric space $(X,d)$ find a near-linear in the size $|R|$ algorithm to compute minimum spanning tree $\MST(R)$.
\end{prob} 

\noindent
Note that $\MST(R)$ is not generally uniquely defined. However, if $R$ is in general position i.e. all pairwise distances of $(R,d)$ have distinct values, then $\MST(R)$ is unique.

\begin{prob}
	\label{pro:MST_approx}
	Given a finite subset $R$ of a metric space $(X,d)$ and a real number $\epsilon > 0$
	find a near-linear in the size $|R|$ algorithm to compute $(1+\epsilon)$-approximate minimum spanning tree $\MST(R)$ 
	i.e. find spanning tree $T$ on $R$ in such a way that $wt(T) \leq (1+\epsilon) \cdot wt(\MST(R))$.
\end{prob}

\noindent
In the naive approach for finding a minimum spanning tree, first, a complete graph $G$ is built on $R$, in such a way that the edge set of $G$ contains all combinations of $R \times R$. Then any classical method, such as Bor\r{u}vka's algorithm \cite{boruuvka1926jistem} or Prim's algorithm \cite{prim1957shortest} to compute $\MST(R)$. Since there are $O(|R|^2)$ edges, the total run-time of the algorithm $|R|^2 \cdot \log(|R|)$. The obtained time worst-case time complexity estimate is too slow for most practical applications. In recent years more sophisticated techniques have been developed. 

\medskip

\noindent
One of the main aims of this Chapter is to show that the original approach \cite{march2010fast} based on dual-tree strategy and cover tree data structure failed to give a worst-case time complexity estimate for their proposed algorithm \cite[Algorithm~]{march2010fast}. In \cite[Theorem~5.1]{march2010fast} it was claimed that Problem \ref{pro:MST} can be resolved in 
$$O(\max\{c(R)^6, c_p^2c^2_l\}\cdot c(R)^{10} \cdot |R| \cdot \log(|R|) \cdot \alpha(|R|)$$
time, where $c(R)$ is expansion constant of Definition~\ref{dfn:expansion_constant}, $\alpha(|R|)$ is slowly growing Ackermann function and new expansion constants $c_p, c_l$ see Remark~\ref{rem:linkage_cluster_expansion_constants}.
However, the provided proof of Theorem 5.1 \cite{march2010fast} was incorrect.
Counterexample \ref{cexa:mst} gives an example, which produces
 $O(|R| \cdot \sqrt{|R|})$ iterations in its first Bor\r{u}vka step of \cite[Algorithm~]{march2010fast}. However, the proof claimed that the number of iterations is bounded by depth multiplied by the number of duplications. The contradiction will be obtained by showing that the given proof implies that the number of iterations should be bounded by $40 \cdot \sqrt{|R|}$, which contradicts previously found  $O(|R| \cdot \sqrt{|R|})$ low-bound for the number of iterations. Section~\ref{sec:emst_challenge_reference_linkage_expansion} will highlight further challenges in the definition of the linkage expansion constant $c_l$ , which is used in another part of the proof of Theorem~5.1.
 
\medskip

\noindent 
Shortcomings of \cite{march2010fast} are resolved in Section~\ref{sec:emst}. The main method of the section is Algorithm~\ref{alg:cover_tree_connected_components}, which is an analog of the minimum spanning tree algorithm \cite[Algorithm~1]{march2010fast} to compressed cover trees.
In Theorem~\ref{thm:dtb_correctness} we prove correctness of Algorithm~\ref{alg:cover_tree_connected_components}.
In Theorem~\ref{thm:single_cover_tree_mst_time} it is shown that the time complexity of Algorithm~\ref{alg:cover_tree_connected_components} is:
\begin{ceqn}
	\begin{equation}
		O\Big ((c_m(R))^{4+ \ceil{\log_2(\rho(R))}} \cdot |H(\T(R))| \cdot |R| \cdot \log(|R|) \cdot \alpha(|R|)\Big ), 
	\end{equation}
\end{ceqn}
where $\rho(R)$ is introduced in Definition \ref{dfn:edge_lengths} and $H(\T(R))$ is the height of Definition \ref{dfn:depth}.





\begin{figure}
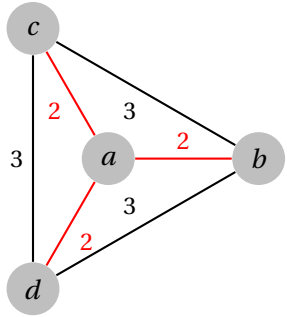

	\centering
	\input images/TriangleGraphMst.tex
	\caption{$\MST$ drawn with red color on a data set arranged triangularly.}
	\label{fig:mst_example}
\end{figure}
\section{Counterexample: An incorrect estimate of reference recursions}
\label{sec:emst_challenge_reference_recursion}

In this section, we assume that $(R,d)$ is a finite metric space. In this section we show that \cite[Theorem~5.1]{march2010fast} was proven incorrectly. Counterexample \ref{cexa:mst} gives a specific example, which contradicts a step in a proof of \cite[Theorem~5.1]{march2010fast}, which similarly to  \cite[Theorem~3.1]{ram2009linear} fails to estimate correctly the number of recursions during the whole execution of Algorithm \ref{alg:FindComponentNeighborsOriginal}. 

\medskip

\noindent
In original Algorithm~\ref{alg:FindComponentNeighborsOriginal}, the node $q_j$ has a level $j$, a reference subset $R_i\subseteq R$ is a subset of $C_i$ for an implicit cover tree $T(R)$. 
The algorithm is called for a pair $r, R_{i} = \{r\}$, where $r$ is the root of the reference tree at the maximal level $j = l_{\max}(T(R))$.
Split Algorithm~\ref{alg:FindComponentNeighborsOriginal} into these blocks: lines 2-8 : FinalCandidates, lines 8-13: reference expansion, and lines 13-27: query expansion.

\medskip

\begin{algorithm}
	\caption{Copy-pasted \cite[Algorithm~1]{march2010fast},  algorithm for dual-tree $\MST$ algorithm. Here $E$ is a set of edges of graph $G$ having the whole reference set $R$ as its vertex set, $N$ is the size of set $R$, and $\text{UpdateTree}(q)$ is a technical method for keeping track of necessary properties of the implicit tree $T(Q)$.}
	\label{alg:DualTreeboruvkaOriginal}
	\begin{algorithmic}[1]
		\STATE \textbf{Dual-tree Bor\r{u}vka}(Implicit cover tree node $q$)
		\STATE $E = \emptyset$
		\WHILE{$|E| < N - 1 $}
		\STATE call Algorithm \ref{alg:FindComponentNeighborsOriginal} with parameters ($q,q,e$).
		\STATE $E = E \cup e $
		\STATE UpdateTree$(q)$
		\ENDWHILE
	\end{algorithmic}
\end{algorithm}
\begin{algorithm}
	\caption{
		Copy-pasted {\cite[Algorithm~3]{march2010fast}}, now deprecated method for finding nearest neighboring components 
		during the iteration of Algorithm \ref{alg:DualTreeboruvkaOriginal}.}
	\label{alg:FindComponentNeighborsOriginal}
	\begin{algorithmic}[1]
		\STATE \textbf{Function} Fcn($T(Q)$-tree node $q_j$, Subset $R_i$ of cover set $C_i$ of $T(R)$,
		Edge set $e$).
		\IF {$i = -\infty$}
		\FOR{$q$ that are descendants of $q_j$ and $r \in R_i$ with $r \nsim q$}
		\IF {$d(q,r) < d(C_q)$}
		\STATE $d(C_q) = d(q,r)$, $e(C_q) = (q,r)$.
		\ENDIF
		\ENDFOR
		\ELSIF{$j < i$}
		\STATE $\C(R_i) = \{r \in \text{Children}(r') \mid r' \in R_i 
		\text{ and } r \nvDash q_j \}$ \label{line:FindComponentNeighbors:ChildrenR}
		\COMMENT{$r \vDash q_j$ means that all the descendants of $r$ and the descendants of $q_j$ belong to the same component }
		\STATE $d = \min\{d(C_q), \min_{r \in C, r \backsim q} \{d(q,r) + 2^{i} \},\min_{r \in \C(R_i), r \nsim q} \{d(q,r)\}\} $
		\label{line:FindComponentNeighbors:d}
		\STATE $R_{i-1} = \{r \in R \mid d(q_j,r) \leq d + 2^{i} + 2^{j+2} \}$
		\STATE \textbf{return} Fcn($q_{j-1}, R_i,e$) \COMMENT{ $q_{j-1}$ is the same point as $q_j$ on one level below}
		\ELSE{}
		\FOR{$p_{j-1} \in \text{Children}(q_j)$}
		\STATE \textbf{return} Fcn($p_{j-1},R_{i},e$)
		\ENDFOR
		\ENDIF
		
	\end{algorithmic}
\end{algorithm}

\noindent
Note that DualTreeBor\r{u}vka algorithm of \cite[Algorithm~1]{march2010fast} designed on implicit cover trees corresponds Algorithm \ref{alg:DualTreeboruvkaOriginal} that uses compressed cover tree. Algorithm \ref{alg:FindComponentNeighborsOriginal} (\cite[Algorithm~3]{march2010fast}) is launched from Line 4 of Algorithm \ref{alg:DualTreeboruvkaOriginal}.

\medskip

\noindent
Counterexample \ref{cexa:mst} shows that \cite[Theorem~5.4]{march2010fast} was proven incorrectly by using slightly modified space $R$ from Example \ref{exa:tall_imbalanced_tree}.

\begin{cexa}[Modification of Example \ref{exa:tall_imbalanced_tree} fails proof of {\cite[Theorem~5.4]{march2010fast}}]
	\label{cexa:mst}
	In this Counterexample we will point out
	that  \cite[Theorem~5.4]{march2010fast} was incorrectly proven by using Example \ref{exa:tall_imbalanced_tree}.
	Let us consider the following quote from its proof:
	
	\medskip 
	
	\noindent
	\emph{Page 608: "\textbf{Theorem 5.4}. Under the assumptions of Theorem. 5.1, the
	FindComponentNeighbors algorithm on a cover tree [Algorithm \ref{alg:DualTreeboruvkaOriginal}]
	finds the nearest neighbor of each component in time bounded by:
	$$O( N + c^4N  + \max\{c^6, c_p^2c^2_l\}\cdot N \alpha(N) + \max\{c^6, c_p^2c^2_l\} \cdot c^{10} \cdot \log(N) \cdot \alpha(N))."$$}
	
	
	\smallskip
	\noindent
	Consider the following part from proof of Theorem 5.4:
	
	\smallskip
	\noindent
	\emph{"
	At each level, $|R| \leq c^4 \max_i|R_i|$. Since the maximum depth of a node is $O(c^2\log(N))$ (depth bound), the number of nodes considered in Line 14 is at most $O(c^6\max_i|R_i|\log(N)$. Considering possible duplication across queries the total number of calls to Line 14 [Our Line 12] is at most $O(c^{10}\max_i|R_i|\log(N))$.
	"}
	\smallskip
	
	\noindent 
	\textbf{Our interpretation:} Let $D(p)$ be the explicit depth of a point $p$ of Definition~\ref{dfn:explicit_depth_for_compressed_cover_tree}. 
	Denote the number of duplications occurring in Algorithm \ref{alg:FindComponentNeighborsOriginal} as $\chi$.
	The above arguments claimed the algorithm runs Line 12 at most this number of times:
	\begin{align}
		\#(\text{Line 12}) &\leq \max_{p \in R}D(p) \cdot (\max_i \C(R_i)) \cdot \chi \\ \label{eqa:AuthorClaimingDualTremarch2010fast}
		&\leq c^2\log(N) \cdot c^4\max_{i}|R_i| \cdot c^4 \leq c^{10}\max_i|R_i|\log(N).
	\end{align}
	Using a cover tree $\T(R)$ based on  Example \ref{exa:tall_imbalanced_tree} it will be shown that $\T(R)$ does not satisfy (\ref{eqa:AuthorClaimingDualTremarch2010fast}), which ultimately crashes the proof for this specific dataset. 
	
	\smallskip
	
	\noindent 
	Let  $X, R, \T(R)$  be as in Example \ref{exa:tall_imbalanced_tree} for some $m > 5$. Let us modify $X$ by collapsing edge $(q,r)$ and by identifying $q$ with $r$. It can be shown that cover tree $\T(R)$ will have exactly same structure as in Example \ref{exa:tall_imbalanced_tree}. Let us run Algorithm \ref{alg:DualTreeboruvkaOriginal} on the dataset $R$ and the cover tree $\T(R)$. We will focus on the first iteration of the loop $3-7$ for $E = \emptyset$. In this case we launch Algorithm \ref{alg:FindComponentNeighborsOriginal} with parameters $(q_j = r, R_m = {r}, E = \emptyset)$.
	
	\medskip
	
	\noindent 
	Note first that $\T(R)$ contains at most one children on every level. Since $c^4$ corresponds to the number of children on level below according to the source, the number of duplications $\chi$ will be at most $2$. Similarly to Lemma \ref{lem:tall_imbalanced_tree_explicit_depth} we obtain $ \max_{p \in R}D(p) \leq 2m+1$. Let us now bound the maximal size of set $\C(R_i)$. For clarity Fcn$(q_j, R_i, E = \emptyset)$ will be denoted as Fcn($j,i,q_j,R_i$) in this counterexample. 
	Since $\T(R)$ contains at most one children on every level $i$, we have $|\C(R_i)| \leq |R_i| + 1$ for any recursion of FindComponentNeighbors algorithm. For any $i > m^2$ denote $r_i$ to be $r$.

	\medskip
	
	\noindent
	Note first that since $l(q_t) = t$ for any $t \in [1,m^2]$, then $r_t$ is recursed into from Fcn($t+1,t+1,p,R_i$), where $p$ is parent node of $r_t$. 
	Let us fix query node $q  = r$ and show by induction that for any recursion ($i$ , $j = i-1$, $q = r$, $R_i$) of FindComponentNeighbors we have $R_i = \{r_{i+2}, r_{i+1}, r_{i}, r\}$. The base case of the induction holds trivially since $R_{m^2+1} = \{r\}$. Assume now that $ R_i = \{r_{i+2}, r_{i+1}, r_{i}, r\}$ for some $i$. Then $\C(R_i) = \{r_{i+2},r_{i+1}, r_{i}, r, r_{i-1}\}$. Recall that $d(r_t, r) = 2^{t+1}$ for all $t$. Since there are no nodes in the same component with $r$, by line \ref{line:FindComponentNeighbors:d} we have $d = 2^{i}$. It follows that
	$R_{i-1} = \{a \in \C(R_i) \mid d(a,r) \leq d + 2^{i} + 2^{i+1}\} =  \{r_{i+1}, r_{i}, r_{i-1}, r\} $.
	Since $R_i$ is not modified during query-expansion (lines 13-17) for node $r$ we have proved the induction claim for Fcn($i-1$ , $j = i-2$, $q = r$, $R_i$). Note that if $i-1$ is divisible by $m$, then $r_{i-1}$ is a child of $r$ on level $i-1$. Therefore the claim is also satisfied for all Fcn($i-1$ , $j = i-2$, $q = r_{i-1}$, $R_i$), where $i-1$ is divisible by $m$.
	
	\medskip 
	
	\noindent
	Let us now show by induction that for parameters ($i$ , $j = i-1$, $q = r_{i}$, $R_i$) we have 
	$\{r, r_{i}\} \subseteq R_i \subseteq \{r_{i+3},r_{i+2}, r_{i+1}, r_{i}, r\}.$ From the previous case $q = r$ we obtain the this claim is satisfied in the base case and when $i$ is divisible by $m$. Assume now that the claim holds for some $i$, because the claim is satisfied for all $i$ that are divisible by $m$, it is sufficient to prove that the claim holds for $i-1$, when $i-1$ is not divisible by $m$. By induction assumption we have $ \{r_i, r_{i-1},r\} \subseteq \C(R_i) = \{r_{i+3}, r_{i+2},r_{i+1}, r_{i}, r, r_{i-1}\}$. 
	By line \ref{line:FindComponentNeighbors:d}: $d = d(r_{i}, r_{i-1}) =2^{i}$. Since $d(r_{i-1},r) = 2^{i}$ and $d(r_{i+3}, r_{i}) \geq 2^{i+2} + 2^{i+1} + 2^{i}$ we have $$\{r, r_{i-1}\} \subseteq R_{i-1} = \{r_t \in \C(R_i) \mid d(r_{i-1},r_t) \leq 2^{i+2}\} \subseteq  \{r_{i+2}, r_{i+1}, r_{i}, r_{i-1}, r\}.$$
	At line 12 we recurse into Fcn($i-1$ , $j = i-1$, $q = r_{i}$, $R_{i-1}$). Note that $i = j$ in this iteration , therefore we proceed into query expansion (lines 13-17) . Note that since $i-1$ is not divisible by $m$, $r_{i-1}$ is a child of $r_i$ at level $i-1$ and therefore FindComponentNeighbors is launched with parameters ($i-1$ , $j = i-2$, $q = r_{i-1}$, $R_{i-1}$), which proves the claim. 
	
	\medskip

	\noindent
	Let $t$ be such that $t \equiv 1 \mod m$. Let us now show that for any $i \leq t$ in recursion Fcn($i$ , $j = i-1$, $q = r_{t}$, $R_i$) we have 
	$$\{r, r_{t}\} \subseteq R_i \subseteq \{r_{i+2}, r_{i+1}, r_{i}, r_t, r\}.$$
	BaseCases $i = t$ follows from the previous paragraph. Assume now that the claim holds for $i$ and let us show that it holds for $i-1$. Again $ \{r_i, r_{i-1},r\} \subseteq \C(R_i) = \{r_{i+3},r_{i+2},r_{i+1}, r_{i}, r, r_{i-1}\}$. By line \ref{line:FindComponentNeighbors:d} $d = d(r_t, r) = 2^t$
	, $d(r_t, r_{i-1}) = 2^{i} + 2^{t}$ and $d(r_{i+3}, r_t) \geq 2^{t} + 2^{i+2}$ we have
	$$\{r, r_{i-1}\} \subseteq R_{i-1} = \{r_t \in \C(R_i) \mid d(r_{i-1},r_t) \leq 2^{i+1} + 2^{i} + 2^{t}\} \subseteq  \{r_{i+1}, r_{i}, r_{i-1}, r\}. $$
	As in previous paragraph we proceed to ($i-1$ , $j = i-1$, $q = r_{t}$, $R_{i-1}$), where ($i-1$ , $j = i-2$, $q = r_{t}$, $R_{i-1}$) is launched. Therefore the claim holds.
	
	\medskip

	\noindent
	It remains to show that Algorithm \ref{alg:FindComponentNeighborsOriginal} that is launched from $(r, R_i = \{r\}, E = \emptyset)$ has $O(m^3)$ low bound on number of times a reference expansions (lines 9-15) was performed.
	Let $\xi$ be the number of times Algorithm \ref{alg:FindComponentNeighborsOriginal} performs 
	reference expansions.  For every $r' \in R$ denote $\xi(r')$ to be the number of reference expansions that were performed for $q_j = r'$. We note that any node $r_u$ is introduced the in a query expansion (lines 9-11) of $(i = u+1, j = u+1)$. Since for all $u$ satisfying $i \equiv 1 \mod m $ set $R_i$ is non empty for all the levels $[1,u]$, we have $\xi(r_u) \geq u - 1$ for such $u$.
	Thus 
	$$\xi \geq \sum_{  i \equiv 1 \mod m } \xi(r_i) \geq \sum^{m+1}_{u = 1} u^2-2 = O(m^3).$$
	By using Inequality \ref{eqa:AuthorClaimingDualTremarch2010fast} we obtain a contradiction $$O(m^3) = \xi \leq  \max_{p \in R}D(p) \cdot (\max_i \C(R_i)) \cdot \chi \leq (2m+1) \cdot 6 \cdot 2 \leq O(m).$$

\end{cexa}

\section{Counterexample: linkage expansion constant $c_l$}
\label{sec:emst_challenge_reference_linkage_expansion}

The purpose of this section is to show linkage expansion constant $c_l$ of \cite[Definition~4]{march2010fast} cannot be used to
estimate maximal size of reference set $R_i$ in proof of \cite[Theorem~5.4]{march2010fast}.
Remark \ref{rem:linkage_cluster_expansion_constants} shows that linkage expansion constant $c_l$ was defined incompletely. 
Counterexample \ref{cexa:mst_component} uses Example \ref{exa:mst_counter_new} to show $c_l$ was incorrectly used in an important step in a proof of \cite[Theorem~5.4]{march2010fast}.

\begin{figure}
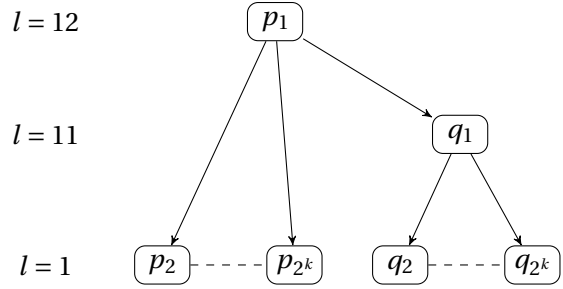

	\centering
	\input images/CoverTreeMSTCounter.tex
	\caption{Compressed cover tree $\T(R)$ of Example \ref{exa:mst_counter_new}. Level $l = 1$ contains points $p_2,p_3,...,p_{2^k}$ that are all connected to $p_1$, as well as points $q_2, q_3, ..., q_{2^k}$ that are all connected to $q_1$.}
	\label{fig:mst_cover_Tree_new}
\end{figure}

\begin{exa}[two separated sets example]
	\label{exa:mst_counter_new}
	Let $(R,d)$ be a finite metric space consisting of $2^{k+1}$ points divided into two subsets $A,B$ having size $2^k$. 
	Let $A = \{p_i \mid i \in 0,...,2^{k} - 1\}$ and let $B = \{p_i \mid i \in 0,...,2^{k}-1\}$. Metric $d(p_i,p_j)$ is defined as follows: For any $1 \leq i < j \leq 2^{k}$, let $(a_1,...,a_k)$ be the binary representation of $i$. Define $J(p_i,p_j) = \min_{t} \{ t \in [1,k]\cap\Z  \mid  a_t \neq b_t \}$ and $d(p_i, p_j) = 1 + \frac{ k + 1  - J(p_i,p_j)}{k+1}$. It follows that for all $p_i, p_j \in A$ we have 
	$1 < d(p_i,p_j) < 2 $.  Similarly for $q_i, q_j \in B$ we define $d(q_i, q_j) = d(p_i,p_j)$ for all $i,j$. For any $i,j$ we define $d(q_i, p_j) = 2^{10}$. 
	
	\medskip
	
	\noindent
	Let us now prove that $(X,d)$ is a metric space. Notice that for any indices $i, j, t \in [1,2^{k}]$ we have 
	$d(p_i,p_j) + d(p_j, p_{t}) > 2$ and $d(p_i, p_t) < 2$. Therefore the triangle inequality is satisfied for all triplets inside $A$, and by the same argument for $B$ as well. Let $p_i \in A$ and $q_j,q_t \in B$ be arbitrary nodes . 
	Then $d(p_i, q_j) = 2^{10} \leq 2^{10} + d(q_t,q_j) \leq d(p_i,q_t) + d(q_t, q_j)$. We have now shown that $(R,d)$ is indeed a metric space. Let us now construct compressed cover tree $\T(R)$.
	
	\medskip
	
	\noindent
	Let $p_1$ be the root of $\T(R)$. Set $l(q_1) = 2$ and let $q_1$ be a child of $p_1$. Set $l(p_i) = 1$ for all $p_i \in A \setminus \{p_1\}$ and set their parent to be $p_1$. Similarly for all $q_i \in B \setminus \{q_1\}$ set $l(q_i) = 1$ and let $q_1$ be their parent.  We have $C_{12} = \{p_1\}$, $C_{11} = \{p_1, q_1\}$ , $C_{10} = \{p_1,q_1\}$ , ... , $C_{2} = \{p_1, q_1\}$, $C_{1} = R$. Since for all $i$ both inequalities $d(p_1, p_i) \leq 2$ and $d(q_1, q_i) \leq 2$ are satisfied we note that $\T(R)$ satisfies the covering condition \ref{dfn:cover_tree_compressed}(b). Note that $d(a,b) > 1$ for all $a,b \in R$, therefore the separation condition \ref{dfn:cover_tree_compressed}(c) holds as well.  \bs 
	
	
\end{exa}

\begin{figure}
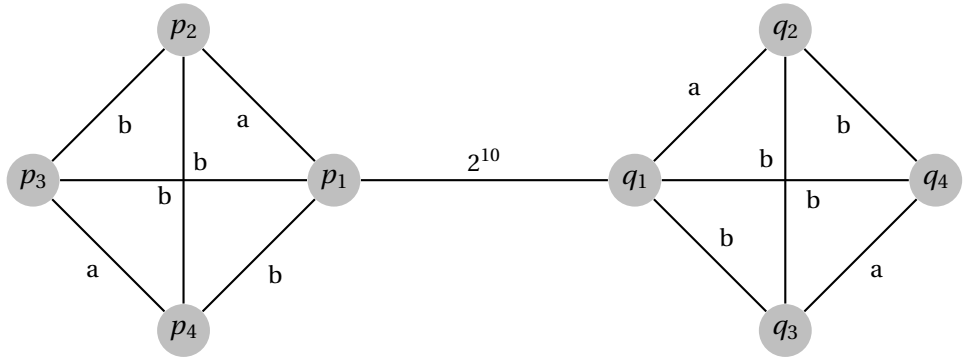

	\centering
	\input images/WellSeparatedPair.tex
	\caption{This figure illustrates clusters $A = \{p_1,p_2,p_3,p_4\}$,$B = \{q_1,q_2,q_3,q_4\}$ of Example \ref{exa:mst_counter_new} for $k = 2$.
		In this picture we have $a = 1\frac{1}{3}$ and $b = 1\frac{2}{3}$.}
	\label{fig:mst_example_2}
\end{figure}

\begin{dfn}[ Bor\r{u}vka clustering, { \cite[Definition~2]{march2010fast} }]
	\label{dfn:boruvkaClustering}
	Given a finite metric space $(R,d)$, \emph{Bor\r{u}vka clustering} is an indexed collection of partitions $\mathcal{F}_i$. Partition $\mathcal{F}_0$ is the partition into isolated points
	$\{ \{p\} \mid p \in R\} $. For any $i$ recall from Definition \ref{dfn:nn_component} that $\NN(U, \mathcal{F}_{i})$ is the set of nearest neighboring components of $U \in \mathcal{F}_i$.
	Partition $\mathcal{F}_{i}$ is obtained from $\mathcal{F}_{i-1}$ by applying \emph{Bor\r{u}vka step} on $\mathcal{F}_{i-1}$:
	In every step for every $V \in \mathcal{F}_{i-1}$ we select $U \in \NN(V, \mathcal{F}_{i-1})$ and merge $V$ with chosen subset $U$. 
	This process, illustrated in Figure \ref{fig:boruvkaClustering} is repeated until $\mathcal{F}_i$ becomes a partition consisting of one element $\{R\}$. \bs
\end{dfn}

\begin{rem}[Cases where linkage expansion constant $c_l$ is undefined  ] 
	\label{rem:linkage_cluster_expansion_constants}
	Recall that linkage expansion constant of \cite[Page~605]{march2010fast} was defined as:
	"
	\cite[Definition~4]{march2010fast}: [ $r$ is replaced with $t$, Bor\r{u}vka clustering $D$ replaced is replaced with $\mathcal{F}$]. Linkage Expansion Constant. Let $C_1$ and $C_2$ be two
	clusters in the Bor\r{u}vka clustering at level $i$ and let $S_1 \subseteq C_1$ and $S_2 \subseteq C_2 $.
	Let $B^l_i(S_1, S_2, t)$ be the set of all pairs $(p,q)$ such that $p \in S_1$, $q \in S_2$ and $d(p,q) \leq t$.

	\smallskip
	
	\noindent
	\textbf{Definition 4:} The linkage expansion constant is the smallest real number $c_l$ such that 
	$$B^l_i(S_1,S_2, 2t) \leq c_l|B^l_i(S_1, S_2,t)|,$$
	for all levels of the Bor\r{u}vka clustering $\mathcal{F}_i$, clusters $C_1$ and $C_2$ at level $i$, subsets $S_1 \subseteq C_1$ and $S_2 \subseteq C_2$, and distances $t> 0$.
	"
	
	\medskip
	
	\noindent
	Bor\r{u}vka clustering $\mathcal{F}$ was originally introduced in \cite[Definition~1]{march2010fast} and restated as in Defintion \ref{dfn:boruvkaClustering} in this chapter.
	Let $R, A, B$ be as in Example \ref{exa:mst_counter_new}. 
	It will be shown in Counterexample \ref{cexa:mst_component}, that in the previous to last Bor\r{u}vka step has $\mathcal{F} = \{A,B\}$. Let $i$ be the index of this step.
	Then for $t = 2^{9}$ set $B^l_i(A, B,2t)$ all edges between $A$ and $B$, therefore $|B^l_i(A, B,2t)| = 2^{2k}$. However, by definition $|B^l_i(A, B,t)| = 0$. Therefore there does not exist real number $c_l$ satisfying $B^l_i(A,B, 2t) \leq c_l|B^l_i(A,B,t)|.$  We can conclude that Definition 4 is undefined for small values $t$. 
	\bs
\end{rem}

\begin{rem}[Cluster expansion constant]
	In this remark, we cite the Cluster expansion constant that will be used in Counterexample \ref{cexa:mst_component}
	
	\smallskip
	
	\noindent
	\cite[Definition~3]{march2010fast}: [ $r$ is replaced with $t$, set $S$ replaced with set $R$, Bor\r{u}vka clustering $D$ replaced is replaced with $\mathcal{F}$]." Given the Bor\r{u}vka clustering, we define the new expansion constants. Let 
	$B_i^c(q,t)$ be the set of all components $C_p$ with a point $p \in C_p$ such
	that $d(q, p) \leq t$. Using this component-wise ball, we define
	the cluster expansion constant. 
	
	\smallskip
	
	\noindent
	\textbf{Definition 3:} The cluster expansion constant is the smallest real number $c_p$ such that 
	$$|B^c_i(q,2t) \leq c_p|B^c_i(q,t)|,$$
	for all points $q \in R$, distances $t > 0$, and each level of the Bor\r{u}vka clustering $\mathcal{F}_i$
	
\end{rem}

\noindent
In Counterexample \ref{cexa:mst_component} it will be shown that there were challenges with using 
linkage expansion constant in the proof of \cite[Theorem~5.4]{march2010fast} 


\begin{cexa}[Linkage expansion constant $c_l$ cannot be used to estimate the size of reference set $\max_i R_i$]
	\label{cexa:mst_component}
	
	\textbf{The Buildup:} Let $R, d, \T(R), k, A,B$ be as in Example \ref{exa:mst_counter_new}. Let $\mathcal{F}_i$ be the Bor\r{u}vka clustering introduced in Definition \ref{dfn:boruvkaClustering}. Let us first show that after $k$-iterations of the while-loop in Algorithm \ref{alg:DualTreeboruvkaOriginal} on an input $\T(R)$ we have $\mathcal{F}_k = {A,B}$. 
	
	\medskip 
	\noindent
	For nodes $p_i,p_j \in A$ define relation $p_i \backsim^A_t p_j$ if indices $i$ and $j$ have exactly the same first $k-t$ numbers in their binary representations. Define relation $\backsim^B_i$ similarly for nodes in $B$.
	Let us prove by induction that for all  $t \in [0,k] $ we have $F_t = [\backsim^A_t] \cup [\backsim^B_t]$, where $[\backsim^A_t]$ and $[\backsim^B_t]$ 
	are collections of classes spanned by relation $\backsim^A_t$ and $\backsim^B_t$, respectively. 
	The basecase of induction follows by noting that 
	$F_0$ has every node in its own separate component. 
	
	\medskip
	\noindent
	Let us now prove the induction step. We will focus only on elements $p_i \in A$, since case for $q_i \in B$ can be proved similarly. 
	Assume that the claim holds for some $u \in [0,k-1]$, let us show that it holds for $u+1$. By definition of metric $d(p_i, p_j) = 1 + \frac{ k + 1 - J(p_i,p_j)}{k+1}$. By induction assumption any pair of nodes that share $u$ common numbers in the binary representations of their indices are already connected. Therefore for any $p_i$ we have $\min_{w \nsim_u p_i }d(p_i, w) > 1 + \frac{ k+1  - u}{k+1} $ . Since $A$ contains all indices of the interval $[1,2^k]$ for every $p_i$ we can find $w \in A$ in such a way that $p$ and $w$ have $k-t-1$ same number in their binary representations , but the number at $k-t$ differs. Since there are no integers between $k-1$ and $k-t-1$ we have $\min_{w \nsim_u p_i }d(p_i, w) = 1 + \frac{ k + 1 - u + 1}{k+1} $ for all $p_i$. It follows that in $F_{t+1}$ all the nodes that have $k-t-1$ same numbers in their binary representations will be merged. Note that $[\backsim^A_k] = \{A\} $  and $[\backsim^A_k] = \{B\} $. Therefore it follows that $\mathcal{F}_k = \{A,B\}$. Consider now the following quote:
	

	\medskip
	
	\noindent
	\emph{\textbf{Exact quote of} \cite[Theorem~5.4]{march2010fast}:[$d'$ replaced with $t$] "
	Consider the other case when $d > 2^{i+2}$ ... }
	
	\medskip
	
	\emph{... We now bound the number of points within a component
	that $q_j$ may have to consider. Let $C_r$ be a component distinct from $C_q$ . Let $L(q_j)$ denote the set of all leaves that are
	descendants of $q_j$ . Let $t = \min_{q \in L(q_j),r \in C_r}d(q, r)$. Then,}
	\begin{align*}
		|B^l_k(C_q, C_r, d + 2^{i+1} + 2^{i})| &\leq |B^l_k(C_q, C_r, 4(d -2^{i+1}))| \\
		&\leq c^2_l |B^l_k(C_q , C_r, d -2^{i+1} )| \\
		&\leq c^2_l |B^l_k(C_q , C_r, t)|.
	\end{align*}
	\emph{
	By the above argument, there can be at most one pair in
	$B^l_k (C_q \cap L(q_j), C_r , t)$. Therefore, there are at most $c^2_l$ points
	in $C_r$ contained in $B(q_j , d + 2^{i+1} + 2^i)$. In the worst case,
	each of the point is at level $C_{i-1}$ of the tree and must
	be considered in $R_{i-1}$. There are at most $c^2_p$ components
	$C_r$ that can contribute points, so the maximum number of
	points in $R_{i-1}$ is $c^2_p c^2_l$.
	"}
	
	\medskip
	
	\noindent
	Note that set of leaves $L(q_j)$ corresponds to nodes $\Sd_{j}(q_j, \T(R))$ of compressed cover tree. We denote the component of $q$ in $\mathcal{F}_k$ by $C_q$. Variable $d$ was introduced in line 10 of Algorithm \ref{alg:FindComponentNeighborsOriginal} and constant $c_p, c_l$ were defined in Remark \ref{rem:linkage_cluster_expansion_constants}.  The quote above states that
	if $d > 2^{i+2}$ number of the points is bounded $c^2_p \cdot c^2_l$ , where $c^2_p$ can be replaced by the number of compoonents at the current stage of the algorithm. Therefore for $k+1$th iteration of lines 3-6 of Algorithm \ref{alg:DualTreeboruvkaOriginal} the claim is 
	$$|R_{i-1}| \leq c^2_l \cdot (\text{Number of components in } \mathcal{F}_k).$$
	\textbf{The contradiction: }In this counterexample it will be shown that $\T(R)$ from Example \ref{exa:mst_counter_new} does not satisfy the above inequality in Algorithm \ref{alg:FindComponentNeighborsOriginal} launched with inputs $\T(R)$ and $\mathcal{F}_k$.
	
	\medskip
	
	\noindent
	By Remark \ref{rem:linkage_cluster_expansion_constants} $|B^l_k(C_q , C_r, t')|$ was indefined for all $t' < d(C_q,c_r)$. Let us change the definition by requiring $t' \geq d(C_q, C_r)$. In \cite{march2010fast} constant $c_l$ was defined as a maximum of constants over all partitions $F_n$. However, nothing prevents us by using the same argument of authors and focusing only on a single partition $\mathcal{F}_k$. In this case $\mathcal{F}_k$ consists of two components $A,B$ which satisfy: for all $a \in A$ and $b \in B$ we have $d(a,b) = 2^{10}$. Therefore for any $S_1 \subseteq A$ and $S_2 \subseteq B$ we have $B(S_1, S_2, x) = S_1 \times S_2$ when $ x \geq 2^{10}$ and indefined else. It follows that $c_{l} = 1$ in this case. 
	\medskip
	
	\noindent
	Let us now perform simulation of $k+1$th iteration of Algorithm \ref{alg:DualTreeboruvkaOriginal}. We are especially interested in iteration $(p_1, R_2, j = 1, i = 2)$ of Algorithm \ref{alg:FindComponentNeighborsOriginal}. Let us start by running Algorithm \ref{alg:FindComponentNeighborsOriginal} with using partition $\mathcal{F}_k$ and parameters $(p_1   R_{12} = \{p_1\}, i = 12, j = 12)$. Since $i = j = 12$ and $q_1$ is children of $p_1$ at level 11 we proceed into query expansion (lines 13-17), where we launch two separate instances 
	$(p_1, R_{12} = \{p_1\}, i = 12, j = 11)$ and  $(q_1, R_{12} = \{p_1\}, i = 12, j = 11)$. Let us focus only on iterations having $p_1$ as the query node. In $(p_1, R_{12} = \{p_1\}, i = 12, j = 11)$  we proceed into reference expansion (lines 8 -13) , where $\C(R_{12}) = \{p_1, q_1\}$, since $q_1$ is the only node that contains descendants that are not already connected to $p_1$ we get $R_{11} = \{q_1\}$. Note that both $q_1$ and $p_1$ do not contain any children on levels $[3, 11]$. Therefore we can skip to the case $i = 2$. Let us continue with iteration  $p_1, R_{2} = \{q_1\}, i = 2, j = 1)$, we proceed into reference expansion (lines 8-13). We have $\C(R_2) = B$. Now since all $b \in B$ are not connected to $p_1$ and $d(p_1, b) = 2^{10}$ for all $b$, we have $d = 2^{10}$ and $R_{1} = B$. Since $d \geq 2^{4} = 2^{i+2}$ we can use the arguments of the quote above. However, now we have:
	$$|R_{1}| \leq c^2_l \cdot (\text{Number of components in } \mathcal{F}_k) \leq 1 \cdot 2 \leq 2.$$
	Since $|R_1| = |B| = \frac{|R|}{2}$, this is a contradiction. \bs
	
\end{cexa}


\section{Time complexity proof for Singletree Bor\r{u}vka MST}
\label{sec:emst} 

The main result of this section is Theorem \ref{thm:single_cover_tree_mst_time} that resolves issues presented in Section \ref{sec:emst_challenge_reference_recursion} and Section \ref{sec:emst_challenge_reference_linkage_expansion} by giving weaker complexity result for Problem \ref{pro:MST}:
$$O(c_m(R)^{O(\log_2(\rho(R)))} \cdot |R| \cdot \log_2(|R|) \cdot \log_2(\Delta(R))),$$
where $\rho(R)$ is the maximal edge length of $\MST(R)$ divided by its minimum edge length and $\Delta(R)$ is aspect ratio of $R$ from
Definition \ref{dfn:radius+d_min}. Algorithm \ref{alg:single_tree_boruvka}, which solves the problem is a single-tree analog of the original dual-tree metric minimum spanning tree {\cite[Algorithm~3]{march2010fast}}.


\medskip

\noindent
Algorithm \ref{alg:single_tree_boruvka}, which resolves Problem \ref{pro:MST} is based on the idea of Bor\r{u}vka algorithm. The computation is started by placing all points of $R$ in their own components. During every Bor\r{u}vka step every component is merged with the component that minimizes the Hausdorff distance. It will be shown that there will be at most $O(\log(|R|))$ Bor\r{u}vka steps. Theorem \ref{thm:single_cover_tree_mst_time} shows that in every Bor\r{u}vka step for every component in $\mathcal{F}$ 
we can find its nearest components in a near-linear time.
\begin{figure*}
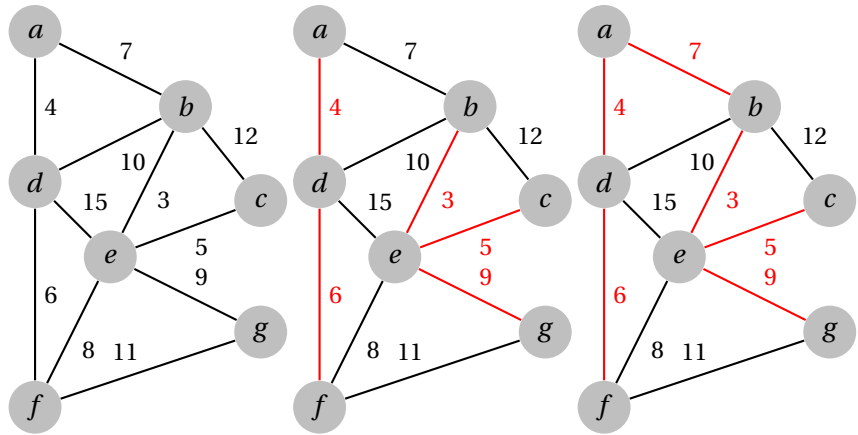

	\centering
	\input \input images/MSTboruvkaPhases.tex
	\caption{Bor\r{u}vka's clustering on an example graph.
		\textbf{Left: } In $\mathcal{F}_{0}$ every point exists in its own separate cluster , \textbf{Middle: } 
		$\mathcal{F}_1$ consists of two clusters $\{a,d,f\}$ and $\{b,c,e,g\}$, 
		\textbf{Right: } In $\mathcal{F}_2$ all the points are in the same cluster.  
	}
	\label{fig:boruvkaClustering}
\end{figure*}

\begin{dfn}[Partition $\mathcal{F}$ of a finite metric space $(R,d$]
	\label{dfn:partition_mst}
	Let $(R,d)$ be a finite metrics space. A \emph{a partition}  $\mathcal{F}$ of $(R,d)$ is a splitting of $R$ into disjoint subsets called \emph{clusters} if $R = \sqcup_{U \in F}U$ and $U \cap V = \emptyset$, for all distinct  $U,V \in \mathcal{F}$. \bs
\end{dfn}



\noindent
Recall that distinctive descendant set $\Sd_j(q, \T(R))$ was introduced in Definition \ref{dfn:distinctive_descendant_set} and that for any partition $\mathcal{F}$ and node $p \in R$ we denote $\mathcal{F}(p)$ to be the cluster of $p$ in the partition~$\mathcal{F}$.

\begin{dfn}[Cluster of descendants]
	\label{dfn:tau}
	Let $R$ be a finite subset of a metric space $(X,d)$. Let $\T(R)$ be a compressed cover and let $\mathcal{F}$ be arbitrary partition of the set $R$. Let $p \in \T(R)$ be arbitrary node , let $\mathcal{F}(p)$ be its cluster and let $i$ be arbitrary level of $\T(R)$. Define \emph{cluster of descendants} $\tau_i(p): R \rightarrow \mathcal{F} \times R$ in the following way:
	\begin{ceqn}
		\begin{equation*}
			\tau_i(p)=\begin{cases}
				(\mathcal{F}(p), \emptyset) & \text{If } \Sd_i(p,\T(R)) \subseteq \mathcal{F}(p)  \\
				(\emptyset, q)  \quad & \text{Else choose } q \in \Sd_i(p,\T(R)) \setminus \mathcal{F}(p) \\
			\end{cases}
		\end{equation*}
	\end{ceqn}
	Note that if $\Sd_i(p,\T(R))$ is not completely contained into some cluster, then $\tau_i(p)$ maps into $(\emptyset,q)$ where $q\in \Sd_i(p,\T(R))$ is arbitrary node that doesn't belong to $\mathcal{F}(p)$. \bs
	
	
	
	
	
\end{dfn}

\noindent
In Figure \ref{fig:uniqueDescendant} assume that $\mathcal{F} = \{U,V\}$, where $U = \{6,8,9\}$ and $V = \{2,3,4,5\}.$ Then
$\tau_2(2) = 6$, $\tau_1(6) = V$, $\tau_1(2) = U$. Let $\mathcal{F}$ be an arbitrary partition of $R$. For every $p \in R$ we denote $p \in \mathcal{F}(p)$ to be the cluster of $p$ in $\mathcal{F}$. 



\begin{algorithm}
	\caption{Singletree Bor\r{u}vka MST based on compressed cover trees.}
	\label{alg:single_tree_boruvka}
	
	\begin{algorithmic}[1]
		\STATE \textbf{Input: } A compressed cover tree $\T(R)$ with a root $r$
		\STATE \textbf{Output: } Minimum spanning tree $\MST(R)$
		\STATE Let $G$ be a graph with vertex set $R$ and an empty edge set. 
		\STATE Let $\mathcal{F}$ be a partition of $R$ into isolated points.
		\WHILE{$|\mathcal{F}| > 1$} \label{line:stb:while_start}
		\STATE Precompute $\tau$ of Definition \ref{dfn:tau} by running Algorithm \ref{alg:cover_tree_connected_components} on $(l_{\max}(\T(R)) ,r)$. \label{line:stb:precompute_tau}
		\FOR {$U \in \mathcal{F}$}\label{line:stb:inner_loop:begin}
		\STATE Set $\mathcal{C}(U)$ = Algorithm \ref{alg:single_tree_boruvka_step} executed on ($\T(R)$, $U$).
		\ENDFOR \label{line:stb:inner_loop:end}
		\FOR {$U \in \mathcal{F}$}
		\label{line:stb:merge_loop:start}
		\STATE Let $(q,p) = \mathcal{C}(U)$.
		\IF {$\mathcal{F}(q) \neq \mathcal{F}(p)$}
		\STATE Add edge $(q,p)$ to $G$.
		\STATE Merge components $\mathcal{F}(q)$ and $\mathcal{F}(p)$ in $\mathcal{F}$. 
		\ENDIF
		\ENDFOR \label{line:stb:merge_loop:end}
		\ENDWHILE \label{line:stb:while_end}
		\STATE \textbf{return} graph $G$
	\end{algorithmic}
\end{algorithm}

\begin{algorithm}
	\caption{This algorithm runs a depth-first traversal for a compressed cover tree $\T(R)$ to find all subtrees, where all the nodes of the subtree belong to the same component.}
	\label{alg:cover_tree_connected_components}
	\begin{algorithmic}[1]
		\STATE \textbf{Function} : FindClusters(a level $i$ of $\T(R)$, a node $p$ of $\T(R)$)
		\STATE \textbf{Output} : A connected component $U \in \mathcal{F}$ or $\emptyset$. 
		\STATE Set $U = F(p)$ \COMMENT{$U$ is the cluster containing $p$}
		\IF{$i > l_{\min(\T(R))}$ }
		\STATE Set $\mathcal{A} = \{p\} \cup \{q \in \Child(p) | l(q) = i-1\}$
		\STATE $j \leftarrow 1 + \nxt(p, i-1,\T(R))$.
		\FOR {$a \in \mathcal{A}$}
		\STATE $(V,b) = \text{FindClusters}(a,j)$.
		\IF{$U \neq V$}
		\STATE Set $\tau_i(p) = (\emptyset, b)$ and  \textbf{return} $(\emptyset, b)$.
		\ENDIF
		\ENDFOR 
		\ENDIF
		\STATE Set $\tau_i(p) = (U, \emptyset)$ and \textbf{return} $(U,\emptyset)$.
	\end{algorithmic}
\end{algorithm}

\begin{algorithm}
	\caption{Single tree Bor\r{u}vka step}
	\label{alg:single_tree_boruvka_step}
	\begin{algorithmic}[1]
		\STATE \textbf{Input} : a compressed cover tree $\T(R)$, a cluster $U \in \mathcal{F}(R)$
		\STATE Set $i \leftarrow l_{\max}(\T(R))$
		\STATE  Let $r$ be the root node of $\T(R)$. Set $R_{i}=\{r\}$.
		\WHILE{$i > l_{\min}$} \label{line:msts:loop_begin}
		\STATE Assign $\mathcal{C}(R_i) \leftarrow \{a \in \Child(p) \text{ for some }p \in R_i \mid l(a) \geq i-1 \}$ \\ \COMMENT{\textbf{Recall} that $\Child(p)$ contains node $p$ } \label{line:msts:dfn_C}
		\STATE Set $\C^{*}(R_i) = \{ p \in \C(R_i) \mid \tau_i(p) \neq U \} $\COMMENT{\textbf{Note: }$\tau$ appears in Definition \ref{dfn:tau}. } \label{line:msts:dfn_C_star}
		\STATE Set $l = \infty$.
		\FOR{$q \in U$}\label{line:msts:innerfor:begin}
		\STATE $l_{1} \leftarrow \min\limits_{p \in \C^{*}(R_i)} \{d(q,p) + 2^{i} \mid \mathcal{F}(p) = U \}$.
		\STATE $l_{2} \leftarrow \min\limits_{p \in \C^{*}(R_i)} \{d(q,p) \mid \mathcal{F}(p) \neq U \}$.
		\STATE $l \leftarrow \min\{l,l_1,l_2\} $.
		\ENDFOR \label{line:msts:innerfor:end}
		\STATE Find $R_{i-1} = \{p \in \C^{*}(R_i) \mid d(U,p) \leq l + 2^{i}\}$ \label{line:msts:dfnRi}
		\STATE Set $j \leftarrow \max_{ a \in R_{i-1}} \nxt(a,i-1,\T(R))$
		\COMMENT{If such $j$ is undefined, we set $j = l_{\min}$} \label{line:msts:dfnindexj}
		\STATE Set $R_j = R_{i-1}$ and $i = j$
		\ENDWHILE \label{line:msts:loop_end}
		\STATE \textbf{returns} pair $(q,p) \in U \times R_{l_{\min}} \setminus U$ minimizing $d(q,p)$
		\label{line:msts:final_line}
		
	\end{algorithmic}
\end{algorithm}

\begin{dfn}[Nearest neighboring component]
	\label{dfn:nn_component}
	Let $\mathcal{F}$ be an arbitrary partition of a finite metric space $(R,d)$. For any two sets $U, V \in \mathcal{F}$ define the Hausdorff distance as $d_H(U, V) =  \min \{d(x,y) \mid x \in U \text{ and }y \in V\}.$ For every component $U \in F$ the nearest neighbor set $\NN(U, \mathcal{F})$ consists of all the components $V \in \mathcal{F}$ satisfying
	$d_H(U,V) \leq \min_{W \in \mathcal{F}}d_H(U,W)$. \bs
\end{dfn}

\begin{lem}[True nearest components are always descendants of $R_i$ for all levels $i$]
	\label{lem:cover_tree_knn_correct_approx}
	Let $R$ be a finite subset of an ambient metric space $(X,d)$. 
	Let $\T(R)$ be a compressed cover tree of $R$. 
	Let $\mathcal{F}$ be arbitrary partition of $R$ and let $U \in \mathcal{F}$ be an arbitrary cluster.
	Then for any iteration $i \in H(\T(R))$ of lines 
	\ref{line:msts:loop_begin}-\ref{line:msts:loop_end} of Algorithm~\ref{alg:single_tree_boruvka_step} there exists $V \in \NN(U,\mathcal{F})$
	and node $\beta \in \bigcup_{p \in R_i}\Sd_i(p, \T(R)) \cap V$ for which $d(U, \beta) = d(U,V)$ for any $V' \in \NN(U,\mathcal{F})$.
	\bs
\end{lem}
\begin{proof}
	Assume contrary.
	Since $R_{l_{\max}} = \{r\}$, where $r$ is the root $\T(R)$ we have $S_{l_{\max}}(r,\T(R)) = R$ and therefore 
	we always have $\bigcup_{p \in R_{l_{\max}}}\Sd_i(p, \T(R)) \cap V \neq \emptyset$. Let $i$ be largest index for which 
	$\bigcup_{p \in R_{i-1}}\Sd_i(p, \T(R)) \cap \NN(U,\mathcal{F}) = \emptyset$. 
	Let $\beta$ be arbitrary point satisfying $\beta \in \bigcup_{p \in R_i}\Sd_i(p, \T(R))$ and $d(U, \beta) = d(U, V)$ for any $V \in \NN(U,\mathcal{F})$.
	By assumption we have:
	$$\beta \in \bigcup_{p \in R_i}\Sd_i(p, \T(R)) \setminus \bigcup_{p \in R_{i-1}}\Sd_{i-1}(p, \T(R)).$$ 
	By Lemma \ref{lem:child_set_equivalence} we have:
	\begin{ceqn}
		\begin{equation}
			\label{eqa:mst_single:neighborsContained}
			\bigcup_{p \in \C(R_i)}\Sd_{i-1}(p, \T(R)) = \bigcup_{p \in R_i}\Sd_{i}(p, \T(R))
		\end{equation}
	\end{ceqn}
	It follows that $\beta \in \bigcup_{p \in \C(R_i)}\Sd_{i-1}(p, \T(R))$. We denote $\alpha \in \C(R_i)$ to be the node that satisfies $\beta \in \Sd_{i-1}(\alpha, \T(R))$. By counter assumption $d(U, \alpha) > l + 2^{i}$. By using triangle inequality we have:
	\begin{ceqn}
		\begin{equation}
			\label{eqa:march2010fast_single:correctness:one}
			d(U,\beta) \geq d(U,\alpha)  - d(\alpha,\beta)   \geq d(U, \alpha) - 2^{i} > l
		\end{equation}
	\end{ceqn}
	Note that $\beta \in \C^{*}(R_i)$ by line \ref{line:msts:dfn_C_star} of Algorithm \ref{alg:single_tree_boruvka_step} , therefore $\C^{*}(R_i)$ is non-empty. By definition $l = \min \{l_1, l_2 \}$.
	Let us assume first that 
	\begin{ceqn}
		\begin{equation}
			\label{eqa:mst_single::correctness:two:zero}
			l = l_1 =  \min_{q \in U}\min\limits_{p \in \C^{*}(R_i)} \{d(q,p) + 2^{i} \mid \mathcal{F}(p) = U \},
		\end{equation}
	\end{ceqn}
	Let $q \in U$ and $\gamma \in \C^{*}(R_i)$ be the points minimizing the distance of (\ref{eqa:mst_single::correctness:two:zero}).
	Let $w \in \Sd_{i-1}(\gamma,\T(R))$ be an arbitrary point for which $w \notin U$. By Lemma \ref{lem:distinctive_descendant_distance} applied on $i-1$ we have  $d(\gamma, w) \leq 2^{i}$. It follows:
	\begin{ceqn}
		\begin{equation}
			\label{eqa:march2010fast_single:correctness:two}
			d(w, q) \leq d(q,\gamma ) + d(\gamma, w) \leq d(q,\gamma ) + 2^{i} = l
		\end{equation}
	\end{ceqn}
	By combining inequalities (\ref{eqa:march2010fast_single:correctness:one}) and (\ref{eqa:march2010fast_single:correctness:two}) and since $q \in U$ we have:
	$$d(U, \mathcal{F}(w)) \leq d(U,w) \leq d(q,w) < d(U,\beta).$$ It follows that there exists $V \in \mathcal{F}$ for which $d(\beta,U) > d(V,U)$, which is a contradiction.
	The case where $l = l_2 = \min\limits_{p \in \C^{*}(R_i)} \{d(U,p) \mid \mathcal{F}(p) \neq U \}$ can be proven similarly. 
\end{proof}


\begin{lem}[the lowest cost edge of cluster is always contained in $\MST(R)$]
	\label{lem:mst_contains_loweset_edge}
	Let $(R,d)$ be a finite metric space and let $\MST(R)$ be a minimum spanning tree of $R$. 
	Let $U \subseteq R$ be a set for which graph $\MST(R)$ restricted to vertex set $U$ is connected. 
	Then for any $(a,b) \in \MST(R)$ having $a \in U$ and $b \in R \setminus U$ minimizing the distance $d(U, R \setminus U)$, there exists possibly another minimum spanning tree $P$ of $R$ having $P | U = R | U$ and
	$(a,b) \in P$.
\end{lem}
\begin{proof}
	Since $R$ is finite, we can pick 
	$a \in U $ and $b \in R \setminus U$ minimizing the distance $d(U, R \setminus U)$
	Let $C$ be a component of $\MST(R) \setminus U$ containing $b$.
	Assume that there is no edges between $C$ to $U$ in $\MST(R)$. Then since $\MST(R) \setminus U$ differs from $\MST(R)$ only by edges having an end point in $U$ it follows that $\MST(R)$ is not connected, which is a contradiction. Therefore there exists some $(f,g) \in \MST(R)$ connecting $U$ to $C$. Note that $P = \MST(R) \setminus (f,g) \cup (a,b)$ is a spanning tree of $R$.
	By assumption $d(a,b) \leq d(f,g)$. Denote the sum of edge lengths by $w(\cdot)$, then $w(P) \leq w(\MST(R))$. It follows that $P$ is a minimum spanning tree of $R$.
\end{proof}

\begin{thmm}[Correctness of Algorithm \ref{alg:single_tree_boruvka}]
	\label{thm:dtb_correctness}
	
	Given any compressed cover tree $\T(R)$ on a reference set $R$ Algorithm~\ref{alg:single_tree_boruvka} solves Problem \ref{pro:MST} by finding some Minimum Spanning Tree $\MST(R)$. 
	
	
\end{thmm}

\begin{proof}
	
	Since $l_{\min}$ is the minimal level of $\T(R)$ we have $\bigcup_{p \in R_{l_{\min}}}\Sd_i(p, \T(R)) = R_{l_{\min}}.$ By Lemma \ref{lem:cover_tree_knn_correct_approx} for any component $U \in \mathcal{F}$ we can find $\beta \in R_{l_{\min}}$ satisfying 
	$d(U, \beta) = d(U, R \setminus U )$.
	Therefore line \ref{line:msts:final_line} of Algorithm \ref{alg:single_tree_boruvka_step} returns correctly a pair $(q,\beta)$ that minimizes distance between $U$ and $R \setminus U$. 
	
	\medskip
	
	Let us now show by induction on Bor\r{u}vka step (lines \ref{line:stb:while_start}-\ref{line:stb:while_end}) that graph $G$ is always a subset of some $\MST(R)$. In base case $G$ is a graph with empty edge set and therefore the claim holds trivially. Assume now that $G$ satisfies the assumption at the beginning of some execution of Bor\r{u}vka step. During every execution of lines \ref{line:stb:merge_loop:start} - \ref{line:stb:merge_loop:end} for every $U \in \mathcal{F}$ 
	we add an edge $(a,b)$ if $\mathcal{F}(a) \neq \mathcal{F}(b)$ , where 
	$a \in U$ and $b \in R \setminus U$ minimizing distance $d(U , R \setminus U)$. By using Lemma \ref{lem:mst_contains_loweset_edge} repetitively we obtain that the resulting graph $G$ is a minimum spanning tree of $R$ at the end of each Bor\r{u}vka step. 
	Since the algorithm is terminated when $|\mathcal{F}| = 1$, using the induction step we conclude that the final output $G$ of the Algorithm \ref{alg:single_tree_boruvka} is a minimum spanning tree of $R$. 
	
\end{proof}

\noindent
Recall that the set of essential levels $\mathcal{E}(p,\T(R))$ was introduced in Definition \ref{dfn:essential_levels_node}.
Lemma \ref{lem:cover_tree_connecetd_components_precompute} shows that for any compressed cover tree $\T(R)$ and any partition $\mathcal{F}$ of $R$ it is possible to precompute $\tau_i(p,\T(R), F)$ for all $p \in \T(R)$ and $i \in \mathcal{E}(p,\T(R))$ in $O(|R|)$ time. 

\begin{lem}
	\label{lem:cover_tree_connecetd_components_precompute}
	Let $R$ be a finite subset of a metric space.
	Let $\T(R)$ be a compressed cover tree on $R$. 
	Let $F$ be an arbitrary partition of set $R$.
	Recall that $\mathcal{E}(\T(R), p)$ was introduced in Definition \ref{dfn:essential_levels_node}.
	Then Algorithm~\ref{alg:cover_tree_connected_components} launched for $(r,l_{\max}(\T(R)))$ finds $\tau(p,i,\T(R), F)$ of Definition \ref{dfn:tau} for all $p \in \T(R)$ and $i \in \mathcal{E}(p,\T(R))$ in $O(|R| \cdot \alpha(|R|))$. 
	\bs
\end{lem}
\begin{proof}
	Let us prove the time complexity claim of Algorithm~\ref{alg:cover_tree_connected_components}. By Lemma \ref{lem:number_of_explicit_levels} we have $|\sum_{p \in R} \mathcal{E}(\T(R), p)| \leq 2 \cdot |R|$. Since
	FindComponents() is launched once for every $(p,i)$, where $p \in Q$ and $i \in \mathcal{E}(\T(R), p)|.$ The total number of iterations is bounded by $2 \cdot |R|$. In every iteration we search for a component of $p$.
	This is done using the union-find data structure that performs the operation in $O(\alpha(|R|))$. 
	Therefore it follows that the total time complexity is: $O(|R| \cdot \alpha(|R|)) . $
\end{proof}

\begin{dfn}
	\label{dfn:edge_lengths}
	Let $\MST(R)$ be an arbitrary minimum spanning tree of a finite metric space $(R,d)$. Define $d_{\min}(R)$ to be the minimal edge length of $\MST(R)$ and  $d_{\max}(R)$ to be the maximal edge length of $\MST(R)$. 
	Define $\rho(R) = 17 + 8 \cdot \frac{d_{\max}(R)}{d_{\min}(R)}$.
	\bs
\end{dfn}
\noindent
Note that $d_{\min}$ was originally defined in Definition \ref{dfn:radius+d_min}.
Edge lengths $d_{\max}(R)$, $d_{\min}(R)$ do not depend on the choice of minimum spanning tree $\MST(R)$ and are therefore well-defined.
Let $\al(|R|)$ be the inverse Ackermann function \cite{seidel2006understanding}, which grows very slowly e.g. $\al(10^{80})\leq 4$.


\begin{lem}[Run-time of Algorithm \ref{alg:single_tree_boruvka_step}]
	\label{lem:single_cover_tree_innerf_time}
	Let $R$ be a finite reference set in a metric space $(X,d)$.
	Given a compressed cover tree $\T(R)$ and any cluster $U \subseteq R$, 
	the run-time of Algorithm \ref{alg:single_tree_boruvka_step} is $$O(c_m(R)^{4+\ceil{\log_2(\rho(R))}} \cdot |H(\T(R))| \cdot |U| \cdot \alpha(|R|)),$$ 
	where all parameters were introduced in Definitions \ref{dfn:expansion_constant}, \ref{dfn:depth}, \ref{dfn:edge_lengths}. 
\end{lem}
\begin{proof}
	
	In while loop (lines \ref{line:msts:loop_begin}-\ref{line:msts:loop_end}) the algorithms encounters all levels $H(\T(R))$.
	Therefore the number of iterations is bounded by the height $|H(\T(R))|$. The total number of nodes in line \ref{line:msts:dfn_C} during a single iteration by Lemma \ref{lem:compressed_cover_tree_width_bound} is at most $(c_m(R))^4 \cdot \max_i|R_i|$. Since $\tau_i$ of Definition \ref{dfn:tau} was precomputed in Algorithm \ref{alg:cover_tree_connected_components}, the time complexity of line \ref{line:msts:dfn_C_star} is the same as line \ref{line:msts:dfn_C}. Lines \ref{line:msts:innerfor:begin} - \ref{line:msts:innerfor:end} takes at most $O(|U| \cdot |\C^{*}(R)|) \leq O(|U| \cdot (c_m(R))^4 \cdot \max_i|R_i|)$ time. Line \ref{line:msts:dfnRi} never does more work that line \ref{line:msts:dfn_C}. In line \ref{line:msts:dfnindexj} we keep track of $\nxt(a,i,\T(R))$, therefore its time complexity is the same as line \ref{line:msts:dfn_C}. It follows that the time complexity of lines \ref{line:msts:loop_begin} - \ref{line:msts:loop_end} is bounded by 
	$O(|U| \cdot (c_m(R))^4 \cdot \max_i|R_i|)$. 
	Line \ref{line:msts:final_line} takes at most $\max_i|R_i| \cdot \alpha(|R|)$ time.
	Consequently, the running time of the whole algorithm is bounded by:
	\begin{ceqn}
		\begin{equation} 
			\label{eqa:stb:one} 
			\centering
			O((c_m(R))^4 \cdot |U| \cdot \max_i|R_i| \cdot D(\T(R)) \cdot \alpha(|R|).
		\end{equation}
	\end{ceqn}
	To finish the proof we will show that $\max_i|R_i| \leq c_m(R)^{\ceil{\log_2(\rho(R))}}$. Since $l \leq d_{\max}(R) + 2^{i}$, 
	by line \ref{line:msts:dfnRi} of Algorithm \ref{alg:single_tree_boruvka_step} we have the following inclusion of subsets: 
	$$R_{i-1} = \{r \in \C^{*}(R_i) \mid d(q,r) \leq l + 2^{i}\} \subseteq \{r \in \C^{*}(R_i) \mid d(q,r) \leq d_{\max}(R) + 2^{i+1}\}$$
	Note that during any iteration we have $i-1 \geq l_{\min}(R)$ and $2^{l_{\min}(R)+1} \geq d_{\min}(R)$. It follows that 
	$2^{i-1} \geq 2^{l_{\min}(R)} \geq \frac{d_{\min}(R)}{2}$.
	Therefore $\frac{1}{2^{i-1}} \leq \frac{2}{d_{\min}(R)}$.
	By Lemma \ref{lem:packing} with $\delta = 2^{i-1}$ and $t = d_{\max}(R) + 2^{i+1}$ we have $$\mu = 4\frac{t}{\delta} + 1 = 4 \cdot \frac{2^{i+1} + d_{\max}(R)}{2^{i-1}} + 1 \leq 17 + 8 \cdot \frac{d_{\max}(R)}{d_{\min}(R)} = \rho(R).$$
	It follows that $\max_i|R_i| \leq c_m(R)^{\ceil{\log_2(\mu)})} \leq c_m(R)^{\ceil{\log_2(\rho(R)})}$. The  by substituting $\max_i|R_i|$ with $c_m(R)^{\ceil{\rho(R)}}$ in (\ref{eqa:stb:one}) we obtain the final bound.
\end{proof}

\noindent
Lemma \ref{lem:num_iterations} is used to show that Algorithm \ref{alg:single_tree_boruvka} performs at most $O(\log(|R|)$ bor\r{u}vka steps, lines \ref{line:stb:inner_loop:begin}- \ref{line:stb:inner_loop:end}. This lemma was skipped in the original work \cite{march2010fast}. 

\begin{lem}[Bor\r{u}vka step halves the number of components]
	\label{lem:num_iterations}
	Let $A$ be a finite set and let $f: A \rightarrow A$ be a function satisfying $f(x) \neq x$ for all $x \in A$. 
	For any $k \in \mathbb{N}$ define $f^{k}$ to be $k$-compositions of function $f$.
	Define relation $\backsim$ on $A$, in such a way that for any $y,z  \in A$ we have $y \backsim z$, if there exists $k \in \mathbb{N}$ having $f^{k}(y) = z$ or $f^{k}(z) = y$. Then $|A / \backsim_f| \leq \frac{|A|}{2}$.  
\end{lem}
\begin{proof}
	Assume that $A = \{a_1,...,a_n\}$
	Let us prove this lemma by induction. Base case $|A| = 2$ is trivial, since $f(a_1) = a_2$ and $f(a_2) = a_1$ is forced. Therefore $a_1 \backsim a_2$ and $|A / \backsim_f| = 1 \leq \frac{|A|}{2}$. For base case $|A| = 3$ consider all possible functions $f$. Since all the cases are symmetrical lets us consider only one variation $f(a_1) = a_2$, $f(a_2) = a_3$ and $f(a_3) = a_2$. By definition we have $a_1 \backsim a_2 \backsim a_3$ and therefore $|A / \backsim_f| = 1 \leq \frac{3}{2} = \frac{A}{2}$.
	
	\medskip
	
	\noindent
	Let us now prove the induction step. Assume that the claim holds for all $|A| \leq n$ and let us prove that it holds for $n+1$. Let $|A| = n+1$ and let $f: A \rightarrow A$ be a function having $f(x) \neq x$ for all $x \in A$. Assume first that $f^{2}(a_{n+1}) = a_{n+1}$. 
	Denote $B = A \setminus \{a_{n+1}, f(a_{n+1})\}$. 
	Define new function $g(x) = f(x)$ if $f(x) \neq a_{n+1}$ and $g(x) = f^{2}(x)$.
	Since $|B| \leq n-1 \leq n$, we can use induction assumption to obtain $|B / \backsim_g| \leq \frac{|B|}{2} \leq \frac{n-1}{2}$. 
	
	\medskip
	
	\noindent
	Let us show that if $u,v \in B$ are such that $u \backsim_f v$, then $u \backsim_g v$. Without loss of generality assume that $u = f^n(v)$ for some $n \in N$. By definition of $g$ for any $u$ we have $v = f^{n}(u) = g_{m}(u)$ for some $m \leq n$. Therefore we have
	$u \backsim_g v$. Since $f^{2}(a_{n+1}) = a_{n+1}$ pair $\{a_{n+1}, f(a_{n+1})\}$ is a component in $A / \backsim_f$. We conclude that
	$A / \backsim_f = \{a_{n+1}, f(a_{n+1})\}  \cup B \backsim_f \subseteq \{a_{n+1}, f(a_{n+1})\}  \cup B \backsim_g$.
	Therefore it follows that:
	$$|A / \backsim_f | \leq |B  / \backsim_g| + 1 \leq  \frac{n-1}{2} + 1 \leq \frac{n+1}{2} \leq \frac{|A|}{2}.$$
	Assume then that $f^{2}(a_{n+1}) \neq a_{n+1}$. Define similarly $g(x) = f(x)$ if $f(x) \neq a_{n+1}$ and $g(x) = f^{2}(x)$ else.
	Let $B = A \setminus \{a_{n+1}\}$. Since $|B| \leq n$, by induction assumption we have $|B / \backsim_g| \leq \frac{|B|}{2} \leq \frac{n}{2}$.
	Let us show that  $ |A / \backsim_f| \leq |B / \backsim_g|$. Let $u, v \in B$ be such that $u \backsim_f v$. Without loss of generality assume that $u = f^{n}(v)$. By definition of $g$ we have $u = f^{n}(v) = g^{m}(u)$ for some $m \leq n$. Therefore $u \backsim_g v$, which shows that $ |A / \backsim_f| \leq |B / \backsim_g|$.
	The claim follows by noting that:
	$$|A / \backsim_f | \leq |B  / \backsim_g| \leq  \frac{n}{2} \leq \frac{n+1}{2} \leq \frac{|A|}{2}. $$
\end{proof}

\begin{thmm}[complexity for exact single tree MST search]
	\label{thm:single_cover_tree_mst_time}
	Let $R$ be a finite reference set in a metric space $(X,d)$.
	Given a compressed cover tree $\T(R)$ Algorithm~\ref{alg:single_tree_boruvka} find $\MST(R)$ in time: 
	$$O\Big ((c_m(R))^{4+ \ceil{\log_2(\rho(R))}} \cdot |H(\T(R))| \cdot |R| \cdot \log_2(|R|) \cdot \alpha(|R|)\Big ),$$ 
	where all parameters were introduced in Definitions \ref{dfn:expansion_constant}, \ref{dfn:depth}, \ref{dfn:edge_lengths}. 
	\bs
\end{thmm}
\begin{proof}
	
	Note first that during an iteration of while loop lines \ref{line:stb:while_start}-\ref{line:stb:while_end} of Algorithm \ref{alg:single_tree_boruvka} we merge every $U \in \mathcal{F}$ with some cluster $V \in \NN(U, \mathcal{F})$.
	Define function $f : \mathcal{F} \rightarrow \mathcal{F}$ by mapping $f(U) = V$ for all $U \in \mathcal{F}$, where $V$ is some cluster in $\NN(U, \mathcal{F})$. By assumptions we have $f(U) \neq U$ for all $U \in \mathcal{F}$, therefore Lemma \ref{lem:num_iterations} can be applied on $\mathcal{F}$ and function $f$. We note that relation $\backsim_f$ 
	corresponds to the merging operation of lines \ref{line:stb:merge_loop:start} - \ref{line:stb:merge_loop:end}. 
	Let $\mathcal{R}$ be the partition $\mathcal{F}$ after executions of lines lines \ref{line:stb:merge_loop:start} - \ref{line:stb:merge_loop:end}. Then by Lemma \ref{lem:num_iterations}
	$$|\mathcal{R}| = |\mathcal{F} / \backsim_f |\leq \frac{|\mathcal{F}|}{2}$$
	Denote by $k$ the total number of times lines \ref{line:stb:while_start}-\ref{line:stb:while_end} were executed.
	Since at the start of execution of Algorithm~\ref{alg:single_tree_boruvka} we have $|\mathcal{F}| = |R|$ and 
	the algorithm is terminated when $|\mathcal{F}| = 1$, we have $\frac{|R|}{2^{k-1}} \geq 1 $. It follows that $k \leq \log(|R|) + 1$.

	
	\medskip
	
	\noindent
	By Lemma \ref{lem:cover_tree_connecetd_components_precompute} the time complexity of line \ref{line:stb:precompute_tau} is $O(|R| \cdot \alpha(|R|))$. To prove the complexity 
	we need to estimate the runtime of lines \ref{line:stb:inner_loop:begin} - \ref{line:stb:inner_loop:end} that launch Algorithm \ref{alg:single_tree_boruvka_step}. 
	By Lemma \ref{lem:single_cover_tree_innerf_time} time complexity of line 8 is $O(c_m(R)^{4+\ceil{\log_2(\rho(R))}} \cdot |H(\T(R))| \cdot |U| \cdot \alpha(|R|)) $. Since $\mathcal{F}$ is partition of $R$, we have $\sum_{U \in \mathcal{F}}|U| = |R|$. It follows that runtime of lines \ref{line:stb:inner_loop:begin} - \ref{line:stb:inner_loop:end} is $O\Big ((c_m(R))^{4+ \ceil{\log_2(\rho(R))}} \cdot |H(\T(R))| \cdot |R|  \cdot \alpha(|R|)\Big ).$ Note that the time complexity of lines 10-16 is $O( |\mathcal{F}| \cdot \alpha(|\mathcal{F}|) ) \leq O(|R| \cdot \alpha(|R|))$. Therefore the time complexity inside while-loop , lines \ref{line:stb:while_start}-\ref{line:stb:while_end} is dominated by lines \ref{line:stb:inner_loop:begin} - \ref{line:stb:inner_loop:end}. The claim of the theorem follows by recalling that lines \ref{line:stb:while_start}-\ref{line:stb:while_end} were executed $O(\log(|R|)$ times.
	
\end{proof}
\begin{cor}
	\label{cor:single_cover_tree_mst_time}
	Let $R$ be a finite reference set in a metric space $(X,d)$.
	Given a compressed cover tree $\T(R)$ Algorithm~\ref{alg:single_tree_boruvka} find $\MST(R)$ in time: 
	$$O\Big ((c_m(R))^{4+ \ceil{\log_2(\rho(R))}} \cdot \log_2(\Delta(R)) \cdot |R| \cdot \log_2(|R|) \cdot \alpha(|R|)\Big ),$$ 
	where all parameters were introduced in Definitions \ref{dfn:expansion_constant}, \ref{dfn:radius+d_min}, \ref{dfn:edge_lengths}. 
	\bs
\end{cor}
\begin{proof}
	Follows from Theorem \ref{thm:single_cover_tree_mst_time} by noting that $|H(\T(R))| \leq \log_2(\Delta(|R|))$.
\end{proof}



\section{Discussions: past challenges, contributions, and future work}
\label{sec:emst_Conclusions}

The motivations were the past challenges in the proof of time complexities in 
\cite[Theorem~5.1]{march2010fast}, which attempted to solve metric minimum spanning tree problem in  $O(|R| \cdot \log_2(|R|))$ time.
Section \ref{sec:emst_challenge_reference_recursion} and \ref{sec:emst_challenge_reference_linkage_expansion} show that the provided proof was incorrect. In Section \ref{sec:emst_challenge_reference_recursion} Counterexample \ref{cexa:mst} provides a concrete example, where the step in the proof of \cite[Theorem~5.1]{march2010fast} claims incorrect bound for the number of total recursions during a single Bor\r{u}vka. Section \ref{sec:emst_challenge_reference_linkage_expansion} deals with problematic definition of linkage expansion constant $c_l$, which is undefined for small radius $t$. In Counterexample \ref{cexa:mst_component} we provide a concrete example, where the reference set size $|R_i|$ exceeds the claimed upper bound $c_p^2 \cdot c_l^2$, where $c_p$ is cluster expansion constant. This ultimately shows that the proof of \cite[Theorem~5.1]{march2010fast} is incorrect and there is no easy way to fix it. 

\medskip

\noindent
In Section \ref{sec:emst} we use a different approach to prove the time complexity of the construction of the exact minimum spanning tree on any finite metric space. To resolve Problem \ref{pro:knn} we use the idea of Bor\r{u}vka algorithm. However, 
to overcome past challenges the search algorithm is replaced with the modification of a single tree search for $k$-nearest neighbors.
 To avoid using multiple expansion constants, especially the undefined linkage expansion constant $c_l$ in the statement and proof of the main Theorem \ref{thm:single_cover_tree_mst_time} we use only a single, minimized expansion constant $c_m(R)$ with possibly big exponent $O(\rho(R)$, which depends on the ratio of maximal edge length and minimal edge length of $\MST(R)$.
 The main result, Theorem \ref{thm:single_cover_tree_mst_time} gives $$O(c_m(R)^{\log_2(\rho(R))} \cdot |R| \cdot \log_2(|R|) \cdot \log_2(\Delta(R)))$$ time complexity for finding $\MST(R)$, where $\Delta(R)$ is the aspect ratio of $R$.

\medskip

\noindent
 In the future we wish to improve the time complexity to $O(c(R)^{O(1)} \cdot |R| \log(|R|))$ by utilizing $\T(R)$ as both, query and reference tree, where $c(R)$ is expansion constant of $R$. However, for such time complexity estimates it would be required to show that a complete paired-tree traversal of monochromatic pair $(\T(R), \T(R))$ could be done in time $c(R)^{O(1)} \cdot |R|$. Another approach for obtaining such time complexity could be by using Kruskal's algorithm as the base algorithm. Assuming that $\T(R)$ has been already constructed. It would be enough to show that any minimal non-zero sized edge in any Bor\r{u}vka clustering $R / \backsim$ can be found in $O(c(R)^{O(1)}\log(|R|))$ time. This is plausible since to find a non-trivial pair $a \in R$, $b \in R$ that minimizes $d(a,b)$ we can use $\T(R)$ as both, query and reference tree, and unlike in the standard dual-tree traversal, we can eliminate nodes that are far-away, from both: the query and the reference tree.

%% file: images/TriangleGraphMst.tex
\begin{tikzpicture}[scale=1.0, auto,swap]
  \foreach \pos/\name in {{(0:0)/a}, {(0:2)/b}, {(120:2)/c},
                            {(240:2)/d}}
        \node[vertex] (\name) at \pos {$\name$};
	 \path[edge,red] (b) -- node[weight] {2} (a);
	  \path[edge,red] (c) -- node[weight] {2} (a);
	   \path[edge,red] (d) -- node[weight] {2} (a);
	 \path[edge] (c) -- node[weight] {3} (b);
	 \path[edge] (b) -- node[weight] {3} (d);
	  \path[edge] (c) -- node[weight] {3} (d);
	 


\end{tikzpicture}

%% file: images/CoverTreeMSTCounter.tex
  \begin{tikzpicture}[align=center, node distance = 1.5cm, scale = 0.45]
    \node (scale4) {$l = 12$};
	\node[below of =scale4] (scale3) {$l = 11$};
	\node[below =35pt of scale3] (scale2) {$l = 1$};
	
		\node [blockzm1,  right=60pt of scale4 ] (nodep1) {$p_1$};
		\node [blockzm1,  right=130pt of scale3 ] (nodeq1) {$q_1$};
		
		\draw[->] (nodep1) -> (nodeq1);
		
		\node [blockzm1,  right=20pt of scale2 ] (nodep2) {$p_2$};
		
		\node [blockzm1,  right=70pt of scale2 ] (nodep3) {$p_{2^k}$};
		
			\node [blockzm1,  right=110pt of scale2 ] (nodeq2) {$q_2$};
		
		\node [blockzm1,  right=160pt of scale2 ] (nodeq3) {$q_{2^k}$};
		
		 \draw[dashed, -] (nodep2) -- (nodep3);
		 
		  \draw[dashed, -] (nodeq2) -- (nodeq3);
		  	\draw[->] (nodep1) -> (nodep2);
		  		\draw[->] (nodep1) -> (nodep3);
		  		
		  			\draw[->] (nodeq1) -> (nodeq2);
		  		\draw[->] (nodeq1) -> (nodeq3);

\end{tikzpicture}

%% file: images/WellSeparatedPair.tex
\begin{tikzpicture}[scale=2.0, auto,swap]
  \node[vertex] (a) at (1,0) {$p_1$};
    \node[vertex] (b) at (0,1) {$p_2$};
      \node[vertex] (c) at (-1,0) {$p_3$};
        \node[vertex] (d) at (0,-1){$p_4$};

          \node[vertex] (d2) at (5,0) {$q_4$};
    \node[vertex] (b2) at (4,1) {$q_2$};
      \node[vertex] (a2) at (3,0) {$q_1$};
        \node[vertex] (c2) at (4,-1){$q_3$};

	 \path[edge] (b) -- node[weight] {a} (a);
	  \path[edge] (c) -- node[weight,yshift = 0.5mm, xshift = -2.5mm] {b} (a);
	   \path[edge] (d) -- node[weight] {b} (a);
	 \path[edge] (c) -- node[weight] {b} (b);
	 \path[edge] (b) -- node[weight,yshift = 2.5mm, xshift = 4.5mm] {b} (d);
	  \path[edge] (c) -- node[weight] {a} (d);
	  
	  	 \path[edge] (b2) -- node[weight] {a} (a2);
	  \path[edge] (c2) -- node[weight] {b} (a2);
	   \path[edge] (d2) -- node[weight,yshift = 0.5mm, xshift = -2.5mm] {b} (a2);
	 \path[edge] (c2) -- node[weight,yshift = -2.5mm, xshift = 1.5mm]  {b} (b2);
	 \path[edge] (b2) -- node[weight] {b} (d2);
	  \path[edge] (c2) -- node[weight] {a} (d2);
	  
	  \path[edge] (a2) --  node[weight] {$2^{10}$} (a);
	 


\end{tikzpicture}

%% file: images/MSTBoruvkaPhases.tex
\begin{tikzpicture}[scale=1.0, auto,swap]
    \foreach \pos/\name in {{(0,2)/a}, {(2,1)/b}, {(3,-0.25)/c},
                            {(0,0)/d}, {(1.0,-1)/e}, {(0,-3)/f}, {(3,-2)/g}}
        \node[vertex] (\name) at \pos {$\name$};
    \foreach \source/ \dest /\weight in {b/a/7, c/b/12,d/a/4,d/b/10,
                                         e/b/3, e/c/5,e/d/15,
                                         f/d/6,f/e/8,
                                         g/e/9,g/f/11}
        \path[edge] (\source) -- node[weight] {$\weight$} (\dest);
\end{tikzpicture}
\begin{tikzpicture}[scale=1.0, auto,swap]
  \foreach \pos/\name in {{(0,2)/a}, {(2,1)/b}, {(3,-0.25)/c},
                            {(0,0)/d}, {(1.0,-1)/e}, {(0,-3)/f}, {(3,-2)/g}}
        \node[vertex] (\name) at \pos {$\name$};
	 \path[edge] (b) -- node[weight] {7} (a);
	 \path[edge] (c) -- node[weight] {12} (b);
	 \path[edge,red] (d) -- node[weight] {4} (a);
	 \path[edge] (d) -- node[weight] {10} (b);
	 \path[edge,red] (e) -- node[weight] {3} (b);
	 \path[edge,red] (e) -- node[weight] {5} (c);
	 \path[edge] (e) -- node[weight] {15} (d);
	 \path[edge,red] (f) -- node[weight] {6} (d);
	 \path[edge] (f) -- node[weight] {8} (e);
	 \path[edge,red] (g) -- node[weight] {9} (e);
	 \path[edge] (g) -- node[weight] {11} (f);


\end{tikzpicture}
\begin{tikzpicture}[scale=1.0, auto,swap]
  \foreach \pos/\name in {{(0,2)/a}, {(2,1)/b}, {(3,-0.25)/c},
                            {(0,0)/d}, {(1.0,-1)/e}, {(0,-3)/f}, {(3,-2)/g}}
        \node[vertex] (\name) at \pos {$\name$};
	 \path[edge,red] (b) -- node[weight] {7} (a);
	 \path[edge] (c) -- node[weight] {12} (b);
	 \path[edge,red] (d) -- node[weight] {4} (a);
	 \path[edge] (d) -- node[weight] {10} (b);
	 \path[edge,red] (e) -- node[weight] {3} (b);
	 \path[edge,red] (e) -- node[weight] {5} (c);
	 \path[edge] (e) -- node[weight] {15} (d);
	 \path[edge,red] (f) -- node[weight] {6} (d);
	 \path[edge] (f) -- node[weight] {8} (e);
	 \path[edge,red] (g) -- node[weight] {9} (e);
	 \path[edge] (g) -- node[weight] {11} (f);


\end{tikzpicture}

%% file: chapters/chapter5.tex
\chapter{Mergegram, extends 0D persistence} 
\label{ch:mergegram} 

\section{Introduction}
\label{sec:intro}

Real-life objects are often represented by unstructured point clouds obtained by laser range scanning or by selecting salient or feature points in images \cite{pauly2002efficient}.
Point clouds are easy to store and can be used as primitives for visualization \cite{zwicker2002pointshop}.
The arguments above give a strong motivation for comparing and classifying unstructured point clouds.
\medskip

\noindent
Rigid objects are naturally studied up to rigid motion or isometry (including reflections), which is any map that preserves inter-point distances.
The recognition of point clouds with fixed number of points is practically solved by the histogram of all pairwise distances, which is a complete isometry invariant in general position \cite{boutin2004reconstructing}.
\medskip

\noindent
Real shapes are often given in a distorted form because of noisy measurements, when points are perturbed, missed or accidentally added. 
One of the first approaches to recognize nearly identical point clouds $A,B$ of different sizes in the same metric space, for example in $\R^m$, is to use the \emph{Hausdorff} distance Definition \ref{dfn:HausdorffDistance}: $\HD(A,B)=\min\ep\geq 0$ such that the first cloud $A$ is covered by $\ep$-balls centered at all points of $B$ and vice versa.  
\medskip

\noindent
However, we also need to take into account infinitely many potential isometries of the ambient space $\R^m$. 
The exact computation of $\inf_{f} \HD(f(A),B)$ minimized over isometries $f$ of $\R^m$ has a high polynomial complexity already for dimension $m=2$ \cite{chew1997geometric}.
An approximate algorithm is cubic in the number of points for $m=3$ \cite{goodrich1999approximate}.
\medskip

\noindent
This chapter extends the 12-page conference version \cite{elkin2020mergegram}, which introduced the new invariant mergegram but didn't prove its continuity under perturbations. In addition to the proof of continuity, another contribution is 
Theorem~\ref{thm:mergegram-to-dendrogram} showing how to reconstruct a dendrogram of single-linkage clustering from a mergegram in general position. 
\medskip

\noindent
The practical novelty is the harder recognition problem including perturbations of isometries within wider classes of affine and projective maps motivated by Computer Vision applications.
Indeed different positions of cameras produce projectively equivalent images of the same rigid shape. 
The new experiments in section~\ref{sec:experiments_mg} extensively compared several approaches on 15000 clouds obtained from real images, see examples in Fig.~\ref{fig:myth_images}.  
\medskip

\newcommand{\scale}{0.13}
\begin{figure}[h!]
    \centering
    \includegraphics[scale = \scale]{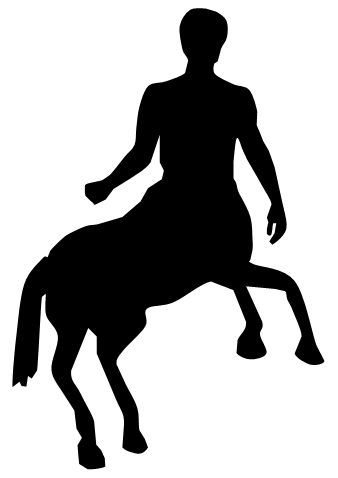}
    \includegraphics[scale = \scale]{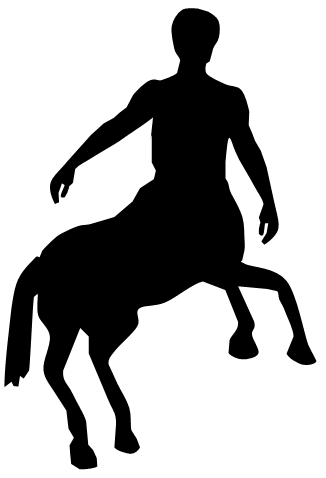}
    \includegraphics[scale = \scale]{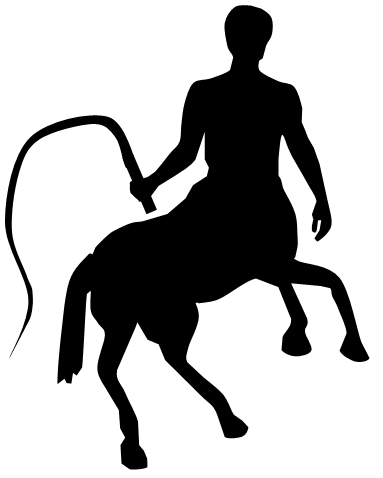}
    \includegraphics[scale = \scale]{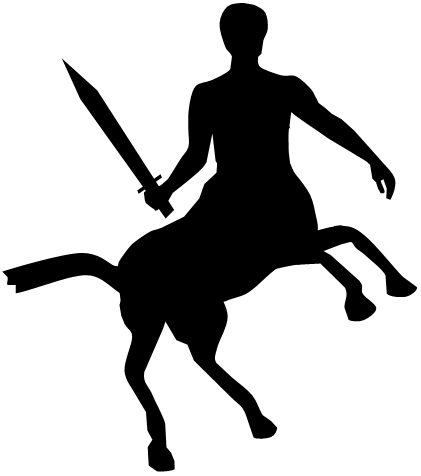}
    \includegraphics[scale = \scale]{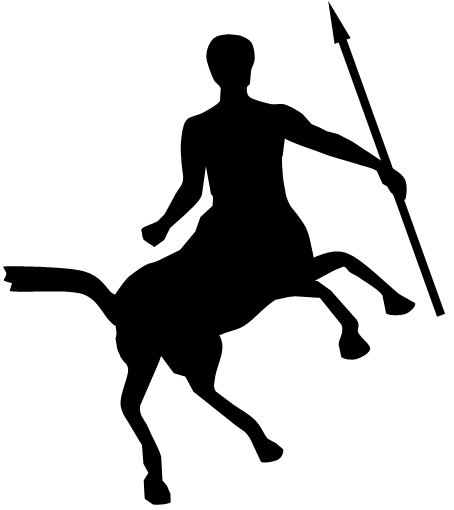}
    \includegraphics[scale = \scale]{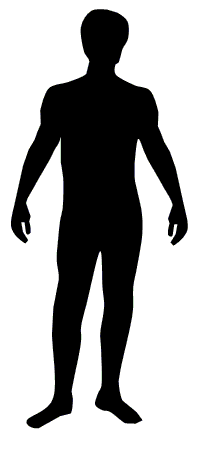}
    \includegraphics[scale = \scale]{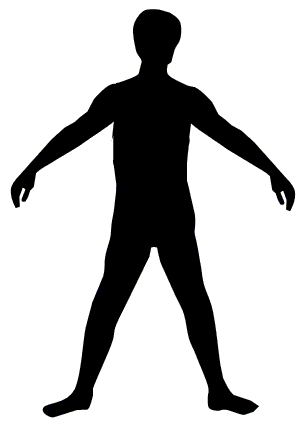}
    \includegraphics[scale = \scale]{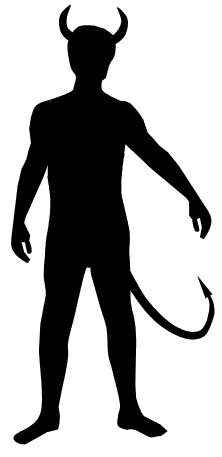}
    \includegraphics[scale = \scale]{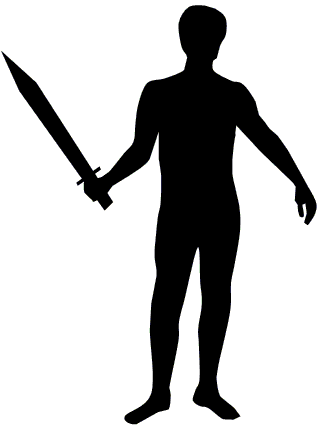}
    \includegraphics[scale = \scale]{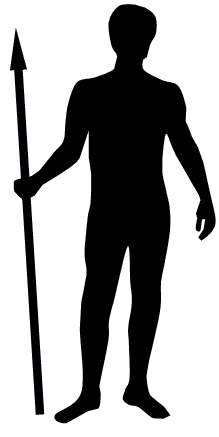}
    \includegraphics[scale = \scale]{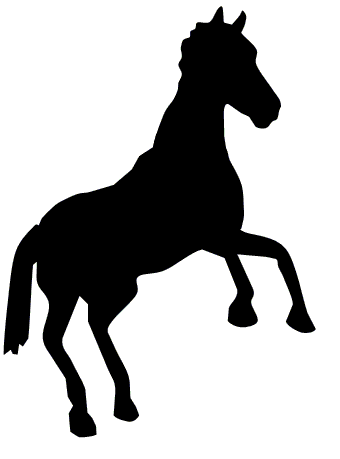}
    \includegraphics[scale = \scale]{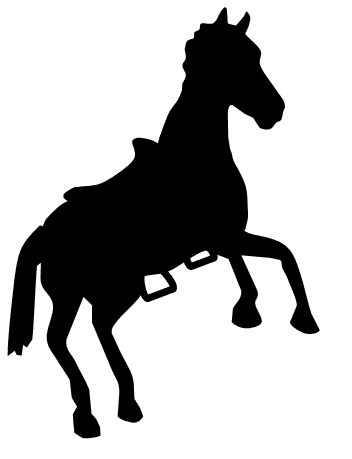}
    \includegraphics[scale = \scale]{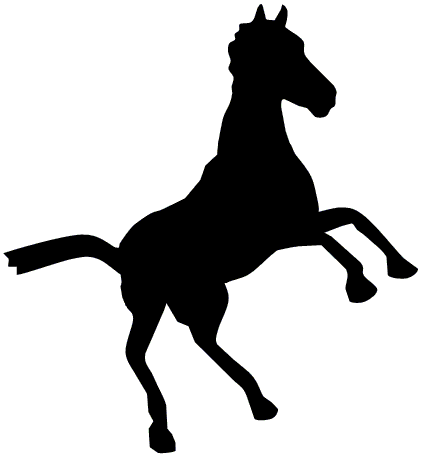}
    \includegraphics[scale = \scale]{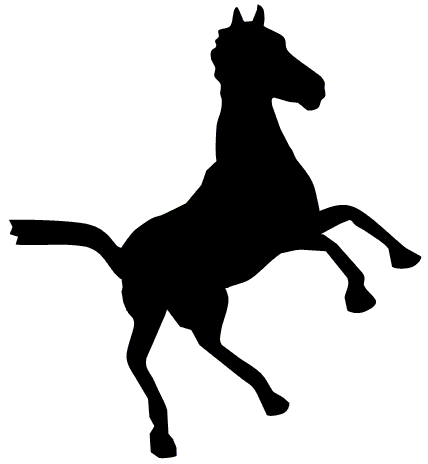}
    \includegraphics[scale = \scale]{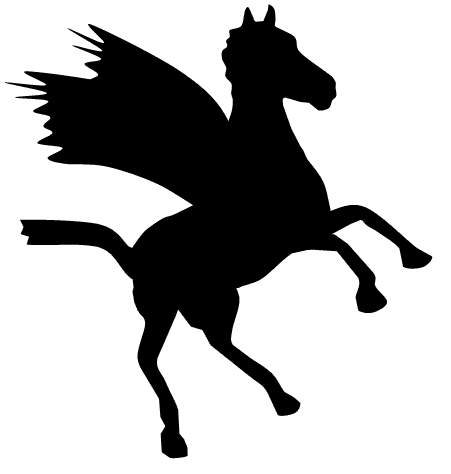}
\caption{Example images from the dataset of mythical creatures, which was introduced in \cite{bronstein2008analysis}.}
    \label{fig:myth_images}
\end{figure}

\begin{prob}[isometry shape recognition under noise]
\label{pro:recognition}
Find an isometry invariant of point clouds in $\R^m$ that is 
\begin{enumerate}[label=(\alph*)]
\item independent of a cloud size
\item provably continuous under perturbations of a cloud
\item computable in a near linear time in the size of a cloud
\item more efficient for recognizing isometry classes of clouds than past invariants.
\end{enumerate}
\bs
\end{prob}

\noindent
The key contributions are proofs of Theorems~\ref{thm:mergegram-to-dendrogram}, \ref{thm:stability_persistence} and Fig.~\ref{fig:max-layer_affine}, \ref{fig:max-layer_projective}, \ref{fig:image-layer_affine} showing that the mergegram achieves a state-of-the-art recognition on substantially distorted images.

\section{Past work}
\label{sec:review}

For the isometry classification of clouds consisting of the same number of points, the easiest invariant is the distribution of all pairwise distances, whose completeness (or injectivity) was proved for all point clouds in general (non-singular) position in $\R^m$ \cite{boutin2004reconstructing}.
\medskip

\noindent
Fixed point clouds $A,B\subset\R^m$ of different sizes can be pairwisely compared by the \emph{Hausdorff} distance \cite{huttenlocher1993comparing}
 $\HD(A,B)=\max\{\sup\limits_{p\in A} d_B(p), \sup\limits_{q\in B} d_A(q)\}$, where $d_B(p)=\inf\limits_{q\in B} d(p,q)$ is the (Euclidean or another) distance from a point $p\in A$ to the cloud $B$.
\medskip

\noindent
The rigid shape recognition problem for non-fixed clouds $A,B$ is harder because of infinitely many potential isometries that can match $A,B$ exactly or approximately. 
Partial cases of this problem were studied for clouds representing surfaces \cite{elad2001bending} and when two clouds have a given isometric matching of one pair of points \cite{ovsjanikov2010one}.
Shape Google by Ovsjanikov et al. practically extends these ideas to non-rigid shape recognition \cite{ovsjanikov2009shape}.
\medskip

\noindent
The most general framework for the isometry shape recognition of point cloud data was proposed by M\'emoi and Sapiro \cite{memoli2005theoretical}.
They studied the Gromov-Hausdorff distance $d_{GH}(A,B)=\inf\limits_{f,g,M} \HD(f(A),g(B))$ minimized over all isometric embeddings $f:A\to M$ and $g:B\to M$ of given point clouds into a metric space $M$.
Since the above definitions involve even more minimizations over infinitely many maps and spaces, $\GH$ can be only approximated.
\medskip

\noindent
The proposed invariant mergegram extends the 0-dimensional persistence in the area of Topological Data Analysis (TDA), which grew from the theory of size functions \cite{verri1993use}.
TDA views a point cloud $A\subset\R^m$ not by fixing any distance threshold but across all scales $s$, for example by blurring given points to balls of a variable radius $s$.
The resulting evolution of topological shapes is summarized by a persistence diagram, which is invariant under isometries of $\R^m$.
TDA can be combined with machine learning and statistical tools due to stability under noise, which was first proved by Cohen-Steiner et al. \cite{cohen2007stability} and then extended to a very general form by Chazal et al. \cite{chazal2016structure}.
\medskip

\noindent
In dimension 0 the \emph{persistence diagram} $\PD(A)$ for distance-based filtrations of a point cloud $A$ consists of the pairs $(0,s)\in\R^2$, where values of $s$ are distance scales at which subsets of $A$ merge by the single-linkage clustering.
These scales equal half-lengths of edges in a \emph{Minimum Spanning Tree} $\MST(A)$.
If distances between all points of $A$ are known, $\MST(A)$ is a connected graph with the vertex set $A$ and a minimal total length.
\medskip

\noindent
Representing a point cloud $A$ by $\PD(A)$ loses a lot of geometry of $A$, but gains stability under perturbations, which can be expressed in the case of point clouds as $\BD(\PD(A),\PD(B))\leq \HD(A,B)$.
Here the \emph{bottleneck distance} $\BD$ between diagrams is defined as a minimum $\ep\geq 0$ such that all pairs of $\PD(A)$ can be bijectively matched to $\ep$-close points of $\PD(B)$ or to diagonal pairs $(s,s)$, and vice versa.
Here $\ep$-closeness of pairs $(a,c)$ and $(b,d)$ in $\R^2$ is measured in the  distance $L_{\infty}=\max\{|a-b|,|c-d|\}$.
\medskip

\noindent
The \emph{mergegram} extends $\PD(A)$ to a stronger invariant whose stability under perturbations in the above sense is proved in section~\ref{sec:stability} for the first time. 
The idea of a mergegram is related to the Reeb graph \cite{parsa2013deterministic} or the merge tree \cite{morozov2013interleaving} for the sublevel set filtration of a scalar function. 
The mergegram $\MG$ is defined at a more abstract level for any clustering dendrogram, which can be reconstructed from $\MG$ in general position.
\medskip


\noindent
Since any persistence diagram and a mergegram are unordered collections of pairs, the experiments in section~\ref{sec:experiments_mg} will use the neural network PersLay \cite{carriere2019perslay} whose output is invariant under permutations of input points by design.
PersLay extends the neural network DeepSet \cite{zaheer2017deep} for unordered sets and introduces new layers to specifically handle persistence diagrams, as well as a new form of representing such permutation-invariant layers.
In other related work deep learning was recently applied to outputs of hierarchical clustering \cite{fu2019learning}, \cite{cirrincione2020gh}, \cite{karim2021deep} and to 0-dimensional persistence \cite{clough2020topological}, \cite{gabrielsson2020topology}. 

\section{Single-linkage clustering and mergegram}
\label{sec:mergegram}

This section introduces dendrogram in Definition~\ref{dfn:dendrogram} and mergegram in Definition~\ref{dfn:mergegram}.

\begin{exa}
\label{exa:5-point_line}
Fig.~\ref{fig:5-point_line} illustrates the key concepts before formal Definitions~\ref{dfn:sl_clustering}, \ref{dfn:dendrogram}, \ref{dfn:mergegram} for the point cloud $A = \{0,1,3,7,10\}$ in the real line $\R$.
Imagine that we gradually blur original data points by growing disks of the same radius $s$ around the given points.

\input figures/5-point_line.txt

\noindent
The disks of the closest points $0,1$ start overlapping at the scale $s=0.5$ when these points merge into one cluster $\{0,1\}$.
This merger is shown by blue arcs joining at the node at $s=0.5$ in the single-linkage dendrogram, see the bottom left picture in Fig.~\ref{fig:5-point_line}.
The persistence diagram $\PD$ in the bottom middle picture of
 Fig.~\ref{fig:5-point_line} represents this merger by the pair $(0,0.5)$ meaning that a singleton cluster of (say) point $1$ was born at the scale $s=0$ and then died later at $s=0.5$ by merging into another cluster of point $0$.
\medskip

\noindent
When clusters $\{0,1,3\}$ and $\{7,10\}$ merge at $s=2$, this event was previously encoded in the persistence diagram by the single pair $(0,2)$ meaning that one cluster inherited from (say) point 7 was born at $s=0$ and died at $s=2$.
The new mergegram in the bottom right picture of Fig.~\ref{fig:5-point_line} represents the above merger by the following two pairs.
The pair $(1,2)$ means that the cluster $\{0,1,3\}$ is merging at the current scale $s=2$ and was previously formed at the smaller scale $s=1$.
The pair $(1.5,2)$ means that another cluster $\{7,10\}$ is merging at the scale $s=2$ and was previously formed at $s=1.5$. 
\medskip

\noindent
The 0D persistence diagram represents the cluster of the whole cloud $A$  by the pair $(0,+\infty)$, because $A$ was inherited from a singleton cluster starting from $s=0$.
The mergegram represents the same cluster $A$ by the pair $(2,+\infty)$, because $A$ was formed during the last merger of $\{0,1,3\}$ and $\{7,10\}$ at $s=2$ and continues to live as $s\to+\infty$.
\medskip
%
\end{exa}

\noindent 
Definition~\ref{dfn:dendrogram} below extends a dendrogram Definition \ref{dfn:dendrogramOriginal} to arbitrary (possibly, infinite) sets $A$.
Since every partition of $A$ is finite by Definition~\ref{dfn:partition}, we don't need to add that an initial partition of $A$ is finite.
Non-singleton sets are now allowed.

\begin{dfn}[dendrogram $\De$ of merge sets]
\label{dfn:dendrogram}
A \emph{dendrogram} $\De$ over any set $A$ is a function $\Delta:[0,+\infty)\to\PS(A)$ of a scale $s\geq 0$ satisfying the following conditions.
\smallskip

\noindent
(\ref{dfn:dendrogram}a)
There exists a scale $r\geq 0$ such that $\De(A;s)$ is the single block partition for $s\geq r$. 
\smallskip

\noindent
(\ref{dfn:dendrogram}b)
If $s\leq t$, then $\De(A;s)$ \emph{refines} $\De(A;t)$, so any set from $\De(t)$ is a subset of some set from $\De(A;t)$.
These inclusions of subsets of $X$ induce the natural map $\De_s^t:\De(s)\to\De(t)$.
\smallskip

\noindent
(\ref{dfn:dendrogram}c)
There are finitely many \emph{merge scales} $s_i$ such that $$s_0 = 0 \text{ and  } s_{i+1} = \text{sup}\{s \mid \text{ the map }  \De_s^t \text{ is identity for } s' \in [s_i,s)\}, i=0,\dots,m-1.$$

\noindent
Since $\De(A;s_{i})\to\De(A;s_{i+1})$ is not an identity map, there is a subset $B\in\De(s_{i+1})$ whose preimage consists of at least two subsets from $\De(s_{i})$.
This subset $B\subset X$ is called a \emph{merge} set and its \emph{birth} scale is $s_i$.
All sets of $\De(A;0)$ are merge sets at the birth scale 0.
The $\life(B)$ is the interval $[s_i,t)$ from its birth scale $s_i$ to its \emph{death} scale $t=\sup\{s \mid \De_{s_i}^s(B)=B\}$.
\bs
\end{dfn}

\noindent
Dendrograms are usually drawn as trees whose nodes represent all sets from the partitions $\De(A;s_i)$ at merge scales.
Edges of such a tree connect any set $B\in\De(A;s_{i})$ with its preimages under $\De(A;s_{i})\to\De(A;s_{i+1})$.
Fig.~\ref{fig:3-point_dendrogram} shows $\De$ for $A=\{0,1,2\}$.
\medskip

\begin{figure}[h]
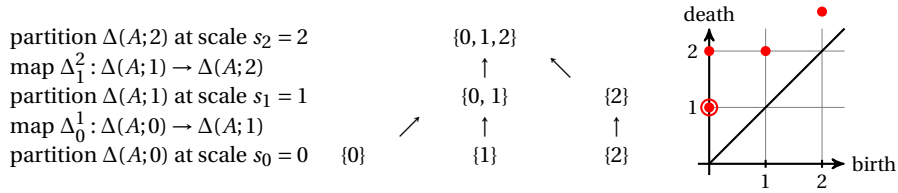

\input figures/3-point_dendrogram.txt
\caption{The dendrogram $\De$ on $A=\{0,1,2\}$ and its mergegram  $\MG(\De)$ from Definition~\ref{dfn:mergegram}.}
\label{fig:3-point_dendrogram}
\end{figure}

\noindent
In Fig.~\ref{fig:3-point_dendrogram} the partition $\De(A;1)$ consists of $\{0,1\}$ and $\{2\}$.
The maps $\De_s^t$ induced by inclusions respect the compositions in the sense that $\De_s^t\circ\De_r^s=\De_r^t$ for any $r\leq s\leq t$.
For example, $\De_0^1(\{0\})=\{0,1\}=\De_0^1(\{1\})$ and $\De_0^1(\{2\})=\{2\}$, so $\De_0^1$ is a well-defined map from the partition $\De(A;0)$ of 3 singleton sets to $\De(A;1)$, but isn't an identity.
\medskip

\noindent
At the scale $s_0=0$ the merge sets $\{0\},\{1\}$ have $\life=[0,1)$, while the merge set $\{2\}$ has $\life=[0,2)$.
At the scale $s_1=1$ the only merge set $\{0,1\}$ has $\life=[1,2)$.
At the scale $s_2=2$ the only merge set $\{0,1,2\}$ has $\life=[2,+\infty)$.
The notation $\De$ is motivated as the first (Greek) letter in the word dendrogram and by a $\De$-shape of a typical tree.
\medskip

\noindent
Condition~(\ref{dfn:dendrogram}a) says that 
a partition of a set $X$ is trivial for all large scales.
Condition~(\ref{dfn:dendrogram}b) means that if the scale $s$ is increasing, then sets from a partition $\De(s)$ can only merge but can not split. 
Condition~(\ref{dfn:dendrogram}c) implies that there are only finitely many mergers, when two or more subsets of $X$ merge into a larger merge set.
\medskip

\begin{lem}\cite[Lemma 3.3]{elkin2020mergegram}
\label{lem:sl_clustering}
Given a metric space $(X,d)$ and a finite set $A\subset X$, the single-linkage dendrogram $\De_{SL}(X)$ from Definition~\ref{dfn:sl_clustering} satisfies Definition~\ref{dfn:dendrogram}.
\bs
\end{lem}

\noindent
A \emph{mergegram} represents life spans of merge sets by pairs 
$(\birth,\death)\in\R^2$.

\begin{dfn}[mergegram $\MG(\De)$]
\label{dfn:mergegram}
The \emph{mergegram} of a dendrogram $\De$ has the pair $(\birth,\death)\in\R^2$ for each merge set $B$ of $\De$ with $\life(B)=[\birth,\death)$.
If any life interval appears $k$ times, the pair (birth,death) has the multiplicity $k$ in $\MG(\De)$.
\bs
\end{dfn}

\noindent
If our input is a point cloud $A$ in a metric space, then the mergegram $\MG(\De_{SL}(A))$ is an isometry invariant of $A$, because $\De_{SL}(A)$ depends only on inter-point distances.
Though $\De_{SL}(A)$ as any dendrogram is unstable under perturbations of points, the key advantage of $\MG(\De_{SL}(A))$ is its stability, which will be proved in Theorem~\ref{thm:stability_mergegram}.


\medskip

\noindent
The dendrogram $\De$ in the first picture of Fig.~\ref{fig:5-point_set_mergegram} generates the mergegram as follows:
\begin{itemize}
\item 
each of the singleton sets $\{b\}$, $\{c\}$, $\{p\}$, $\{q\}$ has pair (0,1), so its multiplicity is 4; 
\item 
each of the merge sets $\{b,c\}$ and $\{p,q\}$ has the pair (1,2), so its multiplicity is 2; 
\item 
the singleton set $\{a\}$ has the pair $(0,3)$;
the merge set $\{b,c,p,q\}$ has the pair (2,3);
\item
the full set $\{a,b,c,p,q\}$ continues to leave up to $s=3$, hence has the pair $(3,+\infty)$.
\end{itemize}

\begin{figure}[H]
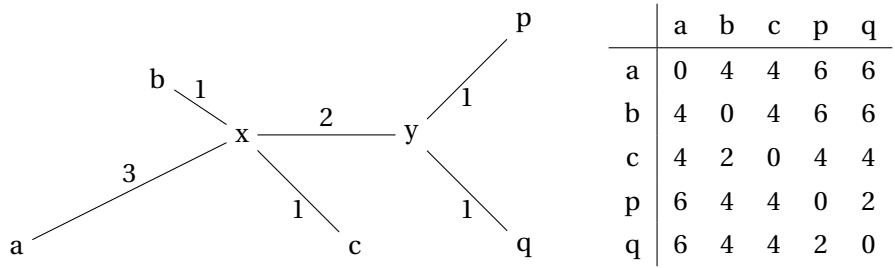
 
\input figures/5-point_set.txt
\caption{The set $X=\{a,b,c,d,e\}$ has the distance matrix defined by the shortest path metric.}
\label{fig:5-point_set}
\end{figure}

\begin{figure}[H]
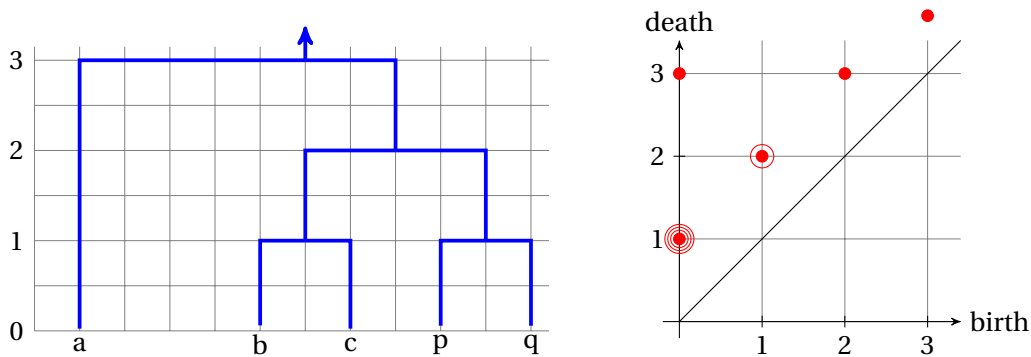

\input figures/5-point_mergegram.txt
\caption{\textbf{Left}: the dendrogram $\De$ for the single linkage clustering of the set 5-point set $X=\{a,b,c,d,e\}$ in Fig.~\ref{fig:5-point_set}.
\textbf{Right}: the mergegram $\MG(\De)$ with one pair (0,1) of multiplicity 4.}
\label{fig:5-point_set_mergegram}
\end{figure}

\section{Relations between 0-dim persistence and mergegram}
\label{sec:relations}

This section recalls the concept of persistence and then shows how any 0D persistence and dendrogram in general position can be reconstructed from a mergegram.

\medskip

\noindent
Key results of this section are
Theorem~\ref{thm:mergegram_to_0D_persistence} which describes how the persistence diagram $\PD$ of the distance-based filtration of any point cloud $A$ can be obtained from the mergegram $\MG(\De_{\SL}(S))$ and Theorem~\ref{thm:mergegram-to-dendrogram} which shows that in some cases dendrogram $\De_{\SL}(S)$ can be reconstructed from mergegram $\MG(\De_{\SL}(S))$.

\begin{thmm}\cite[Theorem~5.3]{elkin2020mergegram}
\label{thm:mergegram_to_0D_persistence} 
For a finite set $A$ in a metric space $(X,d)$, let $d_A:X\to\R$ be the distance to $A$.
Let the mergegram $\MG(\De_{SL}(A))$ be a multiset $\{(b_i,d_i)\}_{i=1}^k$, where some pairs can be repeated.
Then the persistence diagram $\PD\{H_0(X_s^{d_A})\}$ is the difference of the multisets $\{(0,d_i)\}_{i=1}^{k}-\{(0,b_i)\}_{i=1}^{k}$ containing each pair $(0,s)$ exactly $\#b-\#d$ times, where $\#b$ is the number of births $b_i=s$ and  $\#d$ is the number of deaths $d_i=s$.
All trivial pairs $(0,0)$ are ignored, alternatively we take $\{(0,d_i)\}_{i=1}^{k}$ only with $d_i>0$.
\bs
\end{thmm}

\noindent
Theorem~\ref{thm:mergegram_to_0D_persistence} is illustrated by Example~\ref{exa:5-point_line}, where $A=\{0,1,3,7,10\}$ has the diagram $\PD(A)=\{(0,0.5),(0,1),(0,1.5),(0,2),(0,+\infty)\}$ obtained from the mergegram
$$\MG(\De_{\SL}(A))=\{(0,0.5),(0,0.5),(0,1),(0,1.5),(0,1.5),(0.5,1),(1,2),(1.5,2),(2,+\infty)\}$$
as follows: one pair $(0,0.5)\in\PD(A)$ comes from two deaths and one birth $s=0.5$ in $\MG(\De_{\SL}(A))$.
Similarly each of the pairs $(0,1),(0,1.5),(0,2)\in\PD(A)$ comes from two deaths and one birth equal to the same scale $s$.
The cloud $B=\{0,4,6,9,10\}\subset\R$ in \cite[Example~1.1]{elkin2020mergegram} has exactly the same $\PD(B)=\PD(A)$ but different $\MG(\De_{\SL}(B))\neq\MG(\De_{\SL}(A))$.
This example together with Theorem~\ref{thm:mergegram_to_0D_persistence} justify that the mergegram is strictly stronger than 0D persistence as an isometry invariant of a point cloud.
\medskip
 

\noindent
New Reconstruction Theorem~\ref{thm:mergegram-to-dendrogram} below can be contrasted with the weakness of 0D persistence $\PD\{H_0(X_s^{d_A})\}$ consisting of only pairs $(0,s)$ whose finite deaths are half-lengths of edges in a Minimum Spanning Tree $\MST(A)$.
In Example~\ref{exa:5-point_line} these scales $s=0.5,1,1.5,2$ are insufficient to reconstruct the SL dendrogram in Fig.~\ref{fig:5-point_line}. 
Such a unique reconstruction is possible by using the richer invariant mergegram as follows. 

\begin{thmm}[from a mergegram to a dendrogram]
\label{thm:mergegram-to-dendrogram}
Let $A$ be a finite point cloud in \emph{general position} in the sense that all merge scales of $A$ in a dendrogram $\De$ from Definition~\ref{dfn:dendrogram} are different.
Then the dendrogram $\De$ can be reconstructed from its mergegram $\MG(\De)$, uniquely up to a permutation of nodes in $\De$ at scale $s=0$.
\bs
\end{thmm}
\begin{proof}
\noindent
Consider all merge scales one by one in the increasing order starting from the smallest.
The general position implies that only two clusters merge at any merge scale. 
For any current scale $s$, the mergegram contains exactly two pairs $(b_1,s)$ and $(b_2,s)$.
\medskip

\noindent
For a smallest merge scale $s>0$, the births should be $b_1=b_2=0$.
We start drawing a dendrogram $\De$ by merging any two points of $A$ at this smallest scale $s$.
To realize a merger at any larger $s$, we should select two clusters representing $(b_1,s)$ and $(b_2,s)$. 
\medskip

\noindent
If $b_i=0$ then we take any of the unmerged points of $A$.
If $b_i>0$ then the already constructed dendrogram should contain a unique non-singleton cluster determined by the scale $b_i\in(0,s)$.
Hence at any merge scale $s$ we know how to select two clusters to merge.
The only choice comes from choosing points of $A$ or permuting notes of $\De$.
\end{proof}

\noindent 
Following the above proof of Theorem~\ref{thm:mergegram-to-dendrogram} for the cloud $A=\{0,1,3,7,10\}$ in Example~\ref{exa:5-point_line}, the first two pairs $(0,0.5)\in\MG(\De_{\SL}(A))$ indicate that we should merge two points of $A$ at $s=0.5$.
The scale $s=0.5$ uniquely determines this 2-point cluster. 
\medskip

\noindent
The next two pairs $(0,1),(0.5,1)$ mean that the above cluster born at $s=0.5$ should merge at $s=1$ with a singleton cluster (any free point of $A$).
The resulting 3-point cluster is uniquely determined by its merge scale $s=1$.
The further two pairs $(0,1.5),(0,1.5)$ say that a new 2-point cluster is formed at $s=1.5$ by the two remaining points of $A$.
\medskip

\noindent
The final pairs $(1,2),(1.5,2)$ tell us to merge at $s=2$ the two clusters formed earlier at $s=1$ and $s=1.5$.
The resulting dendrogram $\De$ has the expected combinatorial structure as in Fig.~\ref{fig:5-point_line}, though we can draw $\De$ in another way by permuting points of $A$.

\section{Stability of the mergegram}
\label{sec:stability}

This section fully proves the stability of a mergegram in Theorem~\ref{thm:stability_mergegram} with help of Lemmas~\ref{lem:merge_module_decomposition} and \ref{lem:merge_modules_interleaved}.
For simplicity, we consider vector spaces with coefficients only in $\Z_2=\{0,1\}$, which can be replaced by any field.
\medskip

\noindent
Definition~\ref{dfn:homo_modules} introduces homomorphisms between persistence modules, which are needed to state the stability of persistence diagrams $\PD\{H_0(X_s^f)\}$ under perturbations of a function $f:X\to\R$.
This result will imply a stability for the mergegram $\MG(\De_{SL}(A))$ for the dendrogram $\De_{SL}(A)$ of the single-linkage clustering of a set $A\subset X$.

\begin{dfn}[merge module $M(\De)$]
\label{dfn:merge_module}
For any dendrogam $\De$ on a finite set $X$,
the \emph{merge module} $M(\De)$ consists of the vector spaces $M_s(\De)$, $s\in\R$, and linear maps $m_s^t:M_s(\De)\to M_t(\De)$, $s\leq t$.
For any $s\in\R$ and $A\in\De(s)$, the space $M_s(\De)$ has the generator or a basis vector $[A]\in M_s(\De)$.
For $s<t$ and any set $A\in\De(s)$, 
if the image of $A$ under $\De_s^t$ coincides with $A\subset X$, so $\De_s^t(A)=A$, then $m_s^t([A])=[A]$, else $m_s^t([A])=0$. 
\bs
\end{dfn}

\begin{figure}[h]
\begin{tabular}{lccccccccc}
scale $s_3=+\infty$ & 0 & & & & & 0 \\
map $m_2^{+\infty}$ & $\uparrow$ & & & & & $\uparrow$\\
scale $s_2=2$ & $\Z_2$ & & & 0 & 0 & [\{0,1,2\}]\\
map $m_1^2$ & $\uparrow$ & & & $\uparrow$ & $\uparrow$\\
scale $s_1=1$ & $\Z_2\oplus\Z_2$ & 0 & 0 & [\{2\}] & [\{0,1\}] \\
map $m_0^1$ & $\uparrow$ & $\uparrow$ & $\uparrow$ & $\uparrow$ \\
scale $s_0=0$ & $\Z_2\oplus\Z_2\oplus\Z_2$ & [\{0\}] & [\{1\}] & [\{2\}] &
\end{tabular}
\caption{The merge module $M(\De)$ of the dendrogram $\De$ on the set $X=\{0,1,2\}$ in Fig.~\ref{fig:3-point_dendrogram}.}
\label{fig:3-point_module}
\end{figure}

\noindent
In a dendrogram $\De$ from Definition~\ref{dfn:dendrogram}, any merge set $A$ of $\De$ has $\life(A)=[b,d)$ from its birth scale $b$ to its death scale $d$.
Lemmas~\ref{lem:merge_module_decomposition} and~\ref{lem:merge_modules_interleaved} are proved for the first time. 

\begin{lem}[merge module decomposition]
\label{lem:merge_module_decomposition}
For any dendrogram $\De$ from Definition~\ref{dfn:dendrogram}, the merge module $M(\De)\cong\bigoplus\limits_{A}\mathbb{I}(\life(A))$ decomposes over all merge sets $A$.
\bs
\end{lem}
\begin{proof}[Proof of Lemma~\ref{lem:merge_module_decomposition}]
The goal is to prove that $M(\triangle) \cong \bigoplus_{A}\mathbb{I}(\text{life}(A))$. 
Recall that the interval module $\mathbb{I}(\text{life}(A))$
 consists of only vector spaces $0$ and $\Z_2$.
For a scale $r$, let $\mathbb{I}_r(\text{life}(A))$ be its vector space, whose generator is denoted by $[\mathbb{I}_r(\text{life}(A))]$.  
Define
$$\psi_r:M_r(\triangle) \rightarrow \bigoplus_{A}\mathbb{I}_r(\text{life}(A)) \text{ such that } [A] \rightarrow [\mathbb{I}_r(\text{life}(A))] \text{ for all }A \in \triangle(r),$$
$$\phi_r:\bigoplus_{A}\mathbb{I}_r(\text{life}(A)) \rightarrow M_r(\Delta) \text{ such that } [\mathbb{I}_r(\text{life}(A))] \rightarrow [A] \text{ for all }\text{life}(A) \text{ containing }r.$$

\noindent
We will first prove that $\phi_r$ is well-defined. 
If $r \in \text{life}(A)$ then $A \in M_r(\triangle)$.
We know that $M_r(\triangle)$ is generated by elements $A \in \triangle(r)$ for which $r \in \text{life}(A)$. 
Thus the compositions satisfy $\phi_r \circ \psi_r = \text{id}_r$ and $\psi_r \circ \phi_r = \text{id}_r$.
It remains to prove that morphisms correctly behave under the functors $\psi, \phi$. 
The proofs for both cases are essentially the same, thus we will prove it only for $\psi$. 
The goal is to prove that the following diagram commutes:

\begin{figure}[H]
\centering 
\begin{tikzpicture}[scale=1.0]
  \matrix (m) [matrix of math nodes,row sep=3em,column sep=4em,minimum width=2em]
  {
     M_s(\triangle) & M_t(\triangle) \\
     \bigoplus_{A}\I_s(\text{life}(A)) & \bigoplus_{A}\I_t(\text{life}(A)) \\};
  \path[-stealth]
    (m-1-1) edge node [left] {$\psi_s$} (m-2-1)
            edge [->] node [above] {$m^t_s$} (m-1-2)
    (m-2-1.east|-m-2-2) edge node [above] {$i^t_s$}
            node [above] {} (m-2-2)
      (m-1-2) edge node [right] {$\psi_t$} (m-2-2);
\end{tikzpicture}
\end{figure}

\noindent
Here $i^t_s$ is the direct sum of the corresponding maps of interval modules $\bigoplus_{A}(i^t_s)^A $ . 
Let $[A]$ be an arbitrary generator of $M_r(\triangle)$. 
There are two possibilities how $m^t_s$ can map $[A]$. 
If $t \in \text{life}(A)$, then $m^t_s([A]) = [A] \in M_t(\triangle)$ and by definition $$\phi_t \circ m^t_s([A]) = [\mathbb{I}_t(\text{life}(A))].$$
Since both $s,t \in \text{life}(A)$, we also have that
$$m^t_s \circ \phi_t([A]) = [\mathbb{I}_t(\text{life}(A))] = \phi_t \circ m^t_s([A]).$$
Assume now that $t \notin \text{life}(A)$. 
Then $m^t_s([A]) = 0$ and thus $\phi_t(m^t_s([A])) = 0$. On the other hand $i^t_s\circ\phi_s([A]) = [\mathbb{I}_t(\text{life}(A))] = \mathbb{Z}_2$.
Since $t \notin \text{life}(A)$, we get $i^t_s\circ\phi_s([A]) = 0$.
\end{proof}


\begin{lem}[merge modules interleaved]
\label{lem:merge_modules_interleaved}
If subsets $A,B$ of a metric space $(X,d)$ have $\HD(A,B)=\de$, then the merge modules $M(\De_{SL}(A))$, $M(\De_{SL}(B))$ are $\de$-interleaved.
\bs
\end{lem}
\begin{proof}[Proof of Lemma~\ref{lem:merge_modules_interleaved}]
The equality $\HD(A,B)=\de$ means that $A$ is covered by the union of closed balls that have the radius $\de$ and centers at all points $b\in B$.
This union is the preimage is $d_B^{-1}([0,\de])$, i.e. $A\subset d_B^{-1}([0,\de])$.
Extending the distance values by $s\geq 0$, we get $d_A^{-1}([0,s])\subset d_B^{-1}([0,s+\de])$ and similarly $d_B^{-1}([0,s])\subset d_A^{-1}([0,s+\de])$.
\medskip

\noindent
Let $U$ be an arbitrary set in $\De_{SL}(A)$.
Define map $\phi_r:M(A; r) \rightarrow M(B; r+\delta)$ 
$$\phi_r([U]) = \left\{ \begin{array}{ll} 
[U], & \mbox{ if } r+\delta\in \text{life}(\De_{SL}(B),U), \\
0, & \mbox{ otherwise }.
\end{array} \right.\qquad $$

\noindent
Symmetrically for any $V \in \De_{SL}(B)$ we define $\psi_r:M(B; r) \rightarrow M(A; r+\delta) $ 

$$\psi_r([V]) = \left\{ \begin{array}{ll} 
[V], & \mbox{ if } r+\delta\in \text{life}(\De_{SL}(A),V), \\
0, & \mbox{ otherwise }.
\end{array} \right.\qquad $$
In the notation above, $\text{life}(\De_{SL}(B),U)$ is the $\text{life}(U)$ in the dendrogram $\De_{SL}(B)$. 
If $U \notin \De_{SL}(B)(t)$ for all values $t$, then $\text{life}(U) = \emptyset$.
By symmetry it is enough to prove that the following diagrams commute:
\begin{figure}[H]
\centering
\begin{tikzpicture}[scale=1.2]
  \matrix (m) [matrix of math nodes,row sep=3em,column sep=4em,minimum width=2em]
  {
     M_s(\De_{SL}(A)) & M_t(\De_{SL}(A)) \\
     M_{s+\delta}(\De_{SL}(B)) & M_{t+\delta}(\De_{SL}(B)) \\};
  \path[-stealth]
    (m-1-1) edge node [left] {$\phi_s$} (m-2-1)
            edge [->] node [above] {$m^t_s$} (m-1-2)
    (m-2-1.east|-m-2-2) edge node [above] {$m^{t+\delta}_{s+\delta}$}
            node [above] {} (m-2-2)
      (m-1-2) edge node [right] {$\phi_t$} (m-2-2);
\end{tikzpicture}
\end{figure}
\medskip
\begin{figure}[H]
\centering
\begin{tikzcd}
 & M_s(\De_{SL}(B)) \arrow["\psi_s"]{dr}{} \\
M_{s-\delta}(\De_{SL}(A)) \arrow["\phi_{s-\delta}"]{ur}{} \arrow["m^{s+\delta}_{s-\delta}"]{rr}{} && M_{s+\delta}(\De_{SL}(A))
\end{tikzcd}
\end{figure}

\noindent
We note first that if $[a,b) = (\text{life}(\De_{SL}(A),U)$, then $(\text{life}(\De_{SL}(B),U) \subseteq [a-\epsilon, b+\epsilon)$
\medskip

\noindent
We begin by proving commutativity of the first diagram. Let $U$ be arbitrary element of $\De_{SL}(A)(s)$. If $t \notin \text{life}(\De_{SL}(A),U)$ then $\phi_t \circ m^t_s = 0$. If $s+\delta \notin \text{life}(\De_{SL}(B),U)$ or $t+\delta \notin \text{life}(\De_{SL}(B),U)$ then we are done. 
Since $t \notin \text{life}(\De_{SL}(A),U)$, it follows that $t+\delta \notin \text{life}(\De_{SL}(B),U)$.
And thus with given assumptions the diagram commutes.
\medskip

\noindent
Assume now that $t+\delta \notin \text{life}(\De_{SL}(A),U) $. Then both $\phi_t(m^t_s(U)) = 0 = m^{t+\delta}_{s+\delta}(\phi_s(U))$. 
In the last case we assume that $t \in \text{life}(\De_{SL}(A),U)$ and $t+\delta \in \text{life}(\De_{SL}(B),U)$. 
In this case obviously $s+\delta \in \text{life}(\De_{SL}(B),U)$ and thus $\phi_t(m^t_s([U])) = [U] = m^{t+\delta}_{s+\delta}(\phi_s([U]))$.
\medskip

\noindent
For the second diagram assume now that $U \in M(\De_{SL}(A))(s-\delta) $. 
Assume first that $s \notin \text{life}(\De_{SL}(B),U)$, then $s+\delta \notin \text{life}(\De_{SL}(B),U)$ and $m^{s+\delta}_{s-\delta}([U]) = 0 = \psi_s(\phi_{s-\delta}[U])$.
\medskip

\noindent
Assume then that $s \in \text{life}(\De_{SL}(B),U)$. Now outcome of both maps $\psi_s$ and $m^{s+\delta}_{s-\delta}$ depend on if $s+\delta \in \text{life}(\De_{SL}(A),U)$ and thus $m^{s+\delta}_{s-\delta}([U]) = \psi_s(\phi_{s-\delta}[U])$.
Since all the diagrams commute, the required conclusion follows.
\end{proof}

\begin{thmm}[stability of a mergegram]
\label{thm:stability_mergegram}
The mergegrams of any finite point clouds $A,B$ in a metric space $(X,d)$ satisfy $\BD(\MG(\De_{SL}(A)),\MG(\De_{SL}(B))\leq \HD(A,B)$.
Hence any small perturbation of a cloud $A$ in the Hausdorff distance yields a similarly small perturbation in the bottleneck distance for its mergegram $\MG(\De_{SL}(A))$.
\bs
\end{thmm}
\begin{proof}
The given clouds $A,B\subset X$ with $\HD(A,B)=\de$ have $\de$-interleaved merge modules by Lemma~\ref{lem:merge_modules_interleaved}, so $\ID(\MG(\De_{SL}(A)),\MG(\De_{SL}(B))\leq\de$.
Since any merge module $M(\De)$ is finitely decomposable, hence is $q$-tame by Lemma~\ref{lem:merge_module_decomposition}. 
The corresponding mergegram $\MG(M(\De))$ satisfies Theorem~\ref{thm:stability_persistence}, so
$\BD(\MG(\De_{SL}(A)),\MG(\De_{SL}(B))\leq\de$.
\end{proof}

\begin{figure}[h!]
\centering
\includegraphics[width=0.49\linewidth]{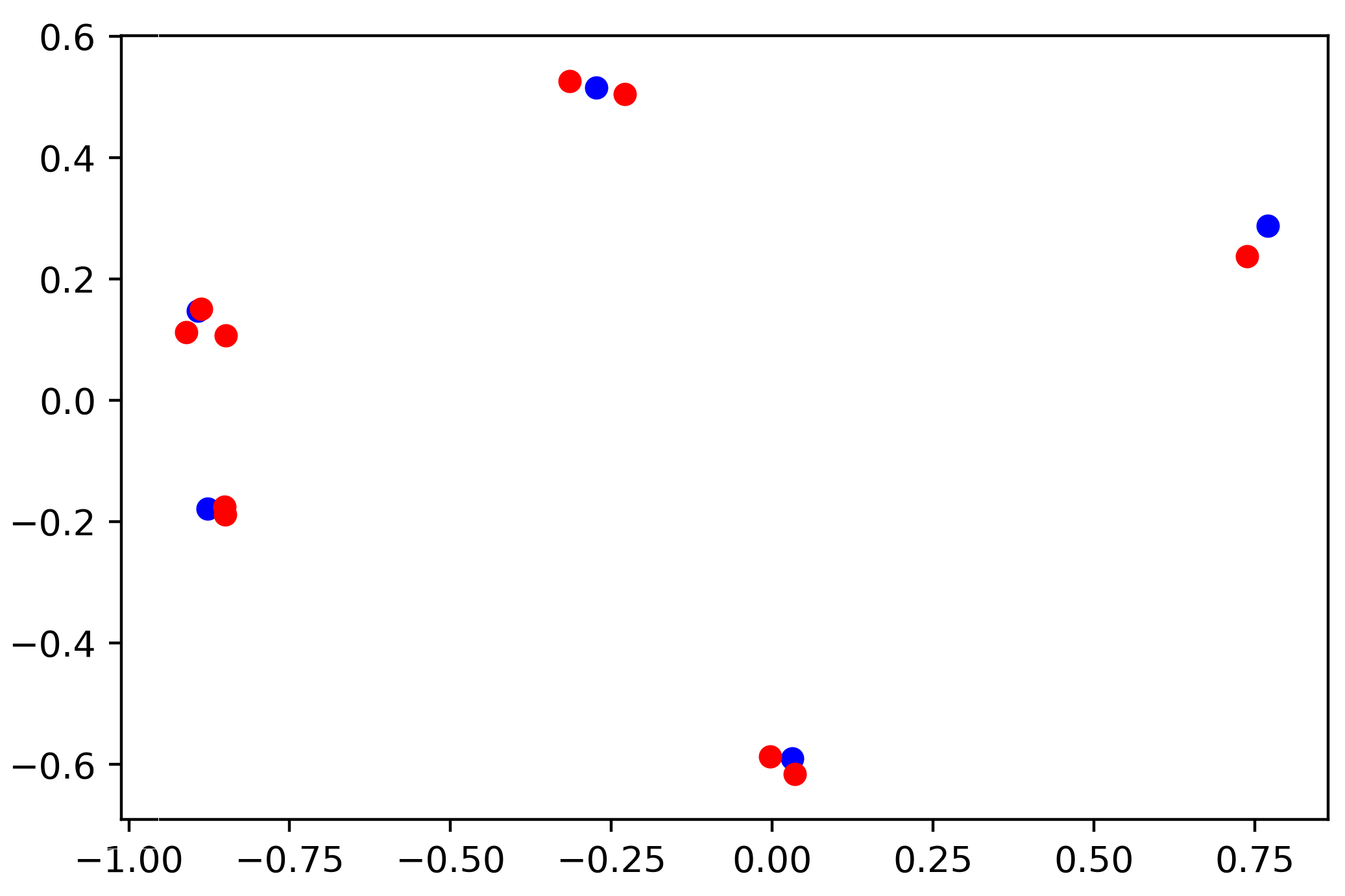}
\includegraphics[width=0.49\linewidth]{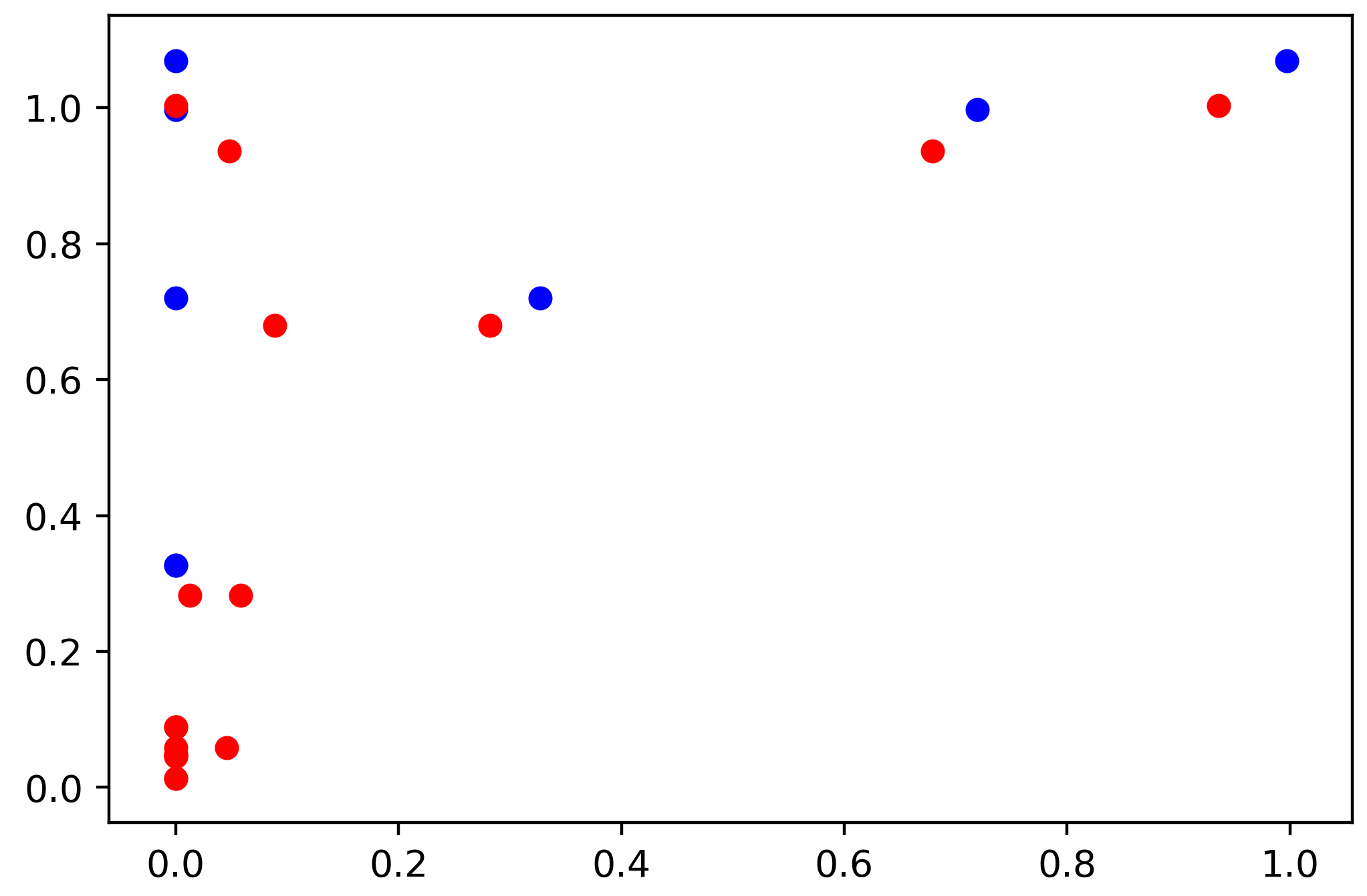}
\caption{\textbf{Left}: the cloud $C$ of 5 blue points is close to the cloud $C'$ of 10 red points in the Hausdorff distance. \textbf{Right}: the mergegrams are close in the bottleneck distance as predicted by Theorem~\ref{thm:stability_mergegram}.}
    \label{fig:perturbed_mergegram}
\end{figure}

\noindent 
Fig.~\ref{fig:perturbed_mergegram} illustrates Theorem~\ref{thm:stability_mergegram} on a cloud and its perturbation by showing their close mergegrams.
The more extensive experiment on 100 clouds  in \cite[Fig.~8]{elkin2020mergegram} similarly confirms that the mergegram is perturbed within expected bounds.  
The computational complexity of the mergegram $\MG(\De_{SL}(A))$ was proved to be near linear in the number $n$ of points in a cloud $A\subset\R^m$, see \cite[Theorem~8.2]{elkin2020mergegram}.
The results above justify that the invariant mergegram satisfies conditions (a,b,c) of Isometry Recognition Problem~\ref{pro:recognition}.

\medskip

\noindent 
Algorithm~\ref{alg:mergegram} computes the mergegram of the SL dendrogram for any finite set $R\subset\R^m$.  Let $\alpha(|R|)$ be the inverse Ackermann function.
Other constants below are defined in \cite{march2010fast}.
\begin{algorithm}
\caption{ This algorithm takes a finite metric space $(R,d)$ as input and 
computes its mergegram, which is returned as the output.  }
  \label{alg:mergegram}
\begin{algorithmic}[1]
   \STATE \textbf{Input} : a finite point cloud $R\subset\mathbb{R}^m$
   \STATE Compute $\MST(R)$ and sort all edges of $\MST(R)$ in increasing order of length
   \STATE Initialize Union-Find structure $U$ over $R$. Set all points of $R$ to be their components.
   \STATE Initialize the function $\text{prev: Components}[U] \rightarrow \mathbb{R}$ by setting $\text{prev}(t) = 0$ for all $t$
   \STATE Initialize the vector Output that will consists of pairs in $\mathbb{R} \times \mathbb{R}$
   \FOR{Edge $e = (a,b)$  in the set of edges (increasing order)}
   \STATE Find components $c_1$ and $c_2$ of $a$ and $b$ respectively in Union-Find $U$
   \STATE Add pairs (prev$[c_1]$,length($e$)), (prev$[c_2]$,length($e$))  $\in \mathbb{R}^2$ to Output
   \STATE Merge components $c_1$ and $c_2$ in Union-Find $U$ and denote the component by $t$
   \STATE Set prev$[t]$ = length($e$)
   \ENDFOR
   \STATE \textbf{return} Output
\end{algorithmic}
\end{algorithm}

\begin{thmm}[a fast mergegram computation]
\label{thm:complexity}
For any cloud $R\subset\mathbb{R}^m$, the mergegram $\MG(\De_{SL}(R))$ can be computed in time 
$$O\Big ((c_m(R))^{4+ \ceil{\log_2(\rho(R))}} \cdot \log_2(\Delta(|R|)) \cdot |R| \cdot \log_2(|R|) \cdot \alpha(|R|)\Big ),$$ 
	where all parameters were introduced in Definitions \ref{dfn:expansion_constant}, \ref{dfn:radius+d_min}, \ref{dfn:edge_lengths}. 
\end{thmm}
\begin{proof}
By Theorem \ref{cor:single_cover_tree_mst_time} the Minimum Spanning Tree $\MST(R)$ needs 
$$O\Big ((c_m(R))^{4+ \ceil{\log_2(\rho(R))}} \cdot \log_2(\Delta(|R|)) \cdot |R| \cdot \log_2(|R|) \cdot \alpha(|R|)\Big ).$$ 
The rest of Algorithm~\ref{alg:mergegram} is dominated by $O(|R|\alpha(|R|))$ Union-Find operations. 
Hence the full algorithm has the same computational complexity as the MST.
\end{proof}


\section{Experiments on a classification of point sets}
\label{sec:experiments_mg_f}

\noindent 
The experiments summarized in Fig.~\ref{fig:100-point_clouds} show that the mergegram curve in blue outperforms other isometry invariants on the isometry classification by the state-of-the-art PersLay.
We generated 10 classes of 100-point clouds within the unit ball $\R^m$ for $m=2,3,4,5$.
For each class, we made 100 copies of each cloud and perturbed every point by a uniform random shift in a cube of the size $2\times\epsilon$, where $\epsilon$ is called a \emph{noise bound}. 
For each of 100 perturbed clouds, we added 25 points such that every new point is $\epsilon$-close to an original point.
Within each of 10 classes all 100 clouds were randomly rotated within the unit ball around the origin, see Fig.~\ref{fig:clouds}. 
For each of the resulting 1000 clouds, we computed the mergegram, 0D persistence diagram and the diagram of pairs of distances to two nearest neighbors for every point. 

\newcommand{\hh}{44mm}
\begin{figure}[h!]
\includegraphics[height=\hh]{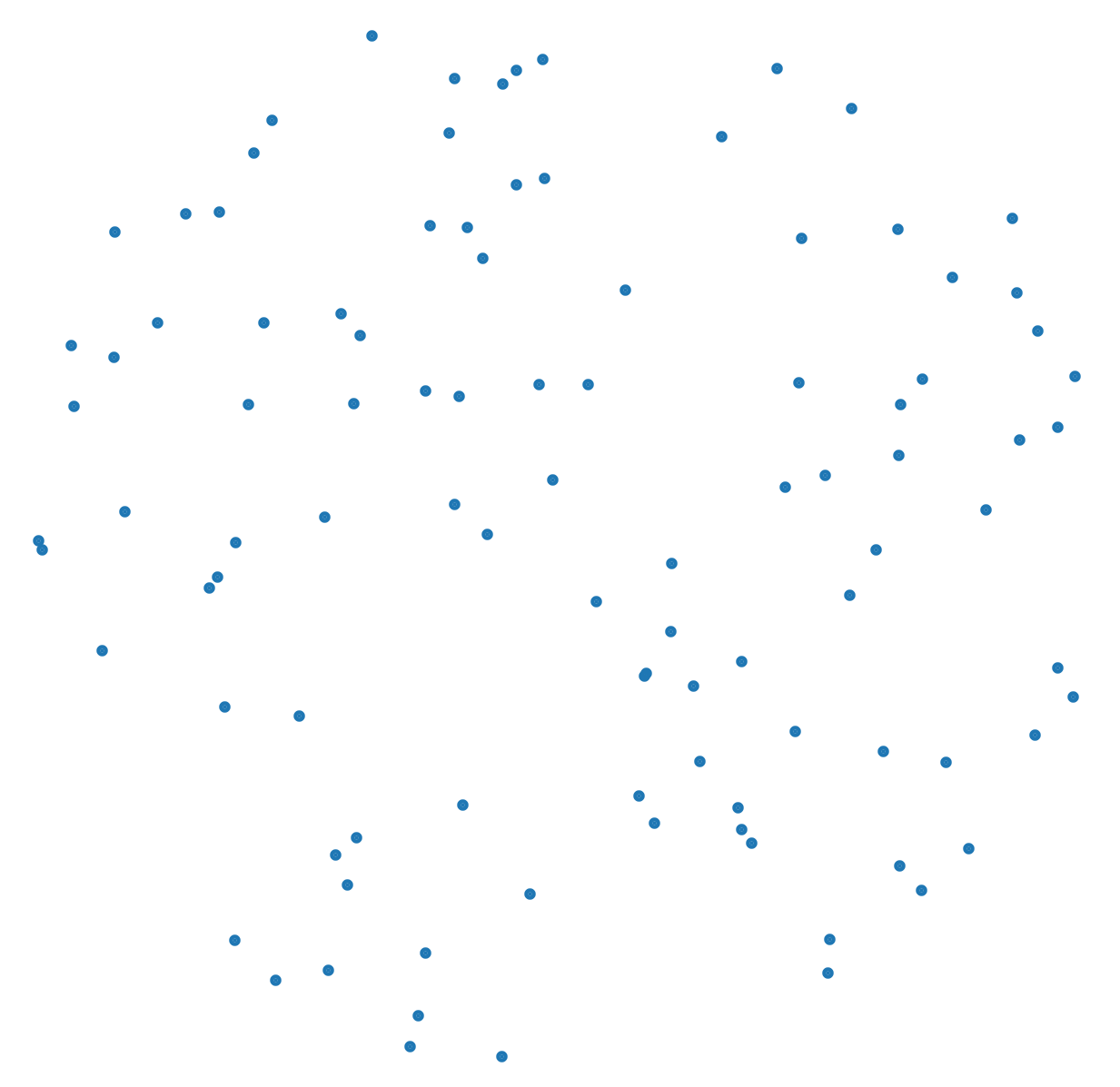}
\includegraphics[height=\hh]{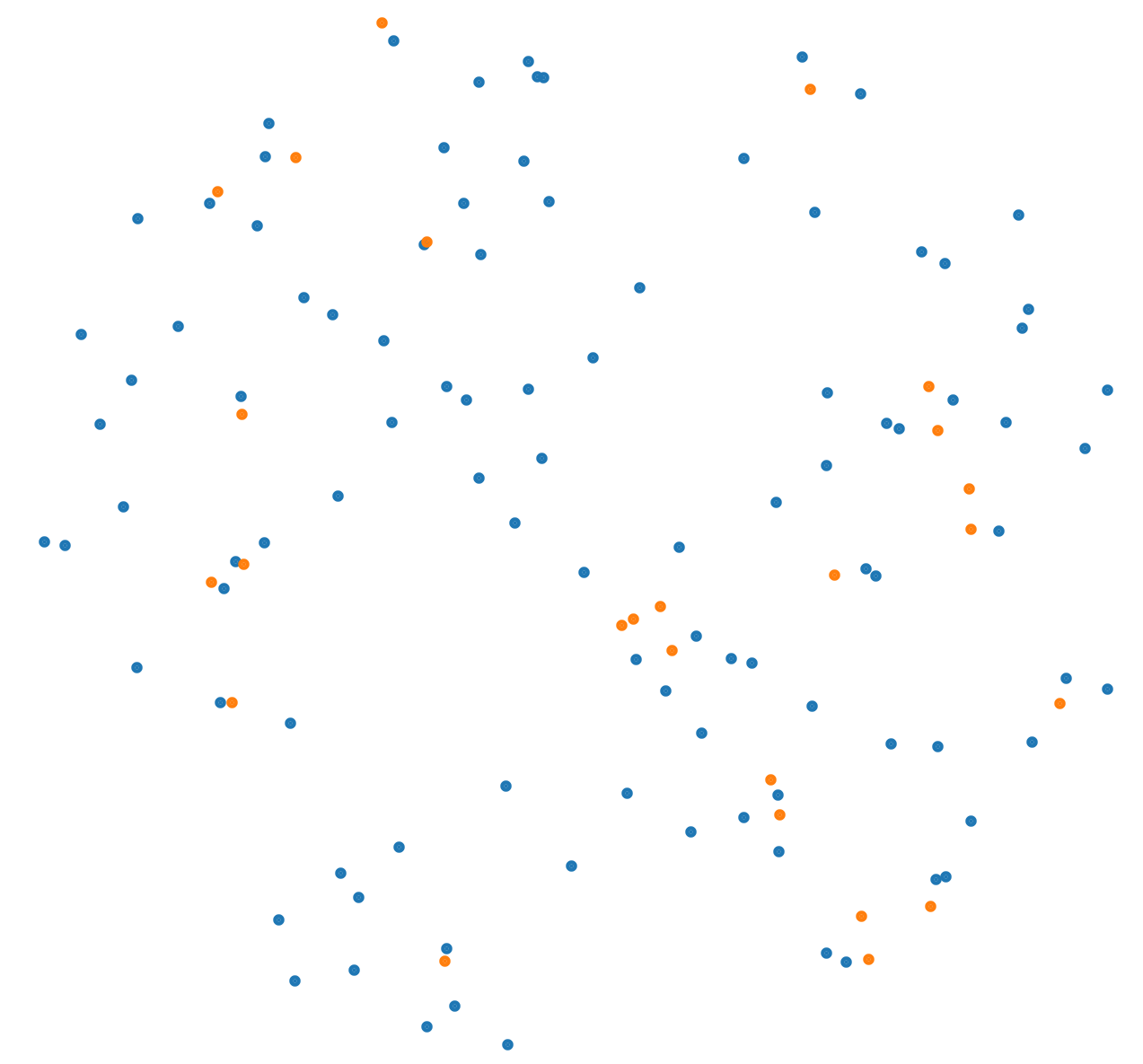}
\includegraphics[height=\hh]{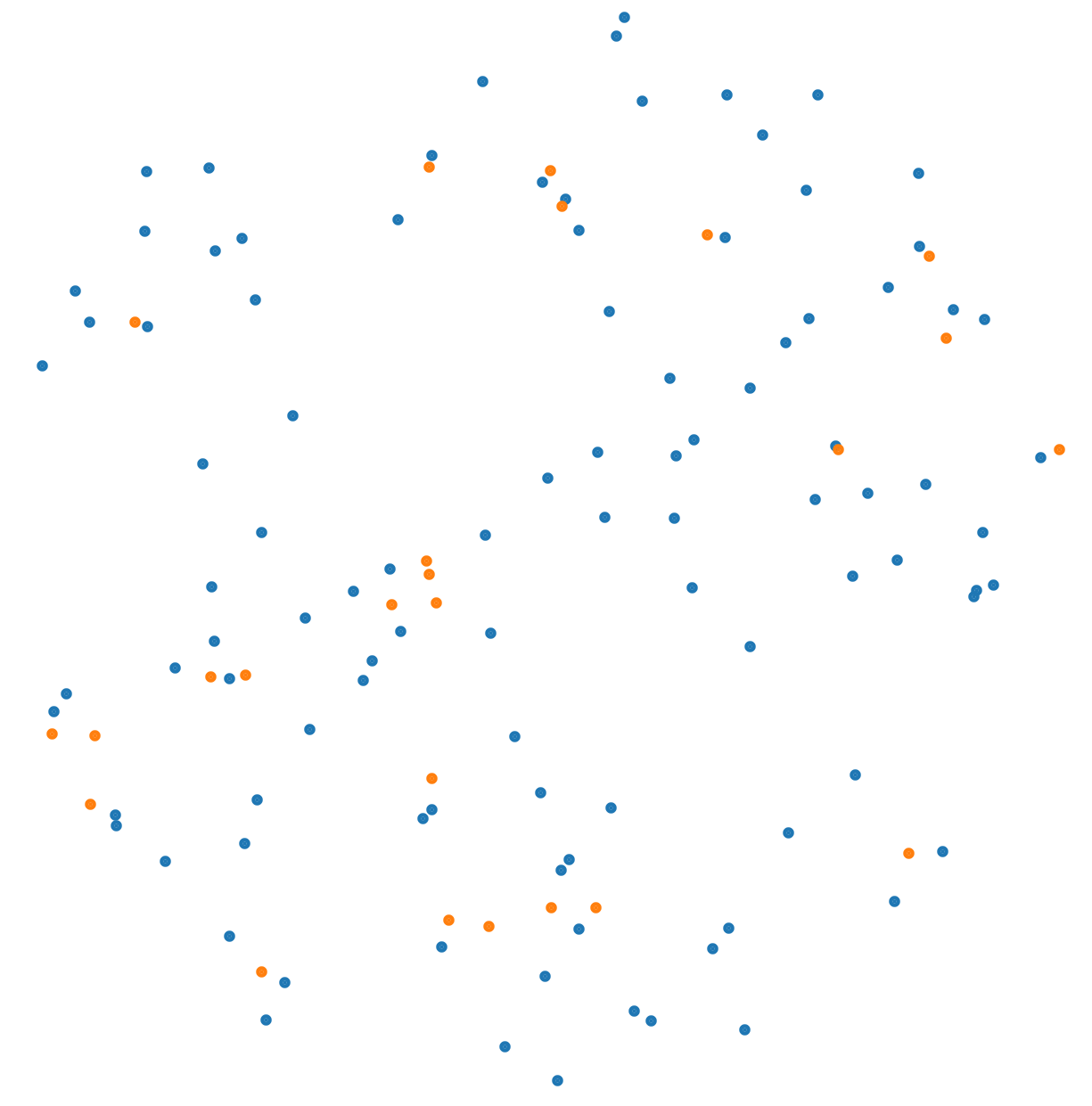}
\caption{\textbf{Left}: an initial random cloud with 100 blue points.
\textbf{Middle}: all blue points are perturbed, 25 extra orange points are added.
\textbf{Right}: a cloud is rotated through a random angle. 
Can we recognize that the initial and final clouds are in the same isometry class modulo small noise?}
\label{fig:clouds}
\end{figure}

\noindent 
The machine learning part has used the obtained diagrams as the input-data for the Perslay \cite{carriere2019perslay}. 
Each dataset was split into learning and test subsets in ratio 4:1. 
The learning loops ran by iterating over mini-batches consisting of 128 elements and going through the full dataset for a given number of epochs. 
The success rate was measured on the test subset.
\medskip

\noindent
The original Perslay module was rewritten in Tensorflow v2 and RTX 2080 graphics card was used to run the experiments.  
The technical concepts of PersLay are explained in \cite{carriere2019perslay}:

\begin{itemize}
    \item Adam(Epochs = 300, Learning rate = 0.01)
    \item Coefficents = Linear coefficents
    \item Functional layer = [PeL(dim=50), PeL(dim=50, operationalLayer=PermutationMaxLayer)]. 
    \item Operation layer = TopK(50)
\end{itemize}

The PersLay training has used the following invariants compared in Fig.~\ref{fig:100-point_clouds}:
\begin{itemize}
\item cloud : the initial cloud $A$ of points corresponds to the baseline curve in black;
\item PD0: the 0D persistence diagram $\PD$ for distance-based filtrations of sublevel sets in red;
\item NN(2) brown curve: for each point $a\in A$ includes distances to two nearest neighbors;
\item the mergegram $\MG(\De_{SL}(A))$ of the SL dendrogram has the blue curve above others.
\end{itemize}

\begin{figure}
\begin{tikzpicture}[scale=0.85]
\begin{axis}[xlabel = 2 dimensions; noise bound, ylabel = success rate,grid]
\addplot table [x=e, y=m, col sep=comma] {100Points25NoiseResults/dim2.csv};
\addlegendentry{mergegram}
\addplot table [x=e, y=h, col sep=comma] {100Points25NoiseResults/dim2.csv};
\addlegendentry{PD0}
\addplot table [x=e, y=n, col sep=comma] {100Points25NoiseResults/dim2.csv};
\addlegendentry{NN(2)}
\addplot table [x=e, y=c, col sep=comma] {100Points25NoiseResults/dim2.csv};
\addlegendentry{cloud}
\end{axis}
\end{tikzpicture}
\begin{tikzpicture}[scale=0.85]
\begin{axis}[xlabel = 2 dimensions; noise bound, grid]
\addplot table [x=e, y=m, col sep=comma] {100Points25NoiseResults/dim3.csv};
\addplot table [x=e, y=h, col sep=comma] {100Points25NoiseResults/dim3.csv};
\addplot table [x=e, y=n, col sep=comma] {100Points25NoiseResults/dim3.csv};
\addplot table [x=e, y=c, col sep=comma] {100Points25NoiseResults/dim3.csv};
\end{axis}
\end{tikzpicture}
\medskip

\begin{tikzpicture}[scale=0.85]
\begin{axis}[xlabel = 4 dimensions; noise bound, ylabel = success rate,grid]
\addplot table [x=e, y=m, col sep=comma] {100Points25NoiseResults/dim4.csv};
\addplot table [x=e, y=h, col sep=comma] {100Points25NoiseResults/dim4.csv};
\addplot table [x=e, y=n, col sep=comma] {100Points25NoiseResults/dim4.csv};
\addplot table [x=e, y=c, col sep=comma] {100Points25NoiseResults/dim4.csv};
\end{axis}
\end{tikzpicture}
\begin{tikzpicture}[scale=0.85]
\begin{axis}[xlabel = 5 dimensions; noise bound,grid]
\addplot table [x=e, y=m, col sep=comma] {100Points25NoiseResults/dim5.csv};
\addplot table [x=e, y=h, col sep=comma] {100Points25NoiseResults/dim5.csv};
\addplot table [x=e, y=n, col sep=comma] {100Points25NoiseResults/dim5.csv};
\addplot table [x=e, y=c, col sep=comma] {100Points25NoiseResults/dim5Cor.csv};
\end{axis}
\end{tikzpicture}
\caption{Success rates of PersLay in identifying isometry classes of 100-point clouds uniformly sampled in a unit ball, averaged over 5 different clouds and 5 cross-validations with 20/80 splits. 
}
\label{fig:100-point_clouds}
\end{figure}
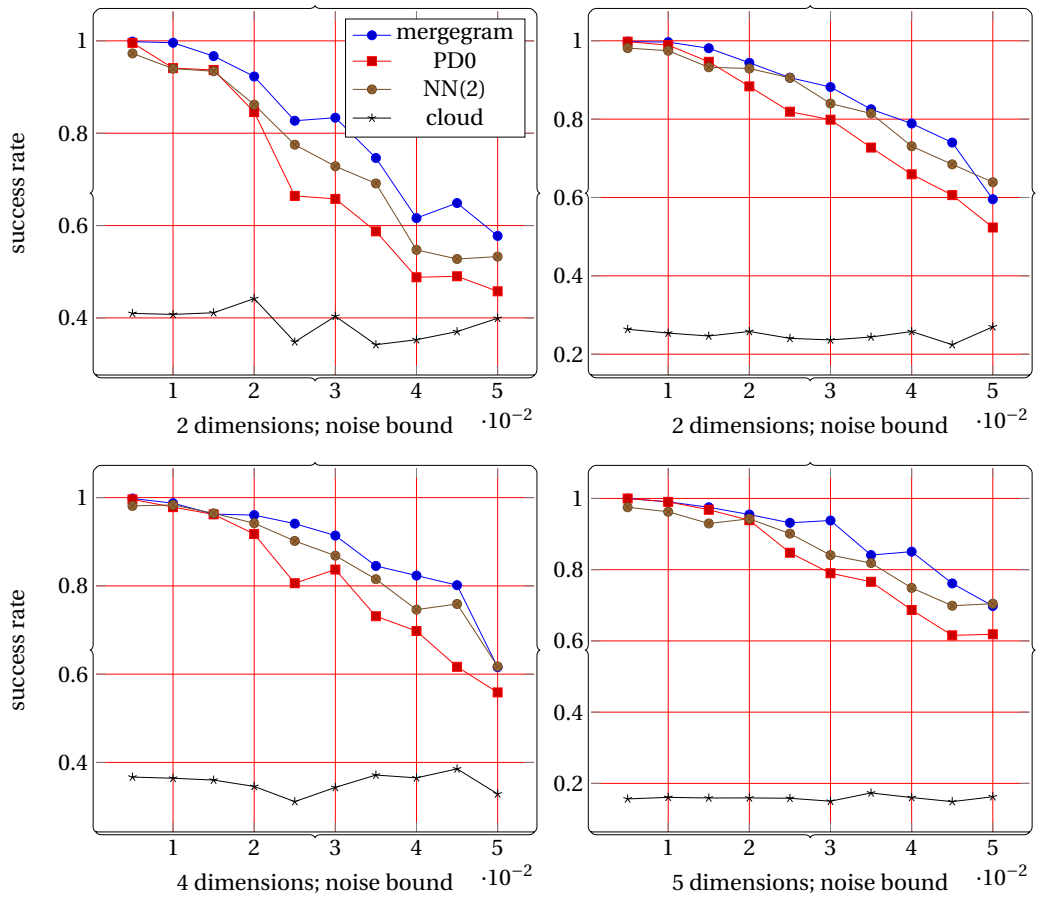

\noindent 
Fig.~\ref{fig:100-point_clouds} shows that the new mergegram has outperformed all other invariants on the isometry classification problem.
The 0D persistence turned out to be weaker than the pairs of distances to two neighbors.
The topological persistence has found applications in data skeletonization with theoretical guarantees \cite{kurlin2015homologically,kalisnik2019higher}. 
We are planning to extend the experiments in section~\ref{sec:experiments_mg} for classifying rigid shapes by comining the new mergegram with the 1D persistence, which has the fast $O(n\log n)$ time for any 2D cloud of $n$ points \cite{kurlin2014fast, kurlin2014auto}.

\newpage

\section{New experiments on isometry recognition of distorted shapes}
\label{sec:experiments_mg}

This section fulfills final condition (d) of Problem~\ref{pro:recognition} by experimentally comparing the mergegram with 0D persistence and distributions of distances to neighbors on 15000 clouds.
Section~\ref{sec:experiments_mg_f} did experiments only on randomly generated clouds.
\medskip

\noindent
We considered 15 classes of shapes represented by black and white images of mythical creatures \cite{bronstein2008analysis}, see Fig.~\ref{fig:myth_images}.
These shapes were chosen to make the shape recognition problem really challenging.
Indeed, similar creatures from this dataset are represented by slightly different shapes, which can be hard to isometrically distinguish from each other.
For example, several images of a horse include only minor differentiating features such as a saddle or a different tails, which makes horses nearly identical.
\medskip

\noindent
\textbf{Shape generation}.
For each image, we generated 1000 perturbed images by affine and projective transformations to get 15000 distorted shapes split into 15 classes.
\medskip

\noindent
First we rotated each image around its central point by an angle generated uniformly in the interval $[0,2\pi)$ using the function cv::rotate from the OpenCV library. 
If needed, we extended the resulting image to fit all black pixels of the rotated shape into a bounding box.
Then both affine and projective transformations distort each image by using a noise parameter $\de$ such that the value $\de=0$ represents the identity transformation.
\medskip

\noindent
Fig.~\ref{fig:distorted_shapes_corner_points} illustrates how an original image is randomly rotated, then randomly distorted by affine or projective transformations depending on the noise parameter $\de$.
\medskip

\begin{figure}
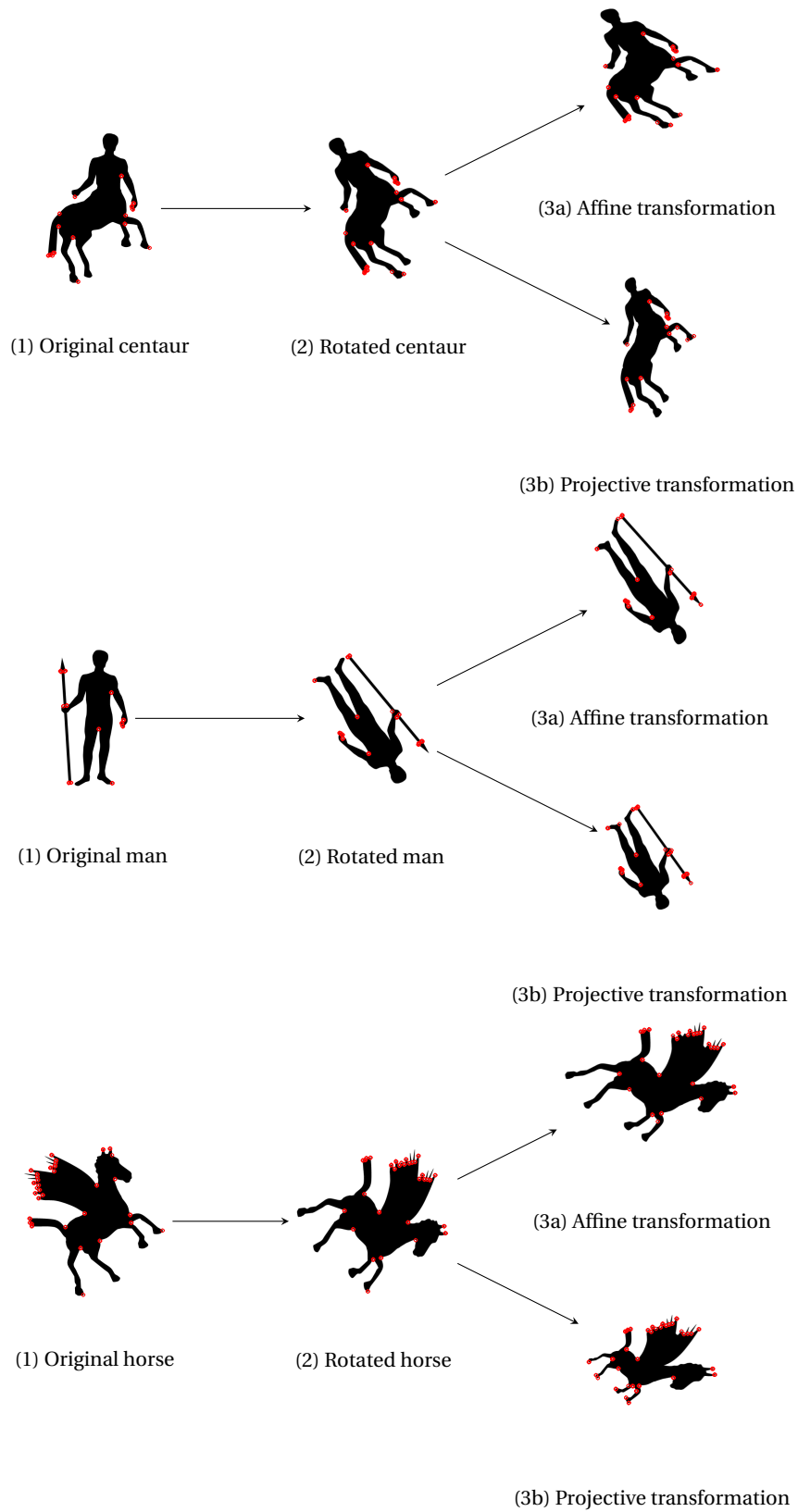

\footnotesize
\centering
\input figures/distorted_shapes_corner_points.txt
\caption{Generating distorted shapes by applying random rotations, affine and projective transformations, which substantially affect the extracted clouds of Harris corner points \cite{sanchez2018analysis} in red.}
\label{fig:distorted_shapes_corner_points}
\end{figure}

\noindent 
\textbf{Affine transformations} are implemented as compositions of the already applied rotations above and the function cv::resize() from the OpenCV library.
This function scales an image of size $w\times h$ by horizontal and vertical factors $a,b$ sampled as follows.
\medskip

\noindent
$\bullet$
\textbf{Uniform noise}:
$a \in [1-\delta w, 1+\delta w]$, 
$b \in [1-\delta h, 1+\delta h]$ have uniform distributions.
\medskip

\noindent
$\bullet$
\textbf{Gaussian noise}: 
$a \in \mathcal{N}(1,\delta h) \cap \R_{+}$ and 
$b \in \mathcal{N}(1,\delta w) \cap \R_{+}$
have Gaussian distributions with mean 1 and standard variance $\de h, \de v$, truncated to positive numbers.
\medskip

\noindent
\textbf{Projective transformation} are implemented as compositions of the already applied rotations above and the OpenCV function  cv::getPerspectiveTransform() function, which is parametrized by 4-dimensional array $v = (a_0, a_1, a_2, a_3)$ consisting of points $a_i\in\Z^2$, $i=0,1,2,3$.
This function maps the corners of the image as follows: $$(0,0)\mapsto a_0, (0,h) \mapsto a_1, (w, 0) \mapsto a_2  \text{ and } (w, h) \mapsto a_3.$$
Then the projective transformation of the rectangle $w\times h$ is uniquely determined by the above corners.
The above points $a_i$ are randomly sampled by using a noise parameter $\de$.
\medskip

\noindent
$\bullet$
\textbf{Uniform noise}: each coordinate has a uniform distribution with a noise parameter $\de$  
$$a_0 \in [0,\delta w] \times [0,\delta h],\qquad 
a_1 \in [0,\delta w] \times [h - \delta h, h],$$ 
$$a_2 \in [w - \delta w, w] \times [0,\delta h],\qquad 
a_3 \in [w - \delta w, w] \times [h - \delta h, h].$$

\noindent
$\bullet$
\textbf{Gaussian noise}: each coordinate has a Gaussian distribution truncated to the image
$$a_0 \in (\mathcal{N}(0,\delta w)\cap [0,w]) \times (\mathcal{N}(0,\delta h)\cap [0,h]),$$
$$a_1 \in (\mathcal{N}(0,\delta w)\cap [0,w]) \times (\mathcal{N}(h,\delta h)\cap [0,h]),$$ 
$$a_2 \in (\mathcal{N}(w,\delta w)\cap [0,w])\times(\mathcal{N}(0,\delta h)\cap [0,h]),$$
$$a_3 \in (\mathcal{N}(w,\delta w)\cap [0,w])\times(\mathcal{N}(h,\delta h)\cap [0,h]).$$

\noindent
\textbf{Point cloud extraction}. For each distorted image, we extract classical Harris point corners \cite{sanchez2018analysis} due to their simplicity, see the red points in Fig.~\ref{fig:distorted_shapes_corner_points}.
For detecting corner points, the OpenCV function cv::cornerHarris was used with the parameters blockSize = 3, apertureSize = 5, k = 0.04, thresh = 120.
However one can use any reliable algorithms such as FAST \cite{rosten2008faster} or scale-invariant feature transform (SIFT) \cite{lowe1999object}. 
\medskip

\noindent
After describing the available point cloud data above, we specify condition~(\ref{pro:recognition}d) of Isometry Recognition Problem~\ref{pro:recognition} in the context of supervised machine learning.

\begin{pro}[experimental recognition]
\label{pro:experimental}
Given a labeled dataset split into classes of similar but projectively distorted shapes, develop a supervised learning tool to recognize a class of distorted shapes with a high accuracy despite substantial noise.  
\end{pro}

\noindent
Since all isometry invariants are independent of point ordering, the most suitable neural network is PersLay \cite{carriere2019perslay} whose output is invariant under permutations by design. 
Each layer is a combination of a coefficient layer $\omega(p):\mathbb{R}^m \rightarrow \mathbb{R}$, a transformation $\phi(p):\mathbb{R}^m \rightarrow \mathbb{R}^q$ and a permutation invariant layer $\text{op}$ combined as follows
$$\text{PersLay}(D) = \text{op}(\{w(p)\phi(p)\}_{p\in D}), \text{ where } D \text{ is a diagram or multiset of points in }\R^m.$$

\noindent
Coordinates of all input points are linearly normalized to $[0,1]$.
We have used the following parameters of the PersLay network for all experiments below.
\medskip

\noindent    
\textbf{The max layer} $\text{MAX}(q)$ consists of the following functions.
    \begin{itemize}
        \item The coefficient layer $w : \mathbb{R}^m \rightarrow \mathbb{R}$ is the weight $w(x_1,\dots,x_m) = k|x_1-x_2|$, where $k$ is a trainable scalar and the dimension is typically $m=2$. 
        \item The transformation layer $\phi: \{\text{diagrams of points in }\mathbb{R}^m\} \rightarrow \mathbb{R}^q$ is the function $\phi(D) = \sum_{p \in D}\lambda p + \gamma\text{maxpool}(D) + \beta$, where $\lambda$,$\gamma$ are $\mathbb{R}^{m\times q}$ trainable matrices, $\beta$ is a $\mathbb{R}^q$ trainable vector and $\text{maxpool}$ returns a maximum value for every $i=1,\dots,m$.
\item The operational layer $\text{op}:\mathbb{R}^q \rightarrow \mathbb{R}^t$ puts all coordinates in increasing order and composes the result with standard densely connected layer \cite{tensorflow2015-whitepaper} $\text{Dense}: \mathbb{R}^q \rightarrow \mathbb{R}^t$.
    \end{itemize}

\noindent
An output is a vector in $\R^t$ for $t=15$ of image classes.
A final prediction is obtained by choosing a class with a largest coordinate in the output vector.
\medskip

\noindent    
\textbf{The image layer} $\text{Im}[x,y]$ for integer parameters $x,y$ and a multiset of points in the unit square $[0,1]^2$ consists of the following functions.    
\begin{itemize}
        \item The coefficient layer $w:\R^2 \rightarrow \R$ is a piecewise constant function trained on $x\cdot y$ parameters, defined on the unit square partition 
        $$\mathcal{P}(x,y) = \left\{ \left[\frac{i}{x}, \frac{i+1}{x}\right] \times  \left[\frac{j}{y}, \frac{j+1}{y}\right]  \mid i = 0,\dots,x-1 \text{ and } j =0,\dots,y-1 \right\}. $$
\item Let $\phi_{p}:\mathbb{R}^2 \rightarrow \mathbb{R}$ be the Gaussian distribution centered at $p\in D$ with a trainable standard deviation $\sigma$. The transformation layer $\phi:\R \rightarrow \R^{xy}$ consists of $xy$ functions $\phi_p$, where $p$ runs over all centroids of the partition $\mathcal{P}(x,y)$.
 \item The operation layer op takes the sum over the given point cloud. A final prediction is made by composing the operation layer with the Dense layer.
    \end{itemize}
\medskip

\noindent
Finally, the PersLay network used the optimizer tf.keras.adam with the standard learning rate 0.01 and 150 epochs, the loss function SparseCategoricalCrossEntropy, the 80:20 of training and testing,
 a 5-fold Monte Carlo cross validation for each run.
\medskip

\noindent 
Fig.~\ref{fig:max-layer_affine}, \ref{fig:max-layer_projective}, \ref{fig:image-layer_affine} show that the mergegram $\MG$ consistently outperforms two other isometry invariants: 0D persistence and the multiset $NN(4)$ consisting of 4-tuple distances to neighbors per given point. 
The simpler multiset $NN(2)$ performed worse.
A given cloud $C\subset\R^2$ was considered as a baseline input.
The noise factors $\de$ reached 25\%, which means that original images were distorted up to a quarter of image sizes. 

\section{A discussion of novel contributions and further open problems}
\label{sec:discussion}

This chapter has further demonstrated that the provably stable-under-noise invariant mergegram of a dendrogram is a fast and efficient tool in the challenging problem of isometry shape recognition, especially for substantially distorted images.
\medskip

\noindent
In comparison with the conference version \cite{elkin2020mergegram}, section~\ref{sec:relations} proved new Theorem~\ref{thm:mergegram-to-dendrogram} describing how to reconstruct a single-linkage dendrogram in general position from its much simpler mergegram.   
It is hard to define a continuous metric between dendrograms, especially because they can be unstable under perturbations.
Theorem~\ref{thm:mergegram-to-dendrogram} allows us to measure a continuous similarity between dendrograms in general position as the bottleneck distance between their unique mergegrams.
This distance can be computed in time $O(n^{1.5}\log n)$ \cite{kerber2016geometry} for diagrams consisting of at most $n$ points.
\medskip

\noindent 
Section~\ref{sec:stability} provided a full proof of stability of the mergegram under perturbations of points, while the earlier paper \cite{elkin2020mergegram} only announced this result without proving highly non-trivial Lemmas~\ref{lem:merge_module_decomposition} and~\ref{lem:merge_modules_interleaved}, which required a heavy algebraic machinery.
\medskip

\begin{figure}[h!]
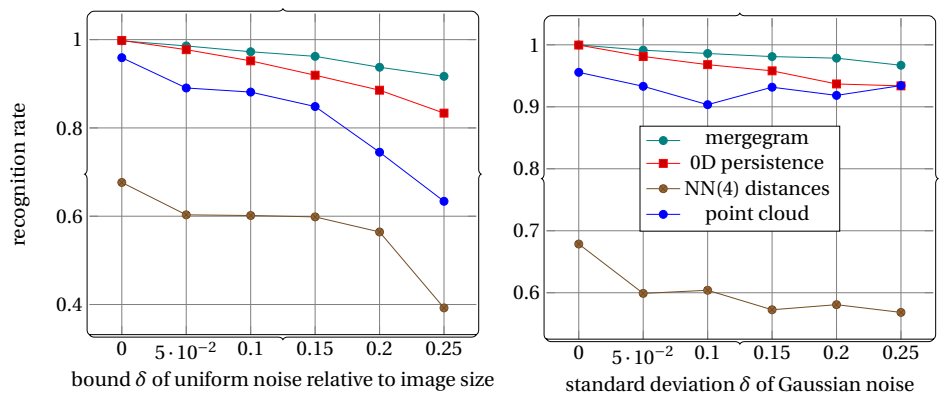

\centering
\input figures/max-layer_affine_uniform.txt
\input figures/max-layer_affine_Gaussian.txt
\caption{Recognition rates are obtained by training the max layer MAX(75) of PersLay on three isometry invariants and a cloud of corner points extracted from 15000 affinely distorted images.}
\label{fig:max-layer_affine}
\end{figure}

\begin{figure}[h!]
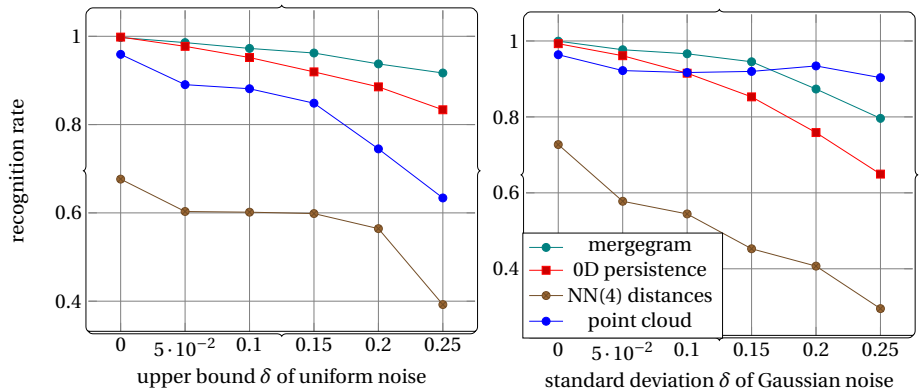

\centering
\input figures/max-layer_projective_uniform.txt
\input figures/max-layer_projective_Gaussian.txt
\caption{Recognition rates are obtained by training the max layer MAX(75) of PersLay on isometry invariants and corner points extracted from 15000 projectively distorted images.}
\label{fig:max-layer_projective}
\end{figure}

\begin{figure}[h!]
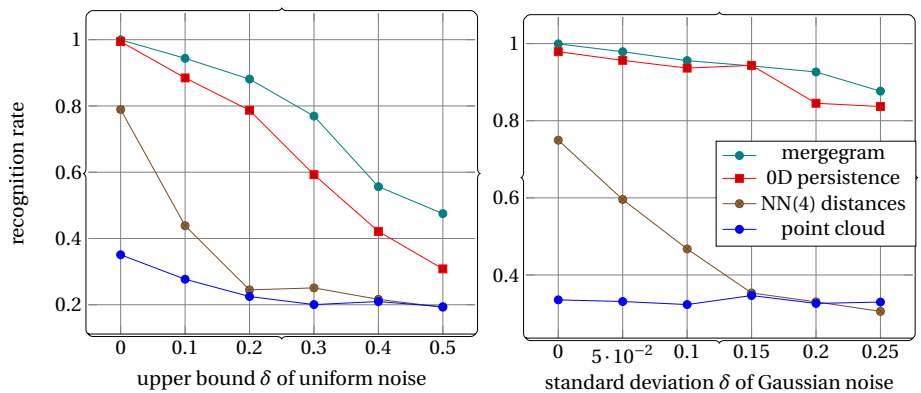

\centering
\input figures/image-layer_affine_uniform.txt
\input figures/image-layer_affine_Gaussian.txt
\caption{Recognition rates are obtained by training the image layer IM[20,20] of PersLay on isometry invariants and a cloud of corner points extracted from 15000 affinely distorted images.}
\label{fig:image-layer_affine}
\end{figure}

\noindent
Example~\ref{exa:5-point_line} and the discussion following Theorem~\ref{thm:mergegram_to_0D_persistence} justify that the invariant mergegram is strictly stronger than the 0D persistence.
This theoretical fact is now confirmed by the new experiments on 15000 point clouds extracted from substantially distorted real shapes.
In Fig.~\ref{fig:max-layer_affine}, \ref{fig:max-layer_projective}, \ref{fig:image-layer_affine} the mergegram outperformed other isometry invariants.
Since the distribution $NN(2)$ of distances to two closest neighbors per point performed badly, we have strengthened this invariant to $NN(4)$ of distances to four nearest neighbors.
However, even $NN(4)$ performed always always worse than the original point cloud, which can not be considered as an isometry invariant.
\medskip

\noindent
For very high levels of 20\% and 25\% distortions in projective transformations, the PersLay network trained on a point cloud achieved high recognition rates, because we have extensively tried many parameters in the layers MAX(75) and Im[20,20] for a best trade-off between accuracy and speed.
The C++ code for the mergegram is at \cite{elkin2020mergegram}.
\medskip


%% file: figures/5-point_line.txt
\begin{figure}[h!]
\centering
\begin{tikzpicture}[scale = 1.1]
  \draw[->] (-1,0) -- (11,0) node[right]{} ;
   \foreach \x/\xtext in {0, 1, 3, 7, 10}
    \draw[shift={(\x,0)}] (0pt,2pt) -- (0pt,-2pt) node[below] {$\xtext$};   
   \filldraw (0,0) circle (2pt);
   \filldraw (1,0) circle (2pt);
   \filldraw (3,0) circle (2pt);
   \filldraw (7,0) circle (2pt);
   \filldraw (10,0) circle (2pt);
\end{tikzpicture}

\begin{tikzpicture}[thick,scale=0.50, every node/.style={transform shape}][sloped]
\draw[style=help lines,step = 1] (-1,0) grid (10.4,4.4);
\draw [->] (-1,0) -- (-1,5) node[above] {{\large scale $s$}};
\foreach \i in {0,0.5,...,2}{ \node at (-1.5,2*\i) {\i}; }
\node (a) at (0,-0.3) {0};
\node (b) at (1,-0.3) {1};
\node (c) at (3,-0.3) {3};
\node (d) at (7,-0.3) {7};
\node (e) at (10,-0.3) {10};
\node (x) at (5,5) {};
\node (ab) at (0.5,1){};
\node (abc) at (1.5,2){};
\node (de) at (8.5,3){};
\node (all) at (5,4){};
\draw [line width=0.5mm, blue ] (a) |- (ab.center);
\draw [line width=0.5mm, blue ] (b) |- (ab.center);
\draw [line width=0.5mm, blue ] (c) |- (abc.center);
\draw [line width=0.5mm, blue ] (d) |- (de.center);
\draw [line width=0.5mm, blue ] (e) |- (de.center);
\draw [line width=0.5mm, blue ] (ab.center) |- (abc.center);
\draw [line width=0.5mm, blue ] (abc.center) |- (all.center);
\draw [line width=0.5mm, blue ] (de.center) |- (all.center);
\draw [line width=0.5mm, blue ] [->] (all.center) -> (x.center);
\end{tikzpicture}
\hspace*{1mm}
\begin{tikzpicture}[thick,scale=0.9, every node/.style={transform shape}]
  \draw[style=help lines,step = 0.5] (0,0) grid (0.5,2.4);
  \draw[->] (-0.2,0) -- (0.8,0) node[right] {birth};
  \draw[->] (0,-0.2) -- (0,2.4) node[above] {};	
  \draw[-] (0,0) -- (1,1) node[right]{};
  \foreach \x/\xtext in {0.5/0.5}
    \draw[shift={(\x,0)}] (0pt,2pt) -- (0pt,-2pt) node[below] {$\xtext$};
  \foreach \y/\ytext in {0.5/0.5,  1/1, 1.5/1.5, 2.0/2}
    \draw[shift={(0,\y)}] (2pt,0pt) -- (-2pt,0pt) node[left] {$\ytext$};  
   \draw [blue,fill] (0,0.5) circle (2pt);
   \draw [blue,fill] (0.0,1) circle (2pt);
   \draw [blue,fill] (0,1.5) circle (2pt);
   \draw [blue,fill] (0,2) circle (2pt);
   \draw [blue,fill] (0,2.7) circle (2pt);
\end{tikzpicture}
\hspace*{1mm}
\begin{tikzpicture}[thick,scale=0.9, every node/.style={transform shape}]
  \draw[style=help lines,step = 0.5] (0,0) grid (2.4,2.4);
  \draw[->] (-0.2,0) -- (2.4,0) node[right] {birth};
  \draw[->] (0,-0.2) -- (0,2.4) node[above] {death};	
  \draw[-] (0,0) -- (2.4,2.4) node[right]{};
  \foreach \x/\xtext in {0.5/0.5, 1/1, 1.5/1.5, 2.0/2}
    \draw[shift={(\x,0)}] (0pt,2pt) -- (0pt,-2pt) node[below] {$\xtext$};
  \foreach \y/\ytext in {0.5/0.5,  1/1, 1.5/1.5, 2.0/2}
    \draw[shift={(0,\y)}] (2pt,0pt) -- (-2pt,0pt) node[left] {$\ytext$};  
   \draw [red,fill] (0,0.5) circle (2pt);
   \draw [red] (0,0.5) circle (4pt);
   \draw [red,fill] (0.0,1.0) circle (2pt);
   \draw [red,fill] (0.0,1.5) circle (2pt);
   \draw [red] (0,1.5) circle (4pt);
   \draw [red,fill] (0.5,1.0) circle (2pt);
   \draw [red,fill] (1,2) circle (2pt);
   \draw [red,fill] (1.5, 2.0) circle (2pt);
   \draw [red,fill] (2, 2.7) circle (2pt);
\end{tikzpicture}
\caption{\textbf{Top}: the 5-point cloud $A = \{0,1,3,7,10\}\subset\R$.
\textbf{Bottom} from left to right: single-linkage dendrogram $\De_{SL}(A)$ from Definition~\ref{dfn:sl_clustering}, the 0D persistence diagram $\PD$ from Definition~\ref{dfn:persistence_diagram} and the new mergegram $\MG$ from Definition~\ref{dfn:mergegram}, where double circles show pairs of multiplicity 2.}
\label{fig:5-point_line}
\end{figure}

%% file: figures/3-point_dendrogram.txt
\parbox{90mm}{
\footnotesize
\begin{tabular}{lccccc}
partition $\De(A;2)$ at scale $s_2=2$ & & & $\{0,1,2\}$ & & \\
map $\De_1^2:\De(A;1)\to\De(A;2)$ & & & $\uparrow$ & $\nwarrow$ & \\
partition $\De(A;1)$ at scale $s_1=1$ & & & \{0, 1\} & & \{2\} \\
map $\De_0^1:\De(A;0)\to\De(A;1)$ & & $\nearrow$ & $\uparrow$ & & $\uparrow$  \\
partition $\De(A;0)$ at scale $s_0=0$ & $\{0\}$ & & $\{1\}$ & & \{2\} 
\end{tabular}}
\parbox{25mm}{
\begin{tikzpicture}[thick,scale=0.75, every node/.style={transform shape}]
  \draw[style=help lines,step = 1] (0,0) grid (2.4,2.4);
  \draw[->] (-0.2,0) -- (2.4,0) node[right] {birth};
  \draw[->] (0,-0.2) -- (0,2.4) node[above] {death};	
  \draw[-] (0,0) -- (2.4,2.4) node[right]{};
  \foreach \x/\xtext in {1/1, 2/2}
    \draw[shift={(\x,0)}] (0pt,2pt) -- (0pt,-2pt) node[below] {$\xtext$};
  \foreach \y/\ytext in {1/1, 2/2}
    \draw[shift={(0,\y)}] (2pt,0pt) -- (-2pt,0pt) node[left] {$\ytext$};
   \draw [red,fill] (0,1) circle (2pt);
   \draw [red] (0,1) circle (4pt);
   \draw [red,fill] (0,2) circle (2pt);
   \draw [red,fill] (1,2) circle (2pt);
   \draw [red,fill] (2,2.7) circle (2pt);
\end{tikzpicture}}

%% file: figures/5-point_set.txt
\parbox{80mm}{
\begin{tikzpicture}[scale = 0.75][sloped]
\node (x) at (5,3) {x};
 \node (a) at (1,1) {a};
 \draw (a) -- node[above]{3} ++ (x);
 \node (b) at (3.5,4.0) {b};
 \draw (b) -- node[above]{1} ++ (x);
 \node (c) at (7,1) {c};
 \draw (c) -- node[below]{1} ++ (x);
 \node (y) at (8,3) {y};
 \draw (x) -- node[above]{2} ++ (y);
 \node (d) at (10,5){p};
 \node (e) at (10,1){q};
 \draw (y) -- node[below]{1} ++ (d);
 \draw (y) -- node[below]{1} ++ (e);
\end{tikzpicture}}
\parbox{40mm}{
\begin{tabular}{c|ccccc}
& a & b & c & p & q \\
\hline
a & 0 & 4 & 4 & 6 & 6 \\
b & 4 & 0 & 4 & 6 & 6 \\
c & 4 & 2 & 0 & 4 & 4 \\
p & 6 & 4 & 4 & 0 & 2 \\
q & 6 & 4 & 4 & 2 & 0 
\end{tabular}}

%% file: figures/5-point_mergegram.txt
\begin{tikzpicture}[scale = 0.6][sloped]
 \draw[style=help lines,step = 1] (-1,0) grid (10.4,6.3);
\foreach \i in {0,...,3.0} { \node at (-1.4,2*\i) {\i}; }
\node (a) at (0,-0.3) {a};
\node (b) at (4,-0.3) {b};
\node (c) at (6,-0.3) {c};
\node (d) at (8,-0.3) {p};
\node (e) at (10,-0.3) {q};
\node (x) at (5,6.75) {};

\node (de) at (9,2){};
\node (bc) at (5.0,2){};
\node (bcde) at (7.0,4){};
\node (all) at (5.0,6){};

\draw [line width=0.5mm, blue ] (a) |- (all.center);
\draw [line width=0.5mm, blue ] (b) |- (bc.center);
\draw [line width=0.5mm, blue ]  (c) |- (bc.center);
\draw [line width=0.5mm, blue ]  (d) |- (de.center);
\draw [line width=0.5mm, blue ]  (e) |- (de.center);
\draw [line width=0.5mm, blue ]  (de.center) |- (bcde.center);
\draw [line width=0.5mm, blue ]  (bc.center) |- (bcde.center);
\draw [line width=0.5mm, blue ]  (bcde.center) |- (all.center);
\draw [line width=0.5mm, blue] [->] (all.center) -> (x.center);
\end{tikzpicture}
\hspace*{1cm}
\begin{tikzpicture}[scale = 1.1]
 \draw[style=help lines,step = 1] (0,0) grid (3.4,3.4);
 \draw[->] (-0.2,0) -- (3.4,0) node[right] {birth};
  \draw[->] (0,-0.2) -- (0,3.4) node[above] {death};	
  \draw[-] (0,0) -- (3.4,3.4) node[right]{};

  \foreach \x/\xtext in {1/1, 2.0/2, 3.0/3}
    \draw[shift={(\x,0)}] (0pt,2pt) -- (0pt,-2pt) node[below] {$\xtext$};

  \foreach \y/\ytext in {1/1, 2.0/2, 3.0/3}
    \draw[shift={(0,\y)}] (2pt,0pt) -- (-2pt,0pt) node[left] {$\ytext$}; 
   \draw [red,fill] (0,1) circle (2pt);
   \draw [red] (0,1) circle (3pt);
   \draw [red] (0,1) circle (4pt);
   \draw [red] (0,1) circle (5pt);
   \draw [red,fill] (1,2) circle (2pt);
   \draw [red] (1,2) circle (4pt);
   \draw [red,fill] (2,3) circle (2pt);
\draw [red,fill] (0,3) circle (2pt);
\draw [red,fill] (3,3.7) circle (2pt);
\end{tikzpicture}

%% file: figures/distorted_shapes_corner_points.txt
    \begin{tikzpicture}[ >=stealth,
    node distance=4cm,
    on grid,
    auto
  ]
        \node [] (start) {\includegraphics[scale = 0.175]{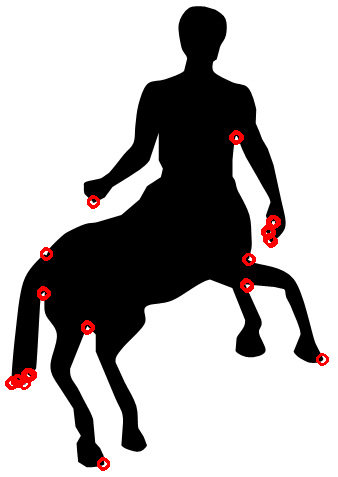} };
        \node [below = 2.0cm of start] (startText) {(1) Original centaur};
      \node [right of= start] (A1) {\includegraphics[scale = 0.175]{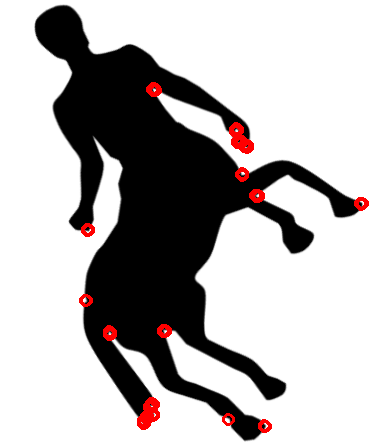}};
        \node [below = 2.0cm of A1] (A1Text) {(2) Rotated centaur };
    	\node [right of= A1] (tmp) {};
        \node [above = 2.0cm of tmp ] (A3){\includegraphics[scale = 0.175]{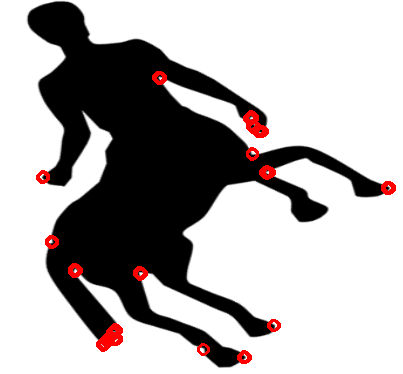}};
                \node [below = 2.0cm of A3] (A3Text) {(3a) Affine transformation };
        \node [below = 2.0cm of tmp] (end){\includegraphics[scale = 0.175]{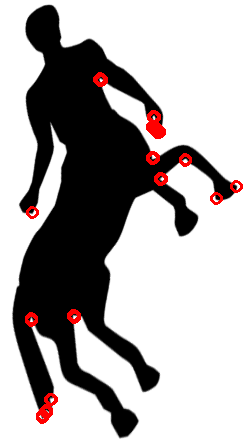}};
                \node [below = 2.0cm of end] (endText) {(3b) Projective transformation };
       \path[draw, ->] (start) -- (A1);
        \path[draw, ->] (A1) -- (A3);
        \path[draw, ->] (A1) -- (end);                    
    \end{tikzpicture}
    \begin{tikzpicture}[ >=stealth,
    node distance=4cm,
    on grid,
    auto
  ]
        \node [] (start) {\includegraphics[scale = 0.175]{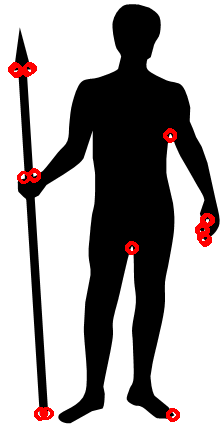} };
        \node [below = 2.0cm of start] (startText) {(1) Original man};
        \node [right of= start] (A1) {\includegraphics[scale = 0.175]{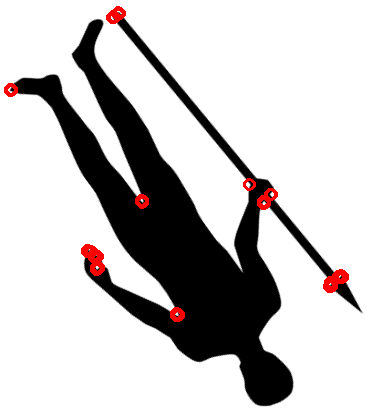}};
        \node [below = 2.0cm of A1] (A1Text) {(2) Rotated man };
        \node [right of= A1] (tmp) {};
        
        \node [above = 2.0cm of tmp ] (A3){\includegraphics[scale = 0.175]{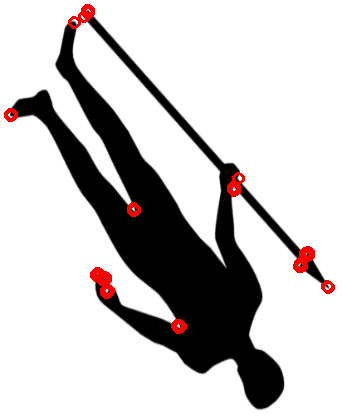}};
                \node [below = 2.0cm of A3] (A3Text) {(3a) Affine transformation };
        \node [below = 2.0cm of tmp] (end){\includegraphics[scale = 0.175]{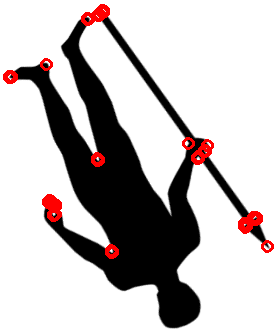}};
                \node [below = 2.0cm of end] (endText) {(3b) Projective transformation };
       \path[draw, ->] (start) -- (A1);
        \path[draw, ->] (A1) -- (A3);
        \path[draw, ->] (A1) -- (end);                    
    \end{tikzpicture}
    \begin{tikzpicture}[ >=stealth,
    node distance=4cm,
    on grid,
    auto
  ]
        \node [] (start) {\includegraphics[scale = 0.175]{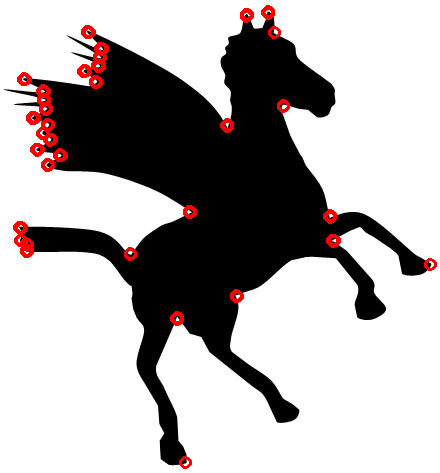} };
        \node [below = 2.0cm of start] (startText) {(1) Original horse};
        \node [right of= start] (A1) {\includegraphics[scale = 0.175]{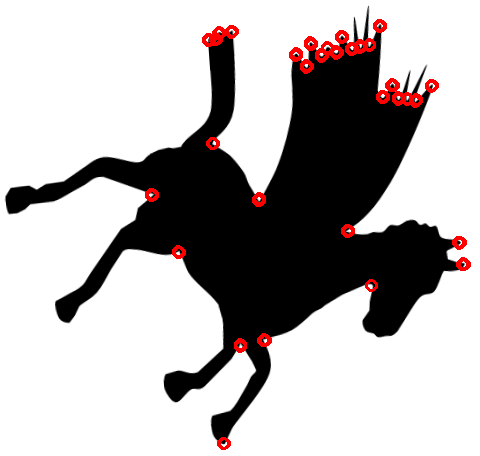}};
        \node [below = 2.0cm of A1] (A1Text) {(2) Rotated horse };
    	\node [right of= A1] (tmp) {};        
        \node [above = 2.0cm of tmp ] (A3){\includegraphics[scale = 0.175]{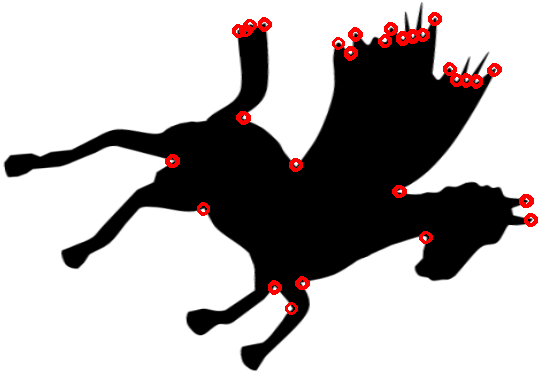}};
                \node [below = 2.0cm of A3] (A3Text) {(3a) Affine transformation };
        \node [below = 2.0cm of tmp] (end){\includegraphics[scale = 0.175]{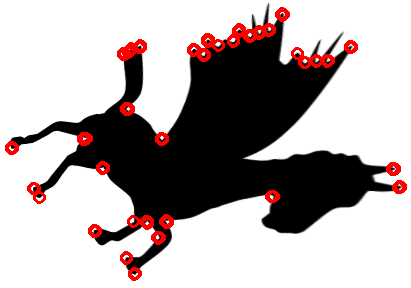}};
                \node [below = 2.0cm of end] (endText) {(3b) Projective transformation };         
       \path[draw, ->] (start) -- (A1);
        \path[draw, ->] (A1) -- (A3);
        \path[draw, ->] (A1) -- (end);                    
    \end{tikzpicture}

%% file: figures/max-layer_affine_uniform.txt
\begin{tikzpicture}[thick,scale=0.75, every node/.style={transform shape}]
\begin{axis}[xlabel = bound $\de$ of uniform noise relative to image size, ylabel = recognition rate, grid, grid style={gray},legend style={at={(0.5,-0.1)},anchor=north}]
\addplot[teal, mark=*] table [x=e, y=m, col sep=comma] {Tables/ProjectiveUniform2.csv};
\addplot table [x=e, y=p, col sep=comma] {Tables/ProjectiveUniform2.csv};
\addplot table [x=e, y=nn4, col sep=comma] {Tables/ProjectiveUniform.csv};
\addplot[blue,mark=*] table [x=e, y=c, col sep=comma] {Tables/ProjectiveUniform2.csv};
\end{axis}
\end{tikzpicture}

%% file: figures/max-layer_affine_Gaussian.txt
\begin{tikzpicture}[thick,scale=0.75, every node/.style={transform shape}]
\begin{axis}[xlabel = standard deviation $\de$ of Gaussian noise, grid,grid style={gray},legend style={at={(0.5,0.5)},anchor=center}]
\addplot[teal, mark=*] table [x=e, y=m, col sep=comma] {Tables/ScaleGaussian.csv};
\addlegendentry{mergegram}
\addplot table [x=e, y=p, col sep=comma] {Tables/ScaleGaussian.csv};
\addlegendentry{0D persistence}
\addplot table [x=e, y=nn4, col sep=comma] {Tables/ScaleGaussian.csv};
\addlegendentry{NN(4) distances}
\addplot[blue,mark=*] table [x=e, y=c, col sep=comma] {Tables/ScaleGaussian.csv};
\addlegendentry{point cloud}
\end{axis}
\end{tikzpicture}

%% file: figures/max-layer_projective_uniform.txt
\begin{tikzpicture}[thick,scale=0.75, every node/.style={transform shape}]
\begin{axis}[xlabel = upper bound $\de$ of uniform noise, ylabel = recognition rate,grid,grid style={gray},legend style={at={(0.5,-0.1)},anchor=north}]
\addplot[teal, mark=*] table [x=e, y=m, col sep=comma] {Tables/ProjectiveUniform2.csv};
\addplot table [x=e, y=p, col sep=comma] {Tables/ProjectiveUniform2.csv};
\addplot table [x=e, y=nn4, col sep=comma] {Tables/ProjectiveUniform.csv};
\addplot[blue,mark=*] table [x=e, y=c, col sep=comma] {Tables/ProjectiveUniform2.csv};
\end{axis}
\end{tikzpicture}

%% file: figures/max-layer_projective_Gaussian.txt
\begin{tikzpicture}[thick,scale=0.75, every node/.style={transform shape}]
\begin{axis}[xlabel = standard deviation $\de$ of Gaussian noise, grid, grid style={gray},legend style={at={(0.25,0.32)},anchor=north}]
\addplot[teal, mark=*] table [x=e, y=m, col sep=comma] {Tables/ProjectiveGaussian.csv};
\addlegendentry{mergegram}
\addplot table [x=e, y=p, col sep=comma] {Tables/ProjectiveGaussian.csv};
\addlegendentry{0D persistence}
\addplot table [x=e, y=nn4, col sep=comma] {Tables/ProjectiveGaussian.csv};
\addlegendentry{NN(4) distances}
\addplot[blue,mark=*] table [x=e, y=c, col sep=comma] {Tables/ProjectiveGaussian.csv};
\addlegendentry{point cloud}
\end{axis}
\end{tikzpicture}

%% file: figures/image-layer_affine_uniform.txt
\begin{tikzpicture}[thick,scale=0.75, every node/.style={transform shape}]
\begin{axis}[xlabel = upper bound $\de$ of uniform noise, ylabel = recognition rate, grid, grid style={gray},legend style={at={(0.5,-0.1)},anchor=north}]
\addplot[teal, mark=*] table [x=e, y=m, col sep=comma] {Tables/RotateScaleImage.csv};
\addplot table [x=e, y=p, col sep=comma] {Tables/RotateScaleImage.csv};
\addplot table [x=e, y=n, col sep=comma] {Tables/RotateScaleImage.csv};
\addplot[blue,mark=*] table [x=e, y=c, col sep=comma] {Tables/RotateScaleImage.csv};
\end{axis}
\end{tikzpicture}

%% file: figures/image-layer_affine_Gaussian.txt
\begin{tikzpicture}[thick,scale=0.75, every node/.style={transform shape}]
\begin{axis}[xlabel = standard deviation $\de$ of Gaussian noise, grid, grid style={gray} ,legend style={at={(0.75,0.45)},anchor=center}]
\addplot[teal, mark=*] table [x=e, y=m, col sep=comma] {Tables/ScaleGaussianImage.csv};
\addlegendentry{mergegram}
\addplot table [x=e, y=p, col sep=comma] {Tables/ScaleGaussianImage.csv};
\addlegendentry{0D persistence}
\addplot table [x=e, y=n, col sep=comma] {Tables/ScaleGaussianImage.csv};
\addlegendentry{NN(4) distances}
\addplot[blue,mark=*] table [x=e, y=c, col sep=comma] {Tables/ScaleGaussianImage.csv};
\addlegendentry{point cloud}
\end{axis}
\end{tikzpicture}

%% file: chapters/conclusions.tex
\chapter{Conclusion and future research work} 
\label{ch:conclusions} 

The main accomplishments of this thesis are the following.

\medskip

\noindent 
Chapter \ref{ch:knn} is dedicated to the $k$-nearest neighborhood problem, Problem \ref{pro:knn}.
In this chapter, we provide a concrete set of counterexamples that illustrates that the time complexity estimates of the past approaches 
based on cover tree data structure \cite{beygelzimer2006cover,ram2009linear} were proven incorrectly.
To overcome past challenges we introduce a new data structure, a compressed cover tree, as well as a new algorithm for the $k$-nearest neighborhood search. At the end of the chapter, we rigorously demonstrate that the new approach resolved Problem \ref{pro:knn} in 
$O(k \cdot \log(k) \cdot \log(R) \cdot |R|)$ time, parametrized by expansion constant $c(R)$ of Definition \ref{dfn:expansion_constant}.

\medskip

\noindent
Chapter \ref{ch:mst} is dedicated to a minimum spanning tree problem of any finite metric space. 
It is shown that in the past approach \cite{march2010fast} the time complexity result that claimed the parametrized $O(|R| \cdot \log(|R|)$ time complexity was proven incorrectly. In this chapter, we present a new simpler algorithm based on a compressed cover tree. 
In the end the problem gets resolved with a weaker $O(|R| \cdot  \log(|R|) \cdot \log(\Delta(R)))$ parametrized time complexity estimate.

\medskip

\noindent
Chapter \ref{ch:mergegram} is dedicated to a new object, mergegram. It is confirmed that mergegram is stable under perturbations of data set $R$ and contains more information than standard $0$-dimensional persistence. In the experimental part of the chapter, we apply mergegram to the recognition of distorted $2$d shapes - problem. It is experimentally verified that mergegram produces excellent results in 2d shape recognition when it is combined with Perslay-neural network architecture.

\medskip

\noindent
As no work is ever complete it is important to highlight a few 
perspective research directions:

\medskip

\noindent
\textbf{Paired-trees:} Worst-time complexities for $k$-nearest neighbors and the minimum spanning tree problem can be potentially improved by utilizing cover trees of both query and reference sets. Given query and reference sets $Q$, $R$ we wish to develop an algorithm that scans through all relevant pairs of nodes in $\T(Q) \times \T(R)$ in 
time  $O(c^{O(1)} \cdot \max \{|Q|,|R|\})$, where $c$ is an expansion constant that depends on sets $Q$,$R$. In \cite[Algorithm~1]{ram2009linear} such algorithm was developed, but Counterexample \ref{cexa:dualtreeproof} exposes its time complexity estimate was incorrect.
In \cite{curtin2015plug} a more general approach was adopted to obtain a time complexity bound for the paired traversal of $(\T(R), \T(Q))$. However the time complexity estimate of the provided algorithm depended on two extra parameters: imbalance $I$ and extra-reference recursions $\theta$. No proper time complexity estimate was provided for the imbalance $I$. This imbalance $I$ can be trivially bounded by $O(|R| \cdot |Q|)$. However this is not sufficient to achieve parametrized worst-time complexity $O(|R|)$. In the future we wish to apply ideas from past approaches \cite{ram2009linear,curtin2015plug}  to generate a new algorithm guaranteeing $O(c^{O(1)} \cdot \max \{|Q|,|R|\})$ complexity for the paired-tree traversal.  

\medskip 

\noindent
\textbf{Improving the minimum spanning tree algorithm:} To improve minimum spanning tree algorithm we can first consider the following smallest-pair problem: given a quotient space $R / \sim$ and compressed cover tree $\T(R)$, find pair $(a,b)$ such that $a$, $b$ belong to distinct components of $R / \sim$ and minimize the distance $d(a,b)$. This could be potentially done by augmented paired tree traversal of $\T(R) \times \T(R)$. It should be noted that this task is significantly easier than finding all non-trivial nearest neighbors of the set $R$.
Using results, similar to Lemma \ref{lem:construction_depth_bound} and Lemma \ref{lem:knn_depth_bound} it might be realistic to prove that such a task could be completed in $O(c(R)^{O(1)} \cdot \log_2(|R|))$ complexity. We can then exploit the idea of Kruskal's algorithm i.e. we complete $|R|$ iteration, in each of which we modify quotient space $R / \sim$ by merging two closest pairs belonging to distinct components. This algorithm could potentially achieve
$O(c(R)^{O(1)} \cdot |R| \cdot \log_2(|R|) \cdot \alpha(|R|))$ worst-case time complexity bound for the construction of a minimum spanning tree on a finite metric space $(R,d)$.

\medskip

\noindent
\textbf{Multiparameter mergegrams} By taking advantage of multiparameter dendrograms  \cite{carlsson2010multiparameter} it is possible to consider multiparameter mergegrams. It will be beneficial to explore different definitions of multiparameter mergegrams and metrics between them. One of the most interesting theoretical problems in multiparamter dendrograms in our view is to find out under which metrics and conditions multiparameter dendrograms would obtain a stability property similar to Theorem~\ref{thm:stability_mergegram}. 